\documentclass[a4paper,10pt]{article}
\usepackage{amsthm,amsfonts,amssymb,euscript}
\usepackage{latexsym, multicol, fancybox}
\usepackage{graphicx}
\usepackage{color}
\usepackage{amsmath, amsthm, amssymb, bm}
\usepackage{epstopdf}
\usepackage{caption}
\usepackage{psfrag}
\usepackage{slashed}
\usepackage{verbatim}

\usepackage{mathrsfs}
 \usepackage{xcolor}
 \usepackage[citebordercolor={green}]{hyperref}
 \usepackage{tikz}
 \usepackage{bbm}
 \usepackage{dsfont}
 \usepackage{bbold}

\numberwithin{equation}{section}

\newtheorem{theorem}{Theorem}[section]
\newtheorem{lemma}[theorem]{Lemma}
\newtheorem{proposition}[theorem]{Proposition}
\newtheorem{corollary}[theorem]{Corollary}
\newtheorem{definition}[theorem]{Definition}
\newtheorem{remark}[theorem]{Remark}

\newcommand{\bea}{\begin{eqnarray}}
\newcommand{\eea}{\end{eqnarray}}
\newcommand{\lab}{\lababel}
\newcommand{\nn}{\nonumber}

\DeclareFontFamily{U}{mathx}{\hyphenchar\font45}
\DeclareFontShape{U}{mathx}{m}{n}{
      <5> <6> <7> <8> <9> <10>
      <10.95> <12> <14.4> <17.28> <20.74> <24.88>
      mathx10
      }{}
\DeclareSymbolFont{mathx}{U}{mathx}{m}{n}
\DeclareFontSubstitution{U}{mathx}{m}{n}
\DeclareMathAccent{\widecheck}{0}{mathx}{"71}

\def\beaa{\begin{eqnarray*}}
\def\eeaa{\end{eqnarray*}}
\def\ba{\begin{array}}
\def\ea{\end{array}}
\def\be#1{\begin{equation} \label{#1}}
\def \eeq{\end{equation}}
\def\bsplit{\begin{split}}

\def\pa{\partial}

\def\dual{{\,^*}}
\def\div{\mathrm{div}\,}
\def\curl{\mathrm{curl}\,}

\def\hot{\widehat{\otimes}}

\def\lab{\label}
\def\f12{\frac 1 2}

\def\a{{\alpha}}

\def\b{{\beta}}
\def\be{{\beta}}
\def\ga{\gamma}
\def\Ga{\Gamma}
\def\de{\delta}

\def\eps{\varepsilon}
\def\la{\lambda}

\def\si{\sigma}
\def\Si{\Sigma}
\def\om{\omega}

\def\Th{\Theta}

\def\vphi{\varphi}
\def\varo{{\varrho}}
\def\th{\theta}
\def\vth{{\vartheta}}
\def\ze{\zeta}

\def\nab{\nabla}

\def\Up{\Upsilon}

\def\B{{\bf B}}
\def\C{{\bf C}}
\def\D{{\bf D}}
\def\E{{\bf E}}
\def\F{{\bf F}}

\def\H{{\bf H}}

\def\J{{\bf J}}

\def\R{{\bf R}}
\def\P{{\bf P}}

\def\N{{\bf N}}

\def\W{{\bf W}}

\def\g{{\bf g}}

\def\u{{\bf u}}

\def\w{{\bf w}}

\def\k{\Th}

\def\CC{{\mathcal C}}
\def\MM{{\mathcal M}}
\def\NN{{\mathcal N}}

\def\LL{{\mathcal L}}

\def\VV{{\mathcal V}}

\def\RR{{\mathcal R}}

\def\RR{\mathcal{R}}

\def\Lie{\LL}

\def\f12{{\frac 1 2}}

\def\Bb{\mathcal{B}}
\def\Bbd{\dual\mathcal{B}}
\def\Kk{\mathcal{K}}
\def\Kkd{\, ^{*}\mathcal{K}}

\def\Cb{\underline{C}}

\def\ub{\underline{u}}
\def\fb{\protect\underline{f}}

\def\hbub{{\underline{H}_{\underline{u}}}}

\def\trch{{\mathrm{tr}}\, \chi}
\def\chih{{\widehat \chi}}
\def\chib{{\underline \chi}}
\def\chibh{{\underline{\chih}}}

\def\etab{{\underline \eta}}
\def\omb{{\underline{\om}}}
\def\bb{{\underline{\b}}}
\def\aa{\protect\underline{\a}}
\def\xib{{\underline \xi}}

\def\trchb{{\tr \,\chib}}
\def\tr{{\mathrm{tr}}\, }

\def\atrch{\, ^{(a)}\trch}
\def\atrchb{\, ^{(a)}\trchb}

\def\rhod{\,\dual\hspace{-2pt}\rho}

\def\sk{\mathfrak{s}}
\def\mm{\mathfrak{m}}

\def\trt{\trh \Th}

\def\Kc{\widecheck{K}}
\def\ah{\hat{a}}

\def\trth{\slashed{\mathrm{tr}}\, \th}
\def\thh{\widehat{\th}\hspace{0.05em}}
\def\thc{\widecheck{\trth}}
\def\trs{\slashed{\mathrm{tr}}}

\def\ss{\mathfrak{s}}

\def\Ricc{\mathbf{Ric}}

\def\kh{\widehat{\k}}
\def\Riem{\mathrm{Riem}}

\def\Ric{\mathrm{Ric}}

\def\0{_0}
\def\1{}
\def\vol{\mathrm{vol}}
\def\ao{\widecheck{a}}
\def\B{\mathscr{B}}
\def\Bd{\dual \B}

 \def\Rh{R\mkern-11mu /\,}
 \def\Rhh{{\widehat{\Rh}}}
 \def\nabh{\nab\mkern-12mu /\,}

\def\laph{\slashed{\Delta}}

\def\divh{\slashed{\mathrm{div}}\hspace{0.1em}}
\def\curlh{\slashed{\mathrm{curl}}\hspace{0.1em}}
\def\trh{\slashed{\mathrm{tr}}\hspace{0.1em}}

\def\K{\mathscr{K}}
\def\Kd{\, ^{*  \hspace{-0.3em}}\K}

\def\trch{\trh\chi}
\def\trchb{\trh\chib}

\def\const{\mathrm{const}}
\def\vol{\mathrm{vol}}

\def\Rb{R}

\def\d{\slashed{\mathcal{D}}}
\def\Kc{\widecheck{K}}

\def\ah{\hat{a}}
\def\ss{\mathfrak{s}}

\def\Ricc{\mathbf{Ric}}

\def\Riem{R}

\def\Ric{\mathrm{Ric}}

 \def\Rh{R\mkern-11mu /\,}
 \def\Rhh{{\widehat{\Rh}}}
 \def\nabh{\nab\mkern-12mu /\,}
 
 \def\trR{\trh\Rh}

\def\Rb{R}

\def\vthA{\vth^{\hspace{-0.1em} A}}
\def\vthB{\vth^{\hspace{-0.1em} B}}

\def\n{^{(n)}}
\def\nn{^{(n+1)}}
\def\nnn{^{(n+2)}}

\def\0{^{(0)}}

\def\gs{\, ^{\mathbb{S}^2}\hspace{-0.3em}\ga}

\def\Ls{\slashed{\Lie}}
\def\ss{\mathfrak{s}}
\def\dgz{{\d}^{(0)}}
\def\dga{{\d}^{\ga}}

\def\lapz{\laph^{\hspace{-0.1em}(0)}}
\def\nabz{\nabh^{(0)\hspace{-0.1em}}}
\def\trz{\trs^{(0)}  }

\def\i{^{(\infty)}}

\def\pp{p}
\def\Bsigma{\,^{(\Si)} \Bb}
 \def\Bdsigma{\,^{(\Si)} \Bbd}
 
 \def\gz{\ga^{(0)}}
 \def\H{\mathfrak{h}}
\def\cc{\mathbf{c}}

\def\muc{\widecheck{\mu}}
\def\v{\mathbf{v}}

\begin{document}


\title{Solving the constraint equation for general free data}
\author{Xuantao Chen and Sergiu Klainerman}


\date{}
\maketitle
\abstract{ We revisit the problem of solving the   Einstein constraint equations in vacuum by a new method which allows us to prescribe four scalar quantities, representing the full dynamical degrees of freedom of the constraint system. We show that once appropriate gauge conditions
have been chosen and four scalars freely specified   (modulo $\ell\le 1$  modes), we can rewrite the constraint equations as a well-posed system of coupled transport and elliptic equations on $2$-spheres,  which we solve by an iteration procedure.   Our method provides a large class of exterior solutions of the constraint equations that can be matched to given interior solutions,   
  according to the existing gluing techniques.   As such, it can be applied to provide a large class of initial  Cauchy data sets evolving to black holes, generalizing the well-known result of the formation of trapped surfaces due to  Li and Yu \cite{LiYu}. 
Though in our Main Theorem \ref{thm:main-precise}, we only specify conditions consistent with $g-g_{Schw}=O(r^{-1-\de})$, $k=O(r^{-2-\de})$, the method is flexible enough to be applied in many other situations. It can, in particular,  be easily adapted to construct arbitrarily fast decaying data. 
 We expect, moreover, that our method can also be applied to construct data with slower decay,  such as used by Shen in \cite{Shen}. 
  In fact, an important motivation for developing our method is to show that the result of \cite{Shen} is sharp, i.e., construct small, smooth initial data sets which violate  Shen's decay conditions,  and for which the stability of the Minkowski space result is wrong. }
  
\parindent = 0 pt
\parskip = 12 pt

\tableofcontents

\section{Introduction}
\subsection{The Einstein constraint equation}
Despite the fundamental role of the (local) well-posedness result \cite{ChoquetBruhat1952}, \cite{ChoquetBruhatGeroch1969} for the Einstein vacuum equation, it remains a challenge to construct  the \textit{full set} of  initial conditions\footnote{Here   $g$ denotes the Riemannian metric on the initial hypersurface $\Si$, with  scalar curvature  $R_g$, and $k$ corresponds to the second fundamental form of $\Si$, as embedded   in the spacetime. }  $(\Si, g, k)$,  verifying the 
constraint equations
\begin{equation}\lab{ece}
    \begin{split}
        \div k-\nab\, \tr k&=0,\\
        R_g+(\tr k)^2-|k|^2&=0.
    \end{split}
\end{equation}
We start by recalling below  some of the main methods  to construct solutions  to \eqref{ece}.


\subsubsection{Solving \eqref{ece} as  an underdetermined  $3D$ elliptic system}

 Given the Riemannian character of the metric  $g$,  it is tempting to   interpret \eqref{ece}  as an underdetermined   $3D$ elliptic system.   The best known method fitting this description, which we briefly review below, is the \textit{conformal method} of Lichnerowicz \cite{Lic44}, Choquet-Bruhat--York \cite{ChoquetYork1979}, Isenberg \cite{Isenberg95}, Maxwell \cite{Max09,Max21}. The idea is that we specify a given choice of a Riemannian metric $g_0$ on the $3$-manifold $\Si$, and a transverse-traceless (TT) symmetric $2$-tensor $\si_0$, i.e., $(\mathrm{div}_{g_0} \si_0)_i=0$, $\mathrm{tr}_{g_0} \si_0=0$. We then seek the solution to the constraint equation of the form
\beaa
g=\phi^4 g_0,\quad k^{ij}=\phi^{-10} (\si_0^{ij}+L[W]^{ij})+\frac 13\phi^{-4}g_0^{ij}H,
\eeaa
where $W$ is a vector field, $L[W]_{ij}:=\, ^{(g_0)}\nab_i W_j+\, ^{(g_0)}\nab_j W_i- \frac 23  (g_0)_{ij} \mathrm{div}_{g_0} W$, and $H$ is a scalar field. The constraint equation then becomes
\begin{equation*}
    \mathrm{div}_{g_0}(L[W])_i=\frac 23 \phi^6 (d H)_i,
\end{equation*}
\begin{equation*}
    \Delta \phi=\frac 18 R_{g_0}\phi-\frac 18\big |\si_0+L[W] \big|^2 \phi^{-7}+\frac 1{12}H^2\phi^5,
\end{equation*}
 which  is a determined  $3D$ elliptic system, and can thus be solved by standard elliptic methods.
By construction, the scalar field $H$ represents the mean curvature $\mathrm{tr}_{g}\, k$  of  $\Si$. In the case  when $\Si$ is a closed (compact without boundary) manifold, taking $H=\mathrm{const}$ and $\si_0=0$ reduces the equation to  a scalar equations  which  can be solved by standard elliptic methods.  
The method  also extends to   the asymptotically flat case, see for example  \cite{Bar86}, which uses the fact that  $\Delta_{g}$ is an isomorphism between the spaces of fields decaying like $r^{-\de}$ and $r^{-2-\de}$ ($0<\de<1$).  The conformal method  also allows one to construct faster decaying   initial data, as  considered  in  the proof of the nonlinear stability of Minkowski space in \cite{CK}. 
  In their recent   work \cite{FangSzeftelTouati2024a,FangSzeftelTouati2024b},  Fang--Szeftel--Touati have extended the method to     construct even more general initial data.   Their result treats arbitrary fast  decay and, as such, provides in particular nontrivial examples for the initial data sets in  \cite{KN03b}.

   Another well-known method, known under the name of \textit{gluing method}, initiated  in the works \cite{Corvino2000},  \cite{CorvinoSchoen}, \cite{ ChruscielDelay2003}, constructs nontrivial  initial data  which are  precisely  Kerr 
   outside a compact region.\footnote{ Note   that the  existence of such solutions  is forbidden for 
   purely elliptic  systems  which  have unique continuation properties.   }
A key observation in that regard, which dates back to Moncrief \cite{Moncrief1975}, is that the linearized constraint equations around the trivial Minkowskian   data set   is uniquely determined  by   a  10-dimensional cokernel space.  The gluing method  resolves the  obstruction  by connecting  this freedom  to the  10-charge family\footnote{These are  the parameters  $m,{\bf a}$, the  linear momentum and center of mass.}  associated  to  Kerr solution, thus matching  data given on a compact set to  a specified Kerr solution.
The gluing method has been used to prove the formation of trapped surfaces from Cauchy initial data \cite{LiYu}. Another important extension of the gluing method, due to  Corlotto--Schoen \cite{CarlottoSchoen2016},  constructs localized-in-angle initial data. We also refer to the further developments in \cite{ChruscielPollack2008}, \cite{ChruscielCorvinoIsenberg2011}, \cite{Cortier2013}, \cite{BieriChrusciel2017}, \cite{AretakisRodnianskiCzimek2023CMP, AretakisRodnianskiCzimek2023AHP, AretakisRodnianskiCzimekDuke}, \cite{MaoTao2022}, \cite{Hintz2024}. The gluing method was further extended in the  work of  Czimek--Rodnianski \cite{CzimekRodnianski2022}   which derived   more flexible  matching solutions.  More precisely, they show  that matching can be done provided  that  a  specific  condition, related  to the  positive mass theorem, is verified.\footnote{The condition can be written as $|\triangle \E|> C|\triangle \P|$ for some (potentially large) $C>0$, where $\triangle \E$ and $\triangle \P$ are respectively the differences of the energy and linear momentum between the two spheres considered for gluing.} A different, more direct approach,   to the obstruction  free gluing   results of  \cite{CzimekRodnianski2022}  was developed  by Mao--Oh--Tao, see  \cite{MaoOhTao2023} and further developed in \cite{IsettMaoOhTao2025}.
 The result in \cite{MaoOhTao2023} have been recently  used in the construction of Cauchy data that evolves into multiple trapped surfaces \cite{ShenWan2025}, \cite{GiorgiShenWanBosonStars}.
 
 In this paper, we  revisit  the  problem  by introducing a new method
  which allows us to   prescribe  four scalar quantities,   representing  the full  dynamical degrees of freedom of the  constraint system. We show  that  once appropriate  gauge conditions
has been  chosen and four    scalars   freely specified   (modulo $\ell\le 1$  modes),  we can   rewrite    the constraint equations   as a  well-posed system of  coupled transport   and   elliptic equations on $2$-spheres, which we solve by an iteration procedure, similar in spirit to the one used in the construction of GCM spheres and hypersurfaces  in \cite{KS-GCM1}, \cite{KS-GCM2}, \cite{ShenGCMKerr}.  In particular, our results  provide   a large  family of 
 exterior solutions of the constraint equations  which can be matched to given interior solutions    according to the existing gluing techniques.\footnote{The main  result,  stated first  in Theorem \ref{Thm:maithm-rough}, see also    Theorem  \ref{thm:main-precise},      constructs solutions  with  prescribed    four  scalars and    specified    asymptotic behavior at  space-like  infinity.  The method can however be  also be applied   in reverse, by  integrating towards  space-like infinity,  from  prescribed data  in a compact region of  $\Si$. }

\subsubsection{The Horizontal Constraint System}

 Though these various versions of the gluing  method  have  provided a great number of  interesting solutions to the constraint equations, they   typically produce solutions  which are exactly Kerr  outside  a compact set.
   The stability results in general relativity study much more general perturbations, and it is thus    an important  to construct initial data with   a lot more  flexibility. Ideally, one would like to have a method  which takes into account the  full  degrees of freedom in \eqref{ece}.

   The goal of this paper is to propose such a method and use it to describe initial  data sets with  more flexible properties. We divide 
    the degrees of freedom  of  \eqref{ece} into \textit{gauge} and  \textit{free}  scalars  and show that for a given choice of the  former,
     we have the freedom to   fully prescribe, up to $\ell\leq 1$ modes,    the remaining four   defining scalars. The constraint equations can then be solved  as  a system of transport and $2D$ elliptic equations, which we call the Horizontal Constraint System (HCS),   similar to the  way one constructs   solutions to the  characteristic
     initial value problem \cite{Chr1}. 
      In particular, this produces a fully  general set of exterior solutions which  can be  matched   to prescribed data on a compact set.

{\bf Connections with the free data.}      An  initial data set          $(\Si, g, k)$,    with $\Si$ a    $3$-manifold and $g, k$ symmetric 2-tensors, is formally specified by      $12$   functions. The constraint equations \eqref{ece} impose $4$ conditions, leaving formally $8$ degrees of freedom. 
      Three of these   are to be accounted   by   the  coordinate  covariance of \eqref{ece} on $\Si$. In our work, we  fix 
      a  \textit{radial} function $r$  whose level surfaces  are  $2$-dimensional spheres.   The other  two coordinates $\vth^1, \vth^2$ 
       can be chosen in a canonical way by   transporting  them from a given sphere $S_0$, where $r=r_0$,  along the integral curves   normal to       the  $r$-foliation.  Beside these    three coordinate  conditions,   one can identify a fourth which   corresponds to the embedding of $\Si$  into the  induced  Einstein vacuum spacetime spacetime.

   The   remaining four   degrees of freedom
    represent   the true dynamical degrees of freedom. We identify them here  
   in terms of      four scalars   obtained from the  Ricci and       curvature coefficients  associated to the  $r$-foliation. 
    Remarkably,  they  happen to provide  the only obstructions to    showing that  the structure  equations
     induced by the constraints,   expressed as  a  system of  transport equations in the direction normal to the foliation,
      is  well-posed.  Thus, once prescribed, modulo $\ell\leq 1$ modes,   one can  derive a  unique solution to \eqref{ece}.
      
   \begin{remark} 
   \lab{remark:comparison}
    It helps to  compare  this     with the characteristic initial data, that is data prescribed on   two  transversal  null    hypersurfaces $C$ and $\Cb$.  In that case,   the free data   is   simply  given by  the  shear tensors on each  hypersurface. 
  The  characteristic constraint equations    have  a simple  reductive structure that allows one  to solve various quantities one-by-one,  avoiding  loss of derivatives; see Chapter 2 of \cite{Chr1} for details.    In contrast, the Cauchy  constraint equations are more heavily coupled and  yet, once  the   defining scalars\footnote{We use the term \textit{defining scalars} to represent the union of gauge and free scalars.} are identified,  we  can recover   a similar reductive structure.
   \end{remark}

   \subsection{Main ideas and first statement of the main theorem}
 
  Given a sphere foliation on $\Si$  with outward  unit normal $N$ and  compatible\footnote{ i.e.  with 
   $\{e_a\}$ tangent to level surfaces  of $r$.}  orthonormal  frame  $\{N, e_a\}_{a=1,2}$, we define  the quantities
\beaa
\th_{ab}:=g(\nab_a N,e_b), \quad \Th_{ab}:=k(e_a,e_b), \quad \Pi:=k(N,N),\quad \Xi_a:= k(N,e_a),\quad Y_a:= \Riem(N,e_b,e_b,e_a),
\eeaa
where $\Riem$ denotes the $3$-dim Riemann curvature tensor. We also define the lapse function $\ah:=(N(r))^{-1}$, and denote the Gauss curvature of the $r$-spheres by $K$. 
   
{\bf Loss of derivatives.}
 In Section \ref{subsect:HCS}, we give the version of constraint equations decomposed with respect to the triad $\{N,e_a\}_{a=1,2}$, called the \textit{Horizontal Constraint System} (HCS).
In Section \ref{section:lossof deriv}, see also Section \ref{subsect:null-frame-formalism} in the spacetime language, we find the following six scalars that appear to be responsible for a loss of derivatives in HCS:
\bea\lab{eq:6-scalars-intro}
\mu:=-\laph(\log\ah)+K-\frac 14 (\trth)^2,\quad \nu:=\divh\Xi, \quad \Pi,\quad \curlh\Xi,\quad \divh Y,\quad \curlh Y.
\eea
Here $\laph$, $\divh$, and $\curlh$ are horizontal Laplacian, divergence, and curl operators defined in Section \ref{subsect:notations}.

{\bf Gauge scalars.} 
Among the scalars in \eqref{eq:6-scalars-intro}, the one that determines  a sphere foliation on $\Si$ is the scalar $\mu$. This has been referred  in \cite{CK} as the mass aspect function  and  used there  to determine the sphere foliation on the last slice $\Si_{t^*}$.
We can prescribe $\mu$ to address the coordinate freedom regarding $r$.
Yet,   even when the coordinates on $\Si$ are fixed,   we can have different initial data sets that  evolve to the same Einstein-vacuum  spacetime (see details in Section  \ref{subsubsection:free-data-on-Sigma}). In our work, we  resolve this ambiguity  by prescribing freely  the scalar field $\nu$. 

\begin{remark}
The traditional  way to  deal with this spacetime ambiguity  is to impose the maximal foliation condition $\mathrm{tr}_g \hspace{0.1em} k=0$, a condition  which 
is more aligned with the 3D elliptic character of the constraints  and is independent of  the choice of  a  foliation on $\Si$. 
 In contrast, our condition on $\nu$    works in tandem  with the    one on $\mu$. 
 Indeed, as stated in Proposition \ref{prop:existence-Si'}, given an initial data set, one can always, at least locally, construct another spacelike hypersurface, embedded in the same vacuum spacetime and with a specific sphere foliation, such that $\mu_{\ell\geq 1}=\nu=0$.\footnote{In fact, we need to impose additional $\ell=0$ conditions to determine a unique gauge; see Section \ref{subsubsect:ell-0-constraints}. For simplicity, we proceed with the vague assertion that $\mu_{\ell\geq 1}=\nu=0$ determines the gauge.}

\end{remark}

{\bf Free scalars.}
Given a gauge choice, specified by  the gauge scalars $(\mu, \nu)$, we show that the remaining degrees of freedom correspond precisely to  the remaining four scalars in \eqref{eq:6-scalars-intro}. Our main result is as  follows.
\begin{theorem}[Main Theorem, rough version]
\lab{Thm:maithm-rough}
Prescribe four scalars in a given exterior region in $\mathbb{R}^3$, denoted $(\Bb, \Bbd, \Kk, \Kkd)$, supported on spherical modes $\ell\geq 2$ (see  Section \ref{sect:spherical-harmonic-Hodge-operators} for the precise definition), and satisfying certain decaying conditions (to be later specified)  as $r\to \infty$.
Then, provided certain $\ell\leq 1$ conditions at spatial infinity,   corresponding   to a specification of the ADM charges (see Definition \ref{def:AF-data}), there exists a solution to the  constraint equation \eqref{ece} such that $\mu_{\ell\geq 1}=\nu=0$, and
\beaa
\left(\divh Y-\Bb\right)_{\ell\geq 2}=0, \quad \left(\curlh Y-\Bbd\right)_{\ell\geq 2}=0,\quad  \left(\laph(\ah\Pi)-\Kk\right)_{\ell\geq 2}=0, \quad \left(r^{-4}\pa_r(r^4\curlh \Xi)-\Kkd\right)_{\ell\geq 2}=0.
\eeaa
\end{theorem}
The precise statement is given in Theorem \ref{thm:main-precise}.

\begin{remark}
Though in  Theorem \ref{thm:main-precise}  we give  conditions  consistent\footnote{Such initial  data  sets were considered   in \cite{LR10} and are  more general than those  of  \cite{CK}.}  
 with $g-g_{Schw}=O(r^{-1-\de})$, $k=O(r^{-2-\de})$, the method can be easily adapted  to construct    arbitrarily  fast decaying data 
 used
  in \cite{KN03b}.   We expect that our  method  can also be applied to construct data with slower decay,  such as  used in \cite{Shen}. In that case, however,  one needs to integrate from a compact domain towards infinity rather than from infinity as we do here. In fact, an   important   motivation for    developing  our  method  here is to   show that the result of \cite{Shen} is sharp, i.e.   construct  small, smooth  initial  data  sets   which violate  Shen's   decay conditions,  and for which   stability of the Minkowski space result is  wrong.
   \end{remark}

 The statement of the theorem  implies  the existence of a  rich  family of vacuum exterior data with prescribed mass and angular momentum, as well as the center of mass.
Combining with Lorentz boosts, which generate nonzero linear momentum, we obtain a generalized 10-charge family of exterior solutions to \eqref{ece} compared with the 10-charge Kerr family constructed in \cite{ChruscielDelay2003}. 
As a consequence, we obtain a much larger class of exterior solutions that can be used for the gluing method.  As mentioned  in the abstract,  our result can be applied to provide a large class of initial  Cauchy data sets evolving to black holes, significantly   extending  the well-known result of the formation of trapped surfaces of  Li and Yu \cite{LiYu}.

\subsection{Acknowledgements}
The first author is supported by ERC-2023 AdG 101141855 BlaHSt. The second  author was funded by the NSF grant 1009079.

\section{Set-up and precise statement of the main theorem}

\subsection{Metrics, connections, and curvature tensors}\lab{subsect:notations}
We adopt the following notations:
\begin{itemize}
    \item The spacetime metric, connection, Riemann curvature tensor, Ricci tensor, and scalar curvature are denoted respectively by $\g$, $\D$, $\R$, $\Ricc$, and $\R_{\g}$. The spacetime coordinate indices are denoted by the Greek letters $\a$, $\b$, etc.
    \item The metric, connection, Riemann curvature tensor, Ricci tensor, and scalar curvature on $3$-dim Riemannian manifolds are denoted respectively by $g$, $\nab$, $R$, $\Ric$, and $R_g$. The corresponding divergence, curl, and trace operators are denoted by $\div$, $\curl$, and $\tr$. The spatial coordinate indices are denoted by the Latin letters $i$, $j$, etc.
    \item The connection with respect to the horizontal structure induced by an $r$-foliation\footnote{Or, more generally, orthogonal  to a   given vectorfield $N$. See also  Section \ref{sect:HorizDecompSi}.}    is denoted by $\nabh$. The corresponding divergence, curl, and trace operators are denoted respectively by $\divh$, $\curlh$, and $\trh$. We always take an orthonormal frame $\{e_a\}_{a=1,2}$ adapted to the horizontal structure. The letters $a$, $b$, etc will be used for such frame indices. 
     When the horizontal structure is integrable, we also denote the induced metric by $\gamma$, and the Gauss curvature by $K=K_\ga$, and the coordinate indices by $A$, $B$, etc.
 \end{itemize}
 Throughout the work, we use the Einstein summation convention on repeated indices. 
To avoid confusion regarding sign conventions, we remark that the Riemann curvature tensor $\Riem$ here is defined through the relation
\beaa
\nab_i \nab_j X^k- \nab_j \nab_i X^k=-\Riem_{i j l}^{\quad  k} X^l, 
\eeaa
and one can also lower the upper index to make it a $(0,4)$-tensor. It satisfies
\beaa
g(\nab_X \nab_Y Z, W)=g(\nab_Y \nab_X Z, W)+\Riem(X,Y,W,Z).
\eeaa
Here $\nab_X \nab_Y$ means $X^i Y^j \nab_i \nab_j$. The Ricci curvature and scalar curvature are then defined as
\beaa
\Ric(X,Y):=g^{ij} \Riem_{Xi Yj},\quad R_g=\mathrm{tr}_g\Ric.
\eeaa
The spacetime Riemann curvature tensor, Ricci tensor, and scalar curvature $\R$, $\Ricc$, $\R_{\g}$ are defined similarly.
The second fundamental form $k$ is defined by
\beaa
k_{ij}=g(\nab_{\pa_i} T,\pa_j).
\eeaa

\subsection{Asymptotically flat data}\lab{subsect:AF-data}
\begin{definition}\lab{def:AF-data}
An initial data set $(\Si,g,k)$ is said to be asymptotically flat, if there exists a coordinate system $(x^1,x^2,x^3)$ defined in a neighborhood of infinity, such that as $r:=\sqrt{(x^1)^2+(x^2)^2+(x^3)^2}\to \infty$, it holds that
\beaa
g_{ij}=\de_{ij}+o(1),\quad k_{ij}=o(1).
\eeaa
Given an $r$-foliation $\{S_r\}$ with the outward normal $N_0$ and induced area element $dA$ with respect to the Euclidean metric $(\de_{ij},0)$, the following quantities are  defined, if the limits exist:
\bea\lab{eq:def-ADM-charges}
\bsplit
\E&:=\frac{1}{16\pi}\lim_{r\to\infty} \int_{S_r} \sum_{i,j} \left(\pa_i g_{ij}-\pa_j g_{ii}\right)  N_0^j \, dA,\\
\P_i &:= \frac{1}{8\pi}\lim_{r\to\infty} \int_{S_r}\sum_j \left(k_{ij}-\mathrm{tr}_\de\hspace{0.1em} k\, g_{ij}\right)  N_0^j \, dA,\\
\J_i &:= \frac 1{8\pi} \lim_{r\to\infty}\int_{S_r} \sum_{j,l,m} \in_{ilm} x^l \left(k_{mj} -\mathrm{tr}_\de\hspace{0.1em} k\, g_{mj}\right) N_0^j dA,\\
\C_i &:= \frac{1}{16\pi m}\lim_{r\to\infty}\int_{S_r}
\sum_{j,k}\Big( x^i\big(\partial_k g_{kj}-\partial_j g_{kk}\big)
-\big((g_{ij}-\de_{ij})-\delta_{ij}(g_{kk}-\de_{kk})\big)\Big)\, N_0^j\, dA,
\end{split}
\eea
see, e.g., Section 1.2 of \cite{MaoOhTao2023}.
The quantities $\E$, $\P$, $\J$, $\C$ are called respectively the ADM energy, linear momentum, angular momentum, and center of mass.
\end{definition}

Throughout this work, we consider $\Si:=(r_0,\infty)\times \mathbb{S}^2$, which can be embedded into the Euclidean space $\mathbb{R}^3$ as an exterior region. This endows   $\Si$ with a  natural $r$-function.

\subsubsection{Horizontal  decomposition on $\Si$}
\lab{sect:HorizDecompSi}

 Assume that  $\Si:=(r_0,\infty)\times \mathbb{S}^2$  is equipped with a metric $g$. A specification of a unit vector field $N$ determines a horizontal structure $H=N^\perp$, defined through the metric $g$.
We then take an orthonormal frame $\{e_1,e_2\}$  spanning  $H$ so that the triad $\{N,e_1,e_2\}$ is an orthonormal frame on $\Si$. We  consider mostly the case of integrable horizontal  structures when  $N$ is orthogonal to an $r$-foliation.  See Section 2  in  \cite{GKS} for a detailed discussion of horizontal structures.

 {\bf Ricci coefficients.} The corresponding Ricci (rotation) coefficients on the $3$-Riemannian manifold are denoted as follows:
 \bea\lab{eq:Ricci-N}
 \pp_a:=g(\nab_N N,e_a),\quad  \th_{ab}:=g(\nab_a N, e_b).
 \eea
 The trace and the traceless part of $\th$ are denoted respectively by $\trth$ and $\thh$.

{\bf Curvature components.} The curvature components are denoted as follows:
\bea\lab{eq:def-curvature-components-Si}
\slashed{R}_{ab}:=R(N,e_a,N,e_b),\quad Y_a:=R(N, e_b, e_b, e_a).
\eea
 The trace and the traceless part of $\Rh$ are denoted respectively by $\trh \Rh$ and $\Rhh$. 
 
 Denoting the horizontal volume form by $\in_{ab}$, and define the dual $\dual Y_a:= \, \in_{ab} Y_b$. Then one directly verifies the relation
 \bea\lab{eq:R-Nabc}
 R_{Nabc} =\, \in_{bc}\hspace{-0.2em} \dual Y_a.
 \eea

{\bf Components of $k$.} Given initial data $(g,k)$ and the triad $\{N,e_1,e_2\}$ on $\Si$, we define
\beaa
\Th^{(N)}_{ab}:=k(e_a,e_b),\quad \Xi^{(N)}_a:=k(N,e_a),\quad \Pi^{(N)} :=k(N,N).
\eeaa
They are well-defined scalars or horizontal tensors once $N$ is specified. In what follows, when there is no danger of confusion, we simply denote $\Th= \Th^{(N)}$, $\Xi=\Xi^{(N)}$, $\Pi= \Pi^{(N)}$.
 The trace and the traceless part of $\k$ are denoted respectively by $\trh \k$ and $\kh$. 

{\bf The lapse function.}
The lapse function of a given $r$-foliation is defined to be 
\bea
\ah:=(N r)^{-1}.
\eea
Given the $r$-foliation, one can always write the metric $g$ in the following form
 \bea
      \lab{eq:g-metric(t,th,vphi)}
            g=\ah^2 dr^2+\ga,\quad \ga=\ga_{AB} d\vth^A d\vth^B,
      \eea
      and note that $\ah:=(N r)^{-1}=|\nab r|_g^{-1}$ are independent of the choice of $(\vth^1,\vth^2)$.
      
The Gauss curvature of the $r$-surfaces is denoted by $K=K_\ga$, where $\ga$ denotes the induced metric. We define
\bea\lab{eq:def-mu}
\mu:=-\slashed{\Delta}(\log\ah)+K-\frac 14(\trth)^2,
\eea
We often denote $\Kc:=K-r^{-2}$. We also define the following scalar field
\bea\lab{eq:def-nu}
\nu:=\divh\Xi=\de^{ab} \nabh_a \Xi_b.
\eea

We have the following simple relation regarding the radial acceleration $1$-form $p$, defined by equation \eqref{eq:Ricci-N}.
\begin{lemma}
For a given $r$-foliation, we have 
\bea\lab{eq:P-loga}
\pp=-\nabh(\log\ah).
\eea
\end{lemma}
\begin{proof}
We write the metric $g$ in the form \eqref{eq:g-metric(t,th,vphi)}, and in addition choose an orthonormal frame $\{e_a\}$ tangent to the $r$-constant spheres. Then since, with respect to the coordinates $r,\vth^1,\vth^2$, $\Ga_{rr}^{\;\; A}=\frac 12 g^{AB}(-\pa_{\vthB} g_{rr})=-\frac 12 g^{AB}\pa_{\vthB} (\ah^2)$,  we have
\begin{equation}\lab{eq:P-same-as-loga}
    \begin{split}
        \pp_a&=g(\nabh_N N,e_a)
        =\ah^{-2}g(\Ga_{rr}^{\;\;\, A} \pa_{\vthA},e_a)=\ah^{-2}(-\frac 12) g^{AB}\pa_{\vthB} (\ah^2) g(\pa_{\vthA}, e_a)\\
        &=-\frac 12\ah^{-2}(e_a)^B \pa_{\vthB} (\ah^2)=-\frac 12\ah^{-2}e_a(\ah^2)=-\nabh_a(\log\ah),
    \end{split}
\end{equation}
as required. Note that the conclusion itself does not depend on the coordinate choice.
\end{proof}

\subsubsection{Hodge operators, Spherical harmonics}\lab{sect:spherical-harmonic-Hodge-operators}
We adopt the following standard notation of horizontal operators for a horizontal $1$-form $\psi$:
\begin{equation*}
    \divh \psi:=\delta^{ab}\, \nabh_a \psi_b,\quad \curlh \psi:=\, \in^{ab} \nabh_a \psi_b,\quad (\nabh\hot \psi)_{ab}:=\nabh_a \psi_b+\nabh_b \psi_a-\delta_{ab}\, \divh\psi.
\end{equation*}
We now recall the Hodge operators defined in \cite{CK} and extended to the non-integrable cases in \cite{GKS}.
\begin{definition}\
  Given a horizontal structure $H$,  we denote by $\ss_0$ the set of scalar fields in the spacetime, by $\ss_1$ the set of $H$-horizontal $1$-forms, and by $\ss_2$ the set of symmetric traceless $H$-horizontal covariant $2$-tensors.
\end{definition}
\begin{definition}\lab{def:Hodge-operators}
We consider the following Hodge operators:
    \begin{itemize}
        \item $\d_1$ takes $\ss_1$ into $\ss_0$:\qquad     $\d_1\xi=(\divh\xi,\curlh\xi),$
        \item $\d_2$ takes $\ss_2$ into $\ss_1$:\qquad $   (\d_2 h)_a=\nabh^b h_{ab},$
        \item $\d_1^*$ takes $\ss_0$ into $\ss_1$:\qquad $    (\d_1^* (f,\dual f))_a=-\nabh_a f+\in_{ab} \nabh_b \dual f,$
        \item $\d_2^*$ takes $\ss_1$ into $\ss_2$:\qquad $  \d_2^* \xi=-\frac 12\nabh\hot\xi$.
    \end{itemize}
    Whenever  we need to be more precise, we will use the notations $\dga_1$, $\dga_2$, $(\dga_1)^*$, $(\dga_2)^*$  to specify the dependence of these operators  on the horizontal  metric $\ga$.
\end{definition}

We focus on the integrable case where $H$ is the tangent bundle of a sphere $(S,\ga)$. The operators $\d_1^*, \d_2^*$ are the formal adjoints $\d_1, \d_2$, i.e., 
\bea\lab{eq:formal-adjoint}
\langle \d_1 \xi, (f,\dual f) \rangle =\langle \xi, \d_1^* (f,\dual f)\rangle,\quad \langle \d_2 h, \xi\rangle =\langle h,\d_2^* \xi\rangle.
\eea
Here $\langle \cdot,\cdot\rangle$ is the inner product of $L^2(S,\ga)$.

We also recall the following identities in \cite{CK}:
\bea
\begin{split}
\lab{eq:dcalident}
\d^*_1 \d_1&=-\laph_1+K,\qquad\,\,\,\, \d_1  \d^*_1=-\laph,\\
\d^*_2  \d_2 &=-\frac 12\laph_2+K,\qquad \d_2 \d^*_2=-\frac 12(\laph_1+K),
\end{split}
\eea
where $K$ denotes the Gauss curvature of the sphere.

{\bf Spherical harmonics.}
We fix a choice of the standard spherical coordinates $(\vth^1,\vth^2)$ on $\mathbb{S}^2$, complemented with $(x^1,x^2)$ near $\vth^1=0,\pi$. This allows us to define the standard spherical harmonics $\{J_{\ell,\mm}\}$, where the integers $\ell$, $\mm$ satisfy $\ell\geq 0$, $-\ell\leq \mm \leq \ell$. 
They form a complete orthonormal basis of the space $L^2(\mathbb{S}^2)$, where $\mathbb{S}^2$ is equiped with the unit round metric 
\beaa
\gs=(d\vth^1)^2+\sin^2 (\vth^1) (d\vth^2)^2. 
\eeaa
We also denote the $r$-weighted round metric
\bea\lab{eq:def-gz}
\gz:=r^2 (\gs)=r^2 \left((d\vth^1)^2+\sin^2 (\vth^1) (d\vth^2)^2\right).
\eea
We denote the following $\ell=1$ basis, which plays a special role as in \cite{KS-GCM1}, \cite{KS-GCM2}:
\beaa
J_0:=J_{1,0}=\sqrt{\frac{3}{4\pi}}\cos\vth^1, \qquad J_+:=J_{1,1}=\sqrt{\frac{3}{4\pi}}\sin\vth^1\cos\vth^2, \qquad J_-:=J_{1,-1}=\sqrt{\frac{3}{4\pi}}\sin\vth^1\sin\vth^2.
\eeaa
For any scalar field $\phi$ on the sphere, one can uniquely decompose 
\bea\lab{eq:decomposition-ell-geq-2}
\phi=\phi_{\ell\leq 1}+\phi_{\ell\geq 2},
\eea
where $\phi_{\ell\leq 1}$ is spanned by $\{1,J_0,J_+,J_-\}$, and is orthogonal to $\phi_{\ell\geq 2}$ with respect to the measure induced by $r^{-2}\gz=\gs$. 
\begin{remark}\lab{rem:omega-i-ell=1}
Note that $J_+$, $J_-$, $J_0$ in fact correspond to the restriction of $x^1$, $x^2$, $x^3$ to the unit sphere modulo a constant factor:
\beaa
J_+=\sqrt{\frac 3{4\pi}} \om_1,\quad J_-=\sqrt{\frac 3{4\pi}}\om_2,\quad J_0=\sqrt{\frac 3{4\pi}} \om_3,\qquad \om_i:= x_i/|x|.
\eeaa
While $\om_i$ are not normalized, to have cleaner constant factors in expressing ADM charges in the $\ell=1$ components (see Appendix \ref{appendix:physical-quantities}), we also introduce the components under $\om_i$
\beaa
\phi_{\ell=1,i}:=\int_{S_r} \phi \, \om_i\, d\vol_{\gs}.
\eeaa 
In contrast, the components $\phi_{\ell,\mm}$ to be introduced in \eqref{eq:def-l-m-modes} are defined with respect to the orthonormal basis $\{J_{\ell,\mm}\}$.
\end{remark}
\begin{lemma}
\lab{Le:Si*-ell=1modes} 
The functions $J_p$ ($p=0,+,-$) verify the following properties on $(S,\gz=r^2(\gs))$:
\begin{equation}
\lab{eq:basicestimatesforJp-onSi_*} 
\bsplit
&\int_{S}J_p = 0,\quad \frac{1}{r^2}\int_{S}J_p J_q = \de_{pq},
\quad \left(r^2 \lapz +2\right)    J_p  = 0, \\
& \quad  (\dgz_2)^* (\dgz_1)^* (J_p,0)= (\dgz_2)^* (\dgz_1)^*(0,J_p)=(0,0).
\end{split}
\end{equation}
\end{lemma}
\begin{proof}
This is a special case of Lemma 5.2.8 in \cite{KS-GCM1}.
For the benefit of the reader, we repeat the proof of  the last statement.  Let  $F:=(\dgz_2)^* (\dgz_1)^* J_p$, where by $(\dgz_1)^* J_p$, we mean either $(\dgz_1)^*(J_p,0)$ or $(\dgz_1)^*(0,J_p)$. 
 Using  the identity $2\dgz_2 (\dgz_2)^*= (\dgz_1)^*\dgz_1-2K_{\gz}=(\dgz_1)^*\dgz_1-2r^{-2}$, we deduce\footnote{In fact, $\laph J_p$ should be replaced by either $(\laph J_p, 0)$ or $(0, \laph J_p)$ depending whether we consider $(\dgz_1)^*( J_p ,0)$ or $(\dgz_1)^*(0, J_p)$.} 
 \beaa
  2 \dgz_1\dgz_2 F&=& \dgz_1((\dgz_1)^*\dgz_1-2r^{-2}) (\dgz_1)^*(J_p) \\
  &=&  \dgz_1(\dgz_1)^* \dgz_1 (\dgz_1)^* J_p - 2r^{-2} \dgz_1 (\dgz_1)^*(J_p) \\
  &=& (\lapz)^2 J_p  +2 r^{-2} \lapz J_p=0,
 \eeaa
 as required.
\end{proof}

\subsubsection{Norms}

Note that $\Si$ is foliated by a family of spheres $S_r:=\{r\}\times \mathbb{S}^2$ by definition. We now define the $L^2$, $L^\infty$, and weighted Sobolev spaces over $S_r$.
\begin{definition}\lab{def:L2-Hs-S_r}
For horizontal covariant rank-$k$ tensors $U_{a_1\cdots a_k}$,
we denote by $L^2(S_r)$ the $L^2$ space through the metric $\gz$ defined in \eqref{eq:def-gz}, and by $\H^s(S_r)$ the Sobolev spaces for positive integers $s$, defined through $r\nabz$ where $\nabz$ is the covariant derivative with respect to $\gz$, i.e., through the norm
\beaa
|| U ||_{\H^s(S_r)}:= \sum_{i\leq s} ||(r\nabz)^i U||_{L^2(S_r)}.
\eeaa
The $L^\infty(S_r)$ space is defined through the norm
\beaa
||U||_{L^\infty(S_r)}:={\mathrm{ess} \sup}_{S_r} |\langle U,U\rangle_{\gz}|^\frac 12.
\eeaa
In this work, whenever we write an $\H^s$, $L^2$, or $L^\infty$ space without specification, we refer to the one over $S_r$ defined here.
\end{definition}
\begin{remark}
Given an initial data set $(\Si, g, k)$ and an $r$-foliation, we can also define similar norms with respect to  $\ga$, i.e. the metric induced on the foliation. These can be related  to the norms defined through $\gz$ in Definition \ref{def:L2-Hs-S_r}, see Lemma \ref{lemma:equivalence-norms}. Consequently, in the iteration scheme, we shall always refer to the norms defined in Definition \ref{def:L2-Hs-S_r}.
\end{remark}

\begin{remark}
Since the $\H^s(S_r)$ norms, in view of the area element of $\gz$, provides an additional $r$ factor, throughout the paper, we will frequently write our estimates for a quantity $\psi$ in the form $r^{-1} ||\psi||_{\H^s}$, reflecting the true $L^\infty$ size of $\psi$ in line with the Sobolev inequality on the sphere.
\end{remark}

\begin{remark}\lab{rem:inverse-metric}
Throughout the paper, we often encounter the difference between $\nabh$ and $\nabz$ on various quantities, which yields the Christoffel symbol of $\ga$ with respect to $\gz$ \[\Ga_{ab}^{\;\;\, c}(\ga;\gz)= \frac 12 (\ga^{-1})^{cd} (\nabz_{a} \ga_{bd}+\nabz_b \ga_{cd}-\nabz_d \ga_{ab}),\]
under a choice of the horizontal orthonormal frame of $\gz$, denoted by $\{\displaystyle e_a\0\}_{a=1,2}$. Our assumption always ensures that 
 $\ga$ is close to $\gz$ in terms of the components in $\{\displaystyle e_a\0\}_{a=1,2}$. Therefore, the inverse of $\ga$ with respect to $\gz$ stays bounded, and hence we have
\beaa
\Ga_{ab}^{\;\;\, c}(\ga;\gz)=O(\nabz(\ga-\gz)),
\eeaa
where the size is defined through the components in $\{\displaystyle e_a\0\}_{a=1,2}$.
\end{remark}

{\bf Sobolev norms in the frequency space.} For a scalar field $\phi$, we denote its $(\ell,\mm)$-modes by
\bea\lab{eq:def-l-m-modes}
\phi_{\ell,\mm}:= \langle \phi,J_{\ell,\mm}\rangle_{\gs}=r^{-2} \langle \phi,J_{\ell,\mm}\rangle_{\gz}.
\eea
It is well-known that the Sobolev space $H^s(S_r,\gs)$ can be alternatively characterized by 
\bea\lab{eq:Hs-norm-modes}
||\phi||_{H^s(S_r,\gs)}^2=\sum_{\ell=0}^\infty \sum_{\mm=-\ell}^\ell (1+\ell^2)^s |\phi_{\ell,\mm} |^2.
\eea
By simple rescaling, $||\phi||_{H^s(S_r,\gs)}=r^{-1} ||\phi||_{\H^s(S_r)}$. Therefore, one has
\bea\lab{eq:Hs-norm-modes-new}
||\phi||_{\H^s(S_r)}^2 =r^2 \sum_{\ell=0}^\infty \sum_{\mm=-\ell}^\ell (1+\ell^2)^s |\phi_{\ell,\mm} |^2.
\eea
In particular, if $\phi$ is supported on $\ell\leq 1$, we have 
\bea\lab{eq:Hs-Linfty-ell-leq-1}
||\phi||_{\H^s(S_r)}\approx r \Big(|\phi_{\ell=0}|+\sum_{m=-1}^1 |\phi_{1,\mm}|\Big) \lesssim r ||\phi||_{L^\infty(S_r)}.
\eea

{\bf Integral Minkowski inequality.} We recall the standard integral Minkowski inequality applied to $L^1(I)$ and sequence-$l^2$ spaces, where $I$ is any interval:
\bea\lab{eq:integral-Minkowski-inequality}
\left\| \int_I \, |a_n(r)| \, dr \right \|_{\ell^2_n}  \leq \int_I \, \left\| a_n(r) \right\|_{\ell^2_n} dr.
\eea
 
\subsection{Horizontal Constraint System}\lab{subsect:HCS}
\subsubsection{Unconditional equations}
In what follows, we restrict our attention to the case of $\Si=(r_0,\infty)\times \mathbb{S}^2$, where $N$ is the outward unit normal to the $r$-foliation $\{S_r\}$. 
The horizontal structure $H=N^\perp$ is then automatically integrable.
Recall that $\Rh_{ab}$, defined in \eqref{eq:def-curvature-components-Si}, can be viewed as an horizontal symmetric $2$-tensor, and,  as such, it can be decomposed as
\begin{equation*}
    \Rh_{ab}=\frac 12\trR\, \ga_{ab}+\Rhh_{ab},
\end{equation*}
where $\Rhh$  is traceless.  The scalar field $\trR$ is, by definition, related to the scalar curvature $R_g$ through the following identity
\begin{equation}\lab{eq:relation-Rg-trR-Rabab}
    R_g=\Ric_{aa}+\Ric_{NN}=\Riem_{abab}+2\Riem_{NaNa}=\Riem_{abab}+2\trR.
\end{equation}
Also recall the horizontal $1$-form $Y_a:=\Riem_{Nbba}$.

We have the following equations, which hold regardless of whether $(g,k)$ solves the constraint equations.
\begin{proposition}[Unconditional equations I]\lab{prop:Unconditional-equations-1}
The following equations hold true:
\bea
\lab{eq:unconditional-N-trth}
\nabh_N \trth&=& \divh \pp-|\thh|^2-\frac 12(\trth)^2-|\pp|^2-\trh\Rh,\\
\lab{eq:unconditional-Gauss}
\Riem_{abab} &=& 2K-\frac 12(\trth)^2+|\thh|^2,\\
\lab{eq:unconditional-Codazzi}
\divh \thh&=&\frac 12 \nabh\trth-Y,\\
\lab{eq:unconditional-transport-thh}
\nabh_N \thh&=&\nabh\hot \pp - \trth\, \thh-\pp\hot \pp-\Rhh,\\
\lab{eq:unconditional-Bianchi}
\nabh_N K &=& -\divh Y-\trth\, K+2\pp\cdot Y
-\thh\cdot (\nabh\hot \pp-\pp\hot \pp) +\frac 12\trth\, (\divh \pp-|\pp|^2).
\eea
\end{proposition}
\begin{proof}
See Appendix \ref{subsubsect:proof-unconditional-equations-1}.
\end{proof}

{\bf Constraint quantities on $\Si$.}
We define the  momentum and    Hamiltonian     constraint quantities
\bea
\lab{eq:def-CC-Mom}
 \CC_{Mom} (g,k)&:=&\div k-\nab\, \tr k ,\\
 \lab{eq:def-CC-Ham}
\CC_{Ham} (g,k)&:=&R_{g}+(\tr k)^2-|k|^2.
\eea
Expanding $(\CC_{Mom})_N$, $(\slashed{\CC}_{Mom})_a:=(\CC_{Mom})_a$, and $\CC_{Ham}$ under the frame $\{N,e_a\}$, we obtain
\begin{proposition}[Unconditional equations II]
\lab{prop:constraint-equation-in-frame}
The following equations hold:
\bea
\lab{eq:structure-constraint-Phi-R}
\hspace{-3em}\nabh_N \trt&=&\divh \Xi + \trth \Pi -\thh\cdot\kh-\frac 12 \trth\trt -2\pp\cdot \Xi-(\CC_{Mom})_N,\\
\lab{eq:structure-constraint-Phi-a}
\hspace{-3em}\nabh_N \Xi &=& -\divh \kh+\pp\cdot\k-\Pi \pp-\frac 32\trth\, \Xi-\thh\cdot \Xi+\frac 12 \nabh \trt +\nabh\, \Pi+ \slashed{\CC}_{Mom},\\
\lab{eq:structure-constraint-H}
\hspace{-3em}\nabh_N \trth &=& \divh \pp-\frac 12 |\thh|^2 -\frac 34(\trth)^2-|\pp|^2+K +\Pi\, \trt+\frac 14 (\trt)^2-|\Xi|^2-\frac 12 |\kh|^2-\frac 12 \CC_{Ham}. 
\eea
\end{proposition}
\begin{proof}
See Appendix \ref{subsect:proof-constraint-equation-in-frame}.
\end{proof}
\begin{definition}\lab{def:HCS}
We call the unconditional equations \eqref{eq:unconditional-N-trth}-\eqref{eq:unconditional-Bianchi}, \eqref{eq:structure-constraint-Phi-R}-\eqref{eq:structure-constraint-H} with $(\CC_{Mom})_N=0$, $\slashed\CC_{Mom}=0$, and $\CC_{Ham}=0$ the \textit{Horizontal Constraint System} (HCS). 
\end{definition}

\subsubsection{Loss of derivatives}
\lab{section:lossof deriv}

At first glance, HCS appears to be ill-posed, i.e., it appears to lose derivatives. 
For example, compared with the Raychaudhuri equation on a null hypersurface (relative to  the geodesic foliation)
\beaa
\nabh_4 \trch=-\frac 12(\trch)^2-|\chih|^2,
\eeaa
the   HCS  equation for $\trth$ (equation \eqref{eq:structure-constraint-H} with $\CC_{Ham}=0$) reads
\bea\lab{eq:N-trth-schematic}
\nabh_N \trth = \divh p+K +\cdots,
\eea
 and the equation for $\trt$ (equation \eqref{eq:structure-constraint-Phi-R} with $(\CC_{Mom})_N=0$) reads
 \bea\lab{eq:N-trk-schematic}
 \nabh_N \trt=\divh \Xi +\cdots,
 \eea
with loss of one derivative for $p$ and $\Xi$.

  Note that there are no $N$-transport equations of $p$. 
A simple way  to avoid the loss of derivative for \eqref{eq:N-trth-schematic} is to prescribe, as a gauge condition,  the scalar field
  \beaa
  \mu=\divh p+K-\frac 14 (\trth)^2.
  \eeaa 
For \eqref{eq:N-trk-schematic}, we consider it together with  equation \eqref{eq:structure-constraint-Phi-a} with $\slashed\CC_{Mom}=0$:
  \beaa
  \nabh_N \Xi = -\divh \kh+p\cdot\k-\Pi p-\frac 32 \trth \Xi-\thh\cdot \Xi+\frac 12 \nabh \trt +\nabh \Pi.
  \eeaa
  There are several terms  on the right that lose derivatives.
     To deal with this, we  first  prescribe the scalar field $\Pi$,  so that the term $\nabh\Pi$ is no longer an issue. The equation then reads
  \bea\lab{eq:N-kn-schematic}
  \nabh_N \Xi = -\divh \kh+\frac 12\nabh \trt +\cdots .
  \eea
Commuting  the equation with $\divh$ and $\curlh$ respectively, we derive
  \beaa
  \nabh_N \divh\Xi &=& -\divh\divh \kh+\frac 12 \laph \trt +\cdots , \\
  \nabh_N \curlh \Xi &=& -\curlh \divh \kh +\cdots . 
  \eeaa
  This motivates us to also interpret $\divh\Xi$, $\curlh\Xi$ as scalars to be prescribed. Indeed, prescribing $\nu=\divh\Xi$ yields  an estimate of $-\divh\divh \kh+\frac 12 \laph \trt$. This also deals with the loss of derivatives in \eqref{eq:N-trk-schematic}, providing an estimate for $\trt$. As a result, one can obtain the estimate of $\divh\divh \kh$. Also, prescribing $\curlh\Xi$
  clearly provides the  control of $\curlh \divh \kh$.
  Since the operator that maps $\kh$ to $\d_1\d_2\kh=(\divh\divh\kh,\curlh\divh\kh)$ is an elliptic Hodge operator with no kernel, we can determine $\kh$.

We also need to estimate the Gauss curvature $K_\ga$; the estimate of $\pp$, which is curl-free in view of \eqref{eq:P-loga}, can then be retrieved from the definition of  $\mu$, using  the Hodge  estimates for $\d_1$.  The transport equation of $K_\ga$, \eqref{eq:unconditional-Bianchi}, again contains a term $\divh Y$ that loses derivatives.  It is hence natural,  in fact necessary, to also prescribe the scalar field $\divh Y$.  
In order to fully determine $Y$, we also prescribe the scalar field $\curlh Y$. Recall that with $\mu$ prescribed,  $\trth$  can be determined  from equation \eqref{eq:N-trth-schematic}. As a consequence,  $\thh$ can also be determined from \eqref{eq:unconditional-Codazzi} using  the  Hodge estimates for $\d_2$.
 
 \begin{remark}
We note that $\Rhh$ is in fact decoupled from the system and can be retrieved from \eqref{eq:unconditional-transport-thh} after all other quantities are determined.
\end{remark}

To summarize, we  were led to  prescribe  the following six scalar fields:
\bea
\lab{eq:6scalars}
\Pi,\quad \mu,\quad \nu,\quad \curlh \Xi,\quad \divh Y,\quad \curlh Y.
\eea
As we have   argued  heuristically above, once these  6 scalars are  prescribed,  there   are no other losses of  derivatives for the HCS.

\subsubsection{Connection with free data}\lab{subsubsection:free-data-on-Sigma}
For a $3$-manifold $\Si$, the initial data $(g,k)\in \Ga(S_+^2 T^*\Si)\times \Ga(S^2 T^*\Si)$ for the Einstein vacuum equations consist of a pair of sections satisfying the Einstein constraint equations \eqref{ece}.
In local coordinates, since both $g$ and $k$ are symmetric, we have $12$ unknowns. The constraint equations \eqref{ece} impose $4$ conditions, leaving formally $8$ degrees of freedom.   
Three of these   are to be accounted 
   by   the  coordinate  covariance of \eqref{ece} on $\Si$, which consist of the following:
   \begin{itemize}
   \item The choice of the sphere foliation, i.e., a specification of a coordinate function $r$ whose level set gives a foliation.\footnote{\lab{ft:r-ambiguity}There is apparently an ambiguity on $r\mapsto F(r)$ with $F$ an increasing function. We will later eliminate this ambiguity in Section \ref{subsubsect:ell-0-constraints}.} We expect to prescribe a scalar field to fix this gauge choice.
   \item The choices of the angular variables $(\vth^1,\vth^2)$. We have chosen to write our metric in the form \eqref{eq:g-metric(t,th,vphi)}. Provided with the boundary condition, i.e., an initial choice of $(\vth^1,\vth^2)$ on a given sphere, this corresponds to the coordinate conditions $N(\vth^1)=N(\vth^2)=0$, where $N$ is the unit normal of the $r$-foliation, in the increasing direction of $r$.
   \end{itemize}
   Therefore, excluding the three coordinate ones, we are left with five degrees of freedom. Among the scalar fields we identified in \eqref{eq:6scalars}, the scalar field $\mu$ plays the role of choosing the $r$-foliation, and the remaining five read $\nu$,
  $\Pi$, $\curlh \Xi$, $\divh Y$, and $\curlh Y$.
   For a given coordinate $(r,\vth^1,\vth^2)$, these five scalars reflect, at least formally, the freedom of the initial data $(g,k)$ on $\Si$ that solves \eqref{ece}. However, not all of them represent the  ``physical" degrees of freedom, as there  is an additional coordinate  choice  to be made  that corresponds to the embedding of $\Si$  into the spacetime.
As we show  below  in Section \ref{section:Spacetime-perspective}, this corresponds to the scalar $\nu=\divh\Xi$.
We therefore  interpret the scalar $\nu$ as a spacetime coordinate choice, and accordingly, call $\nu$ and $\mu$ the \textit{gauge} scalars. Together with the implicit choice $N(\vth^1)=N(\vth^2)=0$,  this exhausts  the   four degrees of freedom
of  solutions to the Einstein-vacuum equations in four spacetime dimensions.

\begin{definition}
Among the six scalars  in \eqref{eq:6scalars},  $\nu=\divh\Xi$ and $\mu$ are  called gauge scalars. 
 The remaining four
   \bea
   \lab{eq:fourscalars}
   \Pi, \quad \curlh \Xi,\quad \divh Y,\quad \curlh Y,
   \eea
 are  called free scalars, indicating that they represent   the true dynamical degrees of freedom of  the Einstein-vacuum equations.
   \end{definition}
   
   While the free scalars describe the dynamical degrees of freedom, as we will see heuristically in \eqref{subsect:ell-0-1-constraints}, the $\ell\leq 1$ parts of the scalars are subject to much more rigid conditions directly related to the ADM charges \eqref{eq:def-ADM-charges} . Therefore, it is in fact the $\ell\geq 2$ part of the free scalars
   \beaa
   \Bsigma:=(\divh Y)_{\ell\geq 2},\quad \Bdsigma:=(\curlh Y)_{\ell\geq 2},\quad \, ^{(\Si)}\Kk:=(\laph(\ah \Pi) )_{\ell\geq 2},\quad \, ^{(\Si)}\Kkd:= (r^{-4}\pa_r(r^4\curlh \Xi) )_{\ell\geq 2},
      \eeaa
   that, as stated in the main theorem (Theorem \ref{thm:main-precise}), are free to prescribe.

\subsection{Spacetime perspective}
\lab{section:Spacetime-perspective}

\subsubsection{The null frame formalism}\lab{subsect:null-frame-formalism}
We now discuss the constraint equations from the spacetime perspective.\footnote{The spacetime perspective helps   to  provide  additional   motivation    for  the   two gauge scalars, but   will not be needed in the rest of the paper.} Indeed, the first and second fundamental forms of  any spacelike hypersurface  in an Einstein-vacuum spacetime  solves the constraint equations \eqref{ece}, and according to \cite{ChoquetBruhat1952}, \cite{ChoquetBruhatGeroch1969}, the converse is also true,  i.e.  regular initial data solving \eqref{ece} is uniquely embedded in its maximal globally hyperbolic development. 

When $\Si$ is an embedded spacelike hypersurface in a spacetime $(\MM,\g)$, one can define the future unit timelike normal vector field $T$ on $\Si$, and the following null pair
\begin{equation}
\lab{eq:e_3e_4-Si}
    e_3:=T-N, \quad e_4:= T+N,\quad \text{on $\Si$}.
\end{equation}
Here $N$,  as before, is  the outward normal vector field to $r$-spheres $S_r$ on $\Si$. 
With such a choice of the null pair, we immediately obtain the following relations of the Ricci coefficients and quantities defined on $\Si$:\footnote{\lab{ft:zeta-Xi}The first two relations are trivial, and the third also follows easily from the calculation 
$    -k(e_a,N)=\g(\D_a N,T)=\g\left(\D_a \left(\frac 12 e_4-\frac 12 e_3\right), \frac 12 e_4+\frac 12 e_3\right)=\frac 14 \g\left(\D_a  e_4, e_3\right)-\frac 14 \g\left(\D_a e_3, e_4\right)=\zeta_a$.
}
\bea\lab{eq:spacetime-chi-chib-zeta}
\chi=\Th+\th,\quad \chib=\Th-\th,\quad \zeta=-\Xi.
\eea
Note that, in contrast to what we discuss below, they do not rely on the extension of the frame beyond $\Si$.

In a spacetime slab containing $\Si$, $S_r$ determines a family of incoming null hypersurfaces, which are the constant leaves of some optical function $\ub$, denoted by $\hbub$. We extend $e_3$ so that it is the null geodesic vector on each $\hbub$.
Regarding the extensions of $e_a$ and $e_4$ beyond $\Si$, we recall the following two choices, both exploited in \cite{KS:Kerr}:
\begin{itemize}
    \item The Principal Geodesic (PG) frame: Each $\hbub$ is foliated by spheres given as the constant leaves of the affine parameter  of  $e_3$, and the horizontal  space   $\{e_a\}_{a=1,2} $  tangent to the corresponding  spheres. This determines a null frame\footnote{In \cite{KS:Kerr}, the corresponding hypersurface $\ub=\const$ are in fact not exactly null, and the definition of the PG structure is more general.} $\{e_3,e_4,e_a\}$ . 
\item The Principal Temporal (PT) frame: We extend $e_4$ by the condition
\begin{equation*}
    \D_{e_3}e_4=0.
\end{equation*}
The null pair  $\{e_3, e_4\}$  determines the horizontal structure spanned by $\{e_a\}_{a=1,2}$, which may, in general,  be non-integrable beyond $\Si$. 
\end{itemize}

In  both cases, the null frame $\{e_3,e_4,e_a\}_{a=1,2}$ is determined in  a  spacetime slab, thereby defining the Ricci coefficients and curvature components:\begin{equation*}
    \chi_{ab}=\g(\D_a e_4,e_b),\quad \chib_{ab}=\g(\D_a e_3,e_b),      \quad  \eta_a=\frac 12 \g(\D_3 e_4,e_a),\quad \etab_a=\frac 12 \g(\D_4 e_3,e_a), \quad \zeta_a=\frac 12 \g(\D_a e_4,e_3), 
\end{equation*}
\begin{equation*}
    \om=\frac 14 \g(\D_4 e_4,e_3),\quad \omb=\frac 14 \g(\D_3 e_3,e_4),\quad  \xi_a=\frac 12 \g(\D_4 e_4,e_a), \quad \xib_a=\frac 12 \g(\D_3 e_3,e_a),
\end{equation*}
\begin{equation*}
    \a_{ab}=\W_{a4b4},\quad \b_a=\frac 12 \W_{a434},\quad \rho=\frac 14 \W_{3434},\quad \dual\rho=\frac 14\dual \W_{3434},\quad \bb_a=\frac 12 \W_{a334},\quad \underline \a_{ab}=\W_{a3b3}.
\end{equation*}
Here $\W$ is the Weyl tensor that can be expressed as
\bea\lab{eq:Weyl-tensor}
\W_{\rho\sigma\mu\de} =  \R_{\rho\sigma\mu\de}
	+	\frac{1}{2} \left( \g_{\rho\mu} \Ricc_{\de\sigma} - \g_{\rho\de} \Ricc_{\mu\sigma} - \g_{\sigma\mu} \Ricc_{\de\rho} + \g_{\sigma\de} \Ricc_{\mu\rho} \right) + \frac{1}{6} \R_g \left( \g_{\rho\mu} \g_{\de\sigma} - \g_{\rho\de} \g_{\mu\sigma} \right).
\eea

\begin{remark}\lab{rem:extrinsic-quantities}
Note that  the  Ricci coefficients $ \omb,\, \xib,\, \eta,\, \om,\, \xi,\, \etab$
are not well-defined\footnote{They cannot be defined  by the choice, $\{e_3,e_4,e_a\}$ on $\Si$, as their definitions contain $e_3$ or $e_4$ derivatives of the frame. }   on $\Si$.
They are however well defined  for   the PG or PT  extension  considered above.  In particular, 
 given that   $e_3$ is geodesic  in both case (in particular  $\omb=0$, $\xib=0$), the choice of the PT frame is equivalent to the condition $\eta=0$. 
\end{remark}

\begin{proposition}\label{prop:relations-om-etab-xi-zeta}
    With the choice of $\{e_3,e_4\}$ given by \eqref{eq:e_3e_4-Si}  on $\Si$, and its extension to the spacetime via the PT condition, we have the following relation on $\Si$:
    \begin{equation}\lab{eq:om-de}
        \om=- \Pi,
    \end{equation}
    \begin{equation}\lab{eq:xi-ep-loga}
    \xi=\Xi+p,
    \end{equation}
    \begin{equation}\lab{eq:etab-ep-loga}
    \etab=\Xi-p.
    \end{equation}
\end{proposition}
\begin{proof}
See Appendix \ref{subsect:Proof-relations-om-etab-xi-zeta}.
\end{proof}
\begin{proposition}\lab{prop:b-bb-expression}
For $\Si$ embedded in a spacetime $(\MM,\g)$ with a specified $r$-foliation,  the following relations hold true between 
the  intrinsic quantities   defined in Section \ref{sect:HorizDecompSi} and the spacetime quantities  defined above:
\bea
\lab{eq:b+bb}
(\b+\bb)_a &=& 2\left(Y+\Xi \cdot \k-\trt \Xi\right)_a+3\Ricc_{Na},\\
\lab{eq:b-bb}
(\b-\bb)_a &=& -2\left(\nabh \Pi-\nabh_N \Xi-2\th\cdot \Xi+p\cdot  \Th-\Pi p \right)_a+(\slashed{\CC}_{Mom})_a,\\
\lab{eq:Gauss-rho}
\rho &=& -K_\ga -\frac 14 (\trt)^2+\frac 14 (\trth)^2 +\frac 12 |\kh|^2-\frac 12 |\thh|^2 \\
\nonumber & &+\frac 12 \CC_{Ham} -\Big(\Ricc-\frac 12 (\R_{\g}) \g\Big)_{NN}+\frac 12 \Big(\Ricc-\frac 12 (\R_{\g}) \g\Big)_{aa}-\frac 23 \R_{\g},\\
\lab{eq:curl-kn-dual-rho}
\dual\rho&=& -\curlh \Xi- \kh\wedge\thh.
\eea
\end{proposition}
\begin{proof}
See Appendix \ref{subsect:spacetime-Ricci-b-bb}.
\end{proof}

{\bf Loss of derivatives.} With the help of Proposition \ref{prop:b-bb-expression}, the HCS system can be re-expressed in terms of  the spacetime quantities.  These can also be derived  directly from the null structure and Bianchi equations,  recorded in full detail in Appendix \ref{subsect:null-str-Bianchi}. Below, we  only  refer to them schematically.  

\begin{remark}The spacetime version of HCS consists of the following types of equations:
\begin{itemize}
\item
The structure equations  that only involve derivatives tangent to $\Si$, e.g., the Codazzi equation
\begin{equation*}
    \divh \chih=\frac 12\nabh \trch-\zeta\cdot\chih+\frac 12\trch\zeta-\b.
\end{equation*}
\item The transport-type equation in the $N$-direction obtained by combining the  $e_3$, $e_4$   transport type equations from the null structure and Bianchi equations. 
Indeed,   suppose that we have $\nabh_3\psi=\underline F$, $\nabh_4\psi=F$, we  can use   use  the formula $N=\frac 12e_4-\frac 12 e_3$ to get
\begin{equation*}
    \nabh_N \psi=\frac 12 F-\frac 12 \underline F.
\end{equation*}
 Note that not all quantities have both $e_3$ and $e_4$ transport equations; this is true only if $\psi$ belongs to  $\big\{\b, \rho, \rhod, \bb, \chi, \chib, \zeta\big\}$ or a combination of these.
\end{itemize}
\end{remark}
The loss of derivatives manifest in  the following spacetime HCS equations:
\beaa
    \nabh_N \trch&=&\frac 14\trch\trchb
-\frac 14(\trch)^2-\om\trch+\divh\xi-\Big(\rho-\frac 12\chih\cdot\chibh\Big)+\cdots, \\
    \nabh_N \trchb&=& -\frac 14\trch\trchb+\frac 14(\trchb)^2+\om\trchb+\divh\etab+\Big(\rho-\frac 12\chih\cdot\chibh\Big)+\cdots, \\
     \nabh_N \Big( \rho-\frac 12\chih\cdot\chibh\Big)&=& \frac 12 \divh (\b+\bb)+\cdots, \\
         \nabh_N\zeta&=& \nabh\om-\frac 12 \b+\frac 12 \bb+\cdots.
\eeaa
In the  first two equations,  the expressions  $\divh \xi-(\rho-\frac 12\chih\cdot\chibh),\,\,  -\divh \etab+(\rho-\frac 12\chih\cdot\chibh)$ are, in view of the relations \eqref{eq:xi-ep-loga}, \eqref{eq:etab-ep-loga}, \eqref{eq:Gauss-rho}, equivalent, modulo lower order terms,  to the scalars $\mu=-\laph (\log\ah)+K_\ga- \frac 14 (\trth)^2$ and $\divh\Xi$, which were prescribed in Section \ref{section:lossof deriv}.
Similarly, in view of  \eqref{eq:b+bb},  the right-hand side of the third equation, $\divh(\b+\bb)$,  is equivalent   to  $\divh Y$. Moreover, by the same relation, the scalar $\curlh (\b+\bb)$ is equivalent to $\curlh Y$, which is also among the prescribed scalars in the list  \eqref{eq:6scalars}.  By elliptic estimates  $\divh(\b+\bb)$ and $\curlh (\b+\bb)$ determines $\b+\bb$. Finally, to resolve the loss of derivatives in the last equation, we commute with  $\divh$ and $\curlh$ to derive
\beaa
         \nabh_N \divh\zeta&=& \laph\om-\frac 12 \divh(\b-\bb)+\cdots \\
                  \nabh_N\curlh\zeta&=& -\frac 12 \curlh(\b-\bb)+\cdots
\eeaa
Note that $\zeta=-\Xi$ and, under the PT condition, $\om=-\Pi$. Then, using the relation \eqref{eq:b-bb}  and the fact  that  $\divh \Xi$, $\curlh \Xi$ and $\Pi$ are all  prescribed in \eqref{eq:6scalars},  we  deduce that  both  $\divh(\b-\bb)$ and $\curlh(\b-\bb)$ are determined,  hence so is  $\b-\bb$.

\subsubsection{Degrees of freedom revisited}\lab{subsect:DOF-revisited}
Using the spacetime formalism, we revisit  the discussion on   degrees of freedom in Section \ref{subsubsection:free-data-on-Sigma}  and explain the role of the  gauge scalar $ \nu=\divh\Xi$. 

As mentioned already in Section \ref{subsubsection:free-data-on-Sigma},   even when the coordinates on $\Si$ are fixed,   we can have different initial data sets that evolve to the same Einstein-vacuum spacetime. The ambiguity is due to the different ways of embedding $\Si$ into the spacetime or, in other words, the choice of time function $t$ that defines $\Si$. 

To explain the relation between this freedom and the gauge scalars, we consider a sphere $S_0\subset \Si$ that is $\eps$-close to the unit sphere, with $\Si$ embedded in a spacetime  $(\MM,\g)$ and $\eps$-close to the constant time slice in Minkowski. By extending the null frame using the PT condition as explained above, we obtain a null frame in a spacetime neighborhood of $S_0$ in $\MM$. Now we consider another spacelike hypersurface $\Si'$ satisfying $S_0\subset \Si'$. Given  a sphere foliation on $\Si'$, passing through  $S_0$, one can also define the outward unit normal $N'$ on $\Si'$, thereby also defining the corresponding primed horizontal operators $\nabh'$, $\divh'$, $\curlh'$, and quantities $\pp'$, $\th'$, $\Xi'$, $\Pi'$, $\mu'$, $\nu'$ as in Section \ref{sect:HorizDecompSi}. 

\begin{proposition}\lab{prop:existence-Si'}
There exists an embedded spacelike hypersurface $\Si'$ in a neighborhood of $S_0$ in $(\MM,\g)$ and a vectorfield  $N'$ on $\Si'$ with $(N')^\perp \subset T\Si'$  integrable such  that, for the integral sphere $S'$ of $(N')^\perp$ foliated by some function $r'$,\footnote{The last condition will be explained in Section \ref{subsubsect:ell-0-constraints}. From the perspective of the null frame transformation $(f,\fb,\la)$, 
this condition in \eqref{eq:two-conditions-Si'} fixes the $\ell=0$ part of $\la$, a part that is constant on a sphere and reflects isotropic change in the choice of the embedding of $\Si$ into the spacetime.}
\bea\lab{eq:two-conditions-Si'}
\mu_{\ell\geq 1}'=0,\quad \nu'=0,\quad \int_{S'} \Pi'=0. 
\eea
Note that here all quantities with $'$ are well-defined on $\Si'$ as the $r'$-foliation is determined. The $\ell\geq 1$ modes are suitably defined by deforming the background spherical coordinates.
\end{proposition}

The proposition is purely motivational and plays no role in the proof of the main results; we postpone its proof to a forthcoming work \cite{Gauge-scalars-preparation}.

To conclude, from the spacetime perspective, there are in fact four coordinate degrees of freedom, and the gauge scalars $\mu$ and $\nu=\divh\Xi$ account for such coordinate ambiguities for those  corresponding to  $t$ and $r$. The remaining four scalars
\beaa
\Pi,\quad \curlh \Xi,\quad \divh Y,\quad \curlh Y,
\eeaa
i.e., the free scalars, correspond to the true dynamical degrees of freedom.  

As mentioned already in Remark \ref{remark:comparison},  one can  compare the situation described above with  the case of 
 the  null  characteristic data  on  $C\cup \Cb$,   as  analyzed  in  \cite{Chr1}.  In that  case also, to specify the free  data
one needs to rely on a specific gauge  choice,  for example  the corresponding two geodesic foliations on $C, \Cb$. 
 These   can be thought  as playing  a role   similar  to that of  $\mu$ in our case, while the role of $\nu$ is replaced by the simple  requirement that $C$, $\Cb$ are null. The dynamical degrees of freedom for the bifurcate characteristic problem are  then  given by the shear tensors $\chih, \chibh$  of  the null hypersurfaces $C, \Cb$, expressed relative  to the geodesic foliations,  each of which contributes $2$ degrees of freedom.

\subsection{Linearization of HCS near Schwarzschild}
According to Proposition \ref{prop:existence-Si'}, we only impose the $\ell\geq 1$ part of the gauge scalar $\mu$. This leaves the $\ell=0$ part undetermined. The other gauge scalar $\nu$ is also, by definition, without a spherical mean. Therefore, we need to impose two additional $\ell=0$ conditions in Section \ref{subsubsect:ell-0-constraints}. We then give the full system in terms of quantities with their Schwarzschildian values subtracted in Section \ref{subsect:HCS-linearized-form}.

\subsubsection{Additional $\ell=0$ conditions}\lab{subsubsect:ell-0-constraints}
We now impose two additional conditions that eliminate the $\ell=0$ ambiguities.

{\bf The average of $\ao$.}
As remarked in footnote \ref{ft:r-ambiguity}, we need to eliminate the ambiguity of the relabeling of the $r$-spheres. We impose the condition
\bea\lab{eq:spherical-mean-ah}
\overline{\ao}=-\frac 12 \Up^{-1} r \overline{\thc}.
\eea
\begin{remark}
In fact, if $r$ is the area radius, then \eqref{eq:spherical-mean-ah} is approximately verified. Indeed, we have the relation
\beaa
1= \pa_r (\sqrt{r^2})=\frac 12\frac{1}{\sqrt{r^2}}\pa_r (r^2)=\frac{1}{8\pi r}\pa_r (\mathrm{Area}(S_r))= \frac{1}{8\pi r} \int_{S_r} \ah\, \trth.
\eeaa
However, due to the slow decay we consider and the fact that we are constructing from spatial infinity, it is impossible to show the converse. Therefore, we relax the requirement that by simply imposing an approximate condition \eqref{eq:spherical-mean-ah} without claiming $r$ to be the area radius.
\end{remark}

{\bf The average of $\Pi$.}
At a heuristic level, taking the $\ell=0$ part of (the linearization of) the equation \eqref{eq:structure-constraint-Phi-R} of $\mathring\trt$ (with $(\CC_{Mom})_N=0$) gives
\bea\lab{eq:pa-r-trt-linearized-l=0}
\pa_r (\mathring\trt)_{\ell=0} &=& 2r^{-1} (\mathring\Pi)_{\ell=0}  -r^{-1} (\mathring\trt)_{\ell=0} .
\eea
There are no other HCS equations that can be used to determine $(\mathring\trt)_{\ell=0}$ or  $(\mathring\Pi)_{\ell=0}$. Therefore, we impose an additional condition on the spherical mean of $\Pi$:
\bea\lab{eq:condition-spherical-mean-Pi}
\overline{\Pi} =0.
\eea
In our context, see Remark \ref{rem:main-thm}, $\mathring\trt$ decays like $r^{-2-\de}$, hence \eqref{eq:pa-r-trt-linearized-l=0} then implies $(\mathring\trt)_{\ell=0}=0$.

\subsubsection{The HCS in perturbative form}\lab{subsect:HCS-linearized-form}

 It is well-known that the presence of mass, which is positive for nontrivial complete asymptotically flat data in view of \cite{SchoenYau1979}, \cite{SchoenYau1981}, \cite{Witten1981}, causes an $r^{-1}$ tail. Such a slow decaying tail would be disastrous when treated as a perturbation, and, as a consequence, it is necessary to linearize around the Schwarzschild data rather than the Minkowski one, even when the mass $m$ is small.\footnote{Our analysis  does  not   in fact requires that $m$ is small.} Recall that for the standard Schwarzschild data, we have
\beaa
\trth\0=2\Up^\frac 12 r^{-1},\quad \ah\0=\Up^{-\frac 12},\quad N\0=\Up^\frac 12 \pa_r, \quad K\0=r^{-2},
\eeaa
 where $\Up=1-2m/r$, and $\psi\0$ refers to the value of the quantity $\psi$ in Schwarzschild.
We denote 
\bea\lab{eq:checked-Schw}
\thc:= \trth-2\Up^\frac 12 r^{-1},\quad \ao\1=\ah-\Up^{-\frac 12},\quad \Kc:=K-r^{-2},\quad \muc:= \mu-2mr^{-3}.
\eea
\begin{definition}[Schematic notations]\lab{def:schematic-notation}
We use the following notations for the appropriately weighted perturbed quantities 
\beaa
\Ga_0=\{\ao\},\quad \Ga_1=\{\thc,\thh, \pp, r^{-1}\Ga_0, \trt,\kh,\Xi,\Pi \},\quad \Ga_2=\{
Y, \Kc, r^{-1}\Ga_1, \nabh\Ga_1 \},
\eeaa
where $k$ indicates the maximal order of differentiation of the metric.
\end{definition}
\begin{remark}\lab{rem:schematic-notations}
In the context of the proof of the main theorem, quantities in $\Ga_k$ are expected to have the decay rate of $r^{-1-k-\de}$.
\end{remark}
\begin{proposition}\lab{prop:linearized-eqns-time-symmetry}
The HCS system, along with the conditions \eqref{eq:spherical-mean-ah}, \eqref{eq:condition-spherical-mean-Pi}, can be expressed in the following form, using the schematic notation in \eqref{eq:checked-Schw}:
\bea
\lab{eq:R-transport-kac}
\pa_r \thc &=& \Up^{-\frac 12} \muc +\ao\muc -2r^{-1}\thc-2 (1-3mr^{-1}) r^{-2}\ao\1+\Ga_1\cdot \Ga_1-\frac 12 \ah\, \CC_{Ham},\\
\lab{eq:R-transport-Kc}
\pa_r \Kc &=& r^{-1} \muc-\ah \divh Y-3r^{-1}\Kc -2\Up^\frac 12 r^{-3} \ao+\Ga_1\cdot \Ga_2,\\
\lab{eq:R-transport-ao}
\Up^\frac 12 \laph \ao&=& \Kc-\Up^\frac 12 r^{-1} \thc-\muc -\laph(\Ga_0\cdot \Ga_0)+\Ga_1\cdot\Ga_1, \\
\lab{eq:spherical-mean-ao}
\overline{\ao\1} &=& -\frac 12 \Up^{-1} r \overline{\thc},\\
\lab{eq:unconditional-Codazzi-1}
\d_1 \d_2 \thh&=& (\frac 12 \laph\trth,0)-(\divh Y,\curlh Y),\\
\lab{eq:unconditional-d1p-1}
\d_1 \pp &=& (-\Up^\frac 12 \laph \ao+\laph(\Ga_0\cdot \Ga_0), 0),\\
\lab{eq:N-transport-go}
\slashed{\Lie}_{\pa_r} (r^{-2}\ga) &=& 2r^{-2}\ah \thh+\ah \thc (r^{-2}\ga)+2\Up^\frac 12 \ao r^{-1} (r^{-2}\ga),\\
\lab{eq:N-trt-linearized}
\pa_r \trt&=& \ah \divh \Xi + 2 r^{-1} \Pi  - r^{-1} \trt +\Ga_1\cdot \Ga_1 - \ah (\CC_{Mom})_N,\\
\lab{eq:N-div-Xi}
\pa_r \divh\Xi &=& - \divh\divh (\ah\kh)-4 r^{-1} \, \divh\Xi+\frac 12 \ah \laph \trt + \laph (\ah\Pi) +\Ga_1\cdot \Ga_2+\divh (\ah\, \slashed\CC_{Mom}),\\
\lab{eq:N-curl-Xi}
\pa_r \curlh \Xi &=& - \curlh\divh (\ah\kh)-4 r^{-1} \, \curlh\Xi+\Ga_1\cdot \Ga_2+\curlh (\ah\, \slashed\CC_{Mom}),\\
\lab{eq:average-Pi}
\overline{\Pi}&=& 0. 
\eea
Here $\overline{f}$ denotes the spherical mean of a scalar field $f$ with respect to the metric $\ga$. 
\end{proposition}
\begin{proof}
The proof is done by simply subtracting the equations in Propositions \ref{prop:Unconditional-equations-1} and \ref{prop:constraint-equation-in-frame} by the corresponding ones in Schwarzschild. For the equation of $\Xi$, we commute it with $\divh$ and $\curlh$ respectively.
See Appendix \ref{appendix-derivation-linearized-equations} for details.
\end{proof}

\subsubsection{The prescribed conditions for the  defining scalars}\lab{subsubsect:linear-operator-L}
In view of the discussion above, we seek solutions of HCS satisfying
\bea
\lab{eq:prescribed-conditions-munu}
\mu_{\ell\geq 1}=0, \qquad \nu=0,
\eea
and 
\bea\lab{eq:prescribed-conditions}
\bsplit
 (\divh Y)_{\ell\geq 2}=\Bb, &\qquad \qquad \quad (\curlh Y)_{\ell\geq 2}=\Bbd,\\
 (\laph (\ah\Pi))_{\ell\geq 2}=\Kk,& \qquad  r^{-4} \pa_r (r^4 \curlh \Xi)_{\ell\geq 2}=\Kkd.
\end{split}
\eea
We then write
\bea
\lab{eq:BbBbdKdKkd}
\bsplit
\divh Y&=\Bb+\B_{\ell\leq 1},\qquad \qquad \qquad\curlh Y=\Bbd+\Bd_{\ell\leq 1}, \\
\laph(\ah\Pi)&=\Kk+\K_{\ell\leq 1},\qquad r^{-4}\pa_r(r^4 \curlh\Xi)=\Kkd-\Kd_{\ell\leq 1}.
\end{split}
\eea
where $\B_{\ell\leq 1}:=(\divh Y)_{\ell\leq 1}$, $\Bd_{\ell\leq 1}:= (\curlh Y)_{\ell\leq 1}$, $\K_{\ell\leq 1}:=(\laph(\ah\Pi))_{\ell\leq 1}$, $\Kd_{\ell\leq 1}:= -(r^{-4}\pa_r(r^4 \curlh\Xi))_{\ell\leq 1}$.

\subsubsection{Triangular block structure of the perturbative form of HCS}
It order to illustrate the structure of the system, it helps to introduce the following notation:
\bea\lab{eq:def-Psi-1-12-intro}
\bsplit
\Psi_1 &=\thc,\quad \Psi_2=\Kc,\quad \Psi_3=\ao,\quad \Psi_4=\thh,\quad \Psi_5=\pp,\quad \Psi_6=Y, \\
\Psi_7 &=\trt,\quad \Psi_8=\kh,\quad \Psi_9 = \Xi,\quad \Psi_{10}=\Pi,\\  \Psi_{11} &=(\B_{\ell\leq 1},\Bd_{\ell\leq 1}),\quad \Psi_{12}=(\K_{\ell\leq 1},\Kd_{\ell\leq 1}).
\end{split}
\eea
Before  writing the HCS system in terms of these new variables,  we make the following substitutions. Projecting \eqref{eq:unconditional-Codazzi-1} to $\ell\leq 1$, we obtain
\bea\lab{eq:heuristic-substitution-ell=1-B}
\B_{\ell\leq 1} := (\divh Y)_{\ell\leq 1}=\frac 12 (\laph\thc)_{\ell=1}+err,
\eea
where $err$ contains nonlinear error terms.\footnote{Indeed, in view of \eqref{eq:unconditional-Codazzi-1}, the terms are the $\ell\leq 1$ parts of $\d_1\d_2\thh$ and the $\ell=0$ part of $\laph\trth$, which are both zero at the linear level.}
Similarly, projecting \eqref{eq:N-div-Xi} to $\ell\leq 1$, using also the gauge condition $\nu=\divh\Xi=0$, we deduce
\bea\lab{eq:heuristic-substitution-ell=1-K}
\K_{\ell\leq 1}:= (\laph (\ah \Pi))_{\ell\leq 1} =-\frac 12(\laph\trt)_{\ell=1}+err.
\eea
In addition, using the condition $\muc_{\ell\geq 1}=0$, we can also write, according to the definitions \eqref{eq:def-mu} and \eqref{eq:checked-Schw},
\bea\lab{eq:heuristic-substitution-ell=0-mu}
\muc=\muc_{\ell=0}=\Kc_{\ell=0}-\Up^\frac 12 r^{-1} (\thc)_{\ell=0}+err,
\eea
where $err$ is quadratic in $\thc$.

Combining these substitutions with \eqref{eq:BbBbdKdKkd}, we can now write the HCS system as
\bea\lab{eq:HCS-schematic-form}
\, ^{(\ga)}L[\Psi] = \begin{pmatrix}
0 \\
-\Up^{-\frac 12} \Bb \\
(0,0) \\
-(\mathcal{B}, \Bbd) \\
0 \\
(\Bb, \Bbd) \\
0 \\
(\Up^{-\frac 12}\Kk, -\Up^\frac 12 \Kkd) \\
\Kkd \\
(\Up^\frac 12 \Kk,0)
\end{pmatrix}+ err,
\eea
where, for a given perturbed horizontal metric $\tilde\ga$, the linear operator $\, ^{(\tilde\ga)} L$ is defined as
\bea\lab{eq:def-L-operator}
\, ^{(\tilde\ga)} L[\Psi]:=\begin{pmatrix}
  (\pa_r+2r^{-1})\Psi_1+2(1-3mr^{-1}) r^{-2} \Psi_3-\Up^{-\frac 12}(\Psi_2-\Up^\frac 12 r^{-1} \Psi_1)_{\ell=0} \\
  (\pa_r+3r^{-1})\Psi_2+2\Up^\frac 12 r^{-3} \Psi_3 + \frac 12 \Up^{-\frac 12} (\laph\Psi_1)_{\ell=1}-r^{-1} (\Psi_2-\Up^\frac 12 r^{-1} \Psi_1)_{\ell=0}\\ 
  (\Up^\frac 12 \laph \Psi_3, \overline{\Psi_3}) -(\Psi_2-\Up^\frac 12 r^{-1} \Psi_1 -\overline{\displaystyle\Psi_2-\Up^\frac 12 r^{-1} \Psi_1}, -\frac 12 \Up^{-1} r \overline{\Psi_1})\\
  \d_1\d_2 \Psi_4-(\frac 12\laph\Psi_1 , 0) +\Psi_{11}  \vspace{0.5ex}\\ 
 \d_1 \Psi_5 + (\Up^\frac 12 \laph \Psi_3,0) \\
  \d_1 \Psi_6 -\Psi_{11} \\
(\pa_r+ r^{-1}) \Psi_7 - 2r^{-1} \Psi_{10} \\
\d_1\d_2 \Psi_8 -(\frac 12 \laph \Psi_7,0)_{\ell\geq 2}-\Up^{-\frac 12}\Psi_{12} \\
(r^{-4}\pa_r(\curlh\Xi))_{\ell=1} + \mathcal P_2\Psi_{12} \\ 
(\laph \Psi_{10},\overline \Psi_{10} ) +\frac 12 ((\laph \Psi_7)_{\ell=1},0)
\end{pmatrix}.
\eea
Here, all the horizontal operators are defined relative to  $\tilde\ga$, and $\Psi_4$ and $\Psi_8$ are traceless with respect to $\tilde\ga$. The notation $\mathcal P_2$ denotes the projection into the second component, i.e., $\mathcal P_2 \Psi_{12}=\Kd_{\ell\leq 1}$.
\begin{remark}[Block-triangular structure]\lab{rem:block-triangular}
Notice that apart from $\Psi_1$, $\Psi_2$, $\Psi_3$, other quantities do not enter the first three rows in the expression of $\, ^{(\tilde\ga)} L[\Psi]$. 
In other words, denoting $\Psi_{main}=(\Psi_1,\Psi_2,\Psi_3)$, the 
linear operator splits into two parts 
\beaa
\, ^{(\tilde\ga)}L[\Psi] = \begin{pmatrix}
\, ^{(\tilde\ga)}L_{main}[\Psi_{main}] \\
 \, ^{(\tilde\ga)} L_{rem}[\Psi]
 \end{pmatrix}.
\eeaa
Equivalently, if we write $\, ^{(\tilde\ga)}L$ in the matrix form, we have a block-triangular structure with respect to the first $3\times 3$ block. Therefore, we can determine $\Psi_{main}=(\Psi_1,\Psi_2,\Psi_3)$ first, independently of other quantities. Once they are determined, taking into account the fact that $\laph\Psi_{10}$ is part of the input (corresponding to the free scalar $\Kk$), 
the second block itself also has a triangular structure.\footnote{Here we mainly refer to the $\ell\geq 2$ parts. The structure of the $\ell=1$ parts of the system is different, as will be discussed in Section \ref{subsect:ell-0-1-constraints} just below.}
\end{remark}
The metric $\ga$ in \eqref{eq:HCS-schematic-form} satisfies \eqref{eq:N-transport-go}, i.e., $\ga$ is in turn determined by $\Psi$. Therefore, we construct the solution through an iteration argument in Section \ref{sect:proof}. In the iteration scheme, the system is solved as if $\ga$ is fixed at each step. 
The block triangular structure pointed out in Remark \ref{rem:block-triangular} allows us, in solving the linear system at each step, to invert the main part $\, ^{(\tilde\ga)}L_{main}[\Psi_{main}]$ first, as we will carry out in Section \ref{subsect:solve-main-part}.

\begin{definition}[Linearized system around Schwarzschild]
We call the system $\, ^{(\gz)}L[\mathring\Psi]=0$ the $\ga\0$-linearized system.
\end{definition}
\begin{remark}\lab{remark:gauge-scalar-not-in-linearized-operator}
Note that in the definition of the $\, ^{(\tilde\ga)}L$ operator, we have already taken the gauge conditions $\mu_{\ell\geq 1}=0$, $\nu=0$ into account, and hence the corresponding terms are not included in the expression of $\, ^{(\tilde\ga)}L[\Psi]$.
\end{remark}

\subsection{The $\ell=1$ constraints}\lab{subsect:ell-0-1-constraints}
In this section, we perform the analysis of $\ell= 1$ modes for the $\gz$-linearized system $\, ^{(\gz)}L[\mathring\Psi]=0$, with $\gz$ the round metric defined in \eqref{eq:def-gz}. 
Therefore, in this section, all horizontal operators $\divh$, $\curlh$, $\cdots$ are defined through $\gz$.
Since $\mathring\Psi_4=\mathring\thh$ and $\mathring\Psi_8=\mathring\kh$ are traceless with respect to $\gz$, hence fully supported on $\ell\geq 2$ modes, they can be disregarded in the analysis below. As we shall see in the following proposition, the $\ell=1$ modes are completely determined by the conditions at spatial infinity. This is unlike the $\ell\geq 2$ modes, where we have to take into account the additional freedom given by the four free scalars.\footnote{This is, of course, under the condition that the gauge scalars are specified.}
\begin{proposition}
Consider the $\gz$-linearized system $\, ^{(\gz)}L[\mathring\Psi]=0$. 
\begin{itemize}
\item[(i)] 
If we impose the conditions
\bea\lab{eq:center-of-mass-zero}
\lim_{r\to \infty} r^3 (\mathring \Psi_1)_{\ell=1,i}=\mathring\cc_i,\quad \lim_{r\to\infty} r^4 (\mathring \Psi_2)_{\ell=1} =0,
\eea
then we have 
\bea\lab{eq:Psi-1-2-3-linearized-ell=1}
(\mathring\Psi_1)_{\ell=1,i}=\mathring\cc_i r^{-3}+O(|\mathring\cc| r^{-4}), \quad (\mathring \Psi_2)_{\ell=1}=O(|\mathring\cc| r^{-5}),\quad 
(\mathring \Psi_3)_{\ell=1,i}= \frac 12 \mathring\cc_i r^{-2}+O(\mathring\cc_i r^{-3}).
\eea
\item[(ii)] 
If, in addition, we impose the conditions
\bea\lab{eq:condition-infinity-linearized-P-J}
\lim_{r\to\infty} r^2 (\mathring \Psi_7)_{\ell=1}=0 ,\quad \lim_{r\to \infty} r^4(\curlh \mathring{\Psi}_9)_{\ell=1,i}=\mathring{\bf a}_i, \quad i=-1,0,1,
\eea
then we have
\beaa
(\mathring\Psi_{10})_{\ell=1}=(\mathring\Psi_{7})_{\ell=1}=0,\quad (\curlh \mathring{\Psi}_9)_{\ell=1,i}=r^{-4} \mathring{\bf a}_i, \quad i=-1,0,1.
\eeaa
\end{itemize}
\end{proposition}
\begin{remark}
We will show in Appendix \ref{appendix:physical-quantities} that the conditions \eqref{eq:center-of-mass-zero} and \eqref{eq:condition-infinity-linearized-P-J} are the linearized version of the conditions $\C_i=-\frac{1}{8\pi m}\cc_i$, $\P_i =0$, $\J_i=\frac{1}{8\pi}{\bf a}_i$ with $\C_i$, $\P_i$, $\J_i$ defined in \eqref{eq:def-ADM-charges}.
\end{remark}

\begin{proof}[Proof of (i)]
The corresponding rows of $\mathring\Psi_1=\mathring\thc$, $\mathring\Psi_2=\mathring\Kc$, and $\mathring\Psi_3=\mathring\ao$ in $\, ^{(\gz)}L[\mathring\Psi]=0$, see \eqref{eq:def-L-operator}, when projected to $\ell=1$, read 
\bea
\pa_r (\mathring\Psi_1)_{\ell=1}&=& 
-2r^{-1}(\mathring\Psi_1)_{\ell=1}-2 (1-3mr^{-1}) r^{-2} (\mathring\Psi_3)_{\ell=1}, \\
\pa_r (\mathring \Psi_2)_{\ell=1} &=& -3r^{-1} (\mathring \Psi_2)_{\ell=1}-2\Up^\frac 12 r^{-3} (\mathring\Psi_3)_{\ell=1}-\frac 12 \Up^{-\frac 12}(\laph \mathring \Psi_1)_{\ell=1},\\
\lab{eq:Psi-3-ell=1}
-2\Up^{\frac 12} r^{-2} (\mathring\Psi_3)_{\ell=1}&=& (\mathring \Psi_2)_{\ell=1}-\Up^\frac 12 r^{-1} (\mathring\Psi_1)_{\ell=1},
\eea
which, by eliminating $(\mathring\Psi_3)_{\ell=1}$, can be reduced to
\beaa
\pa_r (\mathring\Psi_1)_{\ell=1}+3r^{-1} (\mathring\Psi_1)_{\ell=1} &=& O(mr^{-2}) (\mathring\Psi_1)_{\ell=1} +(1+O(mr^{-1})) (\mathring \Psi_2)_{\ell=1},\\
\pa_r (\mathring \Psi_2)_{\ell=1} +2r^{-1} (\mathring \Psi_2)_{\ell=1} &=& \Up^{-\frac 12} (1-\Up) r^{-2} (\mathring\Psi_1)_{\ell=1},
\eeaa
or, in the matrix form,
\beaa
\pa_r \begin{pmatrix} (\mathring\Psi_1)_{\ell=1} \\ (\mathring \Psi_2)_{\ell=1} \end{pmatrix} = \begin{pmatrix} -3r^{-1} & 1 \\ 0  & -2r^{-1} \end{pmatrix} \begin{pmatrix} (\mathring\Psi_1)_{\ell=1} \\ (\mathring \Psi_2)_{\ell=1} \end{pmatrix} +\begin{pmatrix} O(mr^{-1}) (\mathring\Psi_2)_{\ell=1} +O(mr^{-2})(\mathring\Psi_1)_{\ell=1} \\ O(mr^{-2}) (\mathring \Psi_2)_{\ell=1} \end{pmatrix},
\eeaa
and hence,
\beaa
\pa_r \begin{pmatrix} r^3 (\mathring\Psi_1)_{\ell=1} \\ r^{4} (\mathring \Psi_2)_{\ell=1} \end{pmatrix} = \begin{pmatrix} 0 & 1 \\ 0 & 2 \end{pmatrix}r^{-1} \begin{pmatrix} r^3 (\mathring\Psi_1)_{\ell=1} \\ r^{4} (\mathring \Psi_2)_{\ell=1} \end{pmatrix} +O(mr^{-2}) \begin{pmatrix} r^3 (\mathring\Psi_1)_{\ell=1} \\ r^4 (\mathring \Psi_2)_{\ell=1} \end{pmatrix},
\eeaa
or, with the $\mathring\cc$-part subtracted,
\beaa
\pa_r \begin{pmatrix} r^3 (\mathring\Psi_1)_{\ell=1,i}-\mathring\cc_i \\ r^{4} (\mathring \Psi_2)_{\ell=1,i} \end{pmatrix} &=& \begin{pmatrix} 0 & 1 \\ 0 & 2 \end{pmatrix}r^{-1} \begin{pmatrix} r^3 (\mathring\Psi_1)_{\ell=1,i} \\ r^{4} (\mathring \Psi_2)_{\ell=1,i} \end{pmatrix} +O(mr^{-2}) \begin{pmatrix} r^3 (\mathring\Psi_1)_{\ell=1,i} \\ r^4 (\mathring \Psi_2)_{\ell=1,i} \end{pmatrix}+O(|\mathring\cc| r^{-2}) \\
&=& \begin{pmatrix} 0 & 1 \\ 0 & 2 \end{pmatrix}r^{-1} \begin{pmatrix} r^3 (\mathring\Psi_1)_{\ell=1,i}-\mathring\cc_i \\ r^{4} (\mathring \Psi_2)_{\ell=1,i} \end{pmatrix} +O(mr^{-2}) \begin{pmatrix} r^3 (\mathring\Psi_1)_{\ell=1,i}-\mathring\cc_i \\ r^4 (\mathring \Psi_2)_{\ell=1,i} \end{pmatrix}+O(|\mathring\cc| r^{-2}),
\eeaa
where, for the second equality, we use that the first column of $\begin{pmatrix} 0 & 1 \\ 0 & 2 \end{pmatrix}$ is zero.
The matrix $\begin{pmatrix} 0 & 1 \\ 0 & 2 \end{pmatrix}$ is not symmetric, hence not non-negative definite; however, it is accretive with respect to some modified inner product over $\mathbb{R}^2$, see Lemma \ref{lem:Duhamel-round-case}. This allows us to integrate the equation from $r=\infty$, using the condition \eqref{eq:center-of-mass-zero}, and obtain 
\beaa
r^3 (\mathring\Psi_1)_{\ell=1,i}-\mathring\cc_i =O(|\mathring\cc| r^{-1}),\quad r^{4} (\mathring \Psi_2)_{\ell=1}=O(|\mathring\cc| r^{-1}),
\eeaa
as required. The expansion of $(\mathring\Psi_3)_{\ell=1}$ then follows from \eqref{eq:Psi-3-ell=1}, which we used to eliminate $(\mathring\Psi_3)_{\ell=1}$.
\end{proof}

\begin{proof}[Proof of (ii)]
The corresponding rows of $\, ^{(\gz)}L[\mathring\Psi]=0$ in fact come from
projecting the linearized version of equations \eqref{eq:N-trt-linearized}-\eqref{eq:N-curl-Xi} into $\ell=1$ modes, with the condition $\displaystyle{(\divh\mathring\Psi_9)_{\ell=1}=0}$. We have\footnote{Note that in view of \eqref{eq:def-L-operator}, the equation of $\mathring\Psi_8$ implies that $\mathring\Psi_{12}=0$.}
\beaa
\pa_r (\mathring\Psi_7)_{\ell=1} &=& 
2 r^{-1} (\mathring\Psi_{10})_{\ell=1}  - r^{-1} (\mathring\Psi_7)_{\ell=1} ,\\
\pa_r (\curlh \mathring\Psi_9)_{\ell=1} &=& -4 r^{-1} (\curlh \mathring\Psi_9)_{\ell=1}, \\
(\laph \mathring\Psi_{10})_{\ell=1} &=& - \frac 12 (\laph \mathring\Psi_7)_{\ell=1}.
\eeaa
The third equation simply gives $(\mathring\Psi_7)_{\ell=1}=-2(\mathring\Psi_{10})_{\ell=1}$. Combining this with the first equation gives
\beaa
\pa_r (\mathring\Psi_7)_{\ell=1} &=& -2 r^{-1} (\mathring\Psi_7)_{\ell=1},\\
\pa_r (\curlh\mathring\Psi_9)_{\ell=1} &=& -4 r^{-1} (\curlh\mathring\Psi_9)_{\ell=1}.
\eeaa
Therefore, the solution is completely determined from the condition at infinity, which we impose in \eqref{eq:condition-infinity-linearized-P-J}.
Hence, we obtain $(\mathring\Psi_7)_{\ell=1}=0$ and $(\curlh\mathring\Psi_9)_{\ell=1,i}=r^{-4} \mathring{\bf a}_i$. 
\end{proof}

\subsection{Precise statement of the main theorem}
We now state the precise form of the main theorem.
\begin{theorem}[Main Theorem]\lab{thm:main-precise}
There exists a sufficiently small constant $\eps>0$, such that given $m>0$, $r_0>2m$, two constant triplets ${\bf a}=({\bf a}_1,{\bf a}_2,{\bf a}_3)$, ${\bf c}=(\cc_1,\cc_2,\cc_3)$ that are $\eps$-close to zero in $\mathbb{R}^3$, and four scalar functions $\Bb$, $\dual\Bb$, $\Kk$, $\Kkd$, supported on $\ell\geq 2$ in the sense of \eqref{eq:decomposition-ell-geq-2}, satisfying
\bea\lab{eq:main-thm-time-symmetric-B-Bd-bounds}
\sup_{r\in [r_0,\infty)} r^{3+\de} ||(\Bb,\dual\Bb,\Kk)||_{\H^s(S_r)}\leq \eps, \quad \sup_{r\in [r_0,\infty)} r^{4+\de} ||\Kkd||_{\H^s(S_r)}\leq \eps, \quad \text{for some integer $s\geq 3$},
\eea
then there exists a metric $g$ and a symmetric $2$-tensor $k$ on $\Si=(r_0,\infty)\times \mathbb{S}^2$ solving the constraint equation \eqref{ece} such that, under our choice of the frame, for which $\mu_{\ell\geq 1}=\nu=0$, we have
\beaa
\left(\divh Y- \Bb\right)_{\ell\geq 2}= 0,\quad 
\left(\curlh Y-\dual\Bb \right)_{\ell\geq 2}= 0,\quad 
\left(\laph(\ah\Pi)-\Kk \right)_{\ell\geq 2}= 0,\quad 
\left(r^{-4}(\pa_r (r^4\curlh \Xi))-\dual\Kk \right)_{\ell\geq 2}= 0.
\eeaa
Moreover, the ADM charges defined in \eqref{eq:def-ADM-charges} satisfy
\bea\lab{eq:main-thm-ADM-charges}
 \E= m, \quad 
 \J_i=\frac{1}{8\pi}{\bf a}_i,\quad 
 \P_i=0,\quad
 \C_i = -\frac{1}{8\pi m}\cc_i.
 \eea
\end{theorem}
\begin{remark}\lab{rem:main-thm}
Note that the four scalars $(\Bb,\Bbd,\Kk,\Kkd)$ are all at the level of one derivative of curvature (two derivatives of the components of $k$). The theorem, therefore, asserts that we can produce general perturbed initial data with decay rate $O(r^{-1-\de})$ at the metric level. However, compared with $(\Bb, \Bbd,\Kk)$ that is allowed to decay at $O(r^{-4-\de})$, $\Kkd$ must decay one order faster, as is manifest by its alignment with $r^{-4} \pa_r(r^4 \curlh \Xi)$, an expression naturally comes from \eqref{eq:N-curl-Xi}. This is in fact related to the remark in \cite[Page 11]{CK} on the existence of the angular momentum: While the metric is allowed to decay at the $r^{-\frac 32}$ level in \cite{CK}, it is shown through the momentum constraint $\slashed{\CC}_{Ham}=0$ that the angular momentum exists despite the lack of decay at first glance. The equation \eqref{eq:N-curl-Xi} in fact comes from the same momentum constraint, see Appendix \ref{subsect:proof-constraint-equation-in-frame}. 
\end{remark}

\subsection{List of notations and conventions}
For the benefit of the reader, we  recall below  the main  notations we have  introduced:
\bea\lab{eq:list-of-notations}
\bsplit
& p_a:=g(\nab_N N,e_a),\quad  \th_{ab}:=g(\nab_a N, e_b),\quad \slashed{R}_{ab}:=R(N,e_a,N,e_b),\quad Y_a:=R(N, e_b, e_b, e_a),\\ 
& \Th_{ab}:=k(e_a,e_b),\quad \Xi_a:=k(N,e_a),\quad \Pi :=k(N,N),\quad \ah=(N r)^{-1},\\
& \mu=-\slashed{\Delta}(\log\ah)+K-\frac 14(\trth)^2,\quad \nu=\divh\Xi, \\
& g=\ah^2 dr^2+\ga,\quad \ga\0=r^2 (\gs)=r^2 ((d\vth^1)^2+\sin^2 (\vth^1) (d\vth^2)^2).
\end{split}
\eea
We wish to construct solutions such that
\beaa
    \Bb:=(\divh Y)_{\ell\geq 2},\quad \Bbd:=(\curlh Y)_{\ell\geq 2},\quad \Kk:=(\laph\Pi)_{\ell\geq 2},\quad \Kkd:= (r^{-4}\pa_r(r^4\curlh \Xi))_{\ell\geq 2}.
\eeaa
For a general metric $\ga$, we use the notation $\overline{f}^\ga$ for the spherical mean with respect to $\ga$. We drop the $\ga$ when there is no danger of confusion.
We use the notation $\widecheck{\psi}$ for the quantity $\psi$ subtracted by its value in Schwarzschild.

\section{Technical lemmas}
\subsection{Equivalent norms}
We have the following equivalence of the Sobolev norms defined through $\gz$ and $\ga$.
\begin{lemma}\lab{lemma:equivalence-norms}
Consider a $2$-sphere $S_r$ equipped with standard spherical coordinates and the associated rescaled round metric $\gz=r^2(\gs)$. 
Suppose that another metric $\ga$ on $S$ satisfies, for some integer $s\geq 3$,
\bea\lab{eq:bound-metric-equivalence-lemma}
r^{-1} ||\ga-\gz||_{\H^{s+1}(S_r)} \leq \mathring \eps\ll 1.
\eea
Then, denoting by $\nabh$ the covariant derivative of $\ga$, we have,
\begin{itemize}
\item For all $i\leq s+2$ and scalar field $\phi$, we have
\beaa
 ||(r\nabh)^{\leq i} \phi||_{L^2(S_r,\ga)} \sim ||\phi||_{\H^{i}(S_r)};
\eeaa
\item For all $i \leq s+1$ and rank-$k$ covariant tensor $U=U_{a_1\cdots a_k}$,
\beaa
 ||(r\nabh)^{\leq i} \phi||_{L^2(S_r,\ga)} \sim ||\phi||_{\H^{i}(S_r)}.
\eeaa
\end{itemize}
In both cases, the two-sided implicit constant can be taken to be $(1+C\mathring\eps)$ for some constant $C>0$.
\end{lemma}
\begin{proof}
By standard Sobolev embeddings, we have $||(r\nabz)^i U||_{L^\infty(S_r)}\lesssim r^{-1} || U ||_{\H^s(S_r)}$ for $i\leq s-1$. In particular, by \eqref{eq:bound-metric-equivalence-lemma}, we infer $||(r\nabz)^i (\ga-\gz)||_{L^\infty(S_r)}\lesssim \mathring\eps$ for $i\leq s-1$.

Recall that the covariant derivative of $\ga$ (resp. $\gz$) is $\nabh$ (resp. $\nabz$). In view of Remark \ref{rem:inverse-metric},
for a scalar field $\phi$, we have, schematically,
$\nabh \phi=\nabz \phi$, $\nabh^2 \phi = (\nabz)^2 \phi+\nabz(\ga-\gz) \cdot \nabz \phi$, and, inductively,
\bea\lab{eq:inductive-nabh-i-phi}
\nabh^i \phi = (\nabz)^i \phi+\sum_{\substack{i_1+i_2= i,\\i_1\leq i-1, i_2\leq i-1}} (\nabz)^{i_1}(\ga-\gz) \cdot (\nabz)^{i_2} \phi.
\eea
For a general horizontal covariant tensor $U$, we have, schematically,
$\nabh U=\nabz U+\nabz (\ga-\gz)\cdot U$, $\nabh^2 \phi = (\nabz)^2 \phi+\nabz (\ga-\gz) \cdot \nabz U+(\nabz)^2 (\ga-\gz) \cdot U$, and, inductively,
\bea\lab{eq:inductive-nabh-i-U}
\nabh^i U= (\nabz)^i U+\sum_{i_1+i_2= i, i_2\leq i-1} (\nabz)^{i_1}(\ga-\gz) \cdot (\nabz)^{i_2} U.
\eea
Using \eqref{eq:inductive-nabh-i-phi},
\beaa
||\nabh^i \phi-(\nabz)^i \phi ||_{L^2(S_r)}\lesssim \sum_{\substack{i_1+i_2= i,\\i_1\leq i-1, i_2\leq i-1}} ||(\nabz)^{i_1}(\ga-\gz) \cdot (\nabz)^{i_2} \phi||_{L^2(S_r)}.
\eeaa
Since $s+2\geq 5$, for $i\leq s+2$, either $i_1$ or $i_2$ in the sum is no greater than $s-1$, for which we can apply the $L^\infty$ estimate, leaving the other controlled by the $L^2$-type norms. The estimate then easily follows. The case for covariant tensor follows similarly using \eqref{eq:inductive-nabh-i-U}.
\end{proof}

\subsection{Hodge estimates}
\begin{lemma}\lab{lemma:Hodge-estimate-round}
Consider a $2$-sphere $S_r$ equipped with standard spherical coordinates and the associated rescaled round metric $\gz=r^2(\gs)$, and another metric $\ga$ on $S$ which satisfies $r^{-1} ||\ga-\gz||_{\H^{s+1}(S_r)}\leq \mathring\eps \ll 1$ for some $s\geq 3$. Suppose for $\xi\in \ss_1$ and $h\in \ss_2(S,\ga)$ we have
\beaa
\d_1 \xi=(f,\dual f),\quad \d_2 h=F.
\eeaa
Then the following estimates hold for all $i\leq s$:
\bea\lab{eq:Hodge-estimate-round}
||\xi||_{\H^{i+1}(S_r)}\lesssim r ||(f,\dual f)||_{\H^i(S_r)},\quad ||h||_{\H^{i+1}(S_r)} \lesssim r ||F||_{\H^i(S_r)}.
\eea
\end{lemma}
\begin{proof}
We only prove the first inequality, as the second is similar. Commuting the equations with $\nabh^i$, we have, schematically,
\beaa
\d_1 \nabh^i \xi=\nabh^i (f,\dual f)+\nabh^{i-1} (K_\ga \cdot \xi) . 
\eeaa
Here we adopt the convention that $\nabh^{-1}\psi=0$ for any quantity $\psi$. In view of the assumption for $\ga$ and standard Sobolev embeddings, we have $r^{-1} ||(K_\ga-r^{-2}) ||_{L^{\infty}(S_r)}\lesssim \mathring\eps$, i.e., $r^{2} K$ is uniformly close to $1$.
The standard Hodge estimate (\cite[Lemma 2.2.2]{CK}) is then applicable and implies
\beaa
||\nabh^{i+1} \xi||_{L^2(S,\ga)}+ r^{-1} || \nabh^i \xi||_{L^2(S,\ga)} \lesssim || \nabh^i (f,\dual f)+\nabh^{i-1} (K_\ga \cdot \xi)||_{L^2(S,\ga)}.
\eeaa
Therefore, applying Lemma \ref{lemma:equivalence-norms} to $\xi$, $(f,\dual f)$ and the schematic $1$-form $K_\ga\cdot \xi$, we obtain, for $i\leq s$,
\beaa
||\xi||_{\H^{i+1}(S_r)} &\lesssim & ||(r\nabh)^{\leq i+1}\xi||_{L^2(S,\ga)} \lesssim \sum_{j\leq i} ||(r\nabh)^{j+1}\xi||_{L^2(S,\ga)} +||\xi||_{L^2(S,\ga)} \\
& \lesssim & \sum_{j\leq i} ||r (r\nabh)^j (f,\dual f) +r^2(r\nabh)^{j-1} (K_\ga \cdot \xi) ||_{L^2(S,\ga)} +r ||(f,\dual f)||_{L^2(S,\ga)}\\
&\lesssim &  r ||(f,\dual f)||_{\H^i(S_r)}+ r^2 ||(K_\ga \cdot \xi) ||_{\H^{i-1}(S_r)}+r ||(f,\dual f)||_{L^2(S_r)}\\
&\lesssim & r ||(f,\dual f)||_{\H^i(S_r)} +||\xi||_{\H^{i-1}(S_r)} ,
\eeaa
where for the nonlinear term $K_\ga\cdot \xi$, we used the standard $L^2$-$L^\infty$ type estimates, with $L^\infty$ applied to the factor with less derivatives.
The estimate for $\xi$ then follows by induction. The estimate for $h$ follows in a similar way.
\end{proof}

We also have the following estimate regarding the $\ell\leq 1$ part of $\dga_1 \dga_2 h$, which is heuristically mostly supported on $\ell\geq 2$.
\begin{lemma}
Suppose $h$ is a symmetric $2$-tensor on $(S,\ga)$. Then,\footnote{Here we do not assume that $h$ is traceless with respect to $\ga$, but we extend the definition of $\d_2^\ga$ trivially, to all symmetric $2$-tensors, by $\d_2^\ga h:=\divh^\ga h$.}
\begin{equation}\lab{eq:estimate-d1-d2-ell-leq-1}
r^{-1} || (\dga_1 \dga_2 h)_{\ell\leq 1} ||_{\H^s(S_r)} 
\lesssim ||\lapz \trz h||_{L^{\infty}(S_r)}+ r^{-2} ||(r\nabz)^{\leq 2} h||_{L^{\infty}(S_r)} ||(r\nabz)^{\leq 2}(\ga-\gz) ||_{L^{\infty}(S_r)}.
\end{equation}
\end{lemma}
\begin{proof}
Recall from \eqref{eq:Hs-Linfty-ell-leq-1} that $r^{-1} || (\dga_1 \dga_2 h)_{\ell\leq 1} ||_{\H^s(S_r)}\lesssim || (\dga_1 \dga_2 h)_{\ell\leq 1} ||_{L^\infty(S_r)}$.
We have, schematically,
\beaa
\dga_1\dga_2 h=\dgz_1 \dgz_2 h + (\ga-\gz)\cdot (\nabz)^2 h+\nabz (\ga-\gz)\cdot \nabz h+(\nabz)^2 (\ga-\gz)\cdot h .
\eeaa
Note that $\widehat h:=h-\frac 12 (\trz h) \gz$ is traceless, and hence we have $(\dgz_1\dgz_2 \widehat h)_{\ell\leq 1}=0$ by \eqref{eq:basicestimatesforJp-onSi_*}. Therefore,
\beaa
(\dgz_1 \dgz_2 h)_{\ell\leq 1} = \left(\dgz_1 \dgz_2 \Big(\frac 12 (\trz h) \gz\Big)\right)_{\ell\leq 1} =\Big(\frac 12 \lapz \trz h,0\Big)_{\ell\leq 1}.
\eeaa
The estimate \eqref{eq:estimate-d1-d2-ell-leq-1} then follows by combining these relations.
\end{proof}

\subsection{Commutation formulas} 
We first give the commutation formula between $\nabh_N$ and $\nabh_a$ on horizontal covariant tensors.
\begin{lemma}
We have 
\bea\lab{eq:commutation-nab-N-nab-a}
[\nabh_N, \nabh_a] U_{b_1\cdots b_k} = -\th_{ac} \nabh_c U_{b_1\cdots b_k} -\pp_a \nabh_N U_{b_1\cdots b_k} +(p_{b_i} \th_{ac} +p_c \th_{ab_i}+\in_{b_i c} \dual Y_a) U_{b_1\cdots c\cdots b_k}.
\eea
\end{lemma}
\begin{proof}
We only prove the case $k=1$ for simplicity, and the higher rank cases are similar.
For a horizontal $1$-form $\xi$, we have $\nab_a\xi_b 
=\nabh_a \xi_b$, $\nab_N \xi_b = \nabh_N \xi_b$. Therefore, we have, see \eqref{eq:covariant-rules} for the calculation rules,
\beaa
\nab_N \nab_a \xi_b &=&  \nab_N (\nab \xi)_{ab} = \nabh_N (\nab \xi)_{ab} +\pp_a (\nab\xi )_{N b}+\pp_b (\nab\xi)_{a N} =\nabh_N \nabh_a \xi_b+\pp_a (\nab_N \xi_b) +\pp_b (\nab_a \xi_N) \\
&=& \nabh_N \nabh_a \xi_b+\pp_a \nabh_N \xi_b -\pp_b \th_{ac} \xi_c,
\eeaa
and
\beaa
\nab_a\nab_N \xi_b &=& \nab_a (\nab \xi)_{Nb} =\nabh_a (\nab \xi)_{Nb} -\th_{ac} (\nab\xi)_{cb} +\th_{ab} (\nab\xi)_{NN} \\
&=& \nabh_a \nabh_N \xi_b -\th_{ac} \nabh_c \xi_b +\th_{ab} (\pp_c \xi_c).
\eeaa
On the other hand, we have, using \eqref{eq:R-Nabc},
\beaa
\nab_N \nab_a \xi_b - \nab_a\nab_N \xi_b =R_{Nabc} \xi_c=\, \in_{bc} \xi_c \dual Y_a.
\eeaa
Therefore,
\beaa
\nabh_N\nabh_a \xi_b -\nabh_a\nabh_N \xi_b &=& \nab_a\nab_N \xi_b -p_a \nabh_N \xi_b+p_b \th_{ac}\xi_c -(\nab_a\nab_N \xi_b +\th_{ac}\nabh_c \xi_b -\th_{ab} p_c \xi_c) \\
&=&  \in_{bc} \xi_c \dual Y_a -p_a \nabh_N \xi_b+p_b \th_{ac}\xi_c -\th_{ac}\nabh_c \xi_b +\th_{ab} p_c \xi_c,
\eeaa
as required.
\end{proof}
Note that in most situations, we will commute with $\nabz$ rather than $\nabh$; 
When the metric $g=g\0=\Up^{-1} dr^2+\ga\0$, \eqref{eq:commutation-nab-N-nab-a} simplifies to 
$[\nabz_{\pa_r}, \nabz] U = -r^{-1} \nabz U$, or equivalently,
\bea\lab{eq:commutation-nabz-r-nabz}
[\nabz_{\pa_r}, r\nabz] U = 0.
\eea

{\bf Lie derivatives.}
Recall the usual definition for Lie derivatives on $k$-covariant tensor on $\Si$
\beaa
\Lie_X T_{i_1\cdots i_k}=\nab_X T+\nab_{i_1} X^{j} T_{j\cdots i_k}+\cdots +\nab_{i_k} X^j T_{i_1\cdots j}.
\eeaa
Such a definition, as is well-known, is in fact independent of the metric. When $T=U_{a_1\cdots a_k}$ is a horizontal tensor, the Lie derivative $\Lie_X U$
is not necessarily a horizontal tensor. Following  \cite{Chr1}, see also  \cite{GKS}\footnote{\cite{GKS} extends the definition of  \cite{Chr1} to non-integrable  structures.},      we can instead define the projected Lie derivative
\bea\lab{eq:def-projected-Lie-derivatives}
\Ls_X U_{a_1\cdots a_k}:=\nabh_X U_{a_1\cdots a_k}+\nab_{a_1} X^{b} U_{b\cdots a_k}+\cdots +\nab_{a_k} X^b U_{a_1\cdots b}.
\eea
In particular, we have, for $X=f\pa_r$,
\bea\lab{eq:scalar-multiple-projected-Lie-derivatives}
\bsplit
\Ls_{f\pa_r} U_{a_1\cdots a_k} &= \nabh_{f\pa_r} U_{a_1\cdots a_k}+\nab_{a_1} (f\pa_r)^{b} U_{b\cdots a_k}+\cdots +\nab_{a_k} (f\pa_r)^b U_{a_1\cdots b} \\
&= f\nabh_{\pa_r} U_{a_1\cdots a_k}+f\nab_{a_1} (\pa_r)^{b} U_{b\cdots a_k}+\cdots +f\nab_{a_k} (\pa_r)^b U_{a_1\cdots b}\\
&= f\Ls_{\pa_r} U_{a_1\cdots a_k}.
\end{split}
\eea
This is independent of the metric as long as $\pa_r$ is orthogonal to the $r$-spheres. Therefore, we compute using the metric $g\0=\Up^{-1} dr^2+\gz$,  
for which we have $(\nabz_A \pa_r)_B=r^{-1}  \gz_{AB}$. Therefore, using the definition \eqref{eq:def-projected-Lie-derivatives}, we have 
\bea\lab{eq:projected-Lie-equals-covariant}
\Ls_{\pa_r} U_{a_1\cdots a_k} =\nabz_{\pa_r} U_{a_1\cdots a_k}+k r^{-1} U_{a_1\cdots a_k}.
\eea

\subsection{Transport lemma}
\begin{lemma}\lab{lem:transport-lemma}
Suppose a scalar or horizontal covariant tensor $\psi$ satisfies, for some nonnegative integer $i$,
\beaa
\nabz_{\pa_r} \psi+\la r^{-1} \psi=F, \quad \text{and } r^{-1} \|r^\la \psi \|_{\H^i(S_r)}\to 0, \text{ as $r\to\infty$}.
\eeaa
Then we have
\beaa
r^{-1} \|r^\la \psi \|_{\H^i (S_r)} \lesssim \int_r^\infty r'^{-1} \|r'^\la F \|_{\H^i(S_{r'})}\, dr'.
\eeaa
\end{lemma}
\begin{proof}
The equation can be written as $\pa_r(r^\la \psi)=r^\la F$.
Since $d\vol_{r^{-2}\ga\0}$ is independent of $r$, we have
\beaa
\Big|\pa_r \int_{S_r}  |r^{\la-1} \psi|^2 d\vol_{\ga\0}\Big| &=& \Big|\pa_r \int_{S_r}  |r^{\la} \psi|^2 d\vol_{r^{-2}\ga\0}\Big| \lesssim \Big|2\int_{S_r} (\pa_r(r^{\la} \psi)\cdot r^{\la} \psi)\, d\vol_{r^{-2}\ga\0}\Big| \\
&=& \Big| 2\int_{S_r} r^\la F \cdot r^{\la} \psi\, d\vol_{r^{-2}\ga\0} \Big| = \Big|2\int_{S_r} r^{\la-1} F \cdot r^{\la-1} \psi\, d\vol_{\ga\0} \Big| \\
&\lesssim & \| r^{\la-1} F\|_{L^2(S_r)} \| r^{\la-1} \psi\|_{L^2(S_r)},
\eeaa
i.e., $|\pa_r (\|r^{\la-1} \psi\|_{L^2(S_r)}^2)| \lesssim \|r^{\la-1} F\|_{L^2(S_r)} \|r^{\la-1} \psi \|_{L^2(S_r)}$. Therefore, either $\psi=0$, in which case the lemma automatically holds, or we can divide both sides by $\|r^{\la-1} \psi \|_{L^2(S_r)}$ to infer that $|\pa_r (\|r^{\la-1} \psi \|_{L^2(S_r)})| \lesssim \|r^{\la-1} F\|_{L^2(S_r)}$. Therefore, the estimate follows for $i=0$. For positive integers $i$, it follows similarly by commuting the equation with $r\nabz$ using the commutation formula \eqref{eq:commutation-nabz-r-nabz}.
\end{proof}

\subsection{Solvability lemma for the operator \texorpdfstring{$\d_1\d_2\d^*_2\d^*_1$}{d1d2d2stard1star}} 
We study the solvability of the following equation on $(S,\ga)$:
\beaa
\d_1\d_2 h=(F, \dual F), \quad h\in \ss_2.
\eeaa
Recall that, see Section 2.2 in \cite{CK}, as the formal adjoints of the injective elliptic operators $\d_2$ and $\d_1$ on $2$-spheres, $\d_2^*$ and $\d_1^*$ are surjective. Therefore, for each $h$, there exists $(f,\dual f)\in\ss_0$ such that $\d_2^* \d_1^* (f,\dual f)=h$. The equation then becomes
\beaa
\d_1\d_2\d^*_2\d^*_1 (f,\dual f)=(F, \dual F).
\eeaa
Recall that these operators are defined in Definition \ref{def:Hodge-operators}.

We now prove the following lemma, which is a slight generalization of Lemma 2.19 in \cite{KS-GCM1}.
\begin{lemma}\lab{lemma:solvability-div}
Consider the operator $L:=\d_1\d_2\d^*_2\d^*_1$ on $(S,\ga)$, where $S$ is equipped with a standard spherical coordinate and a constant $r>0$, and hence admits the metric $\ga\0$. Suppose that the metric satisfies the estimate $||(r\nabz)^{\leq 4} (\ga-\gz) ||_{L^{\infty}(S)}\leq \mathring\eps\ll 1$. Then the following statements hold:
\begin{itemize}
\item The operator $L=\d_1\d_2\d^*_2\d^*_1$ is a densely-defined self-adjoint operator on $L^2(S,\ga)\times L^2(S,\ga)$. In addition to the zero eigenvalue corresponding to two trivial kernel elements $(1,0)$ and $(0,1)$, there exist six eigenvalues of $L$, denoted by $\la_p$, $\dual \la_p$ with $p=0,+,-$, satisfying $|\la_p|, |\dual\la_p|\lesssim \mathring\eps r^{-4}$, with real-valued eigenfunction pairs 
\beaa
(j,\dual j)_{\la_p}=(J_p,0)+O(\mathring\eps),\quad (j,\dual j)_{\dual\la_p}=(0,J_p)+O(\mathring\eps).
\eeaa 
Any other eigenvalue $\la$ of $L$ satisfies $|\la|\gtrsim r^{-4}$.
\item For the equation 
\bea\lab{eq:4-th-order-equation-h-dual-h}
\d_1\d_2\d^*_2\d^*_1 (f,\dual f)=\sum_{p=0,+,-} (F+c_p J_p+c_0,\dual F+\dual c_p J_p+\dual c_0),
\eea
there exist unique constants $c_0$, $\dual c_0$, $c_p$, $\dual c_p$, for which \eqref{eq:4-th-order-equation-h-dual-h} 
has a  unique solution $(f,\dual f)$   orthogonal to $(1,0)$, $(0,1)$, $(j,\dual j)_{\la_p}$, $(j,\dual j)_{\dual\la_p}$ in $L^2(S,\ga)$.   Moreover, the constants satisfy the estimate
\beaa
|(c_0,\dual c_0)|\lesssim |\overline{(F,\dual F)}^\ga| ,\quad |c_p+ \langle F,J_p\rangle_{r^{-2}\gz} |+|\dual c_p+\langle \dual F,J_p\rangle_{r^{-2}\gz}|\lesssim \mathring \eps r^{-1} ||(F,\dual F)||_{L^2(S,\ga)}.
\eeaa
\end{itemize}
\end{lemma}
\begin{proof}
In view of \eqref{eq:formal-adjoint}, it is clear that $L$ is symmetric on $C^\infty(S)\times C^\infty(S)$ with respect to the inner product of $L^2(S,\ga)\times L^2(S,\ga)$. Since $L$ is also clearly non-negative, there exists a Friedrichs extension, still denoted by $L$, that is densely defined in $L^2(S,\ga)\times L^2(S,\ga)$ and self-adjoint. Note that when $r^{-2}\ga=\gs$, the operator reads $\lapz (\lapz+2r^{-2})$ that acts on scalar pairs, which, in addition to the two constant kernels, has a $6$-dimensional kernel spanned by $(J_p,0)$ and $(0,J_p)$. The first part of the lemma then follows from the fact that $r^{-2}\ga$ is a perturbation of $\gs$. Note also that since constant function pairs lie in the kernel of $L$, we have $\langle (j,\dual j)_{\la_p}, (c_1,c_2)\rangle_\ga =\langle (j,\dual j)_{\dual\la_p}, (c_1,c_2)\rangle_\ga=0$ for any constants $c_1$, $c_2$.

For the equation \eqref{eq:4-th-order-equation-h-dual-h}, we take its inner product with the eigenfunction pair and obtain
\beaa
\langle \d_1\d_2\d^*_2\d^*_1 (f,\dual f),(j,\dual j)_{\la_p}\rangle_\ga =\langle (F,\dual F), (j,\dual j)_{\la_p}\rangle_\ga +  \sum_{q=0,+,-} \langle( c_q   J_q+c_0,  \dual c_q  J_q+\dual c_0),  (j,\dual j)_{\la_p}\rangle_\ga  
\eeaa 
Recall that we require that the solution $(f,\dual f)$ is orthogonal to $(j,\dual j)_{\la_p}$. 
Therefore, since $L$ is self-adjoint and $(j,\dual j)_{\la_p}$ are eigenfunction pairs, the left-hand side is zero, and so are the terms with $c_0$ and $\dual c_0$ as we just remarked. Therefore, we deduce
\beaa
\sum_{q=0,+,-} \left\langle( c_q   J_q,  \dual c_q  J_q),  (j,\dual j)_{\la_p}\right\rangle_{r^{-2}\ga}  = -\left\langle (F,\dual F), (j,\dual j)_{\la_p}\right\rangle_{r^{-2}\ga}. 
\eeaa
Now using the fact that $(j,\dual j)_{\la_p}=(J_p,0)+O(\mathring\eps)$, the left hand side equals $(\de_{pq}+O(\mathring\eps))c_q$. We then also take the inner product of \eqref{eq:4-th-order-equation-h-dual-h} with $(j,\dual j)_{\dual \la_p}$. This gives a linear system of $c_p$, $\dual c_p$ whose coefficient matrix is $O(\mathring\eps)$-close to the identity matrix, and hence we obtain the unique existence of $(c_p,\dual c_p)$. The uniqueness of $(c_0,\dual c_0)$ is then also clear by taking the spherical mean with respect to $\ga$ for \eqref{eq:4-th-order-equation-h-dual-h}. 
The bounds for the constants also follow directly from the relations
\beaa
\left\langle (F,\dual F), (j,\dual j)_{\la_p}\right\rangle_{r^{-2}\ga} &=& \left\langle F,J_p\right\rangle_{r^{-2}\ga}+\left\langle (F,\dual F),O(\mathring\eps)\right\rangle_{r^{-2}\ga},\\
 \left\langle (F,\dual F), (j,\dual j)_{\dual\la_p}\right\rangle_{r^{-2}\ga} &=& \left\langle \dual F,J_p\right\rangle_{r^{-2}\ga}+\left\langle (F,\dual F),O(\mathring\eps)\right\rangle_{r^{-2}\ga},
\eeaa
the bound for $\ga-\gz$, and H\"{o}lder's inequality. 
The existence and uniqueness of $(f,\dual f)$ also follows easily from the fact that $L$ is invertible on the orthogonal complement of $\mathrm{span}\{(1,0),(0,1),(j,\dual j)_{\la_p}, (j,\dual j)_{\dual \la_p}\}$. 
\end{proof}

\begin{corollary}\lab{cor:solvability}
For the equation
\bea\lab{eq:2-th-order-equation-h}
\d_1\d_2 h=\sum_{p=0,+,-} (F+c_p J_p+c_0,\dual F+\dual c_p J_p+\dual c_0),
\eea
there exist unique constants $c_0$, $\dual c_0$, $c_p$, $\dual c_p$ for which \eqref{eq:2-th-order-equation-h}  has a unique solution $h\in\ss_2(S,\ga)$.
Moreover, the constants satisfy the estimate
\bea\lab{eq:estimate-cor-solvability}
|(c_0,\dual c_0)|\lesssim |\overline{(F,\dual F)}^\ga| ,\quad |c_p+ \langle F,J_p\rangle_{r^{-2}\gz} |+|\dual c_p+\langle \dual F,J_p\rangle_{r^{-2}\gz}|\lesssim \mathring \eps r^{-1} ||(F,\dual F)||_{L^2(S,\ga)}.
\eea
\end{corollary}
\begin{proof}
The uniqueness of $c_0$, $\dual c_0$, $c_p$, $\dual c_p$ follows from that $h$ can be expressed in the form $\d_2^*\d_1^* (f,\dual f)$. The uniqueness of $h$ follows from the fact that $\d_1\d_2$ has no kernel. 
\end{proof}

\section{Sketch of the proof of the main theorem}\lab{sect:proof}
\subsection{The linear iteration system}
Recall that we are solving the equations on a base manifold $\Si:=(r_0,\infty) \times \mathbb{S}^2$, and the spherical modes are accordingly defined in Section \ref{sect:spherical-harmonic-Hodge-operators}.
According to the statement of Theorem \ref{thm:main-precise}, at the level of $\ell\geq 2$ modes, we prescribe $(\Bb,\Bbd,\Kk, \Kkd)$, and we iteratively find the data such that
\beaa
\left(\divh Y-\Bb\right)_{\ell\geq 2}=0,\quad \left(\curlh Y-\Bbd\right)_{\ell\geq 2}=0,\quad 
\left(\laph(\ah\Pi)-\Kk \right)_{\ell\geq 2}= 0,\quad 
\left(r^{-4}(\pa_r (r^4\curlh \Xi))-\dual\Kk \right)_{\ell\geq 2}= 0.
\eeaa
More precisely,   we show that the sequence of iterates      $\Psi\n$    of the system   \eqref{eq:HCS-schematic-form}  converge to the desired solution. 
Motivated by    the  equation \eqref{eq:N-transport-go}   in Proposition \ref{prop:linearized-eqns-time-symmetry},    starting with  $\ga^{(0)}$ defined in \eqref{eq:def-gz}, $\ga\n$ is determined iteratively by solving the transport equation\footnote{with the boundary condition at infinity given by $\|\ga\nn-\gz\|_{\H^s}\to 0$. This is ensured in the space where we seek solutions, see \eqref{eq:def-norm-Psi-n}.}
\bea
\lab{eq:metric-iteration}
\slashed{\Lie}_{\pa_r} (r^{-2}\ga\nn)&=2r^{-2}\ah\n\thh\nn+\ah\n\thc\nn(r^{-2}\ga\n)+2\Up^\frac 12 \ao\nn r^{-1} (r^{-2}\ga\n).
\eea
Given $\ga\n$, we can define the horizontal operators $\nabh\n$, $\divh\n$, $\curlh\n$, $\laph\n$, $\d_1\n$, $\d_2\n$, $\cdots$, as well as the spherical mean $\overline{\phi}^{(n)}$ of a scalar field $\phi$ with respect to $\ga\n$. Recalling the definition of $\Psi$    in \eqref{eq:def-Psi-1-12-intro}, the   iterate  $\Psi\n$ reads
\bea\lab{eq:def-Psi-1-12}
\bsplit
\Psi_1\n &=\thc\n,\quad \Psi_2\n=\Kc\n,\quad \Psi_3\n=\ao\n,\quad \Psi_4\n=\thh\n,\quad \Psi_5\n=\pp\n,\quad \Psi_6\n=Y\n, \\
\Psi_7\n &=\trt\n,\quad \Psi_8\n=\kh\n,\quad \Psi_9\n = \Xi\n,\quad \Psi_{10}\n=\Pi\n,\\  \Psi_{11}\n &=(\B_{\ell\leq 1}\n,\Bd_{\ell\leq 1}\n),\quad \Psi_{12}\n=(\K_{\ell\leq 1}\n,\Kd_{\ell\leq 1}\n).
\end{split}
\eea 
We introduce the following norm
\bea\lab{eq:def-norm-Psi-n}
\bsplit
\|(\Psi\n,\ga\n) \|_s&:=\sup_{r\in [r_0,\infty)} \Big(r^{1+\de} \|(\Psi_1\n, \Psi_4\n,\Psi_5\n,\Psi_7\n,\Psi_8\n, \Psi_9\n, \Psi_{10}\n) \|_{\H^{s+1}(S_r)}\\
& \quad \quad \quad +r^{\de} \|\Psi_3\n\|_{\H^{s+2}(S_r)} +r^{2+\de} \|(\Psi_2\n, \Psi_6\n) \|_{\H^s(S_r)} +r^{\de} \|\ga\n-\ga\0\|_{\H^{s+1}(S_r)}\\
& \quad \quad \quad +r^{4+\de} |(\Psi_{11}\n,\Psi_{12}\n)|+r^\de | r^3 (\Psi_1\n)_{\ell=1,i}-\cc_i |\Big),
\end{split}
\eea
where the $\H^s$ norms are defined in Definition \ref{def:L2-Hs-S_r} . 
\begin{remark}
Note that the weights in \eqref{eq:def-norm-Psi-n} are consistent with  the  differentiability order  of   the corresponding quantities, as pointed out in Remark \ref{rem:schematic-notations}.
\end{remark}
We consider the following iteration system, motivated by Proposition \ref{prop:linearized-eqns-time-symmetry}:
\bea
\lab{eq:iteration-kac}
(\pa_r+2r^{-1})\Psi_1\nn &=& \Up^{-\frac 12}\widetilde\mu_{\ell=0}\nn+\Psi_3\n\muc_{\ell=0}\n -2(1-3mr^{-1}) r^{-2}\Psi_3\nn+\Ga_1\n\cdot \Ga_1\n,\\
\lab{eq:iteration-Kc}
(\pa_r+3r^{-1}) \Psi_2\nn &=& r^{-1}\widetilde\mu_{\ell=0}\nn -2\Up^\frac 12 r^{-3}\Psi_3\nn - \Up^{-\frac 12} (\Bb+\widetilde\B_{\ell\leq 1}\nn) \\
\nonumber & & -\Psi_3\n (\Bb+\widetilde\B_{\ell\leq 1,aux}\n)+\Ga_1\n\cdot \Ga_2\n, \\
\lab{eq:iteration-divP}
\Up^\frac 12 \laph\n \Psi_3\nn  &= & \Psi_2\nn-\overline{\Psi_2\nn}^{(n)}
-\Up^\frac 12 r^{-1} (\Psi_1\nn-\overline{\Psi_1\nn}\n)\\
\nonumber & & +\Ga_1\n\cdot\Ga_1\n-\overline{\Ga_1\n\cdot\Ga_1\n}\n
-\laph\n(\Ga_0\n\cdot\Ga_0\n),\\  
\lab{eq:iteration-average-a}
\overline{\Psi_3\nn}\n&=& -\frac 12 \Up^{-1} r \overline{\Psi_1\nn}\n, 
\eea
\vspace{-3.8ex}
\bea
\lab{eq:iteration-h-dualh}
 \d_1\n \d_2\n \Psi_4\nn  &=& \frac 12 \left(\laph\n \Psi_1\nn,0\right)-(\Bb,\Bbd) -\Psi_{11}\nn, \\
\lab{eq:iteration-Delta-ah}
\d_1\n \Psi_5\nn &=& \left(-\Up^\frac 12 \laph\n \Psi_3\nn+\laph\n(\Ga_0\n\cdot \Ga_0\n),0\right),
\eea
\vspace{-4.8ex}
\bea
\lab{eq:iteration-Y}
\d_1\n \Psi_6\nn &=&  (\Bb, \Bbd)+\Psi_{11}\nn  -\overline{(\Bb, \Bbd)+\Psi_{11}\nn}\n,\\
\lab{eq:iteration-trt}
(\pa_r+ r^{-1})\Psi_7\nn &=& 2r^{-1}\Psi_{10}\nn+\Ga_1\n\cdot\Ga_1\n,\\
\lab{eq:iteration-Thh}
\d_1\n \d_2\n \left(\ah\n \Psi_8\nn\right) &=& \frac 12 \left(\ah\n\laph\n \Psi_7\nn,0\right)+(\Kk, -\Kkd)+\Psi_{12}\nn+\Ga_1\n\cdot \Ga_2\n, 
\eea
\vspace{-4.8ex}
\bea
\lab{eq:iteration-Xi}
\d_1\n \Psi_9\nn &=& \bigg(0, \frac{3}{4\pi}  r^{-4} \sum_i{\bf a}_i\om_i+r^{-4} \int_r^\infty r'^4 (\Kkd-\Kd_{\ell\leq 1}\nn)\, dr' \bigg ) \\
\nonumber & & -\bigg(0, \overline{\displaystyle  \frac{3}{4\pi}  r^{-4} \sum_i{\bf a}_i \om_i+r^{-4} \int_r^\infty {r'^4} (\displaystyle{\Kkd}-\Kd_{\ell\leq 1}\nn)\, dr'}\n\bigg), 
\eea
\vspace{-3.5ex}
\bea
\lab{eq:iteration-Pi}
\laph\n \left(\ah\n \Psi_{10}\nn\right) &=& \Kk +\widetilde \K_{\ell\leq 1}\nn-\overline{\Kk+\widetilde \K_{\ell\leq 1}\nn}\n, \\
\lab{eq:iteration-Pi-mean}
\overline{\displaystyle\ah\n\Psi_{10}\nn}\n &=& \overline{\Psi_3\n \Psi_{10}\n}\n,
\eea
along with the metric iterates introduced in  \eqref{eq:metric-iteration}:
\bea\lab{eq:iteration-ga-metric}
\slashed{\Lie}_{\pa_r} (r^{-2}\ga\nn)&=2r^{-2}\ah\n\Psi_4\nn+\ah\n\Psi_1\nn(r^{-2}\ga\n)+2\Up^\frac 12 \Psi_3\nn r^{-1} (r^{-2}\ga\n).
\eea
We explain the notations used here:
\begin{itemize}
\item  The sets  $\Psi^{(n)} $ are iterates of the set $\Psi$  introduced in \eqref{eq:def-Psi-1-12}. For simplicity, in various places, we still denote $\ah\n=1+\Psi_3\n$.
\item For a scalar field $f$, we use $\overline{f}\n$ to denote the spherical mean of $f$ with respect to the metric $\ga\n$.
\item The schematic notations $\Ga_i\n$ for $i=0,1,2$ are defined as in Definition \ref{def:schematic-notation} labeled with $\n$. The dot products in terms like $\Ga_1\n\cdot \Ga_1\n$ are defined with respect to $\ga\n$.
\item The expression $\widetilde\B_{\ell\leq 1} \nn$ stands for
\bea\lab{eq:tilde-BB}
\widetilde\B_{\ell\leq 1} \nn := \frac 12 (\laph\n \Psi_1\nn)_{\ell= 1}+\frac 12(\laph\n \Psi_1\n)_{\ell=0}-\left(\mathcal P_1(\d_1\n\d_2\n \Psi_4\n)\right)_{\ell\leq 1}. 
\eea
We shall also make use of the auxiliary notation
\bea
\lab{eq:tilde-BBNL}
\widetilde\B_{\ell\leq 1,aux} \n := \frac 12 (\laph\n \Psi_1\n)_{\ell= 1}+\frac 12(\laph\n \Psi_1\n)_{\ell=0}-\left(\mathcal P_1(\d_1\n\d_2\n \Psi_4\n)\right)_{\ell\leq 1}.
\eea
Here $\mathcal P_1$ denotes the trivial projection to the first component of a pair of scalars $(\cdot,\cdot)\in \sk_0$. Both $\widetilde\B_{\ell\leq 1} \nn$ and $\widetilde\B_{\ell\leq 1,aux} \n$ behave like\footnote{Recall that $\B_{\ell\leq 1}\nn$ is the first component of $\Psi_{11}\nn$ (defined in \eqref{eq:def-Psi-1-12}).} $\B_{\ell\leq 1}\nn$ in the limit as $n\to\infty$. In particular, in \eqref{eq:tilde-BB}, we distinguish linear and nonlinear terms using $\nn$ and $\n$, see heuristics already in \eqref{eq:heuristic-substitution-ell=1-B}.
\item For a similar reason, we introduce the expression $\widetilde \K_{\ell\leq 1}\nn$:
\begin{equation}\lab{eq:widetilde-K-ell-leq-1}
\widetilde \K_{\ell\leq 1}\nn:= -\frac 12 (\ah\n \laph\n \Psi_7\nn)_{\ell= 1}+\mathcal P_1\left( \d_1\n\d_2\n(\ah\n\Psi_8\n)\right)_{\ell\leq 1}-\frac 12 (\ah\n \laph\n \Psi_7\n)_{\ell=0}+(\Ga_1\n\cdot\Ga_2\n)_{\ell\leq 1},
\end{equation}
where the $\Ga_1\n\cdot\Ga_2\n$ takes the same precise form as the one in \eqref{eq:iteration-Thh}. See heuristics already in \eqref{eq:heuristic-substitution-ell=1-K}.
\item Consistent with the definition of $\muc$ in \eqref{eq:checked-Schw}, $\muc\n$ denotes
\bea\lab{eq:muc-n}
\bsplit
\muc\n:=\mu\n-2mr^{-3}&= -\laph\n (\log\ah\n)+K\n-\frac 14 (\trth\n)^2 -2mr^{-3}\\
&= -\laph\n (\log\Psi_3\n)+\Psi_2\n-\Up^\frac 12 r^{-1}\Psi_1\n-\frac 14 (\Psi_1\n)^2.
\end{split}
\eea
Similar to \eqref{eq:tilde-BB}, \eqref{eq:widetilde-K-ell-leq-1}, we denote
\bea\lab{eq:tilde-mu-ell=0}
\widetilde\mu_{\ell=0}\nn:= (\Psi_2\nn-\Up^\frac 12 r^{-1}\Psi_1\nn)_{\ell=0} -(\laph\n \log\Psi_3\n)_{\ell=0}-\frac 14 ((\Psi_1\n)^2 )_{\ell=0}.
\eea
See heuristics already in \eqref{eq:heuristic-substitution-ell=0-mu}.
\end{itemize}
The $\ell\leq 1$ quantities $\Psi_{11}\nn$, $\Psi_{12}\nn$  will be  determined by equations \eqref{eq:iteration-h-dualh}, \eqref{eq:iteration-Thh}    using  Corollary \ref{cor:solvability}, i.e., by projections on the $\ell\le 1$ modes.

\subsection{Solving the main part $L_{main}$}\lab{subsect:solve-main-part}
To solve the iteration system \eqref{eq:iteration-kac}-\eqref{eq:iteration-ga-metric} at each step, we need to study the linear operator $^{(\tilde\ga)}L$ defined in \eqref{eq:def-L-operator}. As we  have pointed out in Section \ref{subsubsect:linear-operator-L},  $^{(\tilde\ga)}L$ has a  triangular structure,  such that we can focus on the main part   the HCS  system  
\bea
\lab{eq:L-mainn}
\, ^{(\tilde\ga)}L_{main}[\Psi_{main}]= \begin{pmatrix}
  0 \\
  \Bb \\
  0  \end{pmatrix}+err,
\eea
where $\Psi_{main}=(\Psi_1,\Psi_2,\Psi_3)$, and
\bea\lab{eq:L-main}
\, ^{(\tilde\ga)}L_{main}[\Psi_{main}]=\begin{pmatrix}
 (\pa_r+2r^{-1})\Psi_1+2(1-3mr^{-1}) r^{-2} \Psi_3 -\Up^{-\frac 12} (\Psi_2-\Up^\frac 12 r^{-1} \Psi_1)_{\ell=0} \\
  (\pa_r+3r^{-1})\Psi_2 +2\Up^\frac 12 r^{-3} \Psi_3 + \frac 12 \Up^{-\frac 12}  (\laph\Psi_1)_{\ell=1}-r^{-1}(\Psi_2-\Up^\frac 12 r^{-1} \Psi_1)_{\ell=0} \\ 
  (\Up^\frac 12 \laph \Psi_3, \overline{\Psi_3}) -(\Psi_2-\Up^\frac 12 r^{-1} \Psi_1-\overline{\displaystyle \Psi_2-\Up^\frac 12 r^{-1} \Psi_1}, -\frac 12 \Up^{-1} r \overline{\Psi_1})
  \end{pmatrix}.
\eea
In  \eqref{eq:L-mainn}, $err$ denotes lower order terms that only involve nonlinear quantities from the previous step.
In view of the third row of \eqref{eq:L-mainn}, the scalar $\Psi_3=\ao$ can be written as\footnote{\lab{ft:inverse-laph}For given $\tilde\ga$, we extend the definition of $\laph_{\tilde\ga}^{-1}$ by defining $\laph_{\tilde\ga}^{-1} \phi:=\laph_{\tilde \ga}^{-1} (\phi-\overline\phi^{\tilde\ga})$.}
\beaa
\Psi_3=\Up^{-\frac 12} \laph_{\tilde\ga}^{-1} (\Psi_2-\Up^\frac 12 r^{-1} \Psi_1)-\frac 12 \Up^{-1} r \overline{\Psi_1}^{\tilde\ga}+err.
\eeaa
The system \eqref{eq:L-mainn} is then reduced to the following system for $\Psi_1$ and $\Psi_2$:
\bea\lab{eqns:nonlocal-eqn-Kc-kac}
\bsplit
(\pa_r+2r^{-1})\Psi_1 &\; = \; -2(1-3mr^{-1}) r^{-2}\left(\Up^{-\frac 12} \laph_{\tilde\ga}^{-1} (\Psi_2-\Up^\frac 12 r^{-1} \Psi_1)-\frac 12 \Up^{-1} r\overline{\Psi_1}^{\tilde\ga}
\right)\\
& \quad+\Up^{-\frac 12} (\Psi_2-\Up^\frac 12 r^{-1} \Psi_1)_{\ell=0}+F_1,  \\
  (\pa_r+3r^{-1})\Psi_2 &\; = \; -2\Up^\frac 12 r^{-3} \left(\Up^{-\frac 12} \laph_{\tilde\ga}^{-1} (\Psi_2-\Up^\frac 12 r^{-1} \Psi_1)-\frac 12 \Up^{-1} r\overline{\Psi_1}^{\tilde\ga}
 \right)- \frac 12 \Up^{-\frac 12} (\laph_{\tilde \ga}\Psi_1)_{\ell=1}\\
 &\quad +r^{-1} (\Psi_2-\Up^\frac 12 r^{-1} \Psi_1)_{\ell=0}+F_2,
  \end{split}
\eea
with $F_1$ and $F_2$  inhomogeneous terms    depending  the right-hand side of \eqref{eq:L-mainn}, i.e.   free scalars or nonlinear terms. In particular, $F_2$ contains the free scalar $\Bb$.

 We prove  the following proposition regarding  the system \eqref{eqns:nonlocal-eqn-Kc-kac}:
\begin{proposition}\lab{prop:solve-unknowns-kac-Kc}
Consider the system \eqref{eqns:nonlocal-eqn-Kc-kac} with a given metric $\tilde\ga$ satisfying
\bea\lab{eq:bound-condition-nonlocal-metric}
\sup_{r\in [r_0,\infty)} r^{\de} ||\tilde\ga-\gz||_{\H^{s+1}(S_r)}\lesssim \eps_1.
\eea
There exist constants $\eps_0, \eps_1>0$ such that for any $\eps<\eps_0$ and $\cc\in \mathbb{R}^3$ with $|\cc|\leq \eps$, if the following bounds hold true:
\bea\lab{eq:bound-condition-nonlocal-F}
\sup_{r\in [r_0,\infty)} r^{-1} ||r^{3+\de} F_1,r^{4+\de} F_2, r^{4+\de} (F_1)_{\ell=1}, r^{5+\de} (F_2)_{\ell=1}||_{\H^s(S_{r})} \lesssim \eps,
\eea
then for some suitable constant $C>0$, there exists a unique solution to the system \eqref{eqns:nonlocal-eqn-Kc-kac} satisfying
\bea\lab{eq:solve-nonlocal-vanishing-condition}
\sup_{r\in [r_0,\infty)} r^{-1} ||r^{2+\de} \Psi_1, r^{3+\de} \Psi_2||_{\H^s(S_r)}\leq C\eps,\quad \sup_{r\in [r_0,\infty)} r^\de | r^3 (\Psi_1)_{\ell=1,i}-\cc_i , r^4 (\Psi_2)_{\ell=1,i} | \leq C\eps.
\eea
\end{proposition}
\begin{proof}
See  Section \ref{section-nonlocal-eqn-Kc-kac-round}. 
\end{proof}

\subsection{Boundedness estimates}
We are now ready to prove the boundedness result of the iterates.
\begin{proposition}\lab{prop:boundedness}
There exists $\eps>0$ such that for given $m>0$, $\cc,{\bf a}\in\mathbb{R}^3$, and $(\Bb,\Bbd,\Kk,\Kkd)$ as considered in the statement of Theorem \ref{thm:main-precise}, there exists a constant $C_b >0$ and a positive integer $s$, such that
for each nonnegative integer $n$,  if  $(\Psi\n,\ga\n)$ satisfies $||(\Psi\n,\ga\n)||_s \leq C_b \eps$, then there exists a unique solution $(\Psi\nn,\ga\nn)$ to the system \eqref{eq:iteration-kac}-\eqref{eq:iteration-ga-metric} verifying $||(\Psi\nn,\ga\nn)||_s\leq C_b \eps$.
\end{proposition}
\begin{proof}
The proposition relies on the triangular structure, discussed in Remark \ref{rem:block-triangular}, and follows from the following steps below.
\end{proof}

{\bf Step 1.} 
We first apply Proposition \ref{prop:solve-unknowns-kac-Kc} to obtain $\Psi_1\nn$ and $\Psi_2\nn$ by verifying the requirement \eqref{eq:bound-condition-nonlocal-F}. We can then obtain the estimate for $\Psi_1\nn$ with one additional derivative. We then also retrieve the estimate for $\Psi_3\nn$.
\begin{proposition}\lab{prop:step-1}
We have
\begin{equation}\lab{eq:step1-estimate}
\begin{split}
\sup_{r\in [r_0,\infty)} & r^{-1} ||r^{2+\de} \Psi_1\nn ||_{\H^{s+1}(S_r)}+ r^{-1} ||r^{3+\de} \Psi_2\nn||_{\H^s(S_r)}+r^{-1} ||r^{1+\de} \Psi_3\nn||_{\H^{s+2}(S_r)}\\
&+ r^\de |r^3 (\Psi_1\nn)_{\ell=1,i} -\cc_i | \lesssim \eps.
\end{split}
\end{equation}
\end{proposition}
\begin{proof}
See Section \ref{proof:step-1}.
\end{proof}

{\bf Step 2.} Now, since we have obtained $\Psi_1\nn$, we can apply the Codazzi equation \eqref{eq:iteration-h-dualh} to obtain $\Psi_{11}\nn$ and $\Psi_4\nn$. Then we also solve for $\Psi_6\nn$. The estimate for $\Psi_5\nn$ follows from the estimate for $\Psi_3\nn$.
\begin{proposition}\lab{prop:step-2}
We have
\beaa
\sup_{r\in [r_0,\infty)} r^{-1} ||r^{2+\de} (\Psi_4\nn,\Psi_5\nn)||_{\H^{s+1}(S_r)}+ r^{-1} ||r^{3+\de} \Psi_6\nn||_{\H^s(S_r)}+r^{5+\de} |\Psi_{11}\nn | \lesssim \eps.
\eeaa
\end{proposition}
\begin{proof}
See Section \ref{proof:step-2}.
\end{proof}

{\bf Step 3.}
We solve \eqref{eq:iteration-trt} and \eqref{eq:iteration-Pi} for $\Psi_7$ and $\Psi_{10}$. This requires solving a coupled $\ell=1$ part, which we have analyzed at the linear level in Section \ref{subsect:ell-0-1-constraints}, and the remaining part that is decoupled.
\begin{proposition}\lab{prop:step-4}
We have
\beaa
\sup_{r\in [r_0,\infty)} r^{-1} ||r^{2+\de} (\Psi_7\nn, \Psi_{10}\nn) ||_{\H^{s+1} (S_r)}\lesssim \eps.
\eeaa
\end{proposition}
\begin{proof}
See Section \ref{proof:step-4}.
\end{proof}

{\bf Step 4.} We solve the Codazzi equation \eqref{eq:iteration-Thh} for $\Psi_8\nn$, which also determines $\Psi_{12}\nn$, and the div-curl equation \eqref{eq:iteration-Xi} for $\Psi_9\nn$.
\begin{proposition}\lab{prop:step-5}
We have
\beaa
\sup_{r\in [r_0,\infty)} r^{-1} ||r^{2+\de} (\Psi_8\nn, \Psi_9\nn) ||_{\H^{s+1} (S_r)}+r^{5+\de} |\Psi_{12}\nn| \lesssim \eps.
\eeaa
\end{proposition}
\begin{proof}
See Section \ref{proof:step-5}.
\end{proof}

{\bf Step 5.}
We derive the estimate for the spherical metric $\ga\nn$ using \eqref{eq:iteration-ga-metric}.
\begin{proposition}\lab{prop:step-3}
We have
\beaa
\sup_{r\in [r_0,\infty)} r^{-1} ||\ga\nn-\gz||_{\H^{s+1}(S_r)} \lesssim \eps r^{-1-\de}.
\eeaa
\end{proposition}
\begin{proof}
See Section \ref{proof:step-3}.
\end{proof}

\subsection{Contraction estimates}
We use the notation $\de \psi\nn:=\psi\nn-\psi\n$ for a general quantity $\psi$. There should be no difficulty in distinguishing this notation with the constant $\de>0$ appearing in the $r$-weights.
We show that 
\bea\lab{eq:contraction-relation}
||\de(\Psi\nnn,\ga\nnn)||_s \leq C ||\de(\Psi\nn,\ga\nn)||_s,
\eea
 for some positive constant $C<1$. Note that here we define $||\cdot||_s$ as in \eqref{eq:def-norm-Psi-n} but with $\cc_i$ and $\gz$ removed. This is conceptually straightforward and follows in a similar way as the boundednesss result, and hence we leave the details to Section \ref{subsect:contraction-details}.

\subsection{The limit $(g\i,k\i)$}\lab{subsect:thelimit}
The goal of this subsection is to prove
 $\Ga_1\i=\Ga_1(g\i,k\i)$, $\Ga_2\i=\Ga_2(g\i,k\i)$, where $g\i$ and $k\i$ are appropriately identified below.
In other words, all limiting quantities are identified with the corresponding geometric quantities associated with $(g\i,k\i)$.
The fact that $(g\i,k\i)$ solves the constraint equation \eqref{ece} will then be an easy corollary.
\subsubsection{The limiting equations}
In view of \eqref{eq:contraction-relation}, we see that $\{(\Psi\n,\ga\n)\}$ is a Cauchy sequence under the norm $||\cdot||_s$. Therefore, we obtain a limit $(\Psi\i,\ga\i)$
satisfying $||(\Psi\i,\ga\i)||_s\leq C\eps$ by the boundedness statement in Proposition \ref{prop:boundedness}.  
According to our way of introducing the unknowns $\Psi_1$, $\Psi_2$, and $\Psi_3$, we denote $
\ah\i:=\Up^{-\frac 12}+\Psi_3\i$, $K\i:=\Psi_2\i+r^{-2}$, $\trth\i:=\Psi_1\i+2\Up^\frac 12 r^{-1}$, and
\beaa
\mu\i:= -\laph\i(\log\ah\i)+K\i-\frac 14(\trth\i)^2.
\eeaa
We expect that these quantities turn out to be precisely those naturally connected to the limiting initial data set.

The limit $(\Psi\i,\ga\i)$ solves the following system, by taking $n\to \infty$ for the equations \eqref{eq:iteration-kac}-\eqref{eq:iteration-Pi}: 
\bea
\lab{eq:iteration-kac-ii}
(\pa_r+2r^{-1})\Psi_1\i &=& \Up^{-\frac 12}\muc_{\ell=0}\i+\ao\i\muc_{\ell=0}\i -2(1-3mr^{-1}) r^{-2}\Psi_3\i+\Ga_1\i\cdot \Ga_1\i,\\ 
\lab{eq:iteration-Kc-ii}
(\pa_r+3r^{-1}) \Psi_2\i &=& r^{-1}\muc_{\ell=0}\i -2\Up^\frac 12 r^{-3}\Psi_3\i - \Up^{-\frac 12} (\Bb+\widetilde
\B_{\ell\leq 1}\i) \\
\nonumber & & -\Psi_3\i (\Bb+\widetilde
\B_{\ell\leq 1}\i)
 +\Ga_1\i\cdot \Ga_2\i, \\
\lab{eq:iteration-divP-ii}
\Up^\frac 12 \laph\i \Psi_3\i  &=& \Psi_2\i-\overline{\Psi_2\i}\i 
-\Up^\frac 12 r^{-1} (\Psi_1\i-\overline{\Psi_1\i}\i)\\
\nonumber & &+\Ga_1\i\cdot\Ga_1\i-\overline{\Ga_1\i\cdot\Ga_1\i}\i
-\laph\i(\Ga_0\i\cdot\Ga_0\i), 
\eea
\vspace{-3.5ex}
\bea
\lab{eq:iteration-average-a-ii}
\overline{\Psi_3\i}\i&=& -\frac 12 \Up^{-1} r \overline{\Psi_1\i}\i, \\
\lab{eq:iteration-h-dualh-ii}
 \d_1\i \d_2\i \Psi_4\i  &=& \frac 12 (\laph\i \Psi_1\i,0)-(\Bb,\Bbd) -(\B_{\ell\leq 1}\i,\Bd_{\ell\leq 1}\i) \\
\lab{eq:iteration-Delta-ah-ii}
\d_1\i \Psi_5\i &=& \left(-\Up^\frac 12 \laph\i \ao\i+\laph\i(\Ga_0\i\cdot \Ga_0\i),0\right)\\
\lab{eq:iteration-Y-ii}
\d_1\i \Psi_6\i &=& (\Bb+\B_{\ell\leq 1}\i, \Bbd+\Bd_{\ell\leq 1}\i)  -(\overline{\Bb+\B_{\ell\leq 1}\i}\i, \overline{\displaystyle \Bbd+\Bd_{\ell\leq 1}\i}\i),\\
\lab{eq:iteration-trt-ii}
(\pa_r+r^{-1})\Psi_7\i &=& 2r^{-1}\Psi_{10}\i+\Ga_1\i\cdot\Ga_1\i,\\
\lab{eq:iteration-Thh-ii}
\d_1\i \d_2\i (\ah\i\Psi_8\i) &=& \Big(\frac 12 \ah\i\laph\i \Psi_7\i+\Kk, - \Kkd\Big)+(\K_{\ell\leq 1}\i,\Kd_{\ell\leq 1}\i)+\Ga_1\i\cdot \Ga_2\i, 
\eea
\vspace{-4.5ex}
\bea
\lab{eq:iteration-Xi-ii}
\d_1\i \Psi_9\i &=& \bigg(0, \frac{3}{4\pi}  r^{-4} \sum_i{\bf a}_i\om_i+r^{-4} \int_r^\infty r'^4 (\Kkd-\Kd_{\ell\leq 1}\i)\, dr' \bigg ) \\
\nonumber & & -\bigg(0, \overline{\displaystyle  \frac{3}{4\pi}  r^{-4} \sum_i{\bf a}_i \om_i+r^{-4} \int_r^\infty {r'^4} (\displaystyle{\Kkd}-\Kd_{\ell\leq 1}\i)\, dr'}\i\bigg), 
\eea
\vspace{-3.5ex}
\bea
\lab{eq:iteration-Pi-ii}
\laph\i (\ah\i\Psi_{10}\i) &=& \Kk+\K_{\ell\leq 1}\i-\overline{\Kk+\K_{\ell\leq 1}\i}\i.\\
\overline{\displaystyle\ah\i\Psi_{10}\i}\i&=&  \overline{\Psi_3\i \Psi_{10}\i}\i.
\eea
We note the we have used the observation that by taking the limit of \eqref{eq:tilde-BB}, \eqref{eq:tilde-BBNL}, the quantities $\widetilde\B_{\ell\leq 1}\i$, $\widetilde\B_{\ell\leq 1,aux}\i$ are the same:
\bea\lab{eq:limit-tilde-B}
\widetilde\B_{\ell\leq 1} \i=\widetilde\B_{\ell\leq 1,aux}\i=\frac 12 (\laph\i \Psi_1\i)_{\ell= 1}+\frac 12(\laph\i \Psi_1\i)_{\ell=0}-\left(\mathcal P_1(\d_1\i\d_2\i \Psi_4\i)\right)_{\ell\leq 1}.
\eea
Therefore, we write both of them as $\widetilde\B_{\ell\leq 1}\i$.
Similarly, by comparing the limit of \eqref{eq:tilde-mu-ell=0} and \eqref{eq:muc-n}, we see that $\tilde\mu_{\ell=0}\i=\muc_{\ell=0}\i$, and we write both of them as $\muc\i_{\ell=0}$.

Moreover, taking the limit of \eqref{eq:iteration-ga-metric} gives
\bea\lab{eq:metric-limit}
\slashed{\Lie}_{\pa_r} (r^{-2}\ga\i)&=2r^{-2}\ah\i\Psi_4\i+\ah\i\Psi_1\i(r^{-2}\ga\i)+2\Up^\frac 12 \Psi_3\i r^{-1} (r^{-2}\ga\i).
\eea
We now define the metric 
\bea\lab{eq:def-g-limit}
g\i:=(\ah\i)^2 dr^2+\ga\i \text{ on $\Si=(r_0,\infty) \times \mathbb{S}^2$}.
\eea
This provides a choice of the triad $\{N\i,e_a\i\}_{a=1,2}$. We then define the   ``second fundamental form''  $k\i$ through its components:
\begin{equation}\lab{eq:def-k-infty}
k\i(e_a\i,e_b\i):=(\Psi_8\i)_{ab}+\frac 12 \Psi_7\i \de_{ab},\quad k\i(N\i,e_a\i):=(\Psi_9\i)_a,\quad k\i(N\i,N\i):=\Psi_{10}\i.
\end{equation}

\subsubsection{The limit $(g\i,k\i)$ verifies the constraint equation}
It remains to show that $(g\i,k\i)$ solves the constraint equation \eqref{ece}. To prove this, we first need several observations listed in the following lemma:
\begin{lemma}\lab{eq:lemma-limit-properties}
The following statements hold true:
\begin{enumerate}
\item The horizontal tensors $\Psi_4\i$ and $\Psi_8\i$ are traceless with respect to $\ga\i$.
\item\lab{statement:2} With respect to the metric $g\i:=(\ah\i)^2 dr^2+\ga\i$, the quantities $\Psi_4\i$ and $\Psi_1\i+2\Up^\frac 12 r^{-1}$ are exactly the traceless part and the trace of the second fundamental form of the $r$-spheres. We hence denote $\thh\i=\Psi_4\i$ and $\trth\i=\Psi_1\i+2\Up^\frac 12 r^{-1}$ without ambiguity.
\item We have $\Psi_5\i=-\nabh\i (\log\ah\i)$, $\mu_{\ell\geq 1}\i=0$, and the average of $\Psi_3\i+\frac 12 \Up^{-1} r \Psi_1\i$ vanishes with respect to $\ga\i$.
\item For the quantity defined in \eqref{eq:limit-tilde-B}, we have $\widetilde\B_{\ell\leq 1} \i =\B_{\ell\leq 1} \i$, the latter being the first component of $\Psi_{11}\i\in \sk_0$.
\item Denote by $Y(g\i)$ the horizontal tensor $Y$ with respect to $g\i$ defined through \eqref{eq:def-curvature-components-Si}. Then we have $\Psi_6\i=Y(g\i)$.
\item Denote the Gauss curvature of $\ga\i$ by $K(\ga\i)$. Then we have $\Psi_2\i=K(\ga\i)-r^{-2}$.
\end{enumerate}
\end{lemma}
\begin{proof}
See Section \ref{sect:proof-lemma-limit-properties}. 
\end{proof}

\begin{proposition}\lab{prop:limit-solves-constraint}
The data $(g\i,k\i)$ solves the Einstein constraint equation \eqref{ece}.
\end{proposition}
\begin{proof}
This follows from comparing the equations \eqref{eq:iteration-kac-ii}, \eqref{eq:iteration-trt-ii}, \eqref{eq:iteration-Thh-ii} with the unconditional equation \eqref{eq:R-transport-kac}, \eqref{eq:N-trt-linearized}, \eqref{eq:N-div-Xi}, \eqref{eq:N-curl-Xi}, along with the statements in Lemma \ref{eq:lemma-limit-properties}. We leave the details to Section \ref{subsect:proof-limit-solves-constraint}.

\subsection{Conclusions}\lab{sec:conclusions}
We have proved that $(g\i,k\i)$ solves the Einstein constraint equation \eqref{ece}, and under the ambient $r$-foliation, the corresponding geometric quantities satisfy the following estimate:
\beaa
&&r^{-1} ||\ah\i-\Up^{-\frac 12}||_{\H^{s+2}(S_r)}\lesssim \eps r^{-1-\de},\quad r^{-1} ||\ga\i-\gz||_{\H^{s+1}(S_r)}\lesssim \eps r^{-1-\de}, \\
&& r^{-1} || \thc\i,\thh\i,p\i,\trt\i,\kh\i,\Xi\i,\Pi\i||_{\H^{s+1}(S_r)}\lesssim \eps r^{-2-\de}, \\
&& r^{-1} || \Kc\i, Y\i||_{\H^{s}(S_r)} \lesssim \eps r^{-3-\de}.
\eeaa
Moreover, in the proof of Lemma \ref{eq:lemma-limit-properties} and Proposition \ref{prop:limit-solves-constraint}, we obtain the relations
\beaa
&&(\divh\i Y\i)_{\ell\geq 2}=\Bb,\quad (\curlh\i Y\i)_{\ell\geq 2}=\Bbd, \\
&& (\laph\i (\ah\i \Pi\i))_{\ell\geq 2}=\Kk,\quad r^{-4} \pa_r (r^4 (\curlh\i\Xi\i))_{\ell\geq 2}=\Kkd,\\
&& \mu_{\ell\geq 1}\i=0,\quad \divh\i\Xi\i=0.
\eeaa
We also have the following limits
\beaa
\lim_{r\to\infty} r^3(\thc\i)_{\ell=1,i}=\cc_i,\quad \lim_{r\to\infty} r^2(\trt\i)_{\ell=1,i}=0,\quad \lim_{r\to\infty} r^4(\curlh\i\Xi\i)_{\ell=1,i}={\bf a}_i,
\eeaa
and hence, in view of Proposition \ref{prop:relation-ADM-charges}, proves \eqref{eq:main-thm-ADM-charges} regarding the ADM charges.

\section{Details in the proof of the main theorem}

\subsection{Proof of Proposition \ref{prop:solve-unknowns-kac-Kc}}\lab{section-nonlocal-eqn-Kc-kac-round}

We first outline the main ideas  in the  proof of the Proposition.
\begin{enumerate}
\item Since $\tilde\ga$ is in general not round, $\laph_{\tilde\ga}^{-1}$ mixes the different modes. In Section \ref{section-nonlocal-eqn-Kc-kac-perturbed}, we show that, for any scalar field $\phi$,
\beaa
(\laph^{-1}_{\tilde\ga} \phi)_{\ell,\mm}=-\frac {r^2}{\ell(\ell+1)} (r^{-2} \laph\0 \laph^{-1}_{\tilde\ga} \phi)_{\ell,\mm}=-\frac {r^2}{\ell(\ell+1)} (\phi_{\ell,\mm}+\RR(\phi)_{\ell,\mm}),
\eeaa
where $\RR\colon \H^s\to \H^s$ is a linear operator satisfying $||\RR(\phi)||_{\H^s(S_r)}\lesssim \eps_1 r^{-1-\de} ||\phi||_{\H^s(S_r)}$.
\item We study the projection into spherical harmonic modes $J_{\ell,\mm}$ based on the background coordinates. Each mode satisfies the system of the form, for $\u\in\mathbb{R}^2$,
\beaa
\pa_r \u=r^{-1} A(r) \u+r^{-2} B(r) \u+\F,
\eeaa
with a vanishing condition at infinity. It is crucial for the first matrix $A(r)$ on the right to  be  accretive\footnote{i.e., $\langle A\v, \v\rangle_H+\langle \v,A\v\rangle_H\geq 0$ for all $\v\in\mathbb{R}^2$ for some inner product $H$.} for some inner product over $\mathbb{R}^2$. Under this assumption, we first provide a version of Duhamel formula in Lemma \ref{lem:Duhamel-round-case} in Section \ref{sec:Duhamel-formula}. The equations written in modes are derived in Section \ref{sec:equations-in-modes}.
\item In Section \ref{sec:solution-operators-in-modes}, we study the equations projected into different modes. 
\begin{itemize}
\item For $\ell\geq 2$, the equation reads 
 \beaa
\pa_r\begin{pmatrix}
  r^2(\Psi_1)_{\ell,\mm} \\
  r^3 (\Psi_2)_{\ell,\mm}
\end{pmatrix}&=& r^{-1} \begin{pmatrix}
  -\frac{2 }{\ell(\ell+1)} & \frac {2}{\ell(\ell+1)}  \\
  -\frac{2}{\ell(\ell+1)} & \frac{2}{\ell(\ell+1)}
\end{pmatrix} \begin{pmatrix}
  r^2(\Psi_1)_{\ell,\mm} \\
  r^3(\Psi_2)_{\ell,\mm}
\end{pmatrix}+\begin{pmatrix}
  r^2(F_1)_{\ell,\mm} \\
  r^3(F_2)_{\ell,\mm}
\end{pmatrix}+l.o.t.,
\eeaa
where the first matrix on the right is a nilpotent matrix, in particular, not accretive. To deal with this, we consider instead the unknown $\begin{pmatrix}
  r^{2+\de'}(\Psi_1)_{\ell,\mm} \\
  r^{3+\de'} (\Psi_2)_{\ell,\mm}
\end{pmatrix}$ with $0<\de'<\de$, so that the first matrix of the new system becomes, as is shown in Lemma \ref{eq:lemma-Lyapunov-matrix}, positive definite under a certain inner product over $\mathbb{R}^2$ for all $\ell\geq 2$. This verifies the condition of Lemma \ref{lem:Duhamel-round-case} and allows us to construct the solution to such a system. 
\item The corresponding analysis of the matrix for $\ell\leq 1$ parts is easier by incorporating appropriate $r$-weights. Note that, however, as has already appeared in Section \ref{subsect:ell-0-1-constraints}, the $\ell=1$ part contains a non-zero center-of-mass tail $\cc_{\mm} J_{1,\mm} r^{-3}$ that has to be subtracted from $(\Psi_1)_{1,\mm}$.
\end{itemize}
\item The system \eqref{eqns:nonlocal-eqn-Kc-kac} can now be rewritten in the form
\bea\lab{eq:eqn-v-outline}
\pa_r \v_{\ell,\mm} =r^{-1} A_\ell \v_{\ell,\mm}+r^{-2}B_\ell(r)\v_{\ell,\mm}+\F_{\ell,\mm} +\RR^{new}_{\ell,\mm}(\v),
\eea
where $\RR^{new}$ is an appropriate weighted\footnote{with the weight depending on $\ell$.} version of $\RR$. Due to different weights for different $\ell$, one can only expect $\RR^{new}$ to satisfy a relaxed uniform estimate, and it is important that this still provides enough $r$-decaying weights.
In Section \ref{sec:summing-up}, we use \eqref{eq:eqn-v-outline} to prove the existence and uniqueness of the solution by the contraction argument. 
\end{enumerate}

\subsubsection{The perturbed metric and the $\RR$ operator}\lab{section-nonlocal-eqn-Kc-kac-perturbed}
Recall that the assumption on the given perturbed metric $\tilde\ga$ in \eqref{eq:bound-condition-nonlocal-metric} reads
\bea\lab{eq:bound-ga-tilde}
\sup_{r\in [r_0,\infty)} r^{\de} ||\tilde\ga-\gz||_{\H^{s+1}(S_r)}\lesssim \eps_1,
\eea
where $\eps_1$ is a small constant to be determined.

We need to deal with the fact that the operator $\laph_{\tilde\ga}^{-1}$ mixes different modes. For $\ell\geq 1$, we write 
\bea\lab{eq:def-RR-l-m}
(r^{-2} \laph_{\tilde\ga}^{-1} \phi)_{\ell,\mm}=-\frac {1}{\ell(\ell+1)} (\lapz \laph_{\tilde\ga}^{-1} \phi)_{\ell,\mm}=-\frac {1}{\ell(\ell+1)} (\phi_{\ell,\mm}+\RR_{\ell,\mm}(\phi)),
\eea
and for $\ell=0$, we write, schematically,
\bea\lab{eq:def-RR-0}
\left((r^{-2} \laph^{-1} \phi)_{\ell=0},\; \overline{\phi}^{\tilde\ga}-\overline{\phi}^{\gz}\right) = \RR_{\ell=0} (\phi).
\eea
\begin{definition}
\lab{Def:RR-operator}
The linear operator $\RR$ is defined by $\RR(\phi):=\sum_{\ell=0}^\infty \sum_{\mm=-\ell}^\ell \RR_{\ell,\mm}(\phi) J_{\ell,\mm}$.
\end{definition}
\begin{proposition}\lab{prop:bound-RR}
The linear operator $\RR$ satisfies the bound
\beaa
||\RR (\phi) ||_{\H^s(S_r)} \lesssim \eps_1 r^{-1-\de} ||\phi||_{\H^s(S_r)}.
\eeaa
\end{proposition}
\begin{proof}
Since $J_{\ell,\mm}$ and $r^{-2}\gz$ are independent of $r$, by the definition \eqref{eq:def-l-m-modes} of the modes, it remains true that $(\pa_r\phi)_{\ell,\mm}=\pa_r(\phi_{\ell,\mm})$. Since for $\ell\geq 1$, $\lapz J_{\ell,\mm}=-\frac{\ell(\ell+1)}{r^2} J_{\ell,\mm}$,
we have
\bea\lab{eq:projection-inverse-laplacian}
\bsplit
(\laph^{-1}_{\tilde\ga} \phi)_{\ell,\mm} &=  \int_{S_r} (\laph^{-1}_{\tilde\ga} \phi) J_{\ell,\mm}\, d\vol_{\gs}=\int_{S_r} (\laph^{-1}_{\tilde\ga} \phi) \left(-\frac {r^2}{\ell(\ell+1)}\right)\lapz J_{\ell,\mm}\, d\vol_{\gs} \\
&=  -\frac {r^2}{\ell(\ell+1)} \int_{S_r} (\lapz \laph^{-1}_{\tilde\ga} \phi) J_{\ell,\mm}\, d\vol_{\gs}=-\frac {r^2}{\ell(\ell+1)} (\lapz \laph^{-1}_{\tilde\ga} \phi)_{\ell,\mm}.
\end{split}
\eea
Then, combining \eqref{eq:projection-inverse-laplacian} with the definition \eqref{eq:def-RR-l-m}, we have 
\beaa
\RR(\phi)_{\ell,\mm}=-\phi_{\ell,\mm}-\frac{\ell(\ell+1)}{r^2} (\laph_{\tilde \ga}^{-1} \phi)_{\ell,\mm}= (\lapz \laph^{-1}_{\tilde\ga} \phi)_{\ell,\mm} -\phi_{\ell,\mm}.
\eeaa
We then apply Lemma \ref{lemma:estimate-inverse-Laplacian}, as well as Lemma \ref{lem:l=0-RR} for the $\ell=0$ part defined in \eqref{eq:def-RR-0}, to obtain, using \eqref{eq:Hs-norm-modes-new},
\beaa
||\RR (\phi) ||_{\H^s}^2 = r^2 \left(|\RR_{\ell=0} (\phi)| +\sum_{\ell=1}^\infty \sum_{\mm=-\ell}^\ell (1+\ell^2)^s |(\RR(\phi))_{\ell,\mm} |^2\right) \lesssim (\eps_1 r^{-1-\de})^2 ( ||\phi||_{L^2(S_r)}^2 +  ||\phi||_{\H^s(S_r)}^2).
\eeaa
Therefore, it remains to prove Lemma \ref{lemma:estimate-inverse-Laplacian} and Lemma \ref{lem:l=0-RR}.
\end{proof}

\begin{lemma}\lab{lemma:estimate-inverse-Laplacian}
Assume \eqref{eq:bound-ga-tilde} holds. For any integer $s\geq 0$ and scalar field $\phi$, we have
\bea\lab{eq:estimate-inverse-Laplacian}
||\lapz \laph_{\tilde\ga}^{-1}\phi -\phi_{\ell\geq 1}||_{\H^s(S_r)} 
\lesssim \eps_1 r^{-1-\de}  ||\phi||_{\H^s(S_r)}.
\eea
Note from footnote \ref{ft:inverse-laph} that the domain of $\laph_{\tilde\ga}^{-1}$ is extended through
$\laph_{\tilde\ga}^{-1} \phi:= \laph^{-1}_{\tilde\ga} (\phi-\overline{\phi}^{\tilde\ga})$.
\end{lemma}
\begin{proof}
We write $r^2\laph_{\tilde\ga}=r^2\lapz+H$. In view of Remark \ref{rem:inverse-metric}, $H$ is of the form
\beaa
H \phi=O(\tilde\ga-\gz) \cdot (r\nabz)^2 \phi+O(r\nabz \tilde\ga ) \cdot r\nabz \phi.
\eeaa
Therefore, applying the $(r\nabz)$ derivatives  $s$ times, using Definition \ref{def:L2-Hs-S_r}, we obtain, by standard $L^2$-$L^\infty$ estimates, for  $s \ge  3$,
\beaa
||H \phi||_{\H^{s}(S_r) } &\lesssim & ||(r\nabz)^{\leq s}\left(O(\tilde\ga-\gz) \cdot (r\nabz)^2 \phi+O(r\nabz \tilde\ga) \cdot r\nabz \phi\right)||_{L^2(S_r)} \\
&\lesssim & r^{-1} ||\tilde\ga- \gz||_{\H^{s+1}(S_r)} ||\phi||_{\H^{s+2}(S_r)}\\
&\lesssim & \eps_1 r^{-1-\de} ||\phi||_{\H^{s+2}(S_r)}.
\eeaa
We have the identity
\beaa
\lapz \laph_{\tilde\ga}^{-1} \phi- \phi_{\ell\geq 1} &=& (r^2\lapz) (r^2\laph_{\tilde\ga})^{-1} \phi- \phi_{\ell\geq 1}=(r^2\laph_{\tilde \ga}-H)(r^2\laph_{\tilde\ga})^{-1} \phi-\phi+\overline{\phi}^{\gz}\\
&=&(\phi-\overline{\phi}^{\tilde \ga})-H(r^2\laph_{\tilde\ga})^{-1} \phi-\phi+\overline{\phi}^{\gz}\\
&=& (\overline{\phi}^{\tilde \ga}-\overline{\phi}^{\gz}) -  H(r^2\laph_{\tilde\ga})^{-1} \phi.
\eeaa
Therefore,  
\bea\lab{eq:estimate-lapz-laph-inverse-long}
\bsplit
||\lapz \laph^{-1} \phi- \phi_{\ell\geq 1} ||_{\H^s(S_r)} &\lesssim  ||\overline{\phi}^{\tilde \ga}-\overline{\phi}^{\gz}||_{L^2(S_r)} +  ||H(r^2\laph_{\tilde\ga})^{-1} \phi||_{\H^s(S_r)}\\
&\lesssim  
|| \tilde \ga- \gz ||_{L^\infty(S_r)} ||\phi||_{L^2(S_r)} + \eps_1 r^{-1-\de} ||r^{-2} \laph_{\tilde\ga}^{-1} \phi||_{\H^{s+2}(S_r)} \\
&\lesssim \eps_1 r^{-1-\de} \left(||r^{-2} \laph_{\tilde\ga}^{-1} \phi||_{\H^{s+2}(S_r)}+||\phi||_{L^2(S_r)}\right).
\end{split}
\eea
It then remains to estimate $|| r^{-2} \laph_{\tilde\ga}^{-1} \phi ||_{\H^{s+2}(S_r)}$. Notice that the estimate \eqref{eq:estimate-lapz-laph-inverse-long} in fact implies
\beaa
||\lapz \laph_{\tilde\ga}^{-1} \phi||_{\H^s(S_r)} \lesssim  ||\phi_{\ell\geq 1} ||_{\H^s(S_r)}+\eps_1 r^{-1-\de} \left(||r^{-2} \laph_{\tilde\ga}^{-1} \phi||_{\H^{s+2}(S_r)}+||\phi||_{L^2(S_r)}\right).
\eeaa
Sttandard elliptic estimates for $r^2\lapz$  imply $||r^{-2}\laph_{\tilde\ga}^{-1} \phi||_{\H^{s+2}(S_r)}\lesssim ||\phi||_{\H^s(S_r)}+\eps_1 r^{-1-\de} ||r^{-2}\laph_{\tilde\ga}^{-1} \phi||_{\H^{s+2}(S_r)}$, hence $||r^{-2}\laph_{\tilde\ga}^{-1} \phi||_{\H^{s+2}(S_r)}\lesssim ||\phi||_{\H^s(S_r)}$. Plugging this back to \eqref{eq:estimate-lapz-laph-inverse-long}, we obtain the desired estimate.
\end{proof}

\begin{lemma}\lab{lem:l=0-RR}
Suppose that \eqref{eq:bound-ga-tilde} holds. We have
\beaa
|(r^{-2} \laph_{\tilde\ga}^{-1} \phi)_{\ell=0}| , |\overline{\phi}^{\tilde\ga}-\overline{\phi}^{\gz}| \lesssim \eps_1 r^{-1-\de} (r^{-1} ||\phi||_{L^2(S_r)}).
\eeaa
\end{lemma}
\begin{proof}
We have 
\beaa
r^{-2} \int_{S_r} (r^{-2} \laph_{\tilde\ga}^{-1} \phi) \, d\vol_{\gz}& =& r^{-2}\int_{S_r} (r^{-2} \laph_{\tilde\ga}^{-1}) \phi\, (d\vol_{\gz}-d\vol_{\tilde\ga}) \\
&\lesssim & ||r^{-2} \laph_{\tilde\ga}^{-1}\phi||_{L^\infty(S_r)} \Big|r^{-2} \int_{S_r} (\sqrt{\det (\tilde\ga_{ab})}-1) \in_{ab}\0 \Big| \\
& \lesssim & ||\tilde\ga- \gz||_{L^\infty(S_r)} ||r^{-2} \laph_{\tilde\ga}^{-1}\phi||_{L^\infty(S_r)} \lesssim \eps_1 r^{-1-\de} (r^{-1} ||\phi||_{L^2(S_r)}),
\eeaa
where $\in_{ab}=(d\vol_{\gz})_{ab}$ denotes the volume form of $\gz$, which satisfies $\int_{S_r} \in\0 =4\pi r^2$.
The estimate for $|\overline{\phi}^{\tilde\ga}-\overline{\phi}^{\gz}|$ is similar, and in fact we have already used it in the proof of Lemma \ref{lemma:estimate-inverse-Laplacian}.
\end{proof}

\subsubsection{Duhamel's formula, accretiveness of matrices}\lab{sec:Duhamel-formula}
The following Lemma  establishes    a  Duhamel type  representation formula\footnote{This will be applied  to the  specific modes of  the system \eqref{eqns:nonlocal-eqn-Kc-kac}.} for systems of the type 
\eqref{eq:lemma-evolution}.

\begin{lemma}\lab{lem:Duhamel-round-case}
Take an inner product $\langle\cdot,\cdot\rangle_H$ over $\mathbb{R}^2$ independent of $r$. Consider the equation
\bea\lab{eq:lemma-evolution-homogeneous}
\pa_r \u=r^{-1} A(r) \u+r^{-2} B(r) \u
\eea
for $\mathbb{R}^2$-valued vector $\u=\u(r)$. If $N(r)$ is accretive with respect to $H$, i.e., $\langle A\v, \v\rangle_H+\langle \v,A\v\rangle_H\geq 0$ for all $r$ and $\v\in\mathbb{R}^2$, and $B(r)=O(1)$, then the solution operator $U(r,r^*)$ for $r<r^*$, defined through
\bea\lab{eq:def-sol-op-ell}
\pa_r U(r,r^*)=(r^{-1} A(r)+r^{-2}B(r)) U(r,r^*),\quad U(r^*,r^*)=I,
\eea
satisfies $||U(r,r^*)||\leq C$ uniformly for all $r,r^*$ with $r_0\leq r<r^*$.

Moreover, for the inhomogeneous equation
\bea\lab{eq:lemma-evolution}
\pa_r \u=r^{-1} A(r) \u+r^{-2} B(r) \u+\N
\eea
with $\N\in L^1((r_0,\infty),\mathbb{R}^2)$ and the condition  
\bea\lab{eq:lemma-evolution-infinity-condition}
\lim_{r\to \infty} ||\u(r)||_H=0,
\eea
there exists a unique solution $\u\in C^1((r_0,\infty),\mathbb{R}^2 )$ to \eqref{eq:lemma-evolution} satisfying \eqref{eq:lemma-evolution-infinity-condition}. In fact, $\u$ can be expressed as
\bea\lab{eq:Duhamel-linear-1}
\u=-\int_r^\infty U(r,r') \N(r') \, dr'.
\eea
\end{lemma}
\begin{proof}
We first derive the following boundedness estimate
\beaa
\frac{d}{dr} ||U(r,r^*)\u||_{H}^2 &=& \Big\langle U(r,r^*)\u, r^{-1} A U(r,r^*)\u \Big\rangle_H+\Big\langle U(r,r^*)\u, r^{-1} A U(r,r^*)\u \Big\rangle_H+O(r^{-2}) ||U(r,r^*)\u||_H^2  \\
& \geq & - O(r^{-2}) ||U(r,r^*)\u||_H^2, 
\eeaa
where the accretiveness of $A$ is crucially used.
Hence, we have $\frac{d}{dr}\left(\exp\left(-\int_r^{r^*} O(r'^{-2}) \, dr'\right) ||U(r,r^*) \u||_H^2\right)\geq 0$, i.e.,
\bea\lab{eq:Duhamel-boundedness-operator}
||U(r,r^*) \u||_H^2\lesssim ||\u||_H^2 \exp\left(\int_r^{r^*} O(r'^{-2}) \, dr'\right) \lesssim ||\u||_H^2.
\eea
The formula \eqref{eq:Duhamel-linear-1} itself proves the existence of the solution to the inhomogeneous equation \eqref{eq:lemma-evolution}, for which the condition \eqref{eq:lemma-evolution-infinity-condition} is verified using the boundedness \eqref{eq:Duhamel-boundedness-operator} and the integrablity of $\N$. To see the uniqueness, suppose there are two solutions $\u_1$, $\u_2$. Then $\u_1-\u_2$ solves the homogeneous equation, and hence for each $r$, $r'$ with $r<r'$, $(\u_1-\u_2)(r)=U(r,r')(\u_1(r')-\u_2(r'))$. If $|| \u_1(r)- \u_2(r)||=c\neq 0$ for some $r$, then for each $r'>r$, $|| \u_1(r')- \u_2(r')||\gtrsim || \u_1(r)-\u_2(r)||= c>0$, contradicting the covergence $\lim_{s\to \infty} ||\u_j(s)||=0$, $j=1,2$. This proves the uniqueness.
\end{proof}

We will use below the following property of a nilpotent matrix $Q=\begin{pmatrix}
  -1 & 1 \\
  -1 & 1
\end{pmatrix}$. 
\begin{lemma}\lab{eq:lemma-Lyapunov-matrix}
Let $Q=\begin{pmatrix}
  -1 & 1 \\
  -1 & 1
\end{pmatrix}$.
Then, for any given $\de'>0$ and the matrix $A=\de' I+xQ$, there exists a positive-definite matrix $G_{\de'}$ such that the matrix
\beaa
G_{\de'} A+A^T G_{\de'}
\eeaa
is positive-definite for all $x\in (0,1)$. In other words, there exists a positive-definite
inner product $\langle\cdot,\cdot\rangle_{G_{\de'}}$ on $\mathbb{R}^2$ such that the matrix $\frac 12(A+A^*)$ is positive definite with respect to $\langle\cdot,\cdot\rangle_{G_{\de'}}$ for all $x\in (0,1)$, where $A^*$ is the adjoint also with respect to this inner product. In particular, $A$ verifies the accretiveness required in Lemma \ref{lem:Duhamel-round-case}.
\end{lemma}
\begin{proof}
The matrix $A:=\de' I+x Q$ is not symmetric, and its symmetrized matrix is not always positive definite for all $x\in (0,1)$. To deal with this, we consider the following inner product in $\mathbb{R}^2$:
\beaa
\langle \v, \w\rangle_{G_{\de'}}:= \v^T G_{\de'} \w,\quad G_{\de'} :=\begin{pmatrix} 1& -1+(\de')^2 \\ -1+(\de')^2 & 1 \end{pmatrix}.
\eeaa
We compute
\beaa
 Q^T   \begin{pmatrix} 1& -1+(\de')^2 \\ -1+ (\de')^2 & 1 \end{pmatrix} + \begin{pmatrix} 1& -1+(\de')^2 \\ -1+(\de')^2 & 1 \end{pmatrix} Q = \begin{pmatrix} -2(\de')^2 & 0 \\ 0 & 2 (\de')^2 \end{pmatrix}.
\eeaa
Therefore, we have, using $A=\de'I +xQ$,
\beaa
\frac 12(G_{\de'} A+A^T G_{\de'}) &=& \de' G_{\de'}+\frac 12 x (G_{\de'} Q+Q^T G_{\de'}) \\
&=& \begin{pmatrix} \de' -x(\de')^2 & -\de'(1-(\de')^2) \\ -\de'(1-(\de')^2) & \de' +x(\de')^2 \end{pmatrix} =\de' \begin{pmatrix} 1-x \de' & -1+(\de')^2 \\ -1+(\de')^2 & 1 +x\de' \end{pmatrix}.
\eeaa
The last matrix is positive-definite since it is symmetric, $1\pm x\de'>0$, and its determinant is $1-x^2(\de')^2 -1+2(\de')^2-(\de')^4=(2-x^2)(\de')^2 -(\de')^4 >0$ for all $x\in (0,1)$ and $\de\leq 1$. For $\de'>1$ one can simply take $G_{\de'}=I$. This concludes the proof.
\end{proof}
\begin{remark}
The proof is an explicit construction of the solution to the Lyapunov matrix equation; see e.g. \cite{Parks1992} for a historical review.
\end{remark}

\subsubsection{Derivation of the projected equation in modes}\lab{sec:equations-in-modes}
\def\Pc{\widecheck{\Psi}}
We now derive the equations projected into modes. We introduce the notation, with $\cc$,  the prescribed  center-of-mass parameter appeared in \eqref{eq:solve-nonlocal-vanishing-condition}. 
\bea\lab{eq:Psi-1-checked}
\Pc_1 := \Psi_1 - \frac{3}{4\pi}\sum_{i=1}^3 \cc_i \om_i r^{-3},
\eea
where we recall from Remark \ref{rem:omega-i-ell=1} that the functions $\om_i$ only differ from $J_{1,\mm}$ by a constant factor $\sqrt{4\pi/3}$. According to this notation, we have $(\Pc_1)_{\ell\neq 1}=(\Psi_1)_{\ell\neq 1}$, and the last condition in \eqref{eq:solve-nonlocal-vanishing-condition} reads $\lim_{r\to \infty} (\Pc_1)_{\ell=1}=0$.
\begin{proposition}\lab{prop:derivation-eqns-modes}
For the system \eqref{eqns:nonlocal-eqn-Kc-kac}, we denote its components in spherical harmonic modes:
\bea\lab{def:bold-v-components}
\v_{\ell,\mm}=\begin{pmatrix}
  r^{2+\de'} (\Psi_1)_{\ell,\mm} \\
  r^{3+\de'} (\Psi_2)_{\ell,\mm}
\end{pmatrix},\quad \ell=0 \text{ or } \ell\geq 2,\qquad \v_{\ell,\mm}=\begin{pmatrix}
  r^3 (\Pc_1)_{\ell,\mm} \\
  r^{4} (\Psi_2)_{\ell,\mm}
\end{pmatrix},\quad \ell=1,
\eea
where $0<\de'<\de$. 
Then, the system \eqref{eqns:nonlocal-eqn-Kc-kac} is equivalent to the following projected equations into the spherical harmonic modes defined in \eqref{eq:def-l-m-modes}:
\bea\lab{eq:S-inverse-on-v-equation}
\pa_r \v_{\ell,\mm} &=& r^{-1} A_\ell \v_{\ell,\mm}+r^{-2}B_\ell(r)\v_{\ell,\mm}+\F_{\ell,\mm}+O(\ell^{-2}) \RR_{\ell,\mm}\left(r^{2+\de'}\Psi_2,r^{1+\de'} \Pc_1\right),\quad \ell\geq 2, \\
\lab{eq:v-equation-ell=1}
\pa_r \v_{1,\mm}&=& 
r^{-1} A_1 \v_{1,\mm}    +r^{-2} B_1(r) \v_{1,\mm} +\F_{1,\mm}  + O(r^3) \RR_{1,\mm}\left(\Psi_2, r^{-1} \Pc_1\right) ,\\
\lab{eq:v-equation-ell=0}
\pa_r \v_{\ell=0} &=& r^{-1} A_0
\v_{\ell=0}+r^{-2} B_0(r)\v_{\ell=0}+\F_{\ell=0}
+O(r^{2+\de'}) \RR_{\ell=0} \left(\Psi_2,r^{-1}\Pc_1\right).
\eea
Here,
\beaa
A_\ell=\begin{pmatrix}
  \de'-\frac 2{\ell(\ell+1)} & \frac 2{\ell(\ell+1)}  \\
  -\frac 2{\ell(\ell+1)} & \de'+\frac 2{\ell(\ell+1)}
\end{pmatrix} \text{ for }\ell\geq 2,\quad A_1=\begin{pmatrix}
  0 & 1  \\
  0 & 2 
\end{pmatrix},\quad A_0=\begin{pmatrix}
\de' & 1 \\
0 & 1+\de'
\end{pmatrix},
\eeaa
the matrices $B_\ell(r)$ have all their entries bounded uniformly in $r$ and $\ell$, and the inhomogeneous terms read
\bea\lab{eq:def-bold-F-ell-neq-1}
\F_{\ell,\mm}=\begin{pmatrix}
  r^{2+\de'} (F_1)_{\ell,\mm} \\
  r^{3+\de'} (F_2)_{\ell,\mm}
\end{pmatrix}+O(r^{-2+\de'}) \RR_{\ell,\mm} (\textstyle \sum_i \cc_i \om_i), \quad \ell=0 \text{ or } \ell\geq 2,
\eea
and
\bea\lab{eq:def-bold-F-ell=1}
 \F_{1,\mm}:=r^{-2} (\textstyle\sum_i\cc_i \om_i) B_1(r)  \begin{pmatrix} 1 \\ 0\end{pmatrix}   +           \begin{pmatrix}
r^3 (F_1)_{1,\mm} \\ 
 r^{4} (F_2)_{1,\mm} 
 \end{pmatrix}+O(r^{-1}) \RR_{1,\mm} (\textstyle\sum_i\cc_i \om_i),\quad \ell=1.
\eea
Moreover, the following bounds hold true for $\F:=\sum_{\ell=0}^\infty \sum_{\mm=-\ell}^\ell \F_{\ell,\mm} J_{\ell,\mm}$:
\bea\lab{eq:bounds-bold-F}
r^{-1} ||\F ||_{\H^s(S_r)} \lesssim \eps r^{-1-(\de-\de')},\quad |\F_{1,\mm}| \lesssim \eps r^{-1-\de}.
\eea

\end{proposition}
\begin{proof} We proceed as follows:

{\bf   Case $\ell\geq 2$.}
Projecting \eqref{eqns:nonlocal-eqn-Kc-kac} to modes with $\ell\geq 2$ and using \eqref{eq:def-RR-l-m}, we obtain
\bea
\bsplit
(\pa_r+2r^{-1}) (\Psi_1)_{\ell,\mm} &\; = \; -2(1-3mr^{-1}) r^{-2}\cdot \Up^{-\frac 12} (-\frac{r^2}{\ell(\ell+1)} ) (\Psi_2-\Up^\frac 12 r^{-1} \Psi_1)_{\ell,\mm}+(F_1)_{\ell,\mm} \\
&\quad +O(\ell^{-2}) \RR_{\ell,\mm} (\Psi_2,r^{-1} \Psi_1),  \\
  (\pa_r+3r^{-1}) (\Psi_2)_{\ell,\mm} &\; = \; -2\Up^\frac 12 r^{-3} \cdot\Up^{-\frac 12} (-\frac{r^2}{\ell(\ell+1)} ) (\Psi_2-\Up^\frac 12 r^{-1} \Psi_1)_{\ell,\mm} +(F_2)_{\ell,\mm} \\ 
& \quad +O(\ell^{-2}) \RR_{\ell,\mm} (r^{-1} \Psi_2,r^{-2} \Psi_1),
  \end{split}
\eea
or, in the matrix form for $\Psi_1$ and $r\Psi_2$, using that $\Up=1+O(m r^{-1})$,
\beaa
\pa_r\begin{pmatrix}
  (\Psi_1)_{\ell,\mm} \\
  r (\Psi_2)_{\ell,\mm}
\end{pmatrix}&=&\begin{pmatrix}
  -2-\frac{2 }{\ell(\ell+1)} & \frac {2}{\ell(\ell+1)}  \\
  -\frac{2}{\ell(\ell+1)} & -2+\frac{2}{\ell(\ell+1)}
\end{pmatrix}r^{-1}\begin{pmatrix}
  (\Psi_1)_{\ell,\mm} \\
  r (\Psi_2)_{\ell,\mm}
\end{pmatrix}+r^{-2}B_{\ell}(r)\begin{pmatrix}
   (\Psi_1)_{\ell,\mm} \\
  r (\Psi_2)_{\ell,\mm}
\end{pmatrix}+\begin{pmatrix}
  (F_1)_{\ell,\mm} \\
  r(F_2)_{\ell,\mm}
\end{pmatrix} \\
& &+O(\ell^{-2}) \RR_{\ell,\mm} (\Psi_2,r^{-1} \Psi_1),
\eeaa
where $B_\ell(r)$ is a matrix whose entries are bounded uniformly in $r$ and $\ell$.
Mutiplying each row by $r^{2+\de'}$ for some positive $\de'<\de$, we have
\beaa
\pa_r\begin{pmatrix}
  r^{2+\de'} (\Psi_1)_{\ell,\mm} \\
  r^{3+\de'} (\Psi_2)_{\ell,\mm}
\end{pmatrix}
&=&\begin{pmatrix}
  \de'-\frac 2{\ell(\ell+1)} & \frac 2{\ell(\ell+1)}  \\
  -\frac 2{\ell(\ell+1)} & \de'+\frac 2{\ell(\ell+1)}
\end{pmatrix}r^{-1}\begin{pmatrix}
  r^{2+\de'} (\Psi_1)_{\ell,\mm} \\
  r^{3+\de'} (\Psi_2)_{\ell,\mm}
\end{pmatrix} +r^{-2}B_{\ell}(r)\begin{pmatrix}
  r^{2+\de'} (\Psi_1)_{\ell,\mm} \\
  r^{3+\de'} (\Psi_2)_{\ell,\mm}
\end{pmatrix}\\
& &+\begin{pmatrix}
  r^{2+\de'} (F_1)_{\ell,\mm} \\
  r^{3+\de'} (F_2)_{\ell,\mm}
\end{pmatrix}+O(\ell^{-2}) \RR_{\ell,\mm}(r^{2+\de'}\Psi_2,r^{1+\de'} \Psi_1),
\eeaa
and the last term can be further decomposed as, using \eqref{eq:Psi-1-checked},
\beaa
\RR_{\ell,\mm}(r^{2+\de'}\Psi_2,r^{1+\de'} \Psi_1)=O(r^{-2+\de'}) \RR_{\ell,\mm} (\textstyle \sum_i \cc_i \om_i)+\RR_{\ell,\mm}(r^{2+\de'}\Psi_2,r^{1+\de'} \Pc_1).
\eeaa
This proves the expression \eqref{eq:S-inverse-on-v-equation}.

{\bf Case $\ell=0$.}
Projecting the system \eqref{eqns:nonlocal-eqn-Kc-kac} to the $\ell=0$ mode, we obtain
\bea
\bsplit
(\pa_r+2r^{-1})(\Psi_1)_{\ell=0} &\; = \; -2(1-3mr^{-1}) r^{-2}\left(\Up^{-\frac 12} \laph_{\tilde\ga}^{-1} (\Psi_2-\Up^\frac 12 r^{-1} \Psi_1)\right)_{\ell=0}+r^{-1} \overline{\Psi_1}^{\tilde\ga}\\
&\qquad 
+\Up^{-\frac 12}(\Psi_2-\Up^\frac 12 r^{-1} \Psi_1)_{\ell=0}
+(F_1)_{\ell=0},  \\
  (\pa_r+3r^{-1})(\Psi_2)_{\ell=0} &\; = \; -2\Up^\frac 12 r^{-3} \left(\Up^{-\frac 12} \laph_{\tilde\ga}^{-1} (\Psi_2-\Up^\frac 12 r^{-1} \Psi_1)\right)_{\ell=0}+ \Up^{-\frac 12} r^{-2} \overline{\Psi_1}^{\tilde\ga} \\
 & \qquad   +r^{-1} (\Psi_2-\Up^\frac 12 r^{-1} \Psi_1)_{\ell=0}
  +(F_2)_{\ell=0},
  \end{split}
\eea
where we again recall that we extend the definition of $\laph^{-1}_{\tilde\ga}$ using footnote \ref{ft:inverse-laph}.
Using \eqref{eq:def-RR-0}, we write
\beaa
\pa_r\begin{pmatrix}
  (\Psi_1)_{\ell=0} \\
  (\Psi_2)_{\ell=0}
\end{pmatrix}&=&\begin{pmatrix}
  -2r^{-1} & \Up^\frac 12  \\
  (\Up^{-\frac 12}-\Up^\frac 12) r^{-2} & -2r^{-1} 
\end{pmatrix}
\begin{pmatrix}
  (\Psi_1)_{\ell=0} \\
  (\Psi_2)_{\ell=0}
\end{pmatrix}           +           \begin{pmatrix}
 (F_1)_{\ell=0} \\ 
  (F_2)_{\ell=0} 
\end{pmatrix}+\RR_{\ell=0} (\Psi_2,r^{-1}\Psi_1) \begin{pmatrix} 1 \\ O(r^{-1}) \end{pmatrix}.
\eeaa
Denoting $\v_{\ell=0}=\begin{pmatrix}
  r^{2+\de'}(\Psi_1)_{\ell=0} \\
  r^{3+\de'} (\Psi_2)_{\ell=0}
\end{pmatrix}$ and $\F_{\ell=0}=\begin{pmatrix}
  r^{2+\de'}(F_1)_{\ell=0} \\
  r^{3+\de'} (F_2)_{\ell=0}
\end{pmatrix}$, we have
\beaa
\pa_r \v_{\ell=0}=r^{-1} \begin{pmatrix}
\de' & 1 \\
0 & 1+\de'
\end{pmatrix}
\v_{\ell=0}+r^{-2} B_0(r)\v_{\ell=0}+\F_{\ell=0}
+O(r^{2+\de'}) \RR_{\ell=0} (\Psi_2,r^{-1}\Psi_1),
\eeaa
for some matrix $B_0(r)$ with all its entries uniformly bounded in $r$. The last term can be further decomposed as, using \eqref{eq:Psi-1-checked},
\beaa
\RR_{\ell=0}(r^{2+\de'}\Psi_2,r^{1+\de'} \Psi_1)=O(r^{-2+\de'}) \RR_{\ell=0} (\textstyle \sum_i \cc_i \om_i)+\RR_{\ell=0}(r^{2+\de'}\Psi_2,r^{1+\de'} \Pc_1).
\eeaa
This proves the expression \eqref{eq:v-equation-ell=0}.

{\bf Case $\ell=1$.} Projecting the system \eqref{eqns:nonlocal-eqn-Kc-kac} to the $\ell=1$ modes, we obtain
\bea
\bsplit
(\pa_r+2r^{-1})(\Psi_1)_{1,\mm} &\; = \; -2(1-3mr^{-1}) r^{-2}\Up^{-\frac 12} \left(\laph_{\tilde\ga}^{-1} (\Psi_2-\Up^\frac 12 r^{-1} \Psi_1)\right)_{1,\mm}+(F_1)_{\ell,\mm},  \\
  (\pa_r+3r^{-1}) (\Psi_2)_{1,\mm} &\; = \; -2\Up^\frac 12 r^{-3} \Up^{-\frac 12} \left(\laph_{\tilde\ga}^{-1} (\Psi_2-\Up^\frac 12 r^{-1} \Psi_1)\right)_{1,\mm}- \frac 12 \Up^{-\frac 12} (\laph_{\tilde \ga}\Psi_1)_{1,\mm}+(F_2)_{\ell,\mm}.
  \end{split}
\eea
Using \eqref{eq:def-RR-l-m}, we can rewrite the system as
\bea
\bsplit
(\pa_r+2r^{-1})(\Psi_1)_{1,\mm} &\; = \; (1-3mr^{-1})\Up^{-\frac 12}  (\Psi_2-\Up^\frac 12 r^{-1} \Psi_1)_{\ell,\mm}+O(1)\RR_{1,\mm}(\Psi_2,r^{-1}\Psi_1)+(F_1)_{\ell,\mm},  \\
  (\pa_r+3r^{-1})(\Psi_2)_{1,\mm} &\; = \;  r^{-1} (\Psi_2-\Up^\frac 12 r^{-1} \Psi_1)_{\ell,\mm}+r^{-2}\Up^{-\frac 12} (\Psi_1)_{1,\mm}+O(r^{-1})\RR_{1,\mm}(\Psi_2,r^{-1}\Psi_1)+(F_2)_{\ell,\mm},
  \end{split}
\eea
or, in the matrix form,
\beaa
\pa_r\begin{pmatrix}
 (\Psi_1)_{1,\mm} \\
 (\Psi_2)_{1,\mm}
\end{pmatrix}&=&\begin{pmatrix}
  -3r^{-1} +O(mr^{-2}) & 1+O(mr^{-1})  \\
  (\Up^\frac 12 -\Up^{-\frac 12})r^{-2} & -2r^{-1} 
\end{pmatrix}
\begin{pmatrix}
 (\Psi_1)_{1,\mm} \\
 (\Psi_2)_{1,\mm}
\end{pmatrix}            +           \begin{pmatrix}
  (F_1)_{1,\mm}\\
  (F_2)_{1,\mm}
\end{pmatrix} +\RR_{1,\mm}(\Psi_2,r^{-1}\Psi_1)\begin{pmatrix}
  O(1)\\
  O(r^{-1})
\end{pmatrix}.
\eeaa
This can be further written as
\beaa
\pa_r\begin{pmatrix}
 r^3(\Psi_1)_{1,\mm} \\
 r^4 (\Psi_2)_{1,\mm}
\end{pmatrix}&=& r^{-1} \begin{pmatrix}
  0 & 1  \\
  0 & 2 
\end{pmatrix}
\begin{pmatrix}
 r^3(\Psi_1)_{1,\mm} \\
 r^4 (\Psi_2)_{1,\mm}
\end{pmatrix}            +    r^{-2} B_1(r) \begin{pmatrix}
 r^3(\Psi_1)_{1,\mm} \\
 r^4 (\Psi_2)_{1,\mm}
\end{pmatrix} +       \begin{pmatrix}
  r^3(F_1)_{1,\mm}\\
  r^4(F_2)_{1,\mm}
\end{pmatrix} +r^3\RR_{1,\mm}(\Psi_2,r^{-1}\Psi_1),
\eeaa
for some matrix $B_1(r)$ with all its entries uniformly bounded in $r$. 

Recall  that in view of \eqref{eq:solve-nonlocal-vanishing-condition}, $r^3(\Psi_1)_{1,\mm}$ does not vanish at infinity. However, as remarked after \eqref{eq:Psi-1-checked}, $r^3 (\Pc_1)_{1,\mm}$ does, and therefore, we consider $\v_{1,\mm}=\begin{pmatrix}
  r^3 (\Pc_1)_{1,\mm}\\
  r^{4} (\Psi_2)_{1,\mm}
\end{pmatrix}$. Since the first column of $\begin{pmatrix}
  0 & 1  \\
  0 & 2 
\end{pmatrix}$ is zero, we have \beaa\begin{pmatrix}
  0 & 1  \\
  0 & 2 
\end{pmatrix}
\begin{pmatrix}
 r^3(\Psi_1)_{1,\mm} \\
 r^4 (\Psi_2)_{1,\mm}
\end{pmatrix}=\begin{pmatrix}
  0 & 1  \\
  0 & 2 
\end{pmatrix} \v_{1,\mm}.
\eeaa
The system then reads,     
\beaa
\pa_r \v_{1,\mm}&=& \begin{pmatrix}
  0 & 1  \\
  0 & 2 
\end{pmatrix}
r^{-1} \v_{1,\mm}    +r^{-2} B_1(r) \v_{1,\mm}   +r^{-2} (\textstyle\sum_i\cc_i \om_i) B_1(r)  \begin{pmatrix} 1 \\ 0\end{pmatrix}   +           \begin{pmatrix}
r^3 (F_1)_{1,\mm} \\ 
 r^{4} (F_2)_{1,\mm} 
\end{pmatrix} \\
& & + O(r^3) \RR_{1,\mm}(\Psi_2, r^{-1} \Psi_1)  \\
&=& \begin{pmatrix}
  0 & 1  \\
  0 & 2 
\end{pmatrix}
r^{-1} \v_{1,\mm}    +r^{-2} B_1(r) \v_{1,\mm}   +r^{-2} (\textstyle\sum_i\cc_i \om_i) B_1(r)  \begin{pmatrix} 1 \\ 0\end{pmatrix}   +           \begin{pmatrix}
r^3 (F_1)_{1,\mm} \\ 
 r^{4} (F_2)_{1,\mm} 
\end{pmatrix} \\
& & +O(r^{-1}) \RR_{1,\mm} (\textstyle\sum_i\cc_i \om_i) + O(r^3) \RR_{1,\mm}(\Psi_2, r^{-1} \Pc_1) ,
\eeaa
Therefore, denoting 
\bea\lab{eq:def-F-ell=1}
\F_{1,\mm}:=r^{-2} \cc_{\mm} B_1(r)  \begin{pmatrix} 1 \\ 0\end{pmatrix}   +           \begin{pmatrix}
r^3 (F_1)_{1,\mm} \\ 
 r^{4} (F_2)_{1,\mm} 
\end{pmatrix}+O(r^{-1}) \RR_{1,\mm} (\textstyle\sum_i\cc_i \om_i),
\eea
we obtain the expression \eqref{eq:v-equation-ell=0} as required. The equivalence of the equations in modes and the original system \eqref{eqns:nonlocal-eqn-Kc-kac} is also clear since $\{J_{\ell,\mm}\}$ is a complete orthonormal basis over $L^2(S_r)$.

It remains to verify the bounds for $\F$. 
Recall the condition \eqref{eq:bound-condition-nonlocal-F}
\beaa
\sup_{r\in [r_0,\infty)} r^{-1} ||r^{3+\de} F_1,r^{4+\de} F_2, r^{4+\de} (F_1)_{\ell=1}, r^{5+\de} (F_2)_{\ell=1}||_{\H^s(S_{r})} \lesssim \eps.
\eeaa
As a result, by definition \eqref{eq:def-bold-F-ell-neq-1},
\beaa
r^{-1} ||\F_{\ell\neq 1} ||_{\H^s(S_r)} &\lesssim & \eps r^{-1-(\de-\de')}+r^{-2+\de'} \cdot r^{-1} ||\RR_{\ell\neq 1}(\textstyle \sum_i \cc_i \om_i)||_{\H^s} \\
&\lesssim & \eps r^{-1-(\de-\de')}+ \eps_1 |\cc| r^{-3-\de+\de'} \\
&\lesssim & \eps r^{-1-(\de-\de')},
\eeaa
and, by \eqref{eq:def-bold-F-ell=1},
\beaa
|\F_{1,\mm}| &\lesssim & |r^{-2} \cc_{\mm}|+|r^3 (F_1)_{1,\mm}, r^4(F_2)_{1,\mm}|+r^{-1} |\RR_{1,\mm}(\textstyle\sum_i\cc_i \om_i)| \\
& \lesssim & |\cc| r^{-2} + \eps r^{-1-\de} +r^{-1} |\cc| \eps_1 r^{-1-\de} \\
&\lesssim & \eps r^{-1-\de},
\eeaa
where we used that, in view of the bound for $\RR$ established in Proposition \ref{prop:bound-RR},
\beaa
|\RR_{1,\mm}(\textstyle\sum_i\cc_i \om_i)|\lesssim r^{-1}||\RR (\textstyle\sum_i\cc_i \om_i )||_{\H^s(S_r)} \lesssim r^{-1}\cdot \eps_1 r^{-1-\de} || \textstyle\sum_i\cc_i \om_i ||_{\H^s(S_r)} \lesssim \eps_1 |\cc| r^{-1-\de}.
\eeaa
Moreover, the bound for $|\F_{1,\mm}|$ means that we can in fact replace $\F_{\ell\neq 1}$ with $\F$ for the first estimate.
This concludes the proof of Proposition \ref{prop:derivation-eqns-modes}. 
\end{proof}

 {\bf The combined expression.} Since the $r$-weights we put in for different modes are different, we need to derive a uniform bound for the perturbative $\RR$ terms. This is done through the lemma below.
  \begin{lemma}\lab{lem:bound-RR-new-c}
 The system \eqref{eqns:nonlocal-eqn-Kc-kac} can be written in modes as
  \bea\lab{eq:eqn-in-RR-new}
\pa_r \v_{\ell,\mm} =r^{-1} A_\ell \v_{\ell,\mm}+r^{-2}B_\ell(r)\v_{\ell,\mm}+\F_{\ell,\mm}+ \RR^{new}_{\ell,\mm}(\v),
\eea
where the linear operator $\RR^{new}$ satisfies
\beaa
||\RR^{new} (\v)||_{\H^s(S_r)} \lesssim \eps_1 r^{-1-\de-\de'} ||\v||_{\H^s(S_r)}.
\eeaa
 \end{lemma}
 \begin{proof} 
According to Proposition \ref{prop:derivation-eqns-modes}, the system \eqref{eqns:nonlocal-eqn-Kc-kac} is equivalent to
 \bea
\pa_r \v_{\ell,\mm} =r^{-1} A_\ell \v_{\ell,\mm}+r^{-2}B_\ell(r)\v_{\ell,\mm}+\F_{\ell,\mm}+\RR^{new}_{\ell,\mm}(\v),
\eea
where the $\RR^{new}$ term reads, schematically, in terms of $\RR$ defined in Definition \ref{Def:RR-operator},
\beaa
\RR^{new}_{\ell,\mm}(\v)=\begin{cases}
\RR_{1,\mm}\left(r^2 \Pc_1,r^3\Psi_2\right)  & \ell =1 ,\\
\RR_{\ell,\mm} \left(r^{1+\de'}\Pc_1, r^{2+\de'} \Psi_2\right), & \ell=0 \text{ or }\ell\geq 2.
\end{cases}
\eeaa
Since $\de'<\de<1$, relaxing  the $r$ weights for $\ell\neq 1$, we have for each $\ell$ that \beaa
\RR^{new}_{\ell,\mm}(\v) = O(1) \RR_{\ell,\mm} (r^2 \Pc_1 ,r^3\Psi_2).
\eeaa
 Therefore we have, using the bound for $\RR$ in Proposition \ref{prop:bound-RR},  
\beaa
||\RR^{new} (\v)||_{\H^s(S_r)} &\lesssim & ||\RR (r^2 \Pc_1,r^3\Psi_2)||_{\H^s(S_r)} \lesssim \eps_1 r^{-1-\de} || r^2 \Pc_1, r^3\Psi_2||_{\H^s(S_r)} \\
&\lesssim & \eps_1 r^{-1-\de-\de'} ||\v_{\ell \neq 1}||_{\H^s(S_r)}+\eps_1 r^{-2-\de} ||\v_{\ell=1}||_{\H^s(S_r)}\\
&\lesssim & \eps_1 r^{-1-\de-\de'} ||\v||_{\H^s(S_r)}.
\eeaa
 where we used that $\v_{\ell,\mm}=\begin{pmatrix}
  r^{2+\de'} (\Psi_1)_{\ell,\mm} \\
  r^{3+\de'} (\Psi_2)_{\ell,\mm}
\end{pmatrix}$ for $\ell\neq 1$ and $\v_{1,\mm}=\begin{pmatrix}
  r^3 (\Pc_1)_{1,\mm} \\
  r^4 (\Psi_2)_{1,\mm}
\end{pmatrix}$ as defined in \eqref{def:bold-v-components}.
This concludes the proof of Lemma \ref{lem:bound-RR-new-c}.
 \end{proof}

\subsubsection{The solution operators in modes}\lab{sec:solution-operators-in-modes}
In this part, we verify that the matrices $A_\ell$ satisfy the accretiveness required in Lemma \ref{lem:Duhamel-round-case} for all $\ell$, hence giving uniformly bounded backward solution operators introduced in \eqref{eq:def-sol-op-ell}. 

{\bf The case $\ell\geq 2$.}
For simplicity, we denote $x=\frac{1}{\ell(\ell+1)}$ and consider $x\in (0,\frac 16]$, corresponding to $\ell\geq 2$. Denote the matrix 
\bea\lab{eq:matrix-x}
Q:=\begin{pmatrix}
  -1 & 1 \\
  -1 & 1
\end{pmatrix}
\eea
as in Lemma \ref{eq:lemma-Lyapunov-matrix}. Then we have $A_\ell=\de' I+x Q$, which verifies the condition of Lemma \ref{eq:lemma-Lyapunov-matrix}. Therefore, 
$A_\ell:=\de' I+x Q$ is accretive for some inner product $H$ over $\mathbb{R}^2$, and hence by Lemma \ref{lem:Duhamel-round-case}, we obtain a solution operator $U_{\ell}(r,r^*)$ for all $\ell \leq 2$
\bea\lab{eq:sol-op-ell-geq-2}
\pa_r U_{\ell} (r,r^*)=(r^{-1} A_\ell+r^{-2}B_\ell(r) ) U_\ell (r,r^*),\quad U_\ell(r^*,r^*)=I,
\eea
where $U_\ell$ is uniformly bounded.

{\bf The case $\ell\leq 1$.} Since the matrices $A_0=\begin{pmatrix}
\de' & 1 \\
0 & 1+\de'
\end{pmatrix}$, $A_{1}=\begin{pmatrix}
  0 & 1  \\
  0 & 2 
\end{pmatrix}$ can both be diagonalized to a positive definite matrix, they easily verify the accretiveness condition. Hence, we obtain the backward solution operators $U_0(r,r^*)$, $U_1(r,r^*)$ through
\bea\lab{eq:def-sol-op-0}
&&\pa_r U_0(r,r^*)=(r^{-1} A_0+r^{-2} B_0(r))U_0(r,r^*),\quad U_0(r^*,r^*)=I,\\
\lab{eq:def-sol-op-1}
&&\pa_r U_1(r,r^*)=\left(r^{-1} A_{1}+r^{-2} B_1(r) \right) U_1(r,r^*) ,\quad U_1(r^*,r^*)=I,
\eea
and they are both uniformly bounded, as stated in Lemma \ref{lem:Duhamel-round-case}.

\subsubsection{The inhomogeneous solution}
\begin{proposition}\lab{prop:bound-solution-inhomogeneous}
Define $\mathring\v$ using
\beaa
\mathring\v_{\ell,\mm}=-\int_r^\infty U_\ell (r,r') \F_{\ell,\mm}(r')\, dr'.
\eeaa
We have the estimate
\bea\lab{eq:bound-v-0-perturbed}
r^{-1+(\de-\de')} || \mathring\v ||_{\H^s(S_r)} + r^\de |\mathring \v_{1,\mm}|\lesssim \eps.
\eea
\end{proposition}
\begin{proof}
For $\ell=1$, we apply the second bound in \eqref{eq:bounds-bold-F}, which yields
\beaa
|\mathring \v_{1,\mm}|\lesssim \eps r^{-\de}.
\eeaa
Moreover, using the first bound in \eqref{eq:bounds-bold-F}
\beaa
r^{-1} ||\F||_{\H^s(S_r)} \lesssim \eps r^{-1-(\de-\de')},\quad 
\eeaa
we obtain 
\beaa
r^{-1} ||\mathring\v ||_{\H^s(S_r)} &\leq & \Big(\sum_{\ell=0}^\infty \sum_{\mm=-\ell}^\ell (1+\ell^2)^s |\mathring\v_{\ell,\mm}|^2\Big)^\frac 12\leq  C \Bigg(\sum_{\ell=0}^\infty \sum_{\mm=-\ell}^\ell \Big (\int_r^\infty  (1+\ell^2)^s |\F_{\ell,\mm} | (r')\, dr' \Big)^2\Bigg)^\frac 12 \\
& \leq & C \int_r^\infty \Bigg(\sum_{\ell=0}^\infty \sum_{\mm=-\ell}^\ell (1+\ell^2)^s | \F_{\ell,\mm}|^2(r') \Bigg)^\frac 12 \, dr' \\
& \leq & C \int_r^\infty r'^{-1} ||\F||_{\H^s(S_{r'})}\, dr'\\
&\lesssim & \eps r^{-(\de-\de')},
\eeaa
where we used the integral Minkowski inequality \eqref{eq:integral-Minkowski-inequality} from the third inequality.
This concludes the proof of Proposition \ref{prop:bound-solution-inhomogeneous}.
\end{proof}

 \subsubsection{The contraction argument}\lab{sec:summing-up}
Since the $\RR_{\ell,\mm}^{new}$ terms can involve different modes of $\v$, in order to obtain the solution of \eqref{eqns:nonlocal-eqn-Kc-kac}, or equivalently \eqref{eq:eqn-in-RR-new}, 
we need  a physical space  norm  to estimate $\v$, independent of  its modes  $\v_{\ell,\mm}$. We define
\bea\lab{eq:def-contraction-norm}
|| \v ||_{\VV}:= \sup_{r\in [r_0,\infty)} \left(r^{-1+(\de-\de')} ||\v ||_{\H^s(S_r)}+ r^\de |\v_{\ell=1}|\right),
\eea
and seek solutions in the following neighborhood of $\mathring\v$:
\beaa
\VV_{C\eps} :=\{\v\colon  ||\v ||_{\VV} < C \eps\},
\eeaa
where $C$ is a positive constant to be determined. 
In view of \eqref{eq:Duhamel-linear-1}, solutions to \eqref{eq:eqn-in-RR-new} satisfy
\bea\lab{eq:fixed-point-eqn-v}
\v=\Phi(\v),
\eea
where the map $\Phi$ is defined through $\Phi(\v):=\sum_{\ell=0}^\infty \sum_{\mm=-\ell}^\ell \Phi(\v)_{\ell,\mm} J_{\ell,\mm}$, with
\beaa
\Phi(\v)_{\ell,\mm} : =- \int_r^\infty U_\ell (r,r')\left(\F_{\ell,\mm}+\RR^{new}_{\ell,\mm}(\v)\right)\, dr'.
\eeaa
Conversely, any $\v$ satisfying \eqref{eq:fixed-point-eqn-v} gives a solution to the original system \eqref{eqns:nonlocal-eqn-Kc-kac}.
Recall that for $\ell\geq 2$, $U_\ell$ is the same solution operator defined in \eqref{eq:sol-op-ell-geq-2}, and for $\ell=0,1$, $U_\ell$ is defined respectively in \eqref{eq:def-sol-op-0}, \eqref{eq:def-sol-op-1}.

It suffices to show that $\Phi(\VV_{C\eps})\subset \VV_{C\eps}$ and $\Phi$ is a contraction in $\VV_{C\eps}$ with respect to the norm $ ||\cdot ||_{\VV}$. 
We have
\beaa
\Phi(\v)_{\ell,\mm} &=& - \int_r^\infty U_\ell(r,r')\left( \F_{\ell,\mm}+\RR^{new}_{\ell,\mm}(\v)\right)\,  dr'\\
&=& \mathring\v_{\ell,\mm} -\int_r^\infty U_\ell(r,r') \RR^{new}_{\ell,\mm}(\v) \, dr'.
\eeaa
Therefore, 
\beaa
r^{-1} ||\Phi(\v)-\mathring\v ||_{\H^s(S_r)} &\leq & \Big(\sum_{\ell=0}^\infty \sum_{\mm=-\ell}^\ell (1+\ell^2)^s |\Phi(\v)_{\ell,\mm}-\mathring\v_{\ell,\mm}|^2\Big)^\frac 12\\
& \leq & C \Bigg(\sum_{\ell=0}^\infty \sum_{\mm=-\ell}^\ell \Big (\int_r^\infty  (1+\ell^2)^s |\RR^{new}_{\ell,\mm}(\v) | (r')\, dr' \Big)^2\Bigg)^\frac 12 \\
& \leq & C \int_r^\infty \Bigg(\sum_{\ell=0}^\infty \sum_{\mm=-\ell}^\ell (1+\ell^2)^s | \RR^{new}_{\ell,\mm}(\v) |^2(r') \Bigg)^\frac 12 \, dr' \\
& \leq & C \int_r^\infty r'^{-1} ||\RR^{new}(\v)||_{\H^s(S_{r'})}\, dr',
\eeaa
where we used the integral Minkowski inequality \eqref{eq:integral-Minkowski-inequality} from the second line to the third line.
Then, using Lemma \ref{lem:bound-RR-new-c} and \eqref{eq:bound-v-0-perturbed},
\beaa
r^{-1+(\de-\de')} ||\Phi(\v)-\mathring\v||_{\H^s(S_r)} &\leq & C r^{\de-\de'} \int_r^\infty  r'^{-1}\cdot  \eps_1 r'^{-1-\de-\de'}  ||\v||_{\H^s(S_{r'})} \, dr' \\
&\leq & Cr^{\de-\de'} \int_r^\infty C \eps_1 (r'^{-2-\de-\de'})\cdot r'^{1-(\de-\de')}||\v||_{\VV} \, dr' \\
& \leq & C \eps_1 \eps \ll \eps,
\eeaa
for suitable $C>0$.\footnote{We omit writing $\de^{-1}$ since $\de>0$ is a given constant.} We also have
\beaa
r^\de |\Phi(\v)_{\ell=1}-\mathring\v_{\ell=1}| &\leq & C r^\de \int_r^\infty |
\RR^{new}_{1,\mm}(\v) | dr'\leq C r^\de \int_r^\infty r'^{-1} ||\RR^{new}(\v)||_{L^2(S_{r'})} \, dr' \\
&\leq & C r^\de \int_r^\infty \eps_1 r'^{-2-\de-\de'} ||\v||_{\H^s (S_{r'})}\, dr' \leq C r^\de \left(\int_r^\infty \eps_1 r'^{-2-\de-\de'}\cdot r'^{1-(\de-\de')} \, dr'\right) ||\v||_{\VV} \\
&\leq & C \eps_1 \eps \ll \eps.
\eeaa
Therefore, we obtain $||\Phi(\v)-\mathring\v||_{\VV}\ll \eps$. From \eqref{eq:bound-v-0-perturbed} we know that $||\mathring\v||_{\VV}\lesssim \eps$, and hence we see that $\Phi(\VV_{C\eps})\subset \VV_{C\eps}$ for suitable $C>0$. 

To prove that $\Phi$ is a contraction, we note that
\beaa
(\Phi(\v_1)-\Phi(\v_2))_{\ell,\mm}=- \int_r^\infty U_\ell (r,r') \RR^{new}_{\ell,\mm} (\v_1-\v_2)\, dr'.
\eeaa
Hence, by similar estimates using Lemma \ref{lem:bound-RR-new-c}, we obtain
\beaa
r^{-1+(\de-\de')} ||\Phi(\v_1)-\Phi(\v_2)||_{\H^s(S_r)} & \leq & C r^\de \int_r^\infty r'^{-1} ||\RR^{new} (\v_1-\v_2)||_{\H^s(S_{r'})}\, dr' \\
& \leq & C r^\de \int_r^\infty  \eps_1 r'^{-2-\de-\de'}  ||\v_1-\v_2||_{\H^s(S_{r'})} \, dr' \\
&\leq & C \eps_1 r^\de \left(\int_r^\infty r'^{-1-2\de}\, dr'\right)\sup_{r\in [r_0,\infty)}  r^{-1+(\de-\de')} ||\v_1-\v_2||_{\H^s(S_r)} \\
&\leq & C \eps_1 ||\v_1-\v_2||_{\VV},
\eeaa
and
\beaa
r^\de |\Phi(\v_1)_{\ell=1}-\Phi(\v_2)_{\ell=1} |&\leq & C r^\de \int_r^\infty |
\RR^{new}_{1,\mm}(\v_1-\v_2) | dr'\leq C r^\de \int_r^\infty r'^{-1} ||\RR^{new}(\v_1-\v_2)||_{L^2(S_r')} \, dr' \\
&\leq & C r^\de \int_r^\infty \eps_1 r'^{-2-\de-\de'} ||\v_1-\v_2 ||_{\H^s (S_r')}\leq C r^\de \left(\int_r^\infty \eps_1 r'^{-1-2\de}\, dr'\right) ||\v_1-\v_2||_{\VV} \\
&\leq & C \eps_1  ||\v_1-\v_2||_{\VV}.
\eeaa
Therefore, by the fixed point theorem, we obtain a unique solution $\v$ in $\VV_{C\eps}$, which, when expressed in terms of $\Psi_1$ and $\Psi_2$, verifies \eqref{eq:solve-nonlocal-vanishing-condition} in view of the definition of $||\cdot||_{\VV}$ in \eqref{eq:def-contraction-norm}. This concludes the proof of Proposition \ref{prop:solve-unknowns-kac-Kc}.

\subsection{Boundedness estimates: Proof of Proposition \ref{prop:boundedness}}
\begin{remark}
Throughout this proof, the implicit constants in the symbol $\lesssim $ do not include the bootstrap constant $C_b$ stated in the Proposition \ref{prop:boundedness}.
\end{remark}
\begin{remark}
The $L^\infty$ estimates needed in the proof can be easily derived by standard Sobolev embedding from the $L^2$ estimates:
\bea\lab{eq:L-infty-bounds-Psi-n}
\bsplit
& r^{2+\de} || (r\nabz)^{\leq s-1} (\Psi_1\n, \Psi_4\n,\Psi_5\n,\Psi_7\n,\Psi_8\n, \Psi_9\n, \Psi_{10}\n) ||_{L^\infty(S_r)} \\
&+r^{3+\de} || (r\nabz)^{\leq s-2} (\Psi_2\n, \Psi_6\n) ||_{L^\infty(S_r)} +r^{1+\de} ||(r\nabz)^{\leq s}\Psi_3\n||_{L^\infty(S_r)}\\
& +r^{1+\de} || (r\nabz)^{\leq s} (\ga\n-\ga\0) ||_{L^\infty(S_r)} \lesssim C_b \eps.
\end{split}
\eea
\end{remark}
\begin{remark}
Throughout the proof, we will use the following bound, ensured by the assumption of Proposition \ref{prop:boundedness}, without explicit reference:
\bea\lab{eq:iteration-ga-bound-beginning}
r^\de ||\tilde\ga-\gz||_{\H^{s+1}(S_r)}\lesssim C_b \eps.
\eea
In particular, this allows us to apply the Hodge estimate in Lemma \ref{lemma:Hodge-estimate-round}. 
\end{remark}
We now proceed as follows.

\subsubsection{Proof of Proposition \ref{prop:step-1}}\lab{proof:step-1}
We explicitly write down the expression of $\Psi_3\nn$ using \eqref{eq:iteration-divP} and \eqref{eq:iteration-average-a}:
\bea\lab{eq:boundness-proof-Psi-3-expression}
\bsplit
\Psi_3\nn &= \Up^{-\frac 12}(\laph\n)^{-1} \Big(\Psi_2\nn-\Up^\frac 12 r^{-1} \Psi_1\nn+\Ga_1\n\cdot\Ga_1\n\Big)-\Up^{-\frac 12}(\laph\n)^{-1} \laph\n(\Ga_0\n\cdot\Ga_0\n) \\
& \quad  -\frac 12 \Up^{-1} r \overline{\Psi_1\nn}\n.
\end{split}
\eea
Here we again adopt the extended definition of $(\laph\n)^{-1}$ as in footnote \ref{ft:inverse-laph}.

We first apply Proposition \ref{prop:solve-unknowns-kac-Kc} to obtain $\Psi_1\nn$ and $\Psi_2\nn$. 
Denote the error terms
\beaa
\NN\n[\ao]&:=& -\Up^{-\frac 12}(\laph\n)^{-1} \laph\n(\Ga_0\n\cdot\Ga_0\n)+\Up^{-\frac 12} (\laph\n)^{-1} (\Ga_1\n\cdot\Ga_1\n) \\
&=& -\Up^{-\frac 12}(\Ga_0\n\cdot\Ga_0\n-\overline{\Ga_0\n\cdot\Ga_0\n}\n)+\Up^{-\frac 12} (\laph\n)^{-1} (\Ga_1\n\cdot\Ga_1\n), \\
\NN\n[\widetilde \B_{\ell\leq 1}] &:=& \frac 12(\laph\n \Psi_1\n)_{\ell=0}-\left(\mathcal P_1(\d_1\n\d_2\n \Psi_4\n)\right)_{\ell\leq 1},\\
\NN\n[\mu] &:=& -(\laph\n \log\Psi_3\n)_{\ell=0}-\frac 14 ((\Psi_1)^2)_{\ell=0}.
\eeaa
 Then the system of $\Psi_1\nn$ and $\Psi_2\nn$, originating from \eqref{eq:iteration-kac}, \eqref{eq:iteration-Kc}, reads, in view of \eqref{eq:tilde-BB},
 \bea\lab{eq:system-thc-Kc-detailed}
\bsplit
(\pa_r+2r^{-1})\Psi_1\nn &\; = \; -2(1-3mr^{-1}) r^{-2}\left(\Up^{-\frac 12} (\laph\n)^{-1} (\Psi_2\nn-\Up^\frac 12 r^{-1} \Psi_1\nn)-\frac 12 \Up^{-1} r\overline{\Psi_1\nn}\n
\right)\\
& \quad+\Up^{-\frac 12} (\Psi_2\nn-\Up^\frac 12 r^{-1} \Psi_1\nn)_{\ell=0}+\Psi_3\n\muc_{\ell=0}\n+\Ga_1\n\cdot\Ga_1\n \\
&\quad + \Up^{-\frac 12} \NN\n[\mu] -2(1-3mr^{-1}) r^{-2} \NN\n[\ao], \\
  (\pa_r+3r^{-1})\Psi_2\nn &\; = \; -2\Up^\frac 12 r^{-3} \left(\Up^{-\frac 12} (\laph\n)^{-1} (\Psi_2\nn-\Up^\frac 12 r^{-1} \Psi_1\nn)-\frac 12 \Up^{-1} r\overline{\Psi_1\nn}\n
\right) \\
 &\quad +r^{-1} (\Psi_2\nn-\Up^\frac 12 r^{-1} \Psi_1\nn)_{\ell=0}+r^{-1}\NN\n[\mu] -2\Up^\frac 12 r^{-3} \NN\n[\ao] \\
 &\quad   -\Up^{-\frac 12}\Big (\Bb+\frac 12(\laph\n \Psi_1\nn)_{\ell=1}+\NN\n[\widetilde\B_{\ell\leq 1}]\Big)+ \Psi_3\n (\Bb+\widetilde\B_{\ell\leq 1, aux}\n)+\Ga_1\n\cdot \Ga_2\n.
  \end{split}
\eea 
The system \eqref{eq:system-thc-Kc-detailed} is of the form \eqref{eqns:nonlocal-eqn-Kc-kac}, with $\tilde\ga=\ga\n$, and
\beaa
F_1&=& -2(1-3mr^{-1}) r^{-2} \NN\n[\ao] +\Up^{-\frac 12} \NN\n [\mu] +\Psi_3\n \muc\n_{\ell=0} +\Ga_1\n\cdot\Ga_1\n,\\
F_2&=& -\Up^{-\frac 12} (\Bb+\NN\n[\widetilde\B_{\ell\leq 1}])+ \Psi_3\n (\Bb+\widetilde\B_{\ell\leq 1, aux}\n)+\Ga_1\n\cdot\Ga_2\n-2\Up^\frac 12 r^{-3}\NN\n[\ao]+r^{-1} \NN\n [\mu].
\eeaa
We now verify the bounds required in Proposition \ref{prop:solve-unknowns-kac-Kc}. We have
\beaa
r^{-1} ||\NN\n[\ao]||_{\H^{s+1}(S_r)} &\lesssim & r^{-1} ||\Ga_0\n\cdot\Ga_0\n||_{\H^{s+1}(S_r)}+ r^{-1} ||(\laph\n)^{-1} (\Ga_1\n\cdot\Ga_1\n)||_{\H^{s+1}(S_r)} \\
&\lesssim & C_b^2 \eps^2 r^{-2-2\de}.
\eeaa
Applying \eqref{eq:estimate-d1-d2-ell-leq-1}, 
we obtain
\beaa
r^{-1} ||\NN\n[\widetilde\B_{\ell\leq 1}] ||_{\H^s(S_r)} &\lesssim & r^{-1} ||(\laph\n \Psi_1\n)_{\ell=0}||_{\H^s(S_r)}+ r^{-1} ||(\d_1\n\d_2\n \Psi_4\n)_{\ell\leq 1}||_{\H^s(S_r)} \\
&\lesssim &  C_b \eps r^{-1-\de} r^{-2} ||\Psi_1\n||_{L^\infty}+  ||\laph\0 \trh\0 \Psi_4\n||_{L^\infty}\\
& & + C_b \eps r^{-3-\de} ||(r\nabh\0)^{\leq 2} \Psi_4\n||_{L^\infty} \\
& \lesssim & C_b^2 \eps^2 r^{-5-2\de},
\eeaa
\beaa
r^{-1} ||\NN\n[\mu]||_{\H^s(S_r)} &\lesssim & | -(\laph\n \log\Psi_3\n)_{\ell=0}|+\frac 14 |(\Psi_1)^2 | \lesssim C_b \eps^2 r^{-2}\cdot r^{-1-\de} \cdot r^{-1-\de}+C_b^2 \eps^2 r^{-4-2\de}\\
&\lesssim & C_b^2 \eps^2 r^{-4-2\de}.
\eeaa
Therefore, we deduce
\beaa
r^{-1} ||F_1||_{\H^{s+1}(S_r)} &\lesssim & r^{-1}||r^{-2} \NN\n[\ao]||_{\H^{s+1}} +r^{-1}||\NN\n [\mu] ||_{\H^{s+1}} + r^{-1}||\Psi_3\n \muc\n_{\ell=0}||_{\H^{s+1}} +r^{-1} ||\Ga_1\n\cdot\Ga_1\n||_{\H^{s+1} } \\
&\lesssim & C_b^2 \eps^2 r^{-4-2\de}+  C_b \eps^2 r^{-1-\de}\cdot r^{-3-\de} \lesssim C_b^2 \eps^2 r^{-4-2\de},
\eeaa
\beaa
r^{-1} ||F_2 ||_{\H^s(S_r)} &\lesssim & r^{-1} ||\Bb||_{\H^s} +r^{-1}||\NN\n[\widetilde\B_{\ell\leq 1}]||_{\H^s}+ r^{-1}||\Psi_3\n (\Bb+\widetilde\B_{\ell\leq 1, aux}\n)||_{\H^s}\\
& & +r^{-1}||\Ga_1\n\cdot\Ga_2\n||_{\H^s}+r^{-1} || r^{-3}\NN\n[\ao]||_{\H^s}+r^{-1} ||r^{-1}\NN\n [\mu]||_{\H^s} \\
&\lesssim & \eps r^{-4-\de}+C_b^2 \eps^2 r^{-5-2\de}.
\eeaa
Moreover, since $\Bb_{\ell=1}=0$, we have
\beaa
r^{-1} ||(F_2)_{\ell=1} ||_{\H^s(S_r)} &\lesssim & C_b^2 \eps^2 r^{-5-2\de}.
\eeaa
Therefore, for given center of mass value $\cc\in \mathbb{R}^3$, applying Proposition \ref{prop:solve-unknowns-kac-Kc} to \eqref{eq:system-thc-Kc-detailed}, we obtain the unique solution $(\Psi_1\nn,\Psi_2\nn)$ verifying the bounds
\begin{equation}\lab{eq:bounds-Psi-1-2-nn}
\sup_{r\in [r_0,\infty)} r^{-1} ||r^{2+\de} \Psi_1\nn, r^{3+\de} \Psi_2\nn||_{\H^s(S_r)}\lesssim \eps,\quad \sup_{r\in [r_0,\infty)} r^\de | r^3 (\Psi_1\nn)_{\ell=1,i}-\cc_i , r^4 (\Psi_2\nn)_{\ell=1,i} | \lesssim \eps.
\end{equation}
Note that the right-hand side is $\eps$ instead of $C_b \eps$.

To derive the estimate for $\Psi_3\nn$, it suffices to recall the expression \eqref{eq:boundness-proof-Psi-3-expression}, which, again in view of Lemma \ref{lemma:estimate-inverse-Laplacian}, implies
\bea\lab{eq:bound-Psi-3-nn}
r^{-1} ||\Psi_3\nn||_{\H^{s+2}}\lesssim \eps r^{-1-\de}.
\eea
\begin{remark}
Using the second bound in \eqref{eq:bounds-Psi-1-2-nn}, we also easily deduce the behavior of $(\Psi_3\nn)_{\ell=1}$ using \eqref{eq:boundness-proof-Psi-3-expression}:
\bea\lab{eq:expansion-Psi-3-ell=1}
|(\Psi_3\nn)_{\ell=1,i}-\frac 12 \cc_i r^{-2}| \lesssim C_b \eps r^{-1-\de} \cdot ||r^2 \Psi_2\nn, r\Psi_1\nn||_{L^\infty} \lesssim C_b^2 \eps^2 r^{-2-2\de}.
\eea
Such a more precise estimate will be useful in Appendix \ref{appendix:physical-quantities}.
\end{remark}

We now further derive the $\H^{s+1}$ estimates of $\Psi_1\nn$. Commuting the equation \eqref{eq:iteration-kac} with $(r\nabz)^{s+1}$ using \eqref{eq:commutation-nabz-r-nabz}, we have
\beaa
(\pa_r+2r^{-1}) (r\nabz)^{s+1} \Psi_1\nn &=& 
-2(1-3mr^{-1}) r^{-2} (r\nabz)^{s+1} \Psi_3\nn+(r\nabz)^{s+1} (\Ga_1\n\cdot \Ga_1\n).
\eeaa
Directly applying Lemma \ref{lem:transport-lemma} to this equation, using the bound \eqref{eq:bound-Psi-3-nn} we just obtained, we deduce 
\bea\lab{eq:highest-order-Psi-1}
r^{-1} ||r^{2} \Psi_1\nn ||_{\H^{s+1}(S_r)}\lesssim \int_r^\infty r'^{-1}\cdot r'^2  ||r'^{-2} \Psi_3\nn+\Ga_1\n\cdot \Ga_1\n||_{\H^{s+1}(S_{r'})} \, dr' \lesssim \eps r^{-\de}.
\eea
This finishes the proof of Proposition \ref{prop:step-1}.

\subsubsection{Proof of Proposition \ref{prop:step-2}}\lab{proof:step-2}
We proceed to determine $\Psi_4\nn$ and $\Psi_{11}\nn$ (which is supported on $\ell\leq 1$) from the equation \eqref{eq:iteration-h-dualh}:
\beaa
 \d_1\n \d_2\n \Psi_4\nn  = \frac 12 (\laph\n \Psi_1\nn,0)-(\Bb,\Bbd) -\Psi_{11}\nn.
\eeaa
Taking into account that $\Psi_1\nn$ has already been obtained, 
we can apply Corollary \ref{cor:solvability} with $(S,\ga)=(S_r,\ga\n)$ to obtain a unique $\Psi_{11}\nn$ for which \eqref{eq:iteration-h-dualh} is solvable.

Moreover, using the estimate \eqref{eq:estimate-cor-solvability} with $\mathring\eps=C_b \eps r^{-1-\de}$, noticing also that 
$\laph\n \Psi_1\nn$ on the right-hand side of \eqref{eq:iteration-h-dualh} has zero spherical mean over $\ga\n$, we have
\beaa
|(\Psi_{11}\nn)_{\ell=0} | &\lesssim  &|\overline{-(\Bb,\Bbd)}\n| \lesssim ||(\Bb,\Bbd)||_{L^\infty(S_r)} ||\ga\n-\gz||_{L^\infty(S_r)} \lesssim C_b^2 \eps^2 r^{-5-2\de}, 
\eeaa
\beaa
|(\Psi_{11}\nn)_{\ell=1}-\frac 12 (\laph\n \Psi_1\nn,0)_{\ell=1}| &\lesssim & C_b \eps r^{-1-\de} \cdot r^{-1}|| (\Bb,\Bbd) ||_{L^2(S_r,\ga\n)}\\ 
&\lesssim & C_b^2 \eps^2 r^{-5-2\de},
\eeaa
where we have used  
the equivalence of norms from Lemma \ref{lemma:equivalence-norms}.
Then, using the improved estimate \eqref{eq:step1-estimate} for $(\Psi_1\nn)_{\ell=1}$, we obtain
\beaa
|(\Psi_{11}\nn)_{\ell=1}|\lesssim |\cc|r^{-5}+\eps r^{-5-\de}+C_b^2 \eps^2 r^{-5-2\de} \lesssim \eps r^{-5}.
\eeaa
 Corollary \ref{cor:solvability} also implies that the solution $\Psi_4\nn$ to \eqref{eq:iteration-h-dualh} exists and, in view of Lemma \ref{lemma:Hodge-estimate-round} applied to $\d_1\d_2$,
	\beaa
		r^{-1} ||\Psi_4\nn||_{\H^{s+1}(S_r)}\lesssim r^{-1} \cdot r^2 ||(\Bb,\Bbd)||_{\H^{s-1}(S_r)} +r^{-1} || \Psi_1\nn ||_{\H^{s+1}(S_r)} \lesssim \eps r^{-2-\de}.
	\eeaa 	
	Similarly, applying Lemma \ref{lemma:Hodge-estimate-round}, we obtain $\Psi_5\nn$ from \eqref{eq:iteration-Delta-ah} and show that it verifies the estimate 
\beaa
r^{-1}  ||\Psi_5\nn||_{\H^{s+1}(S_r)} \lesssim  ||\laph\n \Psi_3\nn||_{\H^s(S_r)}+ ||\laph\n(\Ga_0\n\cdot\Ga_0\n)||_{\H^s(S_r)} \lesssim  r^{-2} \cdot \eps r^{-\de} \lesssim \eps r^{-2-\de}.
\eeaa	
To conclude this step, we derive the estimate of $\Psi_6\nn$ through the equation \eqref{eq:iteration-Y}. We obtain, by Lemma \ref{lemma:Hodge-estimate-round},
\beaa
r^{-1} ||\Psi_6 \nn||_{\H^{s+1}(S_r)}\lesssim ||(\Bb,\Bbd)||_{\H^{s}(S_r)} +|| \Psi_{11}\nn ||_{\H^{s}(S_r)}\lesssim \eps r^{-3-\de}.
\eeaa

\subsubsection{Proof of Proposition \ref{prop:step-4}}\lab{proof:step-4}
We recall the equations \eqref{eq:iteration-trt}, \eqref{eq:iteration-Pi}, and \eqref{eq:iteration-Pi-mean}:
\beaa
(\pa_r+ r^{-1})\Psi_7\nn &=& 2r^{-1}\Psi_{10}\nn+\Ga_1\n\cdot\Ga_1\n,\\
\laph\n \left(\ah\n \Psi_{10}\nn\right) &=& \Kk+\widetilde \K_{\ell\leq 1}\nn-\overline{\Kk+\widetilde \K_{\ell\leq 1}\nn}\n, \\
\overline{\displaystyle\ah\n\Psi_{10}\nn}\n &=& \overline{\Psi_3\n \Psi_{10}\n}\n.
\eeaa
The equations \eqref{eq:iteration-Pi} and \eqref{eq:iteration-Pi-mean} imply the following expression of $\Psi_{10}\nn$:
\bea\lab{eq:expression-Psi-10}
\ah\n \Psi_{10}\nn=(\laph\n)^{-1}(\Kk+\widetilde \K_{\ell\leq 1}\nn)+\overline{\Psi_3\n \Psi_{10}\n}\n.
\eea
Plugging this into the equation of $\Psi_7\nn$, we obtain
\bea\lab{eq:Psi-7-nonlocal}
(\pa_r+r^{-1}) \Psi_7\nn=2r^{-1} (\ah\n)^{-1} (\laph\n)^{-1}(\Kk+\widetilde \K_{\ell\leq 1}\nn)+\Ga_1\n\cdot\Ga_1\n+2r^{-1} \overline{\Psi_3\n \Psi_{10}\n}\n,
\eea
where we recall \eqref{eq:widetilde-K-ell-leq-1}
\beaa
\widetilde \K_{\ell\leq 1}\nn:= \mathcal P_1\left( \d_1\n\d_2\n(\ah\n\Psi_8\n)\right)_{\ell\leq 1}-\frac 12 (\ah\n \laph\n \Psi_7\nn)_{\ell= 1}-\frac 12 (\ah\n \laph\n \Psi_7\n)_{\ell=0}+(\Ga_1\n\cdot\Ga_2\n)_{\ell\leq 1}.
\eeaa
Therefore, with all quantities labeled with $\, \n$ viewed as known quantities, \eqref{eq:Psi-7-nonlocal} is an equation of $\Psi_7\nn$:
\bea\lab{eq:Psi-7-nonlocal-long}
\bsplit
(\pa_r+r^{-1}) \Psi_7\nn &= 2r^{-1} (\ah\n)^{-1} (\laph\n)^{-1} \left(-\frac 12 (\ah\n \laph\n \Psi_7\nn)_{\ell= 1}\right)+\Ga_1\n\cdot\Ga_1\n+2r^{-1}\overline{\Psi_3\n \Psi_{10}\n}\n \\
&\quad +2r^{-1} (\ah\n)^{-1} (\laph\n)^{-1}\left(\Kk+\mathcal P_1\left( \d_1\n\d_2\n(\ah\n\Psi_8\n)\right)_{\ell\leq 1} 
+(\Ga_1\n\cdot\Ga_2\n)_{\ell\leq 1}\right).
\end{split}
\eea
\begin{lemma}\lab{eq:existence-Psi-7-nn}
There exists a constant $C>0$ such that the equation \eqref{eq:Psi-7-nonlocal-long} has a unique solution $\Psi_7\nn$ verifying $||\Psi_7\nn||_s\leq C\eps$. More precisely, the solution satisfies
\bea\lab{eq:improved-bound-Psi-7-nn}
r^{-1} ||\Psi_7\nn||_{\H^{s+1}} \lesssim \eps r^{-2-\de},\quad r^{-1} ||(\Psi_7\nn)_{\ell=1}||_{\H^{s+1}} \lesssim C_b^2 \eps^2 r^{-3-2\de}.
\eea
\end{lemma}
\begin{proof}
This is a situation similar to, but much simpler than, the one we dealt with in Section \ref{section-nonlocal-eqn-Kc-kac-round}, and hence we only provide a sketch.\footnote{In particular, here we only have a single equation \eqref{eq:Psi-7-nonlocal-long} rather than a system, and the $\ell=1$ condition is zero at infinity, in contrast to the nonzero $\cc$ in Section \ref{section-nonlocal-eqn-Kc-kac-round}.}
Applying Lemma \ref{lemma:estimate-inverse-Laplacian}, 
we can write equation \eqref{eq:Psi-7-nonlocal-long} in the form
\beaa
(\pa_r+r^{-1}) \Psi_7\nn &=& -r^{-1} (\Psi_7\nn)_{\ell=1} + \RR(r^{-1}  \Psi_7\nn)+\Ga_1\n\cdot\Ga_1\n+r^{-1} \Ga_0\n\cdot \Ga_1\n \\
& & +2r^{-1} (\ah\n)^{-1} (\laph\n)^{-1}\left(\Kk+\mathcal P_1\left( \d_1\n\d_2\n(\ah\n\Psi_8\n)\right)_{\ell\leq 1} +(\Ga_1\n\cdot\Ga_2\n)_{\ell\leq 1}\right),
\eeaa
for some error linear operator $\RR$ that has similar properties\footnote{More precisely, the bound in Proposition \ref{prop:bound-RR}. At a heuristic level, $\RR$ provides an additional $\eps r^{-1-\de}$ factor.} as the $\RR$ 
introduced in Definition \ref{Def:RR-operator}.
Alternatively, the equation can be written as
\beaa
(\pa_r+2r^{-1}) \Psi_7\nn &=& r^{-1} (\Psi_7\nn)_{\ell\neq 1} + \RR(r^{-1}  \Psi_7\nn)+\Ga_1\n\cdot\Ga_1\n+r^{-1} \Ga_0\n\cdot \Ga_1\n \\
& & +2r^{-1} (\ah\n)^{-1} (\laph\n)^{-1}\left(\Kk+\mathcal P_1\left( \d_1\n\d_2\n(\ah\n\Psi_8\n)\right)_{\ell\leq 1} 
+(\Ga_1\n\cdot\Ga_2\n)_{\ell\leq 1}\right).
\eeaa
In the latter form, the first term on the right is a positive term, i.e., a special case of the positive definite matrix studied in Section \ref{section-nonlocal-eqn-Kc-kac-round}, and hence can be neglected. We then repeat the contraction argument in Section \ref{sec:summing-up}, in an easier situation, to obtain the existence of $\Psi_7\nn$ in the space consistent with the estimate
\beaa
r^{-1} ||\Psi_7\nn||_{\H^{s+1}} \lesssim \eps r^{-2-\de}.
\eeaa
We then project the equation to $\ell=1$ to obtain an improved estimate for $\ell=1$. The main reason for the improvement is that the free scalar $\Kk$, while only decaying at the rate $r^{-4-\de}$, is not supported on $\ell=1$. Therefore, the $\ell=1$ part of the right-hand side consists of only nonlinear terms. Since the existence of $\Psi_7\nn$ and its $\H^{s+1}$ bound have been obtained, such an improved estimate for $(\Psi_7\nn)_{\ell=1}$ is straightforward using the bound for the error operator $\RR$. 
\end{proof}
To conclude the proof of Proposition \ref{prop:step-4}, we apply the bound \eqref{eq:improved-bound-Psi-7-nn} for $\Psi_7\nn$ we just obtained to \eqref{eq:expression-Psi-10} and derive the estimate for $\Psi_{10}\nn$:
\beaa
r^{-1} ||\Psi_{10}\nn||_{\H^{s+1}} &\lesssim & r^{-1}\cdot r^2 ||\Kk||_{\H^{s-1}}+ r^{-1}||\Psi_{7}\nn||_{\H^{s+1}}+r^{-1}\cdot r^2 ||\big( \d_1\n\d_2\n(\ah\n\Psi_8\n)\big)_{\ell\leq 1}||_{\H^{s+1}}\\
& & +r^2 |(\Ga_1\n\cdot\Ga_2\n)_{\ell\leq 1}|+|\Ga_0\n\cdot\Ga_1\n| \\
&\lesssim & \eps r^{-2-\de}+C_b^2 \eps^2 r^{-3-2\de} \lesssim \eps r^{-2-\de}.
\eeaa
Note that we used that the term $\big( \d_1\n\d_2\n(\ah\n\Psi_8\n)\big)_{\ell\leq 1}$, in view of \eqref{eq:estimate-d1-d2-ell-leq-1}, is in fact nonlinear.

\subsubsection{Proof of Proposition \ref{prop:step-5}}\lab{proof:step-5}
We recall the equations \eqref{eq:iteration-Thh} and \eqref{eq:iteration-Xi}:
\beaa
\d_1\n \d_2\n \left(\ah\n \Psi_8\nn\right) &=& \frac 12 \left(\ah\n\laph\n \Psi_7\nn,0\right)+(\Kk, -\Kkd)+\Psi_{12}\nn+\Ga_1\n\cdot \Ga_2\n, \\
\d_1\n \Psi_9\nn &=& \bigg(0, \frac{3}{4\pi}  r^{-4} \sum_i{\bf a}_i\om_i+r^{-4} \int_r^\infty r'^4 (\Kkd-\Kd_{\ell\leq 1}\nn)\, dr' \bigg ) \\
\nonumber & & -\bigg(0, \overline{\displaystyle  \frac{3}{4\pi}  r^{-4} \sum_i{\bf a}_i \om_i+r^{-4} \int_r^\infty {r'^4} (\displaystyle{\Kkd}-\Kd_{\ell\leq 1}\nn)\, dr'}\n\bigg). 
\eeaa
Since we have determined $\Psi_7\nn$, we can apply Corollary \ref{cor:solvability} to the first equation with $\mathring\eps=C_b \eps r^{-1-\de}$ to obtain
\beaa
|(\Psi_{12}\nn)_{\ell=0} | &\lesssim  &|\overline{(\Kk,-\Kkd)}\n| \lesssim ||(\Kk,-\Kkd)||_{L^\infty(S_r)} ||\ga\n-\gz||_{L^\infty(S_r)} \lesssim C_b \eps^2 r^{-5-2\de}, 
\eeaa
\beaa
&&|(\Psi_{12}\nn)_{\ell=1}+\frac 12 \Up^{-\frac 12}(\laph\n \Psi_7\nn)_{\ell=1}| \\
&\lesssim & C_b \eps r^{-1-\de}\cdot r^{-1}||(\Kk,-\Kkd)+\frac 12 \Psi_3\n (\laph\n \Psi_7\nn,0)+\Ga_1\n\cdot \Ga_2\n ||_{L^2(S_r,\ga\n)} \\ 
&\lesssim & C_b \eps^2 r^{-5-2\de},
\eeaa
and the second estimate implies, in view of the improved $\ell=1$ bound for $\Psi_7\nn$ obtained in Lemma \ref{eq:existence-Psi-7-nn}, 
\beaa
|(\Psi_{12}\nn)_{\ell=1}|\lesssim C_b^2 \eps^2 r^{-5-2\de}+C_b \eps^2 r^{-5-2\de}\lesssim \eps r^{-5-2\de}.
\eeaa
 Corollary \ref{cor:solvability} then, in addition, implies that the solution $\Psi_8\nn$ to \eqref{eq:iteration-Thh} exists and, in view of Lemma \ref{lemma:Hodge-estimate-round} applied to $\d_1\d_2$,
	\beaa
		r^{-1} ||\Psi_8\nn||_{\H^{s+1}(S_r)}\lesssim r^{-1} \cdot r^2 ||(\Kk,-\Kkd)||_{\H^{s-1}(S_r)} +r^{-1} || \Psi_7\nn ||_{\H^{s+1}(S_r)} \lesssim \eps r^{-2-\de}.
	\eeaa 	
This proves the estimate for $\Psi_8\nn$. In view of the assumption on $\Kkd$ in \eqref{eq:main-thm-time-symmetric-B-Bd-bounds}, we have
\beaa
r^{-1} \big \| r^{-4} \int_r^\infty r'^4 (\Kkd) dr' \big\|_{\H^{s}} \lesssim \eps r^{-4-\de},
\eeaa
and hence we obtain, by the Hodge estimate in Lemma \ref{lemma:Hodge-estimate-round} to \eqref{eq:iteration-Xi},
\beaa
r^{-1} ||\Psi_9\nn||_{\H^{s+1}}\lesssim |{\bf a}| r^{-3}+\eps r^{-3-\de} \lesssim \eps r^{-2-\de}.
\eeaa

\subsubsection{Proof of Proposition \ref{prop:step-3}}\lab{proof:step-3}
We now derive the estimate for the spherical metric $\ga\nn$. Since $\slashed{\Lie}_{\pa_r} (r^{-2}\gz)=0$, the left-hand side of \eqref{eq:iteration-ga-metric} can be rewritten as $\slashed{\Lie}_{\pa_r} (r^{-2}\ga\nn-r^{-2}\gz)$. Then, using \eqref{eq:projected-Lie-equals-covariant}, the equation \eqref{eq:iteration-ga-metric} is equivalent to
\beaa
\nabz_{\pa_r} (\ga\nn-\gz) &=2\ah\n\Psi_4\nn+\ah\n\Psi_1\nn \ga\n+2\Up^\frac 12 \Psi_3\nn r^{-1} \ga\n,
\eeaa
where we recall the notations
\beaa
\ah\n=\Up^{-\frac 12}+\ao\n=\Up^{-\frac 12}+\Psi_3\n,\quad \thc\n=\Psi_1\n,\quad \thh\n=\Psi_4\n.
\eeaa
Since we seek solution with $||(\Psi\nn,\ga\nn)||_s<\infty$, we have $\lim_{r\to \infty} r^{-1} ||r (\ga\nn-\gz)||_{\H^{s+1}(S_r)}=0$. This is already stronger than what we need for applying Lemma \ref{lem:transport-lemma} with $\la=0$, and hence, using the improved $\H^{s+1}$ bounds for $\Psi_1\nn$, $\Psi_3\nn$, $\Psi_4\nn$ obtained in previous steps, 
we obtain
\beaa
r^{-1} ||\ga\nn-\gz||_{\H^{s+1}(S_r)} \lesssim \eps r^{-1-\de}.
\eeaa
This proves Proposition \ref{prop:step-3}.

\subsection{Contraction estimates}\lab{subsect:contraction-details}
We use the notation $\de \psi\nn:=\psi\nn-\psi\n$ for a general quantity $\psi$. 
We aim to show the contraction estimate $||\de(\Psi\nnn,\ga\nnn)||_s \leq C ||\de(\Psi\nn,\ga\nn)||_s$ for some positive constant $C<1$. Note again that here we define $||\cdot||_s$ as in \eqref{eq:def-norm-Psi-n} but with $\cc_i$ and $\gz$ removed.

\subsubsection{The main part}\lab{sec:contraction-main-part}
We first analyze the main part regarding $(\de\Psi_1\nnn, \de\Psi_2\nnn, \de\Psi_3\nnn)$. 
\begin{proposition}\lab{prop:contraction-main}
The quantities $\de\Psi_1\nnn$, $\de\Psi_2\nnn$, $\de\Psi_3\nnn$ satisfy the following system
\bea
\lab{eq:de-thc-nn}
(\pa_r+3r^{-1})(\de\Psi_1\nnn) &=& \Up^{-\frac 12}(\de\Psi_2\nnn-\Up^\frac 12 r^{-1}\de\Psi_1\nnn)_{\ell=0}-2(1-3mr^{-1}) r^{-2} (\de\Psi_3\nnn)\\
\nonumber & &+\NN[\de\Psi_1],\\
\lab{eq:de-Kc-nn}
(\pa_r+3r^{-1}) (\de\Psi_2\nnn) &=&  r^{-1} (\de\Psi_2\nnn-\Up^\frac 12 r^{-1}\de\Psi_1\nnn)_{\ell=0} -2\Up^\frac 12 r^{-3}(\de\Psi_3\nnn) \\
\nonumber & & -\frac 12 \Up^{-\frac 12}(\laph\nn \de\Psi_1\nnn)_{\ell= 1}+\NN[\de\Psi_2],\\
\Up^\frac 12 \laph\nn (\de\Psi_3\nnn) &=& (\de\Psi_2\nnn)-(\overline{\de\Psi_2\nnn}^{(n+1)})-\Up^\frac 12 r^{-1} (\de\Psi_1\nnn-\overline{\de\Psi_1\nnn}\nn)\\
\nonumber & & +\NN[\de\Psi_3], \\
\overline{\de\Psi_3\nnn}\nn &=&  -\frac 12 \Up^{-1} r \overline{\de\Psi_1\nnn}\nn+\NN_{av}[\de\Psi_3],
\eea
where the remainders satisfy the bounds
\beaa
r^{-1} ||\NN[\de\Psi_1]||_{\H^{s+1}(S_r)} &\lesssim & \eps r^{-4-2\de} ||\de(\Psi\nn,\ga\nn)||_s,\\
r^{-1}  ||\NN[\de\Psi_3], r\NN[\de\Psi_2]||_{\H^{s}(S_r)} &\lesssim & \eps r^{-4-2\de} ||\de(\Psi\nn,\ga\nn)||_s,\\
r^{-1} ||\NN_{av}[\de\Psi_3]||_{\H^{s+1}(S_r)} &\lesssim & \eps r^{-2-2\de} ||\de(\Psi\nn,\ga\nn)||_s.
\eeaa
\end{proposition}
\begin{proof}
See Appendix \ref{proof:contraction-main}.
\end{proof}

We can then write
\begin{equation}\lab{eq:expression-de-Psi-3}
\de\Psi_3\nnn=\Up^{-\frac 12} (\laph\nn)^{-1} \left(\de\Psi_2\nnn-\Up^\frac 12 r^{-1} \de\Psi_1\nnn+\NN[\de\Psi_3] \right)   -\frac 12 \Up^{-1} r \overline{\de\Psi_1\nnn}\nn+\NN_{av}[\de\Psi_3].
\end{equation}
This reduces the system to the following one for $(\de\Psi_1\nnn,\de\Psi_2\nnn)$:
\beaa
(\pa_r+3r^{-1})(\de\Psi_1\nnn) &=& \Up^{-\frac 12}(\de\Psi_2\nnn-\Up^\frac 12 r^{-1}\de\Psi_1\nnn)_{\ell=0}+(1-3mr^{-1}) r^{-1} \Up^{-1} \overline{\de\Psi_1\nnn}\nn \\
& & -2(1-3mr^{-1}) r^{-2} \Up^{-\frac 12} (\laph\nn)^{-1} \left(\de\Psi_2\nnn-\Up^\frac 12 r^{-1} \de\Psi_1\nnn \right) \\
& &+\NN[\de\Psi_1]-2(1-3mr^{-1}) r^{-2}\left(\Up^{-\frac 12} (\laph\nn)^{-1} (\NN[\de\Psi_3])+\NN_{av}[\de\Psi_3])\right),\\
(\pa_r+3r^{-1}) (\de\Psi_2\nnn) &=&  r^{-1} (\de\Psi_2\nnn-\Up^\frac 12 r^{-1}\de\Psi_1\nnn)_{\ell=0} +\Up^{-\frac 12} r^{-2} \overline{\de\Psi_1\nnn}\nn \\
&& -2 r^{-3} (\laph\nn)^{-1} \left(\de\Psi_2\nnn-\Up^\frac 12 r^{-1} \de\Psi_1\nnn \right)  \\
 & & -\frac 12 \Up^{-\frac 12}(\laph\nn \de\Psi_1\nnn)_{\ell= 1}+\NN[\de\Psi_2]-2r^{-3} (\laph\nn)^{-1}(\NN[\de\Psi_3])\\
 & &  -2\Up^\frac 12 r^{-3} \NN_{av}[\de\Psi_3],
\eeaa
which is already of the form \eqref{eqns:nonlocal-eqn-Kc-kac}. Moreover, we have the bounds
\bea
\nonumber && r^{-1} ||\NN[\de\Psi_1]-2(1-3mr^{-1}) r^{-2}\left(\Up^{-\frac 12} (\laph\nn)^{-1} (\NN[\de\Psi_3])+\NN_{av}[\de\Psi_3])\right)||_{\H^{s+1}} \\
\lab{eq:bounds-contraction-coupled-1}
&\lesssim & \eps r^{-4-2\de} ||\de(\Psi\nn,\ga\nn) ||_s,
\eea
and
\bea
\nonumber && r^{-1} || \NN[\de\Psi_2]-2r^{-3} (\laph\nn)^{-1}(\NN[\de\Psi_3]) -2\Up^\frac 12 r^{-3} \NN_{av}[\de\Psi_3] ||_{\H^s}\\
\lab{eq:bounds-contraction-coupled-2}
 &\lesssim & \eps r^{-5-2\de} ||\de(\Psi\nn,\ga\nn) ||_s.
\eea
Since $\ga\nn$ satisfies the first condition for $\tilde\ga$ in \eqref{eq:bound-condition-nonlocal-metric}, and $(\de\Psi_1\nnn,\de\Psi_2\nnn)$ satisfies the condition \eqref{eq:solve-nonlocal-vanishing-condition} with $\cc_i$ replaced by $0$ in view of the boundedness result, applying Proposition \ref{prop:solve-unknowns-kac-Kc} to this system, we see that $(\de\Psi_1\nnn,\de\Psi_2\nnn)$ must coincide with the unique solution given by the proposition. Moreover, in terms of the resulting estimates for $(\de\Psi_1\nnn,\de\Psi_2\nnn)$, since the estimates in \eqref{eq:bounds-contraction-coupled-1}, \eqref{eq:bounds-contraction-coupled-2} have for each an additional factor of $r$ decay compared with what is needed in Proposition \ref{prop:solve-unknowns-kac-Kc}, it is in fact obvious from the proof of Proposition \ref{prop:solve-unknowns-kac-Kc} that one obtains a corresponding improvement for the solution:\footnote{Or, instead, one could stay content with the improvement for the $\ell=1$ part, which is also enough for the contraction estimates.}
\bea\lab{eq:de-Psi-1-2-improved-s}
r^{-1} ||r^{3+\de} \de\Psi_1\nnn,r^{4+\de} \de\Psi_2\nnn||_{\H^s(S_r)} \lesssim \eps ||\de (\Psi\nn,\ga\nn)||_s.
\eea
Plugging back to \eqref{eq:expression-de-Psi-3}, we also deduce
\bea\lab{eq:de-Psi-3-improved}
r^{-1} ||r^{2+\de} \de\Psi_3\nn ||_{\H^{s+2}(S_r)}\lesssim \eps ||\de(\Psi\nn,\ga\nn) ||_s.
\eea
We also obtain, similar to \eqref{eq:highest-order-Psi-1}, the $\H^{s+1}$ estimate for $\de\Psi_1\nn$ by commuting the equation \eqref{eq:de-thc-nn} with $r(\nabz)^{s+1}$:
\bea\lab{eq:de-Psi-1-improved-s+1}
r^{-1} ||r^{3+\de} \de\Psi_1\nnn||_{\H^{s+1}(S_r)} \lesssim \eps ||\de (\Psi\nn,\ga\nn)||_s.
\eea

\subsubsection{The remaining spatial part}
\begin{proposition}\lab{prop:contraction-remain}
The quantities $\de \Psi_4\nnn$, $\de \Psi_5\nnn$, $\de \Psi_6\nnn$ satisfy the following system
\bea
\lab{eq:de-Psi-4-nn}
\d_1\nn\d_2\nn \de\Psi_4\nnn &=& \frac 12 (\laph\nn \de\Psi_1\nnn,0)-\de\Psi_{11}\nnn +\NN[\de\Psi_4], \\
\lab{eq:de-Psi-5-nn}
\d_1\nn \de \Psi_5\nnn &=& -(\Up^\frac 12 \laph\nn \de\Psi_3\nnn,0 )+\NN[\de\Psi_5], \\
\lab{eq:de-Psi-6-nn}
\d_1\nn \de \Psi_6\nnn &=&  \de \Psi_{11}\nnn-\overline{\de \Psi_{11}\nnn}\nn +\NN[\de\Psi_6].
\eea
where the following bounds hold
\beaa
r^{-1} ||\NN[\de\Psi_4], r^{-1}\NN[\de\Psi_5], \NN[\de\Psi_6] ||_{\H^s} \lesssim \eps r^{-5-2\de} ||\de(\Psi\nn,\ga\nn)||_s.
\eeaa
\end{proposition}
\begin{proof}
See Appendix \ref{proof:contraction-remain}.
\end{proof}

We note that the horizontal tensor $\de\Psi_4\nnn=\Psi_4\nnn-\Psi_4\nn$ is not strictly traceless with respect to $\ga\nn$. We can rewrite \eqref{eq:de-Psi-4-nn} as
\bea\lab{eq:contraction-Codazzi}
\bsplit
& \d_1\nn\d_2\nn \Big(\de\Psi_4\nnn -\frac 12 (\trh\nn \de\Psi_4\nnn) \ga\nn\Big) \\
=&\;  \frac 12 (\laph\nn (\trh\nn \Psi_4\nn),0)+\frac 12 \laph\nn (\de\Psi_1\nnn,0)
-\de\Psi_{11}\nnn+\NN[\de\Psi_4],
\end{split}
\eea
where we used that
\beaa
\d_1\nn\d_2\nn ((\trh\nn \de\Psi_4\nnn)\ga\nn)=(\laph\nn (\trh\nn \de\Psi_4\nnn),0)=- (\laph\nn (\trh\nn \Psi_4\nn),0).
\eeaa
We then also have $\trh\nn \Psi_4\nn=O(\de\ga\nn\cdot \Psi_4\nn)$ since $\Psi_4\nn$ is traceless with respect to $\ga\n$.
The fact that the solution exists for \eqref{eq:contraction-Codazzi} implies, using Corollary \ref{cor:solvability}, that the $\ell\leq 1$ coefficients $\de\Psi_{11}\nnn$ satisfy the estimate 
\beaa
|\de\Psi_{11}\nnn | &\lesssim & r^{-1} ||\laph\nn (\de\ga\nn\cdot \Psi_4\nn)||_{L^2(S_r,\ga\nn)}+r^{-1} ||\laph\nn (\de\Psi_1\nnn)||_{L^2(S_r,\ga\nn)}\\
&& + r^{-1} ||\NN[\de\Psi_4]||_{L^2(S_r,\ga\nn)} \\
&\lesssim & r^{-1} ||\laph\nn (\de\ga\nn\cdot \Psi_4\nn)||_{L^2(S_r)} +  r^{-1} ||\laph\nn(\de\Psi_1\nnn)||_{L^2(S_r)}+ r^{-1} ||\NN[\de\Psi_4]||_{L^2(S_r)} \\
&\lesssim & \eps r^{-5-\de} ||\de(\Psi\nn,\ga\nn)||_s,
\eeaa
where we used the estimate for $\de\Psi_1\nn$ obtained in \eqref{eq:de-Psi-1-2-improved-s}.

To estimate $\de\Psi_4\nnn$, we now apply the Hodge estimate \eqref{eq:Hodge-estimate-round} to \eqref{eq:contraction-Codazzi} and obtain
\beaa
& & r^{-1} ||\de\Psi_4\nnn - \frac 12 (\trh\nn \de\Psi_4\nnn) \ga\nn||_{\H^{s+1}}  \lesssim  ||r^2 \laph\nn (\de\ga\nn\cdot \Psi_4\nn)||_{\H^{s-1}}  \\
& & + r^{-1} ||r^2 \laph\nn (\de\Psi_1\nnn)||_{\H^{s-1}}+ r^2 |\de\Psi_{11}\nnn|+r^{-1} ||r^2\NN[\de\Psi_4]||_{\H^{s-1}} \\
&\lesssim & \eps r^{-3-\de} ||\de(\Psi\nn,\ga\nn)||_s.
\eeaa
This implies 
\beaa
r^{-1} ||\de\Psi_4\nnn ||_{\H^{s+1}} \lesssim \eps r^{-3-\de} ||\de(\Psi\nn,\ga\nn)||_s.
\eeaa
We then apply the Hodge estimates \eqref{eq:Hodge-estimate-round} to  \eqref{eq:de-Psi-5-nn} and \eqref{eq:de-Psi-6-nn} to obtain, using the improved estimates for $\de\Psi_3\nnn$ and $\de\Psi_{11}\nnn$,
\beaa
r^{-1} || \de \Psi_5\nnn||_{\H^{s+1}} &\lesssim & r^{-1} ||r \laph\nn \de\Psi_3\nnn||_{\H^s} +r^{-1} ||r\NN[\de\Psi_5]||_{\H^s} \lesssim   \eps r^{-3-\de} ||\de(\Psi\nn,\ga\nn)||_s, \\
r^{-1} || \de \Psi_6\nnn||_{\H^{s}} &\lesssim & r |\de \Psi_{11}\nnn|+ r^{-1} ||r\NN[\de\Psi_6]||_{\H^{s-1}} \lesssim \eps r^{-4-\de} ||\de(\Psi\nn,\ga\nn)||_s.
\eeaa

\subsubsection{The $k$ part}
\begin{proposition}\lab{prop:contraction-k-part}
The quantities $\de \Psi_7\nnn$, $\de \Psi_8\nnn$, $\de \Psi_9\nnn$, $\de\Psi_{10}\nnn$ satisfy the following system
\bea
\lab{eq:de-Psi-7-nn}
(\pa_r+r^{-1}) \de\Psi_7\nnn &=& 2r^{-1} \de\Psi_{10}\nnn +\NN[\de\Psi_7],\\
\lab{eq:de-Psi-8-nn}
\d_1\nn \d_2\nn \left(\de (\ah\nn\Psi_8\nnn)\right) &=& \frac 12 (\ah\nn\laph\nn \de\Psi_7\nnn,0) +\de\Psi_{12}\nnn +\NN[\de\Psi_8],\\
\lab{eq:de-Psi-9-nn}
\d_1\nn \de\Psi_9\nnn &=& -\bigg(0,r^{-4} \int_r^\infty r'^4 \de\Kd_{\ell\leq 1}\nnn\, dr'\bigg)\\
\nonumber & & +\bigg(0,\overline{\displaystyle r^{-4} \int_r^\infty r'^4 (\de\Kd_{\ell\leq 1}\nnn)\, dr'}\nn\bigg)+\NN[\de\Psi_9], \\
\lab{eq:de-Psi-10-nn}
 \laph\nn \left(\de(\ah\nn\Psi_{10}\nnn)\right) &=& -\frac 12 \left(\ah\nn \laph\nn \de\Psi_7\nnn\right)_{\ell=1} \\
 \nonumber &&+\frac 12 \overline{\left(\ah\nn\laph\nn \de\Psi_7\nnn\right)_{\ell=1}}\nn+\NN[\de\Psi_{10}], \\
 \lab{eq:de-Psi-10-av-nn}
\overline{\displaystyle \de(\ah\nn\Psi_{10}\nnn)}\nn &=& \NN_{av}[\de\Psi_{10}],
\eea
\end{proposition}
where the following bounds hold
\beaa
r^{-1} ||\NN[\de\Psi_7] ||_{\H^{s+1}} &\lesssim & \eps r^{-4-2\de} ||\de(\Psi\nn,\ga\nn)||_s, \\
r^{-1} ||\NN[\de\Psi_8] ||_{\H^{s-1}} &\lesssim & \eps r^{-5-2\de} ||\de(\Psi\nn,\ga\nn)||_s,\\
r^{-1} ||\NN[\de\Psi_9] ||_{\H^s} &\lesssim & \eps r^{-4-2\de} ||\de(\Psi\nn,\ga\nn)||_s, \\
r^{-1} ||\NN[\de\Psi_{10}]||_{\H^{s-1}} &\lesssim & \eps r^{-5-2\de} ||\de(\Psi\nn,\ga\nn)||_s, \\
r^{-1} ||\NN_{av} [\de\Psi_{10}]||_{\H^{s+1}} &\lesssim & \eps r^{-3-2\de} ||\de(\Psi\nn,\ga\nn)||_s.
\eeaa
\begin{proof}
See Appendix \ref{proof:contraction-k-part}.
\end{proof}
The equations \eqref{eq:de-Psi-10-nn}, \eqref{eq:de-Psi-10-av-nn} implies
\beaa
\de(\ah\nn\Psi_{10}\nnn)=(\laph\nn)^{-1} \left(-\frac 12 \left(\ah\nn \laph\nn \de\Psi_7\nnn\right)_{\ell=1}+\NN[\de\Psi_{10}]\right)+\NN_{av}[\de\Psi_{10}].
\eeaa
Note that $\de(\ah\nn\Psi_{10}\nnn)=\ah\nn \de \Psi_{10}\nnn+\de\ah\nn \Psi_{10}\nn$. Therefore, we obtain an expression of $\de\Psi_{10}\nnn$. Plugging it into \eqref{eq:de-Psi-7-nn}, we derive the equation
\beaa
(\pa_r+r^{-1}) \de\Psi_7\nnn &=& 2r^{-1} (\ah\nn)^{-1} (\laph\nn)^{-1} \left(-\frac 12 \left(\ah\nn \laph\nn \de\Psi_7\nnn\right)_{\ell=1}+\NN[\de\Psi_{10}]\right)\\
& & +2r^{-1} (\ah\nn)^{-1}\left(\NN_{av}[\de\Psi_{10}]-\de\ah\nn \Psi_{10}\nn\right) +\NN[\de\Psi_7]
\eeaa
Recall that the boundedness result implies $r^{-1} ||\de\Psi_7\nnn||_{\H^{s+1}}\lesssim \eps r^{-2-\de}$. This provides the vanishing condition we need, and we can proceed as in Lemma \ref{eq:existence-Psi-7-nn} to obtain
\bea\lab{eq:contraction-estimate-Psi-7-nnn}
r^{-1} ||\de\Psi_7\nnn||_{\H^{s+1}} \lesssim \eps r^{-3-\de} ||\de(\Psi\nn,\ga\nn)||_s,
\eea
where we note that compared with the first estimate in \eqref{eq:improved-bound-Psi-7-nn}, the improvement on the decay rate arises from the fact that, unlike for the equation of $\Psi_7\nn$, here the leading contribution from the free scalar in $\Kk$ is cancelled.

Then, plugging back to the expression of $\de\Psi_{10}\nnn$, we obtain
\beaa
r^{-1} ||\de(\ah\nn\Psi_{10}\nnn)||_{\H^{s+1}} &\lesssim & \eps r^{-2-\de} ||\de(\Psi\nn,\ga\nn)||_s.
\eeaa
We now analyze the equation of $\de(\ah\nn\Psi_8\nnn)$. As in \eqref{eq:contraction-Codazzi}, since $\de(\ah\nn\Psi_8\nnn)$ is not necessarily traceless with respect to $\ga\nn$, we write
\bea\lab{eq:contraction-Codazzi-k}
\bsplit
& \d_1\nn\d_2\nn \Big(\de(\ah\nn\Psi_8\nnn) -\frac 12 (\trh\nn \de(\ah\nn\Psi_8\nnn)) \ga\nn\Big) \\
=&\;  \frac 12 (\laph\nn (\trh\nn (\ah\n\Psi_8\nn)),0)+\frac 12 (\laph\nn \de\Psi_7\nnn,0)
-\de\Psi_{12}\nnn+\NN[\de\Psi_8],
\end{split}
\eea
and the first term on the right can be further written in the form $O(\de\ga\nn\cdot \ah\n\Psi_8\nn)$, using that $\Psi_8\nn$ is traceless with respect to $\ga\n$.
Then, the fact that the solution exists for \eqref{eq:contraction-Codazzi-k} implies the following estimate for $\de\Psi_{12}\nnn$, using Corollary \ref{cor:solvability}:
\beaa
|\de\Psi_{12}\nnn | &\lesssim & r^{-1} ||\laph\nn (\de\ga\nn\cdot \ah\n\Psi_8\nn)||_{L^2(S,\ga\nn)}+r^{-1} ||\laph\nn (\de\Psi_7\nnn)||_{L^2(S,\ga\nn)}\\
&& + r^{-1} ||\NN[\de\Psi_8]||_{L^2(S,\ga\nn)} \\
&\lesssim & r^{-1} ||\laph\nn (\de\ga\nn\cdot \ah\n\Psi_8\nn)||_{L^2} +  r^{-1} ||\de\Psi_7\nnn||_{L^2}+ r^{-1} ||\NN[\de\Psi_8]||_{L^2} \\
&\lesssim & \eps r^{-5-\de} ||\de(\Psi\nn,\ga\nn)||_s,
\eeaa
where we used the estimate for $\de\Psi_7\nn$ we just obtained in \eqref{eq:contraction-estimate-Psi-7-nnn}.

We now apply the Hodge estimate \eqref{eq:Hodge-estimate-round} to \eqref{eq:contraction-Codazzi} and obtain
\beaa
& & r^{-1} ||\de(\ah\nn\Psi_8\nnn) -\frac 12 (\trh\nn \de(\ah\nn\Psi_8\nnn)) \ga\nn||_{\H^{s+1}} \\
 &\lesssim &  ||r^2 \laph\nn (\de\ga\nn\cdot \ah\n\Psi_8\nn)||_{\H^{s-1}}  \\
& & + r^{-1} ||r^2 \laph\nn (\de\Psi_7\nnn)||_{\H^{s-1}}+ r^2 |\de\Psi_{12}\nnn|+r^{-1} ||r^2\NN[\de\Psi_8]||_{\H^{s-1}} \\
&\lesssim & \eps r^{-3-\de} ||\de(\Psi\nn,\ga\nn)||_s.
\eeaa
This yields the estimate
\beaa
r^{-1} ||\de \Psi_8\nnn ||_{\H^{s+1}} \lesssim \eps r^{-3-\de} ||\de(\Psi\nn,\ga\nn)||_s.
\eeaa
To conclude, we apply the Hodge estimate \eqref{eq:Hodge-estimate-round} to \eqref{eq:de-Psi-9-nn} to deduce
\beaa
r^{-1} ||\de\Psi_9\nnn||_{\H^{s+1}} &\lesssim & r\cdot r^{-4}\int_r^\infty r'^4 |\de\Psi_{12}\nnn|\, dr'+\eps r^{-3-\de} ||\de(\Psi\nn,\ga\nn)||_s \\
&\lesssim & \eps r^{-3-\de} ||\de(\Psi\nn,\ga\nn)||_s.
\eeaa

\subsubsection{The horizontal metric}
Finally, we derive the equation of $\de \ga\nnn$ using \eqref{eq:iteration-ga-metric}
\beaa
\slashed{\Lie}_{\pa_r} (r^{-2}\de\ga\nnn) &=& 2r^{-2}(\ah\nn)^{-1}\de\Psi_4\nnn+(\ah\nn)^{-1}\de\Psi_1\nnn(r^{-2}\ga\nn)\\
& &+2\Up^\frac 12 \de\Psi_3\nnn r^{-1} (r^{-2}\ga\nn)+\NN[\de\ga],
\eeaa
where
\beaa
\NN[\de\ga] &:= & 2r^{-2} \left((\ah\nn)^{-1}-(\ah\n)^{-1}\right) \Psi_4\nn + r^{-2} \left((\ah\nn)^{-1} \ga\nn-(\ah\n)^{-1} \ga\n\right) \Psi_1\nn\\
& &+2 r^{-3} \Up^\frac 12 (\ga\nn-\ga\n) \Psi_3\nn.
\eeaa
Using \eqref{eq:projected-Lie-equals-covariant}, the equation is equivalent to
\beaa
\nabz_{\pa_r} (\de\ga\nnn) &=& 2(\ah\nn)^{-1}\de\Psi_4\nnn+(\ah\nn)^{-1}\de\Psi_1\nnn \ga\nn +2\Up^\frac 12 \de\Psi_3\nnn r^{-1} \ga\nn+r^2 \NN[\de\ga].
\eeaa
We omit the estimate of $\NN[\de\ga]$ since it contains additional small and decaying factors.
Integrating in the $r$-direction from infinity using Lemma \ref{lem:transport-lemma}, we obtain
\beaa
 r^{-1} ||\de \ga\nnn ||_{\H^{s+1}(S_r)} 
&\lesssim & \int_r^\infty r'^{-1} ||\de\Psi_4\nnn||_{\H^{s+1}(S_{r'})} + r'^{-1} ||\de\Psi_1\nnn||_{\H^{s+1}(S_{r'})}\\
& &+r'^{-2} ||\de\Psi_3\nnn||_{\H^{s+1}(S_{r'})} +r'^{-1} ||r^2\NN[\de\ga]||_{\H^{s+1}(S_{r'})} dr' \\ 
 &\lesssim & \eps r^{-2-\de} ||\de(\Psi\nn,\ga\nn)||_s.
\eeaa

\subsection{The limit $(g\i,k\i)$}
\subsubsection{Proof of Lemma \ref{eq:lemma-limit-properties}}\lab{sect:proof-lemma-limit-properties}
According  to the lemma, we need to verify the following statements:
\begin{enumerate}
\item The horizontal tensors $\Psi_4\i$ and $\Psi_8\i$ are traceless with respect to $\ga\i$.
\item With respect to the metric $g\i:=(\ah\i)^2 dr^2+\ga\i$, see \eqref{eq:def-g-limit},
 the quantities $\Psi_4\i$ and $\Psi_1\i+2\Up^\frac 12 r^{-1}$ are exactly the traceless part and the trace of the second fundamental form of the $r$-spheres. We hence denote $\thh\i=\Psi_4\i$ and $\trth\i=\Psi_1\i+2\Up^\frac 12 r^{-1}$ without ambiguity.
\item We have $\Psi_5\i=-\nabh\i (\log\ah\i)$, $\mu_{\ell\geq 1}\i=0$, and the average of $\Psi_3\i+\frac 12 \Up^{-1} r \Psi_1\i$ vanishes with respect to $\ga\i$.
\item For the quantity defined in \eqref{eq:limit-tilde-B}, we have $\widetilde\B_{\ell\leq 1} \i =\B_{\ell\leq 1} \i$, the latter being the first component of $\Psi_{11}\i\in \sk_0$.
\item Denote by $Y(g\i)$ the horizontal tensor $Y$ with respect to $g\i$ defined through \eqref{eq:def-curvature-components-Si}. Then we have $\Psi_6\i=Y(g\i)$.
\item Denote the Gauss curvature of $\ga\i$ by $K(\ga\i)$. Then we have $\Psi_2\i=K(\ga\i)-r^{-2}$.
\end{enumerate}

Recall  that, see equation \eqref{eq:def-Psi-1-12},
\beaa
\bsplit
\Psi_1\n &=\thc\n,\quad \Psi_2\n=\Kc\n,\quad \Psi_3\n=\ao\n,\quad \Psi_4\n=\thh\n,\quad \Psi_5\n=\pp\n,\quad \Psi_6\n=Y\n, \\
\Psi_7\n &=\trt\n,\quad \Psi_8\n=\kh\n,\quad \Psi_9\n = \Xi\n,\quad \Psi_{10}\n=\Pi\n,\\  \Psi_{11}\n &=(\B_{\ell\leq 1}\n,\Bd_{\ell\leq 1}\n),\quad \Psi_{12}\n=(\K_{\ell\leq 1}\n,\Kd_{\ell\leq 1}\n).
\end{split}
\eeaa
\begin{proof} We proceed as follows:

 {\bf $1^{st}$  statement:}
Since, by construction, we have $(\ga\n)^{AB} (\Psi_4\nn)_{AB}=(\ga\n)^{AB} (\Psi_8\nn)_{AB}=0$, the first statement follows by taking the limit $n\to\infty$.

 {\bf $2^{nd}$  statement:}
Note that the equation \eqref{eq:metric-limit} implies the following reversed derivation of the identity \eqref{eq:identity-metric-kac-ah} in the proof of Proposition \ref{prop:linearized-eqns-time-symmetry}:
\bea
\bsplit
 \slashed{\Lie}_{\pa_r} (r^{-2}\ga\i) 
 &=  2r^{-2}\ah\i\Psi_4\i+\ah\i\Psi_1\i (r^{-2}\ga\i)+2\Up^\frac 12 \Psi_3\i r^{-1} (r^{-2}\ga\i)\\
 &= 2r^{-2}\ah\i\Psi_4\i+\ah\i\Psi_1\i (r^{-2}\ga\i)+2\Up^\frac 12 (\ah\i- \Up^{-\frac 12}) r^{-1} (r^{-2}\ga\i)\\
 &=2r^{-2}\ah\i \Psi_4\i+\ah\i (\Psi_1\i+2\Up^\frac 12 r^{-1}) (r^{-2}\ga\i) -2r^{-1}(r^{-2}\ga\i).
\end{split}
\eea
Therefore, using that $\slashed{\Lie}_{(\ah\i)^{-1}\pa_r} \ga\i=(\ah\i)^{-1} \slashed{\Lie}_{\pa_r} (r^{-2}\ga\i) $ by the form of $g\i$ and \eqref{eq:scalar-multiple-projected-Lie-derivatives}, we deduce
\beaa
\slashed{\Lie}_{(\ah\i)^{-1}\pa_r} \ga\i=2\Psi_4\i+(\Psi_1\i +2\Up^\frac 12 r^{-1}) \ga\i,
\eeaa
and the second statement follows.

 {\bf  $3^{rd}$  statement:}
To prove the third statement, note that using the precise structure of the nonlinear term, \eqref{eq:iteration-Delta-ah-ii} implies $\d_1\i \Psi_5\i=-\laph\i (\log\ah\i)$. Since $\nabh\i (\log\ah\i)=\slashed{d} \log(\ah\i)$ is, with respect to $\ga\i$, the only curl-free $1$-form whose divergence equals $\laph\i (\log\ah\i)$, we have $\Psi_5\i=-\nabh\i (\log\ah\i)$. 
Similarly, 
using the precise structure of the nonlinear term in \eqref{eq:iteration-divP-ii}, in particular Remark \ref{rem:eq-Psi-3-nonlinear-terms} and the fact that the $\laph\i(\Ga_0\i\cdot\Ga_0\i)$ term turns $\Up^\frac 12 \laph\i\ah\i$ to $\laph\i(\log\ah\i)$, we have
\beaa
\laph\i(\log\ah\i)&=& \Psi_2\i -\overline{\Psi_2\i}^{(\infty)}-\frac 14 (\trth\i)^2+\frac 14\overline{\displaystyle(\trth\i)^2}\i.
\eeaa
Therefore, we have
\beaa
\mu_{\ell\geq 1}\i=\Big(-\laph\i(\log\ah\i)+\Psi_2\i-\frac 14(\trth\i)^2\Big)_{\ell\geq 1}=0.
\eeaa
The relation $\overline{\Psi_3\i}\i=-\frac 12 \Up^{-1} r \overline{\Psi_1\i}\i$ is also justified by taking the limit of the equation \eqref{eq:iteration-average-a}.

 {\bf  $4^{th}$  statement:}  We have pointed out in \eqref{eq:limit-tilde-B} that
\beaa
\widetilde\B_{\ell\leq 1} \i &=& \widetilde\B_{\ell\leq 1,aux} \i =\frac 12 (\laph\i \thc\i)_{\ell\leq 1}-\left(\mathcal P_1(\d_1\i\d_2\i \thh\i)\right)_{\ell\leq 1} \\
&=& \left(\frac 12 \laph\i \thc\i-\divh\i\divh\i \thh\i\right)_{\ell\leq 1}.
\eeaa
Comparing this with \eqref{eq:iteration-h-dualh-ii} projected to $\ell\leq 1$ and using that $\Bb_{\ell\geq 1}=0$, we see that $\widetilde\B_{\ell\leq 1} \i =\B_{\ell\leq 1} \i$. This proves the fourth statement.

 {\bf  $5^{th}$  statement:} 
In view of  the  $2^{nd}$ statement,  we  have established that the limit  $\th\i=\thh\i+\frac 12 \trth\i \ga\i$ is  in fact the second fundamental form  of the 
 $r$-foliation with respect to the metric $g\i$. We can therefore make use  of the  unconditional  equation \eqref{eq:unconditional-Codazzi} of Proposition \ref{prop:Unconditional-equations-1}, according to which,
\bea\lab{eq:unconditional-limit-Codazzi}
\divh\i \thh\i &=&\frac 12 \nabh\i\trth\i-Y(g\i).
\eea
Taking $\d_1\i$ of \eqref{eq:unconditional-limit-Codazzi} and comparing it with \eqref{eq:iteration-h-dualh-ii}, we deduce that
\beaa
\d_1\i Y(g\i)=(\Bb,\Bbd)+(\B_{\ell\leq 1}\i, \Bd_{\ell\leq 1}\i).
\eeaa
Comparing it with \eqref{eq:iteration-Y-ii}, we have
\beaa
\d_1\i (\Psi_6\i-Y(g\i))=\overline{\displaystyle(\Bb,\Bbd)+(\B_{\ell\leq 1}\i, \Bd_{\ell\leq 1}\i)}\i.
\eeaa
Taking the spherical mean over $\ga\i$, we see that the right-hand side is in fact zero. Therefore, we obtain that $\Psi_6\i=Y(g\i)$ using the injectivity of $\d_1\i$. This proves the fifth statement.

 {\bf  $6^{th}$  statement:}  Using the previous statements, 
we appeal to the unconditional equation \eqref{eq:R-transport-Kc} for $\Kc(\ga\i):=K(\ga\i)-r^{-2}$, applied to the metric $g\i$,\footnote{The derivation of unconditional equation \eqref{eq:R-transport-Kc} is independent of $k\i$, see Remark \ref{rem:eq-Psi-2-nonlinear-terms}.} as
\begin{equation}\lab{eq:limit-K-linearized-unconditional}
\pa_r \Kc(\ga\i) = r^{-1} \muc(g\i)-\ah\i \divh\i Y\i-3r^{-1}\Kc(\ga\i) -2\Up^\frac 12 r^{-3} \ao\i+
\Ga_1(g\i) \cdot \Ga_2(g\i).
\end{equation}
Here 
\beaa
\muc(g\i):=-\laph\i (\log\ah\i)+K(\ga\i)-\frac 14(\trth\i)^2-2mr^{-3},
\eeaa
and, due to our previous statements as well as  Remark \ref{rem:eq-Psi-2-nonlinear-terms}, we have $\Ga_1(g\i)=\Ga_1\i$ and   $\Ga_2(g\i)=\Ga_2\i$, with the exception that, whenever $\Kc\i$ appears, it is replaced by $\Kc(\ga\i)$. 
On the other hand, since we have proved that $\mu_{\ell\geq 1}\i=0$, $\widetilde\B_{\ell\leq 1}\i=\B_{\ell\leq 1}\i$, and $\divh\i Y\i=\Bb+\B_{\ell\leq 1}\i$, the equation \eqref{eq:iteration-Kc-ii} reads
\bea\lab{eq:limit-Kc-comparison}
(\pa_r+3r^{-1}) \Psi_2\i &=& r^{-1}\muc\i -2\Up^\frac 12 r^{-3}\Psi_3\i -\ah\i \divh\i Y\i
 +\Ga_1\i\cdot \Ga_2\i.
\eea
Note that since they both originate from \eqref{eq:R-transport-Kc}, the schematic forms in \eqref{eq:limit-K-linearized-unconditional} and \eqref{eq:limit-Kc-comparison} have the same expression,\footnote{
The precise expressions of the schematic terms $\Ga_1\i\cdot \Ga_1\i$, $\Ga_1\i\cdot \Ga_2\i$, etc. can be tracked down from the corresponding terms in the derivation of the equations in  Proposition \ref{prop:linearized-eqns-time-symmetry}. 
} apart from the difference between $\Kc(\ga\i)$ and $\Kc\i$. Therefore, taking the difference between \eqref{eq:limit-Kc-comparison} and \eqref{eq:limit-K-linearized-unconditional}, we obtain
\beaa
\pa_r \left(\Kc\i-\Kc(\ga\i)\right) =-2r^{-1} \left(\Kc\i-\Kc(\ga\i)\right)+\Ga_1\i\cdot \left(\Kc\i-\Kc(\ga\i)\right).
\eeaa
We already know that $r^3\Kc\i=r^3\Psi_2\i\to 0$ by the boundedness of $\Psi\i$ in $||\cdot||_s$. Moreover, since $\ga\i$ and its $r\nabz$ derivatives decay at the rate $r^{-1-\de}$, we deduce that $\Kc(\ga\i)=K(\ga\i)-r^{-2}$ decays at the rate $r^{-3-\de}$, so in particular $\lim_{r\to\infty} r^{3} \Kc(\ga\i) \to 0$. 
This is stronger than the condition $\lim_{r\to\infty} r^2 (\Kc\i-\Kc(\ga\i))=0$ needed here, and hence, using that $\Ga_1\i=O(\eps r^{-2-\de})$, we integrate from infinity to obtain $\Kc\i=\Kc(\ga\i)$, i.e., $K\i=K(\ga\i)$.
\end{proof}

\subsubsection{Proof of Proposition \ref{prop:limit-solves-constraint}}\lab{subsect:proof-limit-solves-constraint}
The goal is to prove that $(g\i,k\i)$ solves the Einstein constraint equation \eqref{ece}, where $g\i$ and $k\i$ are defined respectively in \eqref{eq:def-g-limit} and \eqref{eq:def-k-infty}.

Throughout this proof, we use the shorthand notation
\beaa
\CC_{Ham}\i:=\CC_{Ham}(g\i,k\i),\quad \CC_{Mom}\i:=\CC_{Mom}(g\i,k\i),\quad \slashed\CC_{Mom}\i:=\slashed{\CC}_{Mom}(g\i,k\i).
\eeaa
The way of defining $k\i$ in \eqref{eq:def-k-infty} implies 
\beaa
\Psi_7\i=\trt(g\i,k\i),\quad \Psi_8\i = \kh(g\i,k\i),\quad \Psi_9 = \Xi(g\i,k\i),\quad \Psi_{10} = \Pi(g\i,k\i).
\eeaa 
Therefore we will denote them by $\trt\i$, $\kh\i$, $\Xi\i$, and $\Pi\i$ without ambiguity. Together with the statements in Lemma \ref{eq:lemma-limit-properties}, we see that now all quantities in $\Ga_1\i$ and $\Ga_2\i$ have no ambiguities.

Using $\mu_{\ell\geq 1}\i=0$, and $K\i=K(\ga\i)$ from Lemma \ref{eq:lemma-limit-properties}, the equation \eqref{eq:iteration-kac-ii} implies
\beaa
(\pa_r+2r^{-1})\thc\i &=& \ah\i\muc\i  -2(1-3mr^{-1}) r^{-2}\ao\i+\Ga_1\i\cdot \Ga_1\i,
\eeaa
We now apply the unconditional equation \eqref{eq:R-transport-kac}, noting that we have shown in Lemma \ref{eq:lemma-limit-properties} that $\trth\i$ is the $N\i$-expansion with respect to $g\i$ and $\muc(g\i)=\muc\i$, 
\beaa
(\pa_r+2r^{-1})\thc\i &=& \ah\i\muc\i -2(1-3mr^{-2}) r^{-2}\ao\i+\Ga_1\i\cdot \Ga_1\i-\frac 12 \ah\i \CC_{Ham}\i.
\eeaa
Comparing the two equations, using that the schematic terms in fact have the same algebraic expression, we deduce $\CC_{Ham}\i=0$.

We now prove that the momentum constraint also vanishes. The first component of the equation \eqref{eq:iteration-Xi-ii} reads
\beaa
\divh\i \Xi\i =0.
\eeaa
Using this, the unconditional equation \eqref{eq:N-trt-linearized} applied to $(g\i,k\i)$ reads
\bea
\lab{eq:N-trt-linearized-uncond}
\pa_r \trt\i &=& 
 2 r^{-1} \Pi\i  - r^{-1} \trt\i +\Ga_1\i\cdot \Ga_1\i - \ah\i (\CC_{Mom}\i)_{N\i},
\eea
Comparing \eqref{eq:N-trt-linearized-uncond} with \eqref{eq:iteration-trt-ii} and using that $\Ga_1\i\cdot\Ga_1\i$ in the two equations have the same algebraic expression, we obtain $(\CC_{Mom}\i)_{N\i}=0$. 

The unconditional equations \eqref{eq:N-div-Xi}-\eqref{eq:N-curl-Xi} read
\bea\lab{eq:N-div-Xi-uncond-i}
\bsplit
(\pa_r +4r^{-1}) \divh\i\Xi \i &= -\divh\i\divh\i (\ah\i \kh\i) +\frac 12 \ah\i \laph\i \trt\i+\frac 12 \laph\i(\ah\i \Pi\i) \\
&\quad +\Ga_1\i\cdot \Ga_2\i+\divh\i (\ah\i \slashed\CC\i_{Mom}),
\end{split}
\eea
and
\bea
\lab{eq:N-curl-Xi-uncond-i}
r^{-4} \pa_r (r^4 \curlh\i \Xi\i) = - \curlh\i\divh\i (\ah\i \kh\i)+\Ga_1\i\cdot \Ga_2\i+\curlh\i (\ah\i \slashed\CC_{Mom}\i).
\eea
On the other hand, the two components of \eqref{eq:iteration-Thh-ii} read, respectively,
\bea
\lab{eq:div-div-Thh-i}
\divh\i \divh\i (\ah\i\Psi_8\i) &=& \frac 12 \ah\i \laph\i \Psi_7\i+\Kk +\K_{\ell\leq 1}\i +\Ga_1\i\cdot \Ga_2\i, \\
\lab{eq:curl-div-Thh-i}
\curlh\i \divh\i (\ah\i\Psi_8\i) &=&  -\Kkd+\Kd_{\ell\leq 1}\i+\Ga_1\i\cdot \Ga_2\i,
\eea
Note that \eqref{eq:iteration-Pi-ii} implies $\laph\i (\ah\i \Pi\i) =\Kk+\K_{\ell\leq 1}\i-\overline{\Kk+\K_{\ell\leq 1}\i}\i$. Therefore, compare \eqref{eq:div-div-Thh-i} with \eqref{eq:N-div-Xi-uncond-i}, we obtain
\bea\lab{eq:div-CC-r}
\divh\i (\ah\i \slashed \CC_{Mom}\i) = \overline{\Kk+\K_{\ell\leq 1}\i}\i.
\eea
The second component of \eqref{eq:iteration-Xi-ii} implies
\beaa
r^{-4}\pa_r (r^4 \curlh\i\Xi\i)=\Kkd-\Kd\i_{\ell\leq 1}-\pa_r \left(\overline{\displaystyle \int_r^\infty r'^4 (\Kkd-\Kd_{\ell\leq 1}\i) dr' }\i\right),
\eeaa
where the last term on the right is only dependent on $r$, and we denote it by $F(r)$. Plugging this into \eqref{eq:N-curl-Xi-uncond-i}, we obtain
\beaa
\Kkd-\Kd\i_{\ell\leq 1} - F(r)= - \curlh\i\divh\i (\ah\i \kh\i)+\Ga_1\i\cdot \Ga_2\i+\curlh\i (\ah\i \slashed\CC_{Mom}\i).
\eeaa
Comparing this with \eqref{eq:curl-div-Thh-i} and noting that the $\Ga_1\i\cdot\Ga_2\i$ terms of \eqref{eq:N-curl-Xi-uncond-i} and \eqref{eq:curl-div-Thh-i} come from the same equation, hence have identical algebraic expressions, we deduce
\bea\lab{eq:curl-CC-r}
\curlh\i (\ah\i \slashed\CC_{Mom}\i)=F(r).
\eea
Therefore, taking the spherical averages of \eqref{eq:div-CC-r} and \eqref{eq:curl-CC-r} over $\ga\i$, we see that $\overline{\Kk+\K_{\ell\leq 1}\i}\i=F(r)=0$. Then \eqref{eq:div-CC-r} and \eqref{eq:curl-CC-r} together read $\d_1\i \slashed\CC_{Mom}\i=0$, and hence we obtain $\slashed \CC_{Mom}\i=0$.
\end{proof}

\def\N{N}

\def\Ri{\mathrm{Ric}}
\def\Rb{\Riem}

\appendix

\section{Derivation of Horizontal Constraint System}\lab{appendix:derivation-structure-equations}
\subsection{Proof of Proposition \ref{prop:Unconditional-equations-1}}\lab{subsubsect:proof-unconditional-equations-1}
We have
\beaa
\nab_a N=\th_{ab} e_b,\quad \nab_a e_b=\nabh_a e_b-\th_{ab}N,\quad \nab_N e_a=\nabh_N e_a -\pp_a N,\quad \nabh_N N=\pp_a e_a.
\eeaa
Therefore, for a $1$-form $w$ on $\Si$, we have
\bea\lab{eq:covariant-rules}
\bsplit
\nab_N w_N &= N(w_N)-w(\nab_N N) = \nabh_N (w_N) - \pp_a w_a, \\
\nab_N w_a &= N(w_a)-w(\nab_N e_a)=\nabh_N w_a-w(\nab_N e_a-\nabh_N e_a)=\nabh_N w_a+\pp_a w_N, \\
\nab_a w_N &= e_a(w_N)-w(\nab_a N) = \nabh_a (w_N)-\th_{ab} w_b, \\
\nab_a w_b &= e_a(w_b) - w(\nab_a e_b) =\nabh_a w_b -w(\nab_a e_b-\nabh_a e_b) =\nabh_a w_b +\th_{ab} w_N,
\end{split}
\eea
and similar rules apply for tensors of higher ranks.

We now derive\footnote{Within the following displayed equation, $\nab_X \nab_Y$ means $(X^i \nab_i)(Y^j \nab_j)$.}
\beaa
    \nabh_a \pp_b&=&e_a(\pp_b)-\pp(\nabh_a e_b)=e_a(g(\nab_\N \N,e_b))-g(\nab_\N \N,\nabh_a e_b)
    \\
    &=& e_a(g(\nab_\N \N,e_b))-g(\nab_\N \N,\nab_a e_b)
    = g(\nab_a (\nab_\N \N),e_b)\\
    &=&g(\nab_\N (\nab_{a}\N),e_b)+g(\nab_{[e_a,\N]}\N,e_b)+\Rb(e_a,\N,e_b,\N)\\
    &=& g(\nab_\N (\th_{ac} e_c),e_b)+g(\nab_{\nab_a N} N,e_b)-g(\nab_{\nab_N e_a} N, e_b)+R(e_a,\N,e_b,\N) \\
    &=& g(\nab_{\N} (\th_{ac}e_c),e_b)+\th_{ac}g(\nab_c \N,e_b)+\pp_a g(\nab_\N \N,e_a)-g(\nab_{\nabh_N e_a} N,e_b)+R_{a\N b N}\\
     &=& g(N(\th_{ac})e_c+\th_{ac} \nabh_N e_c,e_b)+\th_{ac}g(\nab_c \N,e_b)+\pp_a g(\nab_\N \N,e_a)-g(\nabh_N e_a, e_c) g(\nab_c N,e_b)+\Rb_{a\N b N}\\
    &=& \nabh_N \th_{ab}+ \th_{ac}\th_{cb}+\pp_a \pp_b+\Rb_{a\N b N}+cov,
\eeaa
where, for the term $cov$ that contains $\nabh_N e_a$ type terms,
\beaa
cov&=&g\Big(\th_{ad} g(\nabh_N e_c,e_d)e_c+\th_{dc} g(\nabh_N e_a, e_d)e_c+\th_{ac}\nabh_N e_c,e_b\Big)-g(\nabh_N e_a, e_c) g(\nab_c N,e_b) \\
&=& \th_{ad} g(\nabh_N e_b,e_d)+\th_{db} g(\nabh_N e_a, e_d)+\th_{ac}g(\nabh_N e_c,e_b)-\th_{cb} g(\nabh_N e_a, e_c)\\
&=& \th_{ac} g(\nabh_N e_b,e_c)+\th_{cb} g(\nabh_N e_a, e_c)-\th_{ac}g(\nabh_N e_b,e_c)-\th_{cb} g(\nabh_N e_a, e_c)\\
&=& 0.
\eeaa
Therefore, we obtain
\bea\lab{eq:R-transport-ka}
\nabh_N \th_{ab}=\nabh_a \pp_b - \th_{ac}\th_{cb}-\pp_a \pp_b-\Rb_{a\N b N}.
\eea
Note that $\th_{ac}\th_{cb}=(\thh_{ac}+\frac 12 \trth \de_{ac})(\thh_{cb}+\frac 12\trth\de_{cb})=\thh_{ac}\thh_{cb}+\trth\thh_{ab}+\frac 14(\trth)^2 \de_{ab}$.
The trace part of the equation \eqref{eq:R-transport-ka} reads
\bea\lab{eq:R-transport-trka}
\nabh_N \trth=\divh \pp-|\thh|^2-\frac 12(\trth)^2-|\pp|^2-\trR.
\eea
This proves \eqref{eq:unconditional-N-trth}.
Also note that $\thh_{ac}\thh_{cb}$ only has trace part (equation (2.2.3) in \cite{CK}),
 i.e., $\thh_{ac}\thh_{cb}=\frac 12 |\thh|^2 \de_{ab}$. Therefore the traceless part of \eqref{eq:R-transport-ka} reads
\bea\lab{eq:R-transport-kah}
\nabh_N \thh=\nabh\hot \pp - \trth\, \thh-\pp\hot \pp-\Rhh.
\eea

{\bf The Gauss curvature.} 
The Gauss equation implies
\beaa
\Riem_{abab}&=&\slashed{\Riem}_{abab}-\th_{aa}\th_{bb}+\th_{ab}\th_{ba}\\
&=& 2K-(\trth)^2+\th\cdot\th\\
&=& 2K-\frac 12(\trth)^2+|\thh|^2,
\eeaa
i.e.,
\bea\lab{eq:K-in-Riem-abab}
2K=\Riem_{abab}-|\thh|^2+\frac 12(\trth)^2.
\eea
This proves \eqref{eq:unconditional-Gauss}.

{\bf The Codazzi equation.} We have
\beaa
\Riem_{Nabc}=\nabh_c \th_{ba}-\nabh_b \th_{ca}.
\eeaa
Note that the equation only has two independent components.
Contracting $a$ and $b$ we obtain
\beaa
\Riem_{Naac}&=& \nabh_c \th_{aa}-\nabh_a \th_{ca}=\nabh_c \trth-\nabh_a \Big(\thh_{ac}+\frac 12 \trth \de_{ac}\Big)\\
&=& -(\divh \thh)_c +\frac 12 \nabh_c(\trth),
\eeaa
i.e.,
\bea\lab{eq:Codazzi-divka-unconditional}
\divh \thh=\frac 12 \nabh \trth-Y.
\eea
This proves \eqref{eq:unconditional-Codazzi}.

We also have the following Bianchi-type equations. 
\begin{lemma}\lab{lem:R-abab-thh-squared}
    We have
    \beaa
    \nabh_N (\Riem_{abab}-|\thh|^2 )= -\trth\Riem_{abab}-2\divh Y+4\pp\cdot Y+\trth\,  \trR+2\trth|\thh|^2 - 2\thh\cdot (\nabh\hot \pp-\pp\hot \pp).
    \eeaa
\end{lemma}
\begin{proof}
We have the Bianchi identity
\begin{equation*}
    \nab_{N} \Riem_{abcd}+\nab_b \Riem_{Nacd}+\nab_a \Riem_{bN cd}=0
\end{equation*}
i.e., using the rules \eqref{eq:covariant-rules},
\begin{equation*}
\begin{split}
    \nabh_N \Riem_{abcd}&+\pp_a \Riem_{Nbcd}+\pp_b \Riem_{aNcd}+\pp_c \Riem_{abNd}+\pp_d \Riem_{abcN}\\
    &+\nabh_b \Riem_{Nacd}-\th_{be}\Riem_{eacd}
    +\th_{bc} \Riem_{NaNd}+\th_{bd}\Riem_{NacN}\\
    &+\nabh_a \Riem_{bNcd}-\th_{ae} \Riem_{becd}+\th_{ac} \Riem_{bNNd}+\th_{ad} \Riem_{bNcN}=0.
    \end{split}
\end{equation*}
Then, contracting with both $\de^{ac}$ and $\de^{bd}$, using that $\th_{be}\Riem_{eaab}=\frac 12\trth \Riem_{baab}=-\frac 12\trth \Riem_{abab}$,\footnote{Indeed, we have $R_{1aa1}=R_{2aa2}$ and $R_{1aa2}=R_{2aa1}=0$, and hence $R_{eaab}=\frac 12 R_{caac}\de_{eb}$.} we deduce
\begin{equation*}
\begin{split}
    \nabh_N \Riem_{abab}&-4\pp_a Y_a+\nabh_b Y_b+\frac 12 \trth \Riem_{abab}\\
    &+\th_{ba} \left(\frac 12 \trR \de_{ab}+\Rhh_{ab}\right)-\trth \, \trR
    +\nabh_a Y_a+\frac 12 \trth  \Riem_{abab}-\trth\, \trR+\th_{ab} \left(\frac 12 \trR \de_{ba}+\Rhh_{ba}\right)=0.
    \end{split}
\end{equation*}
Therefore, we obtain
\beaa
\nabh_N \Riem_{abab}+\trth \Riem_{abab}=-2\divh Y+4\pp\cdot Y+\trth \, \trR-2\thh\cdot \Rhh.
\eeaa
Using \eqref{eq:R-transport-kah}, we also obtain
\beaa
\nabh_N (\Riem_{abab}-|\thh|^2)&=& -\trth\Riem_{abab}-2\divh Y+4\pp\cdot Y+\trth \, \trR-2\thh\cdot \Rhh \\
& &-2\thh\cdot (\nabh\hot \pp-\trth \thh-\pp\hot \pp-\Rhh)\\
&=& -\trth\Riem_{abab}-2\divh Y+4\pp\cdot Y+\trth \, \trR+2\trth |\thh|^2-2\thh\cdot (\nabh\hot \pp-\pp\hot \pp),
\eeaa
as required. This concludes the proof of Lemma \ref{lem:R-abab-thh-squared}.
\end{proof}
We then further derive the equation of $K$ using \eqref{eq:K-in-Riem-abab} and \eqref{eq:R-transport-trka}:
\beaa
2\nabh_N K &=& \nabh_N \left(\Riem_{abab}-|\thh|^2+\frac 12 (\trth)^2\right) \\
&=& -\trth\Riem_{abab}-2\divh Y+4\pp\cdot Y+\trth \, \trR+2\trth |\thh|^2-2\thh\cdot (\nabh\hot \pp-\pp\hot \pp)\\
& & +\trth \left(\divh \pp-|\thh|^2-\frac 12 (\trth)^2-|\pp|^2-\trR\right)\\
&=& -\trth \left(2K-\frac 12(\trth)^2+|\thh|^2\right)-2\divh Y+4\pp\cdot Y+\trth \, \trR+2\trth |\thh|^2-2\thh\cdot (\nabh\hot \pp-\pp\hot \pp)
\\
& &+\trth \left(\divh \pp-|\thh|^2-\frac 12 (\trth)^2-|\pp|^2-\trR\right)\\
&=& -2\divh Y-2\trth K+4\pp\cdot Y
-2\thh\cdot (\nabh\hot \pp-\pp\hot \pp)+\trth (\divh \pp-|\pp|^2).
\eeaa
Therefore,
\beaa
\nabh_N K &=& -\divh Y-\trth K+2\pp\cdot Y
-\thh\cdot (\nabh\hot \pp-\pp\hot \pp) +\frac 12\trth (\divh \pp-|\pp|^2).
\eeaa
This proves \eqref{eq:unconditional-Bianchi}. In particular, the identity does not contain $\trR$.

\subsection{Proof of Proposition \ref{prop:constraint-equation-in-frame}}\lab{subsect:proof-constraint-equation-in-frame}
We can express the constraint quantities defined in \eqref{eq:def-CC-Mom} and \eqref{eq:def-CC-Ham} as follows, using the rules in \eqref{eq:covariant-rules},
\beaa
(\CC_{Mom})_N &=& \nab_i k_{iN}-N(\tr k)= \nab_a k_{aN}+\nab_N k_{NN}-N(\trt+k_{NN}) \\
	&=& \nabh_a k_{aN}+\trth k_{NN}-\th_{ab} k_{ab}+N\Pi-\pp_a k_{Na}-\pp_a k_{aN}-N(\trt)-N\Pi \\
	&=& \divh \Xi+\trth \Pi-\th\cdot \k-2\pp\cdot \Xi-\nabh_N \trt,\\
	&=& \divh \Xi+\trth \Pi-\thh\cdot \kh-\frac 12 \trth\, \trt-2\pp\cdot \Xi-\nabh_N \trt,
\eeaa
and
\beaa
(\slashed\CC_{Mom})_a:=(\CC_{Mom})_a &=& \nab_i k_{ia}-\nab_a(\tr k) = \nab_N k_{Na} +\nab_{b} k_{ba}-\nabh_a (\tr k)\\
	&=& \nabh_N k_{Na} - \pp_b k_{ba}+ \pp_a k_{NN} +\nabh_b k_{ba}+\trth k_{Na}+\th_{ba} k_{bN}-\nabh_a \trt -\nabh_a(k_{NN})\\
	&=& \left(\nabh_N \Xi +\divh\k-\pp\cdot \k+\Pi \pp+\trth \Xi +\th\cdot \Xi-\nabh \trt -\nabh \Pi\right)_a\\
	&=& \Big(\nabh_N \Xi +\divh\kh-\pp\cdot \k+\Pi p+\frac 32\trth \Xi +\thh\cdot \Xi-\frac 12 \nabh \trt -\nabh \Pi\Big)_a\, .
\eeaa
Then the first two equations in Proposition \ref{prop:constraint-equation-in-frame} follow. To prove the third one, we first calculate
\beaa
\CC_{Ham} &=& R_g+(k_{NN}+\trt)^2-(k_{NN})^2-2|\Xi|^2-|\k|^2\\
&=& R_g+2k_{NN}\trt+(\trt)^2-2|\Xi|^2-|\kh|^2-\frac 12(\trt)^2\\
&=& R_g+2\Pi\, \trt+\frac 12(\trt)^2-2|\Xi|^2-|\kh|^2.
\eeaa
Combining this relation with the unconditional equations \eqref{eq:unconditional-N-trth}, \eqref{eq:unconditional-Gauss}, and the identity $\trh\Rh=-\frac 12 \Riem_{abab}+\frac 12 R_g$ from \eqref{eq:relation-Rg-trR-Rabab}, we get
\beaa
\nabh_N \trth&=&\divh \pp-|\thh|^2-\frac 12(\trth)^2-|\pp|^2-\trh\Rh\\
&= & \divh \pp-|\thh|^2-\frac 12(\trth)^2-|\pp|^2+\frac 12 \Riem_{abab}-\frac 12 R_g\\
&= & \divh \pp-|\thh|^2-\frac 12(\trth)^2-|\pp|^2+K-\frac 14(\trth)^2+\frac 12 |\thh|^2 -\frac 12 R_g\\
&=& \divh \pp-|\thh|^2-\frac 12(\trth)^2-|\pp|^2+K-\frac 14(\trth)^2+\frac 12 |\thh|^2 \\
& & +\frac 12 \Big(2\Pi\trt+\frac 12 (\trt)^2-2|\Xi|^2-|\kh|^2 -\CC_{Ham}\Big)\\
&=& \divh \pp-\frac 12 |\thh|^2-\frac 34(\trth)^2-|\pp|^2+K +\Pi \trt+\frac 14 (\trt)^2-|\Xi|^2-\frac 12 |\kh|^2 -\frac 12 \CC_{Ham}.
\eeaa
This proves the third identity in Proposition \ref{prop:constraint-equation-in-frame}.

\subsection{Proof of Proposition \ref{prop:linearized-eqns-time-symmetry}}\lab{appendix-derivation-linearized-equations}
For the standard Schwarzschild, we have, with $\Up:= 1-2mr^{-1}$,
\beaa
\trth\0=2\Up^\frac 12 r^{-1},\quad \ah\0=\Up^{-\frac 12},\quad N\0=\Up^\frac 12 \pa_r, \quad K\0=r^{-2}.
\eeaa
The structure equations hold true:
\bea
\lab{eq:pa-r-trth-Schw}
\pa_r \trth\0&=&\ah\0 r^{-2}-\frac 34 \ah\0 (\trth\0)^2,\\
\lab{eq:pa-r-K-Schw}
\pa_r K\0&=& -\ah\0\trth\0 r^{-2}.
\eea
In view of the expression of $\mu$ in \eqref{eq:def-mu} and the relation \eqref{eq:P-loga} between $p$ and $\ah$, the equation \eqref{eq:structure-constraint-H} can be written as
\beaa
\nabh_N \trth = \mu- \frac 12 (\trth)^2 -\frac 12 |\thh|^2 -|\pp|^2 +\Pi\, \trt+\frac 14 (\trt)^2-|\Xi|^2-\frac 12 |\kh|^2-\frac 12 \CC_{Ham}.
\eeaa
Therefore, using $N=\ah^{-1}\pa_r$, subtracting \eqref{eq:pa-r-trth-Schw}, 
we have, recalling the schematic notations introduced in Definition \ref{def:schematic-notation},
\def\mut{\tilde\mu}
\beaa
\pa_r \thc&=&\ah\mu - \frac 12 \ah (\trth)^2 -\ah\0 r^{-2} +\frac 34 \ah\0 (\trth\0)^2+\Ga_1\cdot \Ga_1-\frac 12 \ah\, \CC_{Ham}\\
&=& \ah\mu +\left(\frac 14 \ah (\trth)^2-\ah\0 r^{-2}\right)+\left(\frac 34 \ah\0 (\trth\0)^2 -\frac 34 (\ah\0+\ao)(\trth)^2\right)+\Ga_1\cdot \Ga_1 -\frac 12\ah\, \CC_{Ham}\\
&=& \ah\mu+\left(\frac 14 \ah\0 (\trth)^2+\frac 14\ao (\trth)^2-\ah\0 r^{-2}\right)+\frac 34 \ah\0\left((\trth\0)^2-(\trth)^2\right)-\frac 34 \ao (\trth)^2 \\
& & +\Ga_1\cdot \Ga_1 -\frac 12\ah\, \CC_{Ham} \\
&=& \ah\mu+\ah\0 \left(\frac 14(\trth)^2-r^{-2}\right)+\frac 14 \ao\left(\thc+2\Up^\frac 12 r^{-1}\right)^2-\frac 34 \ah\0 \thc \left(2\trth\0+\thc\right)-\frac 34 \ao (\trth)^2\\
& & +\Ga_1\cdot \Ga_1 -\frac 12\ah\, \CC_{Ham} \\
&=& \ah\mu+\ah\0 \left(\Up r^{-2} +\Up^\frac 12 r^{-1} \thc-r^{-2}\right)+\frac 14 \ao \left(\thc+2\Up^\frac 12 r^{-1}\right)^2-\frac 34 \ah\0 \thc \left(2\trth\0\right)\\
& &-\frac 34 \ao \left(\thc+2\Up^\frac 12 r^{-1}\right)^2 -\frac 12 \ah\, \CC_{Ham}+\Ga_1 \cdot \Ga_1 \\
&=& \ah\mu+\ah\0 (\Up-1) r^{-2} +r^{-1}\thc +\Up r^{-2}\ao- 3 r^{-1} \thc-3\Up r^{-2}\ao+\Ga_1\cdot \Ga_1 -\frac 12\ah\, \CC_{Ham}\\
&=& (\Up^{-\frac 12}+\ao) \left(\muc+2m r^{-3}\right) +\Up^{-\frac 12} (\Up-1)r^{-2} -2 r^{-1}\thc -2\Up r^{-2}\ao+\Ga_1\cdot \Ga_1 -\frac 12\ah\, \CC_{Ham}  \\
&=& (\Up^{-\frac 12}+\ao) \muc+2m r^{-3} \ao  -2 r^{-1}\thc -2\Up r^{-2}\ao+\Ga_1\cdot \Ga_1 -\frac 12\ah\, \CC_{Ham} \\
&=& (\Up^{-\frac 12}+\ao) \muc  -2 r^{-1}\thc -2(1-3mr^{-1}) r^{-2}\ao+\Ga_1\cdot \Ga_1 -\frac 12\ah\, \CC_{Ham}.
\eeaa
This proves \eqref{eq:R-transport-kac}. 

We now derive the equation for $\Kc$ by subtracting \eqref{eq:pa-r-K-Schw} from \eqref{eq:unconditional-Bianchi}.
Using $\divh \pp=\mu+\frac 14 (\trth)^2 -\Kc-r^{-2}$, we have
\beaa
\pa_r \Kc &=& -\ah \divh Y-\ah \trth\, K  +\frac 12 \ah \trth\, \divh \pp+\Ga_1\cdot \Ga_2 +\ah\0 \trth\0 r^{-2} \\
&=& -\ah \div Y -\ah \trth \Kc -(\ah\trth-\ah\0 \trth\0)r^{-2}+\frac 12 \ah \trth \divh \pp+\Ga_1\cdot \Ga_2 \\
&=& -\ah \div Y-\ah\0 \trth\0 \Kc-\ao\1\trth r^{-2}-\ah\0\thc r^{-2}+\frac 12\ah\0 \trth\0 \divh \pp+\Ga_1\cdot \Ga_2,\\
&=& -\ah \div Y-2r^{-1}\Kc -2\Up^\frac 12 r^{-3} \ao\1-\Up^{-\frac 12} r^{-2}\thc+r^{-1} (\mu+\frac 14 (\trth)^2 -\Kc-r^{-2})+\Ga_1\cdot \Ga_2\\
&=& -\ah \div Y-3r^{-1}\Kc -2\Up^\frac 12 r^{-3} \ao\1-\Up^{-\frac 12} r^{-2}\thc+r^{-1} \mu+\Up^\frac 12 r^{-2} \thc+(\Up -1) r^{-3} +\Ga_1\cdot \Ga_2 \\
&=& -\ah \div Y-3r^{-1}\Kc -2\Up^\frac 12 r^{-3} \ao\1+r^{-1} \mu-(2m r^{-1}) r^{-3}+\Ga_1\cdot \Ga_2 \\
&=& -\ah \div Y-3r^{-1}\Kc -2\Up^\frac 12 r^{-3} \ao\1+r^{-1} \muc+\Ga_1\cdot \Ga_2.
\eeaa
This proves \eqref{eq:R-transport-Kc}.
\begin{remark}\lab{rem:eq-Psi-2-nonlinear-terms}
We note that here $\Ga_1$ only involves $r^{-1}\ao$, $\thc$, $\thh$, $p$, and $\Ga_2$ only involves $Y$, $\Kc$, $\nabh p$.
\end{remark}

To prove \eqref{eq:R-transport-ao}, note that
\bea\lab{eq:laph-loga-laph-a}
\bsplit
\laph(\log(\ah))&= \laph(\log(\ah\0+\ao\1))=\laph \left(\log \ah\0+\log(1+\frac{\ao\1}{\ah\0})\right)=\laph\left(\frac{\ao\1}{\ah\0}+\Ga_0\cdot \Ga_0\right)\\
&= \Up^\frac 12 \laph \ao\1+\laph(\Ga_0\cdot \Ga_0).
\end{split}
\eea
Therefore, 
\beaa
\Up^\frac 12 \laph \ao &=& -\laph(\Ga_0\cdot \Ga_0)+K-\frac 14 (\trth)^2-\mu = -\laph(\Ga_0\cdot \Ga_0)+K-\Up^\frac 12 r^{-1} \thc-\Up r^{-2}+\Ga_1\cdot \Ga_1-\mu \\
&=& \Kc-\Up^\frac 12 r^{-1} \thc-\mu +(1-\Up)r^{-2} -\laph(\Ga_0\cdot \Ga_0)+\Ga_1\cdot\Ga_1.
\eeaa
\begin{remark}\lab{rem:eq-Psi-3-nonlinear-terms}
We note that here the $\Ga_1\cdot \Ga_1$ term in fact only consists of $(\thc)^2$.
\end{remark}

To prove \eqref{eq:N-transport-go}, recall that
\beaa
\slashed{\Lie}_N \ga_{AB}=2\th_{AB}=2\thh_{AB}+\trth \ga_{AB}.
\eeaa
From \eqref{eq:scalar-multiple-projected-Lie-derivatives} we know that $\slashed{\Lie}_{N} \ga_{AB}=\ah^{-1} \slashed{\Lie}_{\pa_r} \ga_{AB}$. 
We then have
\bea\lab{eq:identity-metric-kac-ah}
\bsplit
 \slashed{\Lie}_{\pa_r} (r^{-2}\ga) &=2r^{-2}\ah \thh+\ah (\thc+2\Up^\frac 12 r^{-1}) (r^{-2}\ga) -2r^{-1}(r^{-2}\ga)\\
 &= 2r^{-2}\ah \thh+\ah\thc (r^{-2}\ga)+2\ah \Up^\frac 12 r^{-1} (r^{-2}\ga)-2r^{-1}(r^{-2}\ga) \\
&= 2r^{-2}\ah\thh+\ah\thc (r^{-2}\ga)+\Up^\frac 12 (\ah- \Up^{-\frac 12}) 2r^{-1} (r^{-2}\ga) \\
&= 2r^{-2}\ah\thh+\ah\thc (r^{-2}\ga)+2\Up^\frac 12 \ao r^{-1} (r^{-2}\ga).
\end{split}
\eea
The equations \eqref{eq:unconditional-Codazzi-1} and \eqref{eq:unconditional-d1p-1} directly follow by taking $\d_1$ of \eqref{eq:unconditional-Codazzi} and \eqref{eq:P-loga} respectively.\footnote{For the latter, we also use \eqref{eq:laph-loga-laph-a}.} The equation \eqref{eq:N-trt-linearized} is the same as \eqref{eq:structure-constraint-H} with both sides multiplied by $\ah$, except also taking into account that $\trth=2\Up^\frac 12 r^{-1}+\thc$ and $\ah=\Up^{-\frac 12}+\ao$. 

To derive the equations \eqref{eq:N-div-Xi} and \eqref{eq:N-curl-Xi}, we first note that, using $p=-\nabh(\log\ah)$ by \eqref{eq:P-loga},
\beaa
-\divh\kh+p\cdot \kh =-\divh\kh-\ah^{-1} \nab(\ah)\cdot \kh=-\ah^{-1} \divh(\ah \kh),\quad \nabh\Pi-p\Pi = \nabh\Pi+\ah^{-1} \nab(\ah)\Pi = \ah^{-1} \nabh(\ah\Pi).
\eeaa
Similarly, we also have $\nabh_a (\slashed\CC_{Mom})_b -p_a\cdot (\slashed\CC_{Mom})_b =\ah^{-1} \nabh_a (\ah \slashed\CC_{Mom})_b$. Then \eqref{eq:structure-constraint-Phi-a} can be rewritten as
\bea\lab{eq:structure-constraint-Phi-a-appendix}
\nabh_N \Xi &=& -\ah^{-1} \divh (\ah\kh)+\frac 12 \trt\pp -\frac 32\trth\, \Xi-\thh\cdot \Xi+\frac 12 \nabh \trt +\ah^{-1}\nabh (\ah \Pi)+ \slashed{\CC}_{Mom}.
\eea
Then we recall the commutation formula \eqref{eq:commutation-nab-N-nab-a}, which, when applied to \eqref{eq:structure-constraint-Phi-a-appendix}, gives
\beaa
\nabh_N \nabh_a \Xi_b &=& -\nabh_a (\ah^{-1}\divh (\ah\kh))_b -\frac 32 \trth\, \nabh_a\Xi_b +\nabh_a(\Ga_1\cdot\Ga_1)+\frac 12 \nabh_a\nabh_b \trt +\nabh_a(\ah^{-1} \nabh_b(\ah \Pi))+ \nabh_a (\slashed{\CC}_{Mom})_b \\
& & -\frac 12\trth \nabh_a \Xi_b +\Ga_1\cdot \nabh\Ga_1 -\pp_a (\nabh_N \Xi_b) +\Ga_2 \cdot \Xi \\
&=& -\nabh_a (\ah^{-1}\divh (\ah\kh))_b -\frac 32 \trth\, \nabh_a\Xi_b +\nabh_a(\Ga_1\cdot\Ga_1)+\frac 12 \nabh_a\nabh_b \trt +\nabh_a(\ah^{-1} \nabh_b(\ah \Pi))+ \nabh_a (\slashed{\CC}_{Mom})_b \\
& & -\frac 12\trth \nabh_a \Xi_b -\Ga_1\cdot (\nabh \Ga_1,r^{-1}\Ga_1) +p_a (\ah^{-1} \divh(\ah \kh)-\ah^{-1}\nabh(\ah\Pi))_b+\Ga_2 \cdot \Xi -\pp_a (\slashed\CC_{Mom})_b \\
&=& -\ah^{-1} \nabh_a (\divh (\ah\kh))_b -\frac 32 \trth\, \nabh_a\Xi_b +\nabh_a(\Ga_1\cdot\Ga_1)+\frac 12 \nabh_a\nabh_b \trt +\ah^{-1} \nabh_a( \nabh_b(\ah \Pi))+ \ah^{-1} \nabh_a (\ah \slashed{\CC}_{Mom})_b \\
& & -\frac 12\trth \nabh_a \Xi_b -\Ga_1\cdot (\nabh \Ga_1,r^{-1}\Ga_1) +\Ga_2 \cdot \Xi  
\eeaa
where we used $\nabh_N \Xi= -\ah^{-1} \divh (\ah\kh) +\ah^{-1}\nabh (\ah \Pi)+(\nabh\Ga_1,r^{-1}\Ga_1)+ \slashed{\CC}_{Mom}$.
Hence,
\beaa
\nabh_N \div\Xi &=& -\ah^{-1}\divh\divh (\ah\kh) -2\trth\, \divh \Xi +\frac 12 \laph \trt+\frac 12 \ah^{-1}\laph(\ah\Pi) +\ah^{-1} \divh (\ah \slashed\CC_{Mom})+\Ga_1\cdot \Ga_2, \\
\nabh_N \curlh\Xi &=& -\ah^{-1} \curlh\divh (\ah\kh) -2\trth\, \curlh \Xi  +\ah^{-1}\curlh (\ah\slashed\CC_{Mom})+\Ga_1\cdot \Ga_2.
\eeaa
The equations \eqref{eq:N-div-Xi} and \eqref{eq:N-curl-Xi} then follow by multiplying both sides by $\ah$, and using that $\trth=2\Up^\frac 12 r^{-1}+\Ga_1$, $\ah=\Up^{-\frac 12}+\Ga_0$.

\section{Computation in spacetime notations}
\subsection{The null structure and Bianchi equations}\lab{subsect:null-str-Bianchi}
We recall the null structure and Bianchi equations for Einstein-vacuum spacetime, given in full generality in \cite{KS:Kerr}, \cite{GKS}.

\def\c{\cdot}
\newcommand{\hch}{\widehat{\chi}}
\def\varoc{\check{\varrho}}
\begin{proposition}[Null structure equations]\lab{prop:null-str-eqns}
The  connection coefficients  verify the following   equations:
\label{prop-nullstr}
\beaa
\nabh_3\trchb&=&-|\chibh|^2-\frac 1 2 \big( \trchb^2-\atrchb^2\big)+2\divh\xib  - 2\omb \trchb +  2 \xib\c(\eta+\etab-2\ze),\\
\nabh_3\atrchb&=&-\trchb\atrchb +2\curlh \xib -2\omb\atrchb+ 2 \xib\wedge(-\eta+\etab+2\ze),\\
\nabh_3\chibh&=&-\trchb\,  \chibh+  \nabh\hot \xib- 2 \omb \chibh+    \xib\hot(\eta+\etab-2\ze)-\aa,
\eeaa
\beaa
\nabh_3\trch
&=& -\chibh\c\chih -\frac 1 2 \trchb\trch+\frac 1 2 \atrchb\atrch    +   2   \divh \eta+ 2 \omb \trch + 2 \big(\xi\c \xib +|\eta|^2\big)+ 2\rho,\\
\nabh_3\atrch
&=&-\chibh\wedge\chih-\frac 1 2(\atrchb \trch+\trchb\atrch)+ 2 \curlh \eta + 2 \omb \atrch + 2 \xib\wedge\xi  -  2 \dual \rho,\\
\nabh_3\chih
&=&-\frac 1 2 \big( \trch \chibh+\trchb \chih\big)-\frac 1 2 \big(-\dual \chibh \, \atrch+\dual \chih\,\atrchb\big)
+ \nabh\hot \eta +2 \omb \chih+ \xib\hot\xi +\eta\hot\eta,
\eeaa
\beaa
\nabh_4\trchb
&=& -\chih\c\chibh -\frac 1 2 \trch\trchb+\frac 1 2 \atrch\atrchb    +  2   \divh \etab+ 2 \om \trchb + 2\big( \xi\c \xib +|\etab|^2\big)+2\rho,\\
\nabh_4\atrchb
&=&-\chih\wedge\chibh-\frac 1 2(\atrch \trchb+\trch\atrchb)+ 2 \curlh \etab + 2 \om \atrchb + 2 \xi\wedge\xib+2 \dual \rho,\\
\nabh_4\chibh
&=&-\frac 1 2 \big( \trchb \chih+\trch \chibh\big)-\frac 1 2 \big(-\dual \chih \, \atrchb+\dual \chibh\,\atrch\big)
+\nabh\hot \etab +2 \om \chibh+ \xi\hot\xib + \etab\hot\etab,
\eeaa
\beaa
\nabh_4\trch&=&-|\chih|^2-\frac 1 2 \big( \trch^2-\atrch^2\big)+ 2 \divh\xi  - 2 \om \trch + 2   \xi\c(\etab+\eta+2\ze),\\
\nabh_4\atrch&=&-\trch\atrch + 2 \curlh \xi - 2 \om\atrch+ 2 \xi\wedge(-\etab+\eta-2\ze),\\
\nabh_4\chih&=&-\trch\,  \chih+ \nabh\hot \xi- 2 \om \chih+    \xi\hot(\etab+\eta+2\ze)-\a.
\eeaa
Also,
\beaa
\nabh_3 \ze+2\nabh\omb&=& -\chibh\c(\ze+\eta)-\frac{1}{2}\trchb(\ze+\eta)-\frac{1}{2}\atrchb(\dual\ze+\dual\eta)+ 2 \omb(\ze-\eta)\\
&&+\hch\c\xib+\frac{1}{2}\trch\,\xib+\frac{1}{2}\atrch\dual\xib +2 \om \xib -\bb,
\\
\nabh_4 \ze -2\nabh\om&=& \chih\c(-\ze+\etab)+\frac{1}{2}\trch(-\ze+\etab)+\frac{1}{2}\atrch(-\dual\ze+\dual\etab)+2 \om(\ze+\etab)\\
&& -\chibh\c\xi -\frac{1}{2}\trchb\,\xi-\frac{1}{2}\atrchb\dual\xi -2 \omb \xi -\b,
\\
\nabh_3 \etab -\nabh_4\xib &=& -\chibh\c(\etab-\eta) -\frac{1}{2}\trchb(\etab-\eta)+\frac{1}{2}\atrchb(\dual\etab-\dual\eta) -4 \om \xib  +\bb, \\
\nabh_4 \eta    -    \nabh_3\xi &=& -\chih\c(\eta-\etab) -\frac{1}{2}\trch(\eta-\etab)+\frac{1}{2}\atrch(\dual\eta-\dual\etab)-4\omb \xi -\b,\\
\eeaa
and
\beaa
\nabh_3\om+\nabh_4\omb -4\om\omb -\xi\c \xib -(\eta-\etab)\c\ze +\eta\c\etab&=&   \rho.
\eeaa
Also,
\beaa
\divh\chih +\ze\c\chih &=& \frac{1}{2}\nabh\trch+\frac{1}{2}\trch\ze -\frac{1}{2}\dual\nabh\atrch-\frac{1}{2}\atrch\dual\ze -\atrch\dual\eta-\atrchb\dual\xi -\b,\\
\divh\chibh -\ze\c\chibh &=& \frac{1}{2}\nabh\trchb-\frac{1}{2}\trchb\ze -\frac{1}{2}\dual\nabh\atrchb+\frac{1}{2}\atrchb\dual\ze -\atrchb\dual\etab-\atrch\dual\xib +\bb,
\eeaa
and
\beaa
\curlh\ze&=&-\frac 1 2 \chih\wedge\chibh   +\frac 1 4 \big(  \trch\atrchb-\trchb\atrch   \big)+\om \atrchb -\omb\atrch+\dual \rho.
\eeaa
\end{proposition}

   \begin{proposition}[Null Bianchi identities]\label{prop:bianchi} 
       The  curvature components  verify the following   equations:
    \beaa
    \nabh_3\a-  \nabh\hot \b&=&-\frac 1 2 \big(\trchb\a+\atrchb\dual \a)+4\omb \a+
  (\ze+4\eta)\hot \b - 3 (\rho\chih +\rhod\dual\chih),\\
\nabh_4\beta - \divh\a &=&-2(\trch\beta-\atrch \dual \b) - 2  \om\b +\a\c  (2 \ze +\etab) + 3  (\xi\rho+\dual \xi\rhod),\\
     \nabh_3\b+\divh\varrho&=&-(\trchb \b+\atrchb \dual \b)+2 \omb\,\b+2\bb\c \chih+3 (\rho\eta+\rhod\dual \eta)+    \a\c\xib,\\
 \nabh_4 \rho-\divh \b&=&-\frac 3 2 (\trch \rho+\atrch \rhod)+(2\etab+\ze)\c\b-2\xi\c\bb-\frac 1 2 \chibh \c\a,\\
   \nabh_4 \rhod+\curlh\b&=&-\frac 3 2 (\trch \rhod-\atrch \rho)-(2\etab+\ze)\c\dual \b-2\xi\c\dual \bb+\frac 1 2 \chibh \c\dual \a, \\
     \nabh_3 \rho+\divh\bb&=&-\frac 3 2 (\trchb \rho -\atrchb \rhod) -(2\eta-\ze) \c\bb+2\xib\c\b-\frac{1}{2}\chih\c\aa,
 \\
   \nabh_3 \rhod+\curlh\bb&=&-\frac 3 2 (\trchb \rhod+\atrchb \rho)- (2\eta-\ze) \c\dual \bb-2\xib\c\dual\b-\frac 1 2 \chih\c\dual \aa,\\
     \nabh_4\bb-\divh\varoc&=&-(\trch \bb+ \atrch \dual \bb)+ 2\om\,\bb+2\b\c \chibh
    -3 (\rho\etab-\rhod\dual \etab)-    \aa\c\xi,\\
     \nabh_3\bb +\divh\aa &=&-2(\trchb\,\bb-\atrchb \dual \bb)- 2  \omb\bb-\aa\c(-2\ze+\eta) - 3  (\xib\rho-\dual \xib \rhod),\\
     \nabh_4\aa+ \nabh\hot \bb&=&-\frac 1 2 \big(\trch\aa-\atrch\dual \aa)+4\om \aa+
 (\ze-4\etab)\hot \bb - 3  (\rho\chibh -\rhod\dual\chibh).
\eeaa
Here,
\beaa
\divh\varo&=&- (\nabh\rho+\dual\nabh\rhod),\\
\divh\varoc&=&- (\nabh\rho-\dual\nabh\rhod).
\eeaa
    \end{proposition}

\subsection{Proof of Proposition \ref{prop:relations-om-etab-xi-zeta}}\lab{subsect:Proof-relations-om-etab-xi-zeta}
    We have, using $\D_3 e_3=\D_3 e_4=0$,
\begin{equation*}
\begin{split}
    \om&=\frac 14 \g(\D_4 e_4,e_3)=\frac 14 \g(\D_{2 N+ e_3} e_4,e_3)=\frac 12  \g \left(\D_N ( T+N),T-N \right)+\frac{1}{4} \g(\D_3 e_4,e_3)\\
    &=\frac 12 \g\left(\D_N( T+N),T-N\right)=\frac 12 \g(\D_N T,-N)+\frac 12 \g(\D_N N, T)\\
    &=-k(N,N).
    \end{split}
\end{equation*}
This proves \eqref{eq:om-de}. We also have
\begin{equation}\lab{eq:relations-k-Na}
\begin{split}
    \Xi_a &=\g(\D_N T,e_a)=\g\left(\D_{\frac 12 e_4-\frac 12 e_3} \left(\frac 12  e_3+\frac 12 e_4\right), e_a\right)\\
    &=\frac 14 \g(\D_4 e_3,e_a)+\frac 14 \g(\D_4 e_4,e_a) \\
    &=\frac 12\xi_a+\frac 12 \etab_a,
    \end{split}
\end{equation}
and
\begin{equation}\lab{eq:relations-p-xi-etab}
\begin{split}
  p_a &=   g(\nab_N N,e_a)=\g(\D_N N,e_a)=\g\left(\D_{N}\Big(\frac 12 e_4-\frac 12 e_3\Big),e_a\right)=\frac 12 \g(\D_N e_4,e_a)-\frac 12 \g(\D_N e_3,e_a)\\
    &=\frac 12 \left(\frac 12 \g(\D_4 e_4,e_a)-\frac 12 \g(\D_3 e_4,e_a)\right)-\frac 12 \left(\frac 12 \g(\D_4 e_3,e_a)-\frac 12 \g(\D_3 e_3,e_a)\right)\\
    &=\frac 12 \xi_a-\frac 12\etab_a.
    \end{split}
\end{equation}
Therefore, combining \eqref{eq:relations-k-Na} and \eqref{eq:relations-p-xi-etab}, we obtain
\beaa
\xi=\Xi+p,\quad \etab=\Xi-p,
\eeaa
This proves \eqref{eq:xi-ep-loga}-\eqref{eq:etab-ep-loga}. 

\subsection{Proof of Proposition \ref{prop:b-bb-expression}}\lab{subsect:spacetime-Ricci-b-bb}
We first state the following general relations.
\begin{lemma}\lab{lem:spacetime-curvatures}
The following equations hold true:
\bea
\lab{eq:Gauss-spacetime}
\R_{ijkl}&=& R_{ijkl}+k_{ik} k_{jl}-k_{il} k_{jk},\\
\lab{eq:Codazzi-spacetime-1}
\R_{Tabc}&=&\nabh_c \k_{ab}-\nabh_b \k_{ac}+\th_{ca}\Xi_{b}-\th_{ba}\Xi_{c},\\
\lab{eq:Codazzi-spacetime-2}
\R_{TaNb}&= & \nabh_b \Xi_a-\nabh_N \k_{ab}+\Pi \th_{ba}-\th_{bc}\k_{ac}-p_a \Xi_{b}-p_b \Xi_{a},\\
\lab{eq:Codazzi-spacetime-curl}
\R_{TNab} &=& \nabh_b\Xi_a - \nabh_a\Xi_b-\k_{ac}\th_{bc}+\k_{bc}\th_{ac},\\
\lab{eq:Codazzi-spacetime-3}
\R_{TNNa}&= & \nabh_a \Pi-\nabh_N \Xi_a-2\th_{ab}\Xi_{b}+p_b \k_{ba}-p_a \Pi.
\eea
\end{lemma}
\begin{proof}
The first equation is the Gauss equation (note the sign flip due to the Lorentzian signature). We have the Codazzi equation\footnote{Note again our convention $k_{ij}=\g(\D_i T,\pa_j)$.}
\bea\lab{eq:Codazzi-ijl}
\R_{Tijl}=\nab_l k_{ij}-\nab_j k_{il}.
\eea
This yields the following relations, again using the rules in \eqref{eq:covariant-rules},
\bea\lab{eq:Codazzi-expanded}
\bsplit
\R_{Tabc}&=\nab_c k_{ab}-\nab_b k_{ac}=\nabh_c k_{ab}+\th_{ca}k_{Nb}+\th_{cb} k_{aN}-\nabh_b k_{ac}-\th_{ba}k_{Nc}-\th_{bc}k_{aN},\\
&=\nabh_c k_{ab}-\nabh_b k_{ac}+\th_{ca}k_{Nb}-\th_{ba}k_{Nc},\\
\R_{TaNb}&= \nab_b k_{aN}-\nabh_N k_{ab}=\nabh_b \Xi_a+\th_{ba} k_{NN}-\th_{bc} k_{ac}-(\nabh_N k_{ab}+p_a k_{Nb}+p_b k_{Na}), \\
\R_{TNab} &=\nab_b k_{Na}-\nab_a k_{Nb}=\nabh_b k_{Na}-\th_{bc} k_{ca} + \th_{ba} k_{NN}-(\nabh_a k_{Nb}-\th_{ac} k_{cb} + \th_{ab} k_{NN}) \\
&= \nabh_b k_{Na} - \nabh_a k_{Nb}-k_{ac}\th_{bc}+k_{bc}\th_{ac}, \\
\R_{TNNa}&= \nab_a k_{NN}-\nabh_N k_{Na}=\nabh_a k_{NN}-2\th_{ab} k_{bN}-(\nabh_N k_{Na}-p_b k_{ba}+p_a k_{NN})\\
&= \nabh_a k_{NN}-\nabh_N k_{Na}-2\th_{ab}k_{bN}+p_b k_{ba}-p_a k_{NN}.
\end{split}
\eea
Therefore, the equations \eqref{eq:Gauss-spacetime}-\eqref{eq:Codazzi-spacetime-3} follow. 
This concludes the proof of Lemma \ref{lem:spacetime-curvatures}.
\end{proof}

\begin{remark}\lab{rem:relation-momentum-constraint-Ric-T-i}
In view of the relations \eqref{eq:Gauss-spacetime}, \eqref{eq:Codazzi-ijl}, and the definitions \eqref{eq:def-CC-Mom}-\eqref{eq:def-CC-Ham}, we have
\beaa
\Ricc_{Ti} &=& (\div k)_i -\nab_i (\tr k) = (\CC_{Mom})_i, \\
\R_g &=& -\Ricc_{TT}+\Ricc_{ii}=-2\R_{TiTi}+R_{ijij}+(\tr k)^2-|k|^2 \\
&=& -2\Ricc_{TT}+\CC_{Ham}.
\eeaa
\end{remark}

We now start to prove \eqref{eq:b+bb}-\eqref{eq:b-bb}.
From the spacetime notations, we have
\begin{equation*}
\begin{split}
    \bb(\R):=\frac 12 \R_{a334}&=\frac 12 \R(e_a,T-N,T-N,T+N)=\R_{aTTN}-\R_{aNTN},      
      \end{split}
\end{equation*}
\begin{equation*}
\begin{split}
    \b(\R):=\frac 12 \R_{a434}&=\frac 12 \R(e_a,T+N,T-N,T+N)=\R_{aTTN}+\R_{aNTN},
        \end{split}
\end{equation*}
and hence
\beaa
\b(\R)+\bb(\R) &=& -2\R_{aTNT}=2\Ricc_{aN}-2\R_{abNb},\\
\b(\R)-\bb(\R) &=& 2\R_{TNaN}.
\eeaa
Using the definition of the Weyl tensor \eqref{eq:Weyl-tensor}, we have
\beaa
\bb&=& \frac 12 \W_{a334}=\bb(\R)+\frac 14 (-2) \Ricc_{3a}=\bb(\R)-\frac 12 \Ricc_{3a},\\
\b&=& \frac 12 \W_{a434}=\b(\R)-\frac 14 (-2) \Ricc_{4a} =\b(\R)+\frac 12\Ricc_{4a},
\eeaa
hence
\beaa
\b+\bb&=&\b(\R)+\bb(\R)+\Ricc_{Na}=3\Ricc_{Na}+2\R_{Nbba},\\
\b-\bb&=& \b(\R)-\bb(\R)+\Ricc_{Ta}=2\R_{TNaN}+\Ricc_{Ta}.
\eeaa
Using \eqref{eq:Gauss-spacetime} and \eqref{eq:Codazzi-spacetime-3}, along with Remark \ref{rem:relation-momentum-constraint-Ric-T-i}, we have
\beaa
(\b+\bb)_a -3\Ricc_{Na} &=& 2\R_{Nbba}=2(R_{Nbba}+k_{Nb} k_{ba}-k_{Na} k_{bb})=2\left(Y_a+\Xi_b \cdot \k_{ba}-\trt \Xi_a\right),\\
(\b-\bb)_a - (\slashed\CC_{Mom})_a &=& -2\R_{TNNa}=-2\left(\nabh_a \Pi-\nabh_N \Xi_a-2\th_{ab}\Xi_b+p_b \k_{ba}-p_a \Pi\right),
\eeaa
as required.

For $\rho$, first note that by \eqref{eq:Weyl-tensor}
\beaa
\rho=\frac 14 \W_{3434}=\frac 14 (\R_{3434}+2\Ricc_{34}-\frac 23\R_{\g})=\rho(\R)+\frac 12 \Ricc_{34}-\frac 16 \R_{\g}.
\eeaa 
Using the Gauss equation for codimension 2, we have, again noting the sign flip for the $T$-direction,
\beaa
\R_{abab}=2K_\ga+\frac 12 (\trt)^2-\frac 12 (\trth)^2-|\kh|^2+|\thh|^2.
\eeaa
Therefore, using $\Ricc_{34}=-\frac 12 \R_{3443}+\R_{3a4a}=2\rho(\R)+\R_{3a4a}$ and $\Ricc_{aa}=-\R_{3a4a}+\R_{abab}$, we have
\beaa
\rho&=&\rho(\R)+\frac 12 \Ricc_{34}-\frac 16 \R_g=\frac 12\left(\Ricc_{34}+\Ricc_{aa}-\R_{abab}\right)+\frac 12 \Ricc_{34}-\frac 16 \R_g \\
&=& -\frac 12 \R_{abab}+\Ricc_{TT}-\Ricc_{NN}+\frac 12 \Ricc_{aa}-\frac 16 \R_{\g} \\
&=& -\frac 12 \R_{abab}+\frac 12 (\CC_{Ham}- \R_{\g})-(\Ricc-\frac 12 (\R_{\g}) \g)_{NN}-\frac 12 \R_{\g}+\frac 12 (\Ricc-\frac 12 (\R_{\g}) \g)_{aa}+\frac 12 \R_{\g}-\frac 16 \R_{\g} \\
&=& -K_\ga -\frac 14 (\trt)^2+\frac 14 (\trth)^2 +\frac 12 |\kh|^2-\frac 12 |\thh|^2 \\
& & +\frac 12 \CC_{Ham} -(\Ricc-\frac 12 (\R_{\g}) \g)_{NN}+\frac 12 (\Ricc-\frac 12 (\R_{\g}) \g)_{aa}-\frac 23 \R_{\g}.
\eeaa
This proves \eqref{eq:Gauss-rho}.

For $\dual\rho$ we have 
\beaa
\dual\rho=\frac 14 \dual \W_{3434}=\frac 14 (\frac 12 \in_{34ab} \W_{ab 34})=\frac 12 \W_{1234}=\frac 12 \R_{1234}=\R_{12TN}=-\curlh \Xi- \kh\wedge\thh.
\eeaa
where we used the relation \eqref{eq:Codazzi-spacetime-curl} for the last equality. This proves \eqref{eq:curl-kn-dual-rho} and concludes the proof of Proposition \ref{prop:b-bb-expression}.

\section{Physical quantities}\lab{appendix:physical-quantities}
In this appendix, we prove the following alternative expression of the ADM charges:
\begin{proposition}\lab{prop:relation-ADM-charges}
Under the class of initial data $(g,k)$ we construct in Theorem \ref{thm:main-precise}, we have
\beaa
\E=m,\quad \C_i=-\frac{1}{8\pi m}\lim_{r\to \infty} r^3 (\thc)_{\ell=1,i},\quad \P_i=-\frac 1{8\pi} \lim_{r\to\infty} r^2 (\trt)_{\ell=1,i},\quad \J_i=\frac 1{8\pi} \lim_{r\to\infty} r^4 (\curlh\Xi)_{\ell=1,i}.
\eeaa
\end{proposition}
We prove the relations for $\E$, $\C_i$ in the first subsection and for $\P_i$, $\J_i$ in the next.

In this appendix, we frequently use the vector field $\overline\pa_i:=\pa_i -\om_i \pa_r$, where $\om_i:=x^i/r$. We record the well-known relations
\beaa
\gz(\overline\pa_i,\overline\pa_j)=\de_{ij}-\om_i\om_j,\quad \pa_i \om_j=\frac{1}{r}(\de_{ij}-\om_i\om_j).
\eeaa

\subsection{Energy and center of mass}
We relate the ADM energy $\E$ and center of mass $\C_i$ defined in \eqref{eq:def-ADM-charges} with the conditions we impose at infinity. The definitions in \eqref{eq:def-ADM-charges} are written in Cartesian coordinates, and hence we first need the following lemma.
\begin{lemma}
Under the assumption that $g=\ah^2 dr^2+\ga_{AB} d\th^A d\th^B$, we have
\beaa
\E &=& \frac{1}{16\pi}\lim_{r\to\infty} \int_{S_r} \frac 2r g_{rr} -\frac 1r \ga_{kk} -\pa_r (\ga_{kk})\, dA,\\
\C_i &=& \frac{1}{16\pi m}\lim_{r\to\infty}\int_{S_r}
  \om^i \big(2 (g_{rr} -1) - r\pa_r (\ga_{kk})\big)\, dA,
\eeaa
where $\ga_{kk}:=\sum_{k=1}^3\ga(\overline \pa_k,\overline\pa_k)$.
\end{lemma}
\begin{proof}
In Cartesian coordinates, we have, using that $g(\pa_r,\overline \pa_i)=0$,
\beaa
g_{ij} =  g(\overline \pa_i+\om_i \pa_r,\overline \pa_j+\om_j \pa_r)=\om_i \om_j g_{rr} +\ga(\overline\pa_i, \overline\pa_j). 
\eeaa
We have, recalling that $\om^k \pa_k=\pa_r$,
\beaa
g_{kr}&=& \om_k g_{rr},\\
(\pa_k g_{kj})(\pa_r)^j &=& \pa_k (g_{kr})-g_{kj} \pa_k (\pa_r)^j =\pa_k (\om_k g_{rr})-g_{kj} \frac 1r (\de_{kj}-\om_k \om_j) \\
&=& \frac 2r g_{rr} +\pa_r (g_{rr})-\frac 1r \ga(\overline\pa_k,\overline\pa_k), \\
\pa_r g_{kk} &=& \pa_r (g_{rr}+\ga(\overline \pa_k,\overline \pa_k)).
\eeaa
Therefore,
\beaa
(\partial_k g_{kj}-\partial_j g_{kk}) (\pa_r)^j &=& \frac 2r g_{rr} -\frac 1r \ga(\overline\pa_k,\overline\pa_k) -\pa_r (\ga(\overline \pa_k,\overline \pa_k)),\\
x^i\big(\partial_k g_{kj}-\partial_j g_{kk}\big) (\pa_r)^j &=& r \om^i \left(\frac 2r g_{rr} -\frac 1r \ga(\overline\pa_k,\overline\pa_k) -\pa_r (\ga(\overline \pa_k,\overline \pa_k))\right),
\eeaa
and
\beaa
(g_{ir}-\de_{ir})-\delta_{ir}(g_{kk}-\de_{kk}) &=& \om_i (g_{rr}-1) -\om_i \left(g_{rr}-1+\ga(\overline \pa_k,\overline \pa_k)-\gz(\overline \pa_k,\overline \pa_k)\right)\\
&=& -\om_i \left(\ga(\overline \pa_k,\overline \pa_k)-\gz(\overline \pa_k,\overline \pa_k)\right).
\eeaa
Therefore,
\beaa
\E=\frac{1}{16\pi}\lim_{r\to\infty} \int_{S_r} \sum_{i,j} (\pa_i g_{ij}-\pa_j g_{ii}) (\pa_r)^j dA=\frac{1}{16\pi}\lim_{r\to\infty} \int_{S_r} \frac 2r g_{rr} -\frac 1r \ga_{kk} -\pa_r (\ga_{kk})\, dA,
\eeaa
and, using $\gz_{kk}=2$,
\beaa
\C_i &=& \frac{1}{16\pi m}\lim_{r\to\infty}\int_{S_r}
\Big( x^i\big(\partial_k g_{kj}-\partial_j g_{kk}\big)
-\big((g_{ij}-\de_{ij})-\delta_{ij}(g_{kk}-\de_{kk})\big)\Big) (\pa_r)^j\, dA \\
&=& \frac{1}{16\pi m}\lim_{r\to\infty}\int_{S_r}
\Big( r \om^i \big(\frac 2r g_{rr} -\frac 1r \ga_{kk} -\pa_r (\ga_{kk})\big)+\om_i \big(\ga_{kk}-\gz_{kk}\big)
\Big)\, dA \\
&=& \frac{1}{16\pi m}\lim_{r\to\infty}\int_{S_r}
 \om^i \big(2 g_{rr} -2 - r \pa_r (\ga_{kk})\big)\, dA,
\eeaa
as required.
\end{proof}
We now derive the expansion of $\ga_{kk}$. 
\begin{proposition}\lab{prop:pa-r-ga-kk}
We have
\bea\lab{eq:pa-r-ga-kk}
\pa_r (\ga_{kk})= 2\Up^{-\frac 12} (\thc)_{\ell=1} +4\Up^\frac 12 r^{-1} (\ao)_{\ell=1} +O_{\ell\neq 1} (r^{-2-\de})+ O(r^{-3-2\de}).
\eea
Here, the $O_{\ell\neq 1}$ refers to a term supported on $\ell\neq 1$ and bounded by the quantities in the parentheses.
\end{proposition}
\begin{proof}
Recall from \eqref{eq:N-transport-go} that
\beaa
\slashed{\Lie}_{\pa_r} (r^{-2}\ga) &=& 2r^{-2}\ah \thh+\ah \thc (r^{-2}\ga)+2\Up^\frac 12 \ao r^{-1} (r^{-2}\ga) \\
&=& 2r^{-2} \Up^{-\frac 12} \thh+\Up^{-\frac 12} \thc (r^{-2}\gz)+r^{-2} O\left(\left((\ga-\gz),\ao\right) \cdot (\thh,\thc)\right)\\
&& +2\Up^\frac 12 r^{-1} \ao (r^{-2}\gz)+r^{-3} O\left((\ga-\gz)\cdot \ao\right). 
\eeaa  
We have
\beaa
[\pa_r, \overline\pa_i] &=& [\pa_r, \pa_i-\om_i \pa_r]=[\pa_r,\pa_i]-\pa_r(\om^i) \pa_r=-\pa_i (\om_j) \pa_j -\om_j \pa_j (\om^i) \pa_r \\
&=& -\frac 1r (\pa_i-\om_i \pa_r)+0 = -\frac 1r \overline \pa_i.
\eeaa
Therefore, $\slashed{\Lie}_{\pa_r} \overline\pa_i=-\frac 1r \overline \pa_i$, and hence
\beaa
\pa_r (\ga_{kk})&=& \slashed{\Lie}_{\pa_r} \ga (\overline \pa_k,\overline\pa_k)-2\ga(\overline\pa_k, \frac 1r \overline\pa_k)=r^2 \slashed{\Lie}_{\pa_r} (r^{-2} \ga)(\overline\pa_k,\overline\pa_k) \\
&=& 2\Up^{-\frac 12} \thc +4\Up^\frac 12 r^{-1} \ao +2 \Up^\frac 12 \thh(\overline \pa_k,\overline\pa_k)+ O(r^{-1-\de}\cdot r^{-2-\de}) \\
&=& 2\Up^{-\frac 12} (\thc)_{\ell=1} +4\Up^\frac 12 r^{-1} (\ao)_{\ell=1} +O_{\ell\neq 1} (r^{-2-\de})+ O(r^{-3-2\de})
\eeaa
where we use that the scalar $\thh(\overline \pa_k,\overline\pa_k)=\trz \thh=O(\ga-\gz)\cdot \thh$, i.e., can be controlled by the size of metric perturbation $O(r^{-1-\de})$ multiplied by the size of $\thh=O(r^{-2-\de})$.
\end{proof}

Recall from \eqref{eq:solve-nonlocal-vanishing-condition} and \eqref{eq:expansion-Psi-3-ell=1} that
\beaa
(\thc)_{\ell=1,i}=\cc_i r^{-3}+O(r^{-3-\de}),\quad (\ao)_{\ell=1,i}=\frac 12 \cc_i r^{-2}+O(r^{-2-\de}).
\eeaa
Therefore, we have, using $\Up=1+O(mr^{-1})$,
\beaa
(\pa_r (\ga_{kk}))_{\ell=1,i} &=& 2\Up^{-\frac 12} (\thc)_{\ell=1,i} +4\Up^\frac 12 r^{-1} (\ao)_{\ell=1,i} + O(r^{-3-2\de}) \\
&=& 4\cc_i r^{-3} +O(r^{-3-\de}).
\eeaa
We also know that $g_{rr}=\ah^2=\ah_0^2+2\ah_0 \ao+O(\ao^2)=\Up^{-1}+2\Up^{-\frac 12} \ao_{\ell=1}+O(r^{-2-2\de})$. Hence, 
again using $\Up=1+O(mr^{-1})$,
\beaa
\C_i &=& \frac{1}{16\pi m}\lim_{r\to\infty}\int_{S_r}
  \om^i \big(2 (g_{rr} -1) - r\pa_r (\ga_{kk})\big)\, dA  \\
&=& \frac{1}{16\pi m}\lim_{r\to\infty}  r^2 \left(4 \ao_{\ell=1,i}-r (\pa_r(\ga_{kk}))_{\ell=1,i}\right) \\
&=& \frac{1}{16\pi m} (2\cc_i-4\cc_i)=-\frac{1}{8\pi m}\cc_i.
\eeaa
This proves the relation for $\C_i$ with $\cc_i=\lim_{r\to\infty} r^3 (\thc)_{\ell=1,i}$.

The equation \eqref{eq:pa-r-ga-kk} also implies the rougher form
\beaa
(\pa_r (\ga_{kk}))_{\ell=0} =O(r^{-2-\de}).
\eeaa
Therefore, using $\gz_{kk}=2$, we deduce $(\ga_{kk})_{\ell=0}=2+O(r^{-1-\de})$, and
\beaa
\E=\frac{1}{16\pi}\lim_{r\to\infty} \int_{S_r} \frac 2r g_{rr} -\frac 1r \ga_{kk} -\pa_r (\ga_{kk}) dA=\frac{1}{16\pi}\lim_{r\to\infty} \int_{S_r} \frac 2r (\frac{2m}r)+O(r^{-2-\de}) dA=m.
\eeaa
This proves the relation for $\E$.


\subsection{Linear momentum and angular momentum}
\begin{lemma}\lab{lemma:nab-ell=1}
We have
\beaa
(\nabz \om_i)^\#=\frac 1r \overline \pa_i.
\eeaa
Here, the index raising $\#$ is defined with respect to $\ga\0$.
\end{lemma}
\begin{proof}
By rotational symmetry, we assume $i$ corresponds to the $z$ direction. In this case,  $\om_i=\cos\vth$ in the standard spherical coordinates $(\vth,\vphi)$, and $\overline\pa_i=\overline{\pa}_z=\pa_z-g(\pa_z,\pa_r)\pa_r=\pa_z-\cos\vth \pa_r$. In Cartesian coordinates,
\beaa
\pa_r=(\sin\vth\cos\vphi,\sin\vth\sin\vphi,\cos\vth),\quad \pa_\vth=r(\cos\vth \cos\vphi,\cos\vth \sin\vphi,-\sin\vphi).
\eeaa
It is then straightforward to verify that $\overline{\pa}_z=-r^{-1} \sin\vth\pa_\vth$.
We compute, 
\beaa
(\nabh^{(0)} \cos\vth)^\vth=r^{-2} \pa_\vth \cos\vth=-r^{-2} \sin\vth,\quad (\nabz \cos\vth)^\vphi=0.
\eeaa
Therefore, we obtain $(\nabz \cos\vth)^\#=r^{-1}\overline\pa_z$.
\end{proof}

For notational simplicity, we define $\Xi\0$ by $\displaystyle \Xi\0_a:=k(\pa_r, e_a\0)$, where $\{e_a\0\}$ is a horizontal orthonormal frame with respect to $\gz$. We have, using the bounds we obtain in Section \ref{sec:conclusions}, 
\bea\lab{eq:estimate-div-divz-Xi-Xiz}
|r\divh\Xi-r\divh\0 \Xi\0|\lesssim |(r\nabz)^{\leq 1}(\ga-\gz)||\Xi|+|(r\nabz)^{\leq 1} k||\ao, (\ga-\gz)| \lesssim \eps^2 r^{-3-2\de}.
\eea
We have, using \eqref{lemma:nab-ell=1},
\beaa
(\divh^{(0)} \Xi\0)_{\ell=1,i}=r^{-2} \int_{S_r} (\divh^{(0)} \Xi\0) \om_i\, dA=-r^{-2} \int_{S_r} \Xi\0 \cdot \nabh^{(0)}\om_i\, dA=-r^{-3} \int_{S_r} k(\pa_r,\overline \pa_i) dA.
\eeaa
Note that we have imposed in Theorem \ref{thm:main-precise} that $(\divh\Xi)_{\ell=1}=0$.  Therefore, we have $\int_{S_r} k(\pa_r,\overline \pa_i) dA=O(r^2 (\divh\Xi-\divh\0\Xi\0)_{\ell=1})=O(r^{-1-2\de})$.
We now calculate
\beaa
\P_i &=& \frac{1}{8\pi}\lim_{r\to\infty} \int_{S_r} (k_{ij}-\mathrm{tr}_\de\hspace{0.1em}  k \de_{ij}) (\pa_r)^j dA=\frac{1}{8\pi}\lim_{r\to\infty} \int_{S_r} (k(\pa_i,\pa_r)-\mathrm{tr}_\de\hspace{0.1em} k\, \om_i) dA \\
&=& \frac{1}{8\pi}\lim_{r\to\infty} \int_{S_r} k(\overline \pa_i,\pa_r) -(\trh\0 k) \om_i dA
\eeaa
Here $\trh\0 k$ is the $\gz$-spherical trace of $k$, and hence $\trh\0 k=\trt+O(r^{-3-2\de})$. 
As a result, the spherical integral becomes
\beaa
\P_i= \frac{1}{8\pi}\lim_{r\to\infty}\int_{S_r} -(\trt) \om^i\, dA = - \frac{1}{8\pi} r^2 (\trt)_{\ell=1,i}.
\eeaa
For the angular momentum, we have
\beaa
\J_i := \frac 1{8\pi} \lim_{r\to\infty}\int_{S_r} \in_{ilm} x^l (k^{mj} -\de^{mj}\mathrm{tr}_\de\hspace{0.1em} k) (\pa_r)_j\, dA= \frac 1{8\pi} \lim_{r\to\infty}\int_{S_r}  k(\pa_r, \in_{ilm} x^l \pa_m)\, dA,
\eeaa
where we use $\in_{ilm}x^\ell(\pa_r)^m=0$. Now, note that, using Lemma \ref{lemma:nab-ell=1},
\beaa
\in_{ilm} x^l \pa_m=r \in_{ab}\0  (\overline\pa_i)^b e_a\0 =r^2 \in_{ab}\0 (\nabz \om_i)_b\, e_a\0.
\eeaa
We also have the relation
\beaa
(\curlh\0 \Xi\0)_{\ell=1,i}=-r^{-2} \int_{S_r}  \in_{ab}\0 (\nabz_b \Xi\0_a) \om_i\, dA= r^{-2} \int_{S_r}  \in_{ab}\0 \Xi\0_a (\nabz_b \om_i) dA, 
\eeaa
and, similar to \eqref{eq:estimate-div-divz-Xi-Xiz}, we have $\curlh \Xi=\curlh\0 \Xi\0+O(r^{-4-2\de})$.
Therefore, we deduce
\beaa
\J_i =\frac 1{8\pi} \lim_{r\to\infty}\int_{S_r}  k(\pa_r, \in_{ab}\0 e_a\0 (\nabz \om_i)_b) dA=\frac 1{8\pi} \lim_{r\to\infty} r^4 (\curlh \Xi)_{\ell=1,i}.
\eeaa

\section{Contraction estimates}
We will frequently use the relation
\beaa
\psi\nn \cdot \phi\nn-\psi\n\cdot \phi\n = \de\psi\nn\cdot  \phi\nn +\psi\n \cdot \de\phi\nn.
\eeaa
Here, the dots may be with respect to $\ga\nn$ or $\ga\n$. The difference, however, generates lower order terms of $\de\ga\nn$ and hence we omit the estimates.

We also have the following calculation for any sequence of linear operators $L\n$:
\bea\lab{eq:general-formula-L-difference}
\bsplit
L\nn[\phi\nnn]-L\n[\phi\nn] &= L\nn[\phi\nnn]-L\nn[\phi\nn]+L\nn[\phi\nn]-L\n[\phi\nn] \\
&= L\nn{\de\phi\nnn} + (L\nn-L\n)[\phi\nn].
\end{split}
\eea
This applies for $L\n=\laph\n$, $\d_1\n$, $\d_2\n$, as well as the spherical mean operator $\phi\mapsto \overline{\phi}\n$. Moreover, the following schematic relations hold, in view of Remark \ref{rem:inverse-metric}, 
\beaa
(\d^2)\nn \psi-(\d^2)\n \psi &=& \de\ga\nn \cdot \nabh^2 \psi +\nabh(\de\ga\nn)\cdot\nabh \psi, \quad \d^2=\laph, \d_1\d_2, \\
\d\nn \psi-\d\n\psi &=& \de \ga\n \psi, \quad \d=\d_1,\d_2,\divh,\curlh, \\
\overline\phi\nn-\overline \phi\n &=& \de\ga\n \cdot \phi, \\
\trh\nn \psi-\trh\n \psi &=& \de\ga\n \cdot \psi,\quad \psi\in \ss_2.
\eeaa

\subsection{Proof of Proposition \ref{prop:contraction-main}}\lab{proof:contraction-main}
Taking the differences of the equations \eqref{eq:iteration-kac}-\eqref{eq:iteration-average-a} between $n\mapsto n+1$ and $n$, we obtain
\beaa
(\pa_r+2r^{-1})(\de\Psi_1\nnn) &=& \Up^{-\frac 12} (\de\tilde\mu\nnn)_{\ell=0}-2(1-3mr^{-1}) r^{-2} (\de\Psi_3\nnn)+\Ga_1\nn\cdot \Ga_1\nn-\Ga_1\n\cdot\Ga_1\n,\\ 
(\pa_r+3r^{-1}) (\de\Psi_2\nnn) &=& r^{-1}(\de\tilde\mu\nnn)_{\ell=0} -2\Up^\frac 12 r^{-3}(\de\Psi_3\nnn) -\Up^{-\frac 12}(\widetilde\B_{\ell\leq 1}\nnn-\widetilde\B_{\ell\leq 1}\nn)\\
&  & + \Ga_1\nn\cdot \Ga_2\nn - \Ga_1\n\cdot \Ga_2\n - \de\Psi_3\nn \Bb-\Psi_3\nn \widetilde\B_{\ell\leq 1,aux}\nn+\Psi_3\n \widetilde\B_{\ell\leq 1,aux}\n,\\ 
\Up^\frac 12 \laph\nn (\de\Psi_3\nnn) &=& (\de\Psi_2\nnn)-(\overline{\de\Psi_2\nnn}^{(n+1)}) -\Up^\frac 12 r^{-1} (\de\Psi_1\nnn-\overline{\de\Psi_1\nnn}\nn)\\
& & -\Up^\frac 12 (\laph\nn-\laph\n)(\Psi_3\nn) -( \overline{\Psi_2\nn}\nn -\overline{\Psi_2\nn}\n)\\
& & +\Up^\frac 12 r^{-1} (\overline{\Psi_1\nn}\nn-\overline{\Psi_1\nn}\n) \\
& & -\laph\nn(\Ga_0\nn\cdot\Ga_0\nn)+\laph\n(\Ga_0\n\cdot\Ga_0\n) \\
& & +\Ga_1\nn\cdot \Ga_1\nn-\Ga_1\n\cdot \Ga_1\n+\overline{\Ga_1\nn\cdot \Ga_1\nn}\nn-\overline{\Ga_1\n\cdot \Ga_1\n}\n,\\
\overline{\de\Psi_3\nnn}\nn &=&  -\frac 12 \Up^{-1} r \overline{\de\Psi_1\nnn}\nn\\
& & -(\overline{\Psi_3\nn}\nn-\overline{\Psi_3\nn}\n)+\frac 12 \Up^{-1} r (\overline{\Psi_1\nn}\nn-\overline{\Psi_1\nn}\n).
\eeaa
Note that, by \eqref{eq:tilde-BB} and \eqref{eq:tilde-mu-ell=0}, we have
\beaa
\widetilde\B_{\ell\leq 1}\nnn-\widetilde\B_{\ell=1}\nn &=& \frac 12 (\laph\nn \thc\nnn)_{\ell= 1}-\frac 12 (\laph\n \thc\nn)_{\ell\leq 1} \\
& & +\frac 12 (\laph\nn \thc\nn)_{\ell=0}-\frac 12 (\laph\n \thc\n)_{\ell=0} \\
& & -\left(\mathcal P_1(\d_1\nn\d_2\nn \thh\nn)\right)_{\ell\leq 1} +\left(\mathcal P_1(\d_1\n\d_2\n \thh\n)\right)_{\ell\leq 1},
\eeaa
and 
\beaa
\de \tilde\mu\nnn_{\ell=0}=\tilde\mu\nnn_{\ell=0}-\tilde\mu\nn_{\ell=0} &=& (\de\Psi_2\nnn-\Up^\frac 12 r^{-1}\de\Psi_1\nnn)_{\ell=0} \\
& & -(\laph\nn \log(\Up^{-\frac 12}+\Psi_3\nn))_{\ell=0}+(\laph\n \log(\Up^{-\frac 12}+\Psi_3\n))_{\ell=0} \\
& & -\frac 14 ((\Psi_1\nn)^2 )_{\ell=0}+\frac 14 ((\Psi_1\n)^2 )_{\ell=0}.
\eeaa
Note that the last line is of the form $\Ga_1\nn\cdot \Ga_1\nn-\Ga_1\n\cdot\Ga_1\n$.
Therefore, the system \eqref{eq:de-thc-nn}-\eqref{eq:de-Kc-nn} holds, with
\beaa
\NN[\de\Psi_1] &=&  -(\laph\nn \log(\Up^{-\frac 12}+\Psi_3\nn))_{\ell=0}+(\laph\n \log(\Up^{-\frac 12}+\Psi_3\n))_{\ell=0} \\
& & +\Ga_1\nn\cdot \Ga_1\nn-\Ga_1\n\cdot\Ga_1\n,\\
\NN[\de\Psi_2]&=& -\frac 12 \Up^{-\frac 12}\Big((\laph\nn-\laph\n)\thc\nn\Big)_{\ell=1}-\frac 12\Up^{-\frac 12} \Big(\laph\nn\thc\nn-\laph\n\thc\n\Big)_{\ell=0}\\
& & -\Up^{-\frac 12}\left(\mathcal P_1(\d_1\nn\d_2\nn \thh\nn-\d_1\n\d_2\n \thh\n)\right)_{\ell\leq 1} \\
& & +r^{-1} (-(\laph\nn \log(\Up^{-\frac 12}+\Psi_3\nn))_{\ell=0}+(\laph\n \log(\Up^{-\frac 12}+\Psi_3\n))_{\ell=0})\\
& &  + \Ga_1\nn\cdot \Ga_2\nn - \Ga_1\n\cdot \Ga_2\n - \de\Psi_3\nn \Bb-\Psi_3\nn \widetilde\B_{\ell\leq 1,aux}\nn+\Psi_3\n \widetilde\B_{\ell\leq 1,aux}\n, \\
\NN[\de\Psi_3] &=& -\Up^\frac 12 (\laph\nn-\laph\n)(\Psi_3\nn) -( \overline{\Psi_2\nn}\nn -\overline{\Psi_2\nn}\n)+\Up^\frac 12 r^{-1} (\overline{\Psi_1\nn}\nn-\overline{\Psi_1\nn}\n) \\
& & -\laph\nn(\Ga_0\nn\cdot\Ga_0\nn)+\laph\n(\Ga_0\n\cdot\Ga_0\n) \\
& & +\Ga_1\nn\cdot \Ga_1\nn-\Ga_1\n\cdot \Ga_1\n+\overline{\Ga_1\nn\cdot \Ga_1\nn}\nn-\overline{\Ga_1\n\cdot \Ga_1\n}\n,\\
\NN_{av}[\de\Psi_3] &=& -(\overline{\Psi_3\nn}\nn-\overline{\Psi_3\nn}\n)+\frac 12 \Up^{-1} r (\overline{\Psi_1\nn}\nn-\overline{\Psi_1\nn}\n).
\eeaa
We first note that, since $\thh\nn$ and $\thh\n$ are traceless with respect to $\ga\nn$ and $\ga\n$, we have
\beaa
\trh\0 (\thh\nn-\thh\n)&=& (\trh\0- \trh\nn) \thh\nn- (\trh\0-\trh\n) \thh\n\\
&=& (\trh\0-\trh\nn)\de\thh\nn-(\trh\nn-\trh\n)\thh\n \\
&=& (\ga\nn-\gz)\cdot \de\thh\nn-\de\ga\nn \thh\n.
\eeaa 
Therefore, using \eqref{eq:estimate-d1-d2-ell-leq-1},
\bea\lab{eq:appendix-d1-d2-difference}
\bsplit
&\quad  r^{-1} ||(\d_1\nn\d_2\nn \thh\nn-\d_1\n\d_2\n \thh\n)_{\ell\leq 1} ||_{\H^s} \\
&\lesssim  r^{-1} ||(\d_1\nn\d_2\nn (\de\thh\nn))_{\ell\leq 1} ||_{\H^s} + r^{-1} || (\de\ga\nn\cdot \nabh^2 \thh\n+\nabh(\de\ga\nn)\cdot \nabh\thh\n)_{\ell\leq 1} ||_{\H^s} \\
&\lesssim ||\lapz\trh\0 \de\thh\nn ||_{L^\infty} + ||(r\nabz)^{\leq 2} \de\thh\nn||_{L^\infty} ||(r\nabz)^{\leq 2} (\ga\nn-\gz) ||_{L^\infty} \\
& \quad + r^{-4-\de} (r^{-1} ||\de\ga\nn||_{\H^s}) \\
&\lesssim  \eps r^{-5-2\de} ||\de(\Psi\nn,\ga\nn) ||_s.
\end{split}
\eea
We then have, using the standard $L^2$-$L^\infty$ type estimate,
\beaa
r^{-1} ||\NN[\de\Psi_1]||_{\H^{s+1}} &\lesssim & ||\laph\nn \log(\Up^{-\frac 12}+\Psi_3\nn))_{\ell=0}+(\laph\n \log(\Up^{-\frac 12}+\Psi_3\n))_{\ell=0}||_{L^\infty} \\
& & +r^{-1} ||(\Ga_1\nn,\Ga_1\n)\cdot\de\Ga_1\nn||_{\H^{s+1}} \\
&\lesssim & \eps r^{-3-\de} ||\de\ga\nn||_{L^\infty}+\eps r^{-2-\de} (r^{-1} ||\de\Ga_1\nn||_{\H^{s+1}})\\
&\lesssim & \eps r^{-4-2\de} ||\de(\Psi\nn,\ga\nn) ||_s,
\eeaa
\beaa
r^{-1} ||\NN[\de\Psi_2]||_{\H^s} &\lesssim & ||(\laph\nn-\laph\n)\thc\nn||_{L^\infty} +\eps r^{-5-2\de} ||\de(\Psi\nn,\ga\nn) ||_s \\ 
& & +r^{-1}||\Ga_1\nn\cdot\Ga_2\nn-\Ga_1\n\cdot \Ga_2\n||_{\H^{s+1}}+\eps r^{-4-\de} (r^{-1} ||\de\Psi_3\nn||_{\H^s}) \\
& & + r^{-1} || \widetilde \B_{\ell\leq 1,aux}\nn \de\Psi_3\nn||_{\H^s} + r^{-1} ||\Psi_3\n \de(\widetilde \B_{\ell\leq 1,aux}\nn)||_{\H^s} \\
&\lesssim & \eps r^{-5-2\de} ||\de(\Psi\nn,\ga\nn) ||_s,
\eeaa
Similarly, we obtain
\beaa
r^{-1} ||\NN[\de\Psi_3]||_{\H^s} &\lesssim & \eps r^{-4-2\de} ||\de(\Psi\nn,\ga\nn) ||_s,\\
|\NN_{av}[\de\Psi_3]| &\lesssim & \eps r^{-2-2\de} ||\de(\Psi\nn,\ga\nn) ||_s,
\eeaa
and we omit the details since the reasoning is totally the same.

\subsection{Proof of Proposition \ref{prop:contraction-remain}}\lab{proof:contraction-remain}
Taking the differences of the equations \eqref{eq:iteration-h-dualh}-\eqref{eq:iteration-Y} between $n\mapsto n+1$ and $n$, we have
\beaa
\d_1\nn \d_2\nn \Psi_4\nnn-\d_1\n \d_2\n \Psi_4\nn &=& \frac 12 (\laph\nn \Psi_1\nnn,0)-\frac 12 (\laph\n \Psi_1\nn,0) -\de\Psi_{11}\nnn, \\
\d_1\nn \Psi_5\nnn-\d_1\n \Psi_5\nn &=& -(\Up^\frac 12 \laph\nn \Psi_3\nnn,0 )+(\Up^\frac 12 \laph\n \Psi_3\nn,0 )\\
& &+ \left(\laph\nn(\Ga_0\nn\cdot\Ga_0\nn),0\right)-\left(\laph\n(\Ga_0\n\cdot\Ga_0\n),0\right), \\
\d_1\nn \Psi_6\nnn- \d_1\n \Psi_6\nn &=&  \de\Psi_{11}\nnn-(\overline{\Psi_{11}\nnn}\nn- \overline{\Psi_{11}\nn}\n) \\
& & -(\overline{(\Bb, \Bbd)}\nn- \overline{(\Bb, \Bbd)}\n).
\eeaa
The equations can be rewritten as
\beaa
\d_1\nn\d_2\nn \de\Psi_4\nnn &=& \frac 12 (\laph\nn \de\Psi_1\nnn,0)-\de\Psi_{11}\nnn \\
& & -(\d_1\nn \d_2\nn-\d_1\n \d_2\n) \Psi_4\nn +\frac 12((\laph\nn-\laph\n)\Psi_1\nn,0) \\
\d_1\nn \de \Psi_5\nnn &=& -(\Up^\frac 12 \laph\nn \de\Psi_3\nnn,0 )\\
& & -(\d_1\nn-\d_1\n)\Psi_5\nn -\Up^\frac 12 ((\laph\nn-\laph\n) \Psi_3\nn,0) \\
& &+ \left(\laph\nn(\Ga_0\nn\cdot\Ga_0\nn),0\right)-\left(\laph\n(\Ga_0\n\cdot\Ga_0\n),0\right) \\
\d_1\nn \de \Psi_6\nnn &=&  \de \Psi_{11}\nnn-\overline{\de \Psi_{11}\nnn}\nn  -(\overline{(\Bb, \Bbd)}\nn- \overline{(\Bb, \Bbd)}\n)\\
& & -(\d_1\nn-\d_1\n) \Psi_6\nn-(\overline{\Psi_{11}\nn}\nn-\overline{\Psi_{11}\nn}\n).
\eeaa
This is of the form \eqref{eq:de-Psi-4-nn}-\eqref{eq:de-Psi-6-nn} with
\beaa
\NN[\de\Psi_4] &=& -(\d_1\nn \d_2\nn-\d_1\n \d_2\n) \Psi_4\nn +\frac 12((\laph\nn-\laph\n)\Psi_1\nn,0), \\
\NN[\de\Psi_5] &=& -(\d_1\nn-\d_1\n)\Psi_5\nn -\Up^\frac 12 ((\laph\nn-\laph\n) \Psi_3\nn,0) \\
& &+ \left(\laph\nn(\Ga_0\nn\cdot\Ga_0\nn),0\right)-\left(\laph\n(\Ga_0\n\cdot\Ga_0\n),0\right), \\
\NN[\de\Psi_6] &=& -(\d_1\nn-\d_1\n) \Psi_6\nn-(\overline{\Psi_{11}\nn}\nn-\overline{\Psi_{11}\nn}\n) -(\overline{(\Bb, \Bbd)}\nn- \overline{(\Bb, \Bbd)}\n).
\eeaa
We then proceed as in Section \ref{proof:contraction-main} to estimate
\beaa
r^{-1} ||\NN[\de\Psi_4]||_{\H^s} &\lesssim & \eps r^{-5-2\de} ||\de(\Psi\nn,\ga\nn)||_s,\\
r^{-1} ||\NN[\de\Psi_5]||_{\H^s} &\lesssim & \eps r^{-4-2\de} ||\de(\Psi\nn,\ga\nn)||_s, \\
r^{-1} ||\NN[\de\Psi_6]||_{\H^s} &\lesssim & \eps r^{-5-2\de} ||\de(\Psi\nn,\ga\nn) ||_s.
\eeaa
This concludes the proof of Proposition \ref{prop:contraction-remain}.

\subsection{Proof of Proposition \ref{prop:contraction-k-part}}\lab{proof:contraction-k-part}
We denote $\de(\ah\nn\psi\nnn):=\ah\nn\psi\nnn-\ah\n\psi\nn$. 

Taking the differences of the equations \eqref{eq:iteration-trt}-\eqref{eq:iteration-Pi-mean} between $n\mapsto n+1$ and $n$, we have
\beaa
(\pa_r+r^{-1}) \de\Psi_7\nnn &=& 2r^{-1} \de\Psi_{10}\nnn +\Ga_1\nn\cdot\Ga_1\nn -\Ga_1\n\cdot \Ga_1\n,
\eeaa
\beaa
\d_1\nn \d_2\nn (\ah\nn \Psi_8\nnn)- \d_1\n \d_2\n (\ah\n \Psi_8\nn) &=& \frac 12 (\ah\nn\laph\nn \Psi_7\nnn-\ah\n\laph\n\Psi_7\nn,0)\\
& &+\de\Psi_{12}\nn+\Ga_1\nn\cdot \Ga_2\nn-\Ga_1\n\cdot \Ga_2\n, 
\eeaa
\beaa
\d_1\nn\Psi_9\nnn-\d_1\n\Psi_9\nnn &=& -\bigg(0,r^{-4} \int_r^\infty r'^4 \de\Kd_{\ell\leq 1}\nnn\, dr'\bigg) \\
& & -\bigg(0,\overline{\displaystyle r^{-4} \int_r^\infty r'^4 (\Kkd)\, dr'}\nn-\overline{\displaystyle r^{-4} \int_r^\infty r'^4 (\Kkd)\, dr'}\n\bigg) \\
& & +\bigg(0,\overline{\displaystyle r^{-4} \int_r^\infty r'^4 (\Kd_{\ell\leq 1}\nnn)\, dr'}\nn-\overline{\displaystyle r^{-4} \int_r^\infty r'^4 (\Kd_{\ell\leq 1}\nn)\, dr'}\n\bigg),
\eeaa
\beaa
\laph\nn(\ah\nn\Psi_{10}\nnn)-\laph\n(\ah\n\Psi_{10}\nn) &=& \de\widetilde \K_{\ell\leq 1}\nnn -(\overline{\Kk}\nn-\overline{\Kk}\n) -(\overline{\widetilde \K_{\ell\leq 1}\nnn}\nn-\overline{\widetilde \K_{\ell\leq 1}\nn}\n), \\
\overline{\displaystyle\ah\nn\Psi_{10}\nnn}\nn -\overline{\displaystyle\ah\n\Psi_{10}\nn}\n &=& \overline{\Psi_3\nn\Psi_{10}\nn}\nn -\overline{\Psi_3\n\Psi_{10}\n}\n.
\eeaa
Using the formula \eqref{eq:general-formula-L-difference}, we have the following relations
\beaa
&& \d_1\nn \d_2\nn (\ah\nn \Psi_8\nnn)- \d_1\n \d_2\n (\ah\n \Psi_8\nn) \\
&=& \d_1\nn\d_2\nn \de (\ah\nn\Psi_8\nnn)+(\d_1\nn\d_2\nn-\d_1\n\d_2\n) (\ah\n \Psi_8\nn),
\eeaa
\beaa
&& \ah\nn\laph\nn \Psi_7\nnn-\ah\n\laph\n\Psi_7\nn\\
&=& \ah\nn\laph\nn \de\Psi_7\nnn+(\ah\nn\laph\nn-\ah\n\laph\n)\Psi_7\nn,
\eeaa
\beaa
\d_1\nn\Psi_9\nnn-\d_1\n\Psi_9\nnn &=& \d_1\nn \de\Psi_9\nnn+(\d_1\nn-\d_1\n)\Psi_9\nn,
\eeaa
\beaa
&& \laph\nn(\ah\nn\Psi_{10}\nnn)-\laph\n(\ah\n\Psi_{10}\nn) \\
&=& \laph\nn \de (\ah\nn\Psi_{10}\nnn)+(\laph\nn-\laph\n) (\ah\n \Psi_{10}\nn),
\eeaa
\beaa
&& \overline{\ah\nn\Psi_{10}\nnn}\nn -\overline{\ah\n\Psi_{10}\nn}\n \\
&=& \overline{\de(\ah\nn\Psi_{10}\nnn)}+\overline{\ah\n\Psi_{10}\nn}\nn -\overline{\ah\n\Psi_{10}\nn}\n.
\eeaa
Moreover, we have, in view of \eqref{eq:widetilde-K-ell-leq-1},
\beaa
\de\widetilde\K_{\ell\leq 1}\nnn &=& \mathcal P_1\left( \d_1\nn\d_2\nn(\ah\nn\Psi_8\nn)-\d_1\n\d_2\n(\ah\n\Psi_8\n)\right)_{\ell\leq 1}\\
& & -\frac 12 (\ah\nn \laph\nn \Psi_7\nnn-\ah\n \laph\n \Psi_7\nn)_{\ell= 1}\\
& & -\frac 12 (\ah\nn \laph\nn \Psi_7\nn-\ah\n \laph\n \Psi_7\n)_{\ell=0}+(\Ga_1\nn\cdot\Ga_2\nn-\Ga_1\n\cdot\Ga_2\n)_{\ell\leq 1} \\
&=& -\frac 12(\ah\nn\laph\nn \de\Psi_7\nnn)_{\ell=1} +\de \RR_{lower}\nn,
\eeaa
where
\beaa
\de\RR_{lower}\nn&:=& \mathcal P_1\left( \d_1\nn\d_2\nn(\ah\nn\Psi_8\nn)-\d_1\n\d_2\n(\ah\n\Psi_8\n)\right)_{\ell\leq 1}\\
& & -\frac 12 (\ah\nn \laph\nn \Psi_7\nn-\ah\n \laph\n \Psi_7\nn)_{\ell= 1}\\
& & -\frac 12 (\ah\nn \laph\nn \Psi_7\nn-\ah\n \laph\n \Psi_7\n)_{\ell=0}+(\Ga_1\nn\cdot\Ga_2\nn-\Ga_1\n\cdot\Ga_2\n)_{\ell\leq 1} .
\eeaa
Using this notation, we write
\beaa
&&\overline{\widetilde \K_{\ell\leq 1}\nnn}\nn-\overline{\widetilde \K_{\ell\leq 1}\nn}\n = \overline{\de\widetilde\K_{\ell\leq 1}\nnn}\nn +(\overline{\widetilde\K_{\ell\leq 1}\nn}\nn-\overline{\widetilde\K_{\ell\leq 1}\nn}\n) \\
&=& -\frac 12 \overline{(\ah\nn\laph\nn \de\Psi_7\nnn)_{\ell=1}}\nn +\overline{\de\RR_{lower}\nn}\nn+(\overline{\widetilde\K_{\ell\leq 1}\nn}\nn-\overline{\widetilde\K_{\ell\leq 1}\nn}\n).
\eeaa
Then, the system is already of the form \eqref{eq:de-Psi-7-nn}-\eqref{eq:de-Psi-10-av-nn}, with
\beaa
\NN[\de\Psi_7]&:=& \Ga_1\nn\cdot\Ga_1\nn -\Ga_1\n\cdot \Ga_1\n,\\
\NN[\de\Psi_8]&:=& -\left(\d_1\nn\d_2\nn-\d_1\n\d_2\n\right) (\ah\n \Psi_8\nn)+\frac 12 \left((\ah\nn\laph\nn-\ah\n\laph\n)\Psi_7\nn,0\right) \\
& & +\Ga_1\nn\cdot \Ga_2\nn-\Ga_1\n\cdot \Ga_2\n, \\
\NN[\de\Psi_9] &:=& -(\d_1\nn-\d_1\n)\Psi_9\nn-\left(0,\overline{\displaystyle r^{-4} \int_r^\infty r'^4 (\Kkd) dr'}\nn-\overline{\displaystyle r^{-4} \int_r^\infty r'^4 (\Kkd) dr'}\n \right) \\
& & +\left(0,\overline{\displaystyle r^{-4} \int_r^\infty r'^4 (\Kd_{\ell\leq 1}\nn) dr'}\nn-\overline{\displaystyle r^{-4} \int_r^\infty r'^4 (\Kd_{\ell\leq 1}\nn) dr'}\n\right),\\
\NN[\de\Psi_{10}] &:=& -(\laph\nn-\laph\n) (\ah\n \Psi_{10}\nn)+\de\RR_{lower}\nn\\
& & -(\overline{\Kk}\nn-\overline{\Kk}\n) -(\overline{\widetilde\K_{\ell\leq 1}\nn}\nn-\overline{\widetilde\K_{\ell\leq 1}\nn}\n)-\overline{\de\RR_{lower}\nn}\n,\\
\NN_{av}[\de\Psi_{10}] &:= & -\overline{\displaystyle\ah\n\Psi_{10}\nn}\nn +\overline{\displaystyle\ah\n\Psi_{10}\nn}\n.
\eeaa
We then proceed as in Section \ref{proof:contraction-main} to estimate
\beaa
r^{-1} ||\NN[\de\Psi_7] ||_{\H^{s+1}} &\lesssim & \eps r^{-4-2\de} ||\de(\Psi\nn,\ga\nn)||_s, \\
r^{-1} ||\NN[\de\Psi_8] ||_{\H^{s-1}} &\lesssim & \eps r^{-5-2\de} ||\de(\Psi\nn,\ga\nn)||_s,\\
r^{-1} ||\NN[\de\Psi_9] ||_{\H^s} &\lesssim & \eps r^{-4-2\de} ||\de(\Psi\nn,\ga\nn)||_s, \\
r^{-1} ||\NN[\de\Psi_{10}]||_{\H^{s-1}} &\lesssim & \eps r^{-5-2\de} ||\de(\Psi\nn,\ga\nn)||_s, \\
r^{-1} ||\NN_{av} [\de\Psi_{10}]||_{\H^{s+1}} &\lesssim & \eps r^{-3-2\de} ||\de(\Psi\nn,\ga\nn)||_s.
\eeaa
This concludes the proof of Proposition \ref{prop:contraction-k-part}.

\end{document}